\documentclass[12pt]{article}
\usepackage[letterpaper, margin=1in]{geometry}

\usepackage{setspace}
\usepackage{fancyhdr}
\usepackage{authblk}

\usepackage{amsmath}
\allowdisplaybreaks
\usepackage{amsthm}
\usepackage{amssymb}
\usepackage{mathrsfs}
\usepackage{bm}

\usepackage[colorlinks, citecolor=blue, urlcolor=blue]{hyperref}

\usepackage{enumitem}
\usepackage[markup=underlined]{changes}

\usepackage[numbers,sort&compress]{natbib}

\usepackage{longtable}
\usepackage{adjustbox}
\usepackage{threeparttable}
\usepackage{booktabs}

\usepackage{algorithm}
\usepackage{algpseudocode}

\usepackage{tikz-cd}
\usetikzlibrary{arrows.meta, positioning}

\usepackage{placeins}

\DeclareMathOperator{\var}{var}
\DeclareMathOperator{\diag}{diag}
\DeclareMathOperator{\argmax}{argmax}
\DeclareMathOperator{\Bernoulli}{Bernoulli}
\DeclareMathOperator{\Row}{Row}

\DeclareMathOperator{\Cov}{Cov}
\DeclareMathOperator{\diam}{diam}

\numberwithin{equation}{section}

\theoremstyle{plain}
\newtheorem{theorem}{Theorem}[section]
\newtheorem{lemma}{Lemma}[section]
\newtheorem{corollary}{Corollary}[section]
\newtheorem{assumption}{Assumption}[section]

\newtheorem*{proposition}{Proposition}

\theoremstyle{definition}
\newtheorem{definition}{Definition}[section]
\newtheorem{remark}{Remark}[section]
\newtheorem{example}{Example}[section]

\begin{document}

\title{CBARA: Covariate-Balanced-and-Adjusted Response-Adaptive Randomization}
\author[1]{Hengjia Fang\thanks{hengjiafang@ruc.edu.cn}}
\author[1]{Wei Ma\thanks{mawei@ruc.edu.cn}}
\affil[1]{Institute of Statistics and Big Data, Renmin University of China, Beijing, China}

\date{} 

\maketitle

\begin{abstract}
    We propose the covariate-balanced-and-adjusted response-adaptive randomization (CBARA) procedure for adaptive design in clinical trials, which integrates the complementary strengths of covariate-adjusted response-adaptive randomization (CARA) and covariate-adaptive randomization (CAR).
    The CBARA procedure updates the target allocation ratio according to observed responses and patient covariate profiles without requiring a correctly specified model, thereby retaining CARA's ethical and efficiency considerations while improving robustness.
    In addition, the CBARA procedure extends the CAR principle from fixed target allocation ratios to covariate-adjusted adaptive target allocation ratios, yet still pursues balance in treatment allocation with respect to covariate features.
    This integration is enabled by a newly defined imbalance vector and three interrelated components: the allocation function, parameter estimation and update mechanism.
    We establish the asymptotic properties of covariate imbalance and the estimators under the CBARA procedure.
    The results demonstrate that the CBARA procedure can improve balance for both observed and unobserved covariates while preserving the consistency of the allocation ratio.
    The theoretical analysis is developed through a pseudo-Markov chain framework, where a new discrepancy measure for transition kernels is introduced to handle the continuity of Poisson equation solutions with respect to parameters.
\end{abstract}

\tableofcontents

\section{Introduction}

\subsection{Background}

In the context of precision medicine, the design of clinical trials must efficiently exploit both treatment responses and covariate information, such as biomarkers, in order to improve the allocation mechanism while maintaining a balance between ethical considerations and inferential efficiency.
The covariate-adjusted response-adaptive randomization (CARA) procedure achieves this objective by dynamically updating treatment allocation probabilities as data accumulate \cite{rosenbergerRandomizationClinicalTrials2016,sverdlovModernAdaptiveRandomized2016}.
This approach serves two fundamental purposes.
From an ethical standpoint, it reduces the number of units assigned to treatments that are emerging as inferior for their specific covariate.
From an inferential perspective, this design enhances statistical efficiency by adaptively allocating a larger proportion of units to the treatment group exhibiting higher response variance.
This feature is consistent with the Neyman allocation principle and leads to a reduction in the asymptotic variance of the treatment effect estimators.

Most existing CARA procedures achieve these two objectives by adjusting the conditional allocation ratio of units given any covariate value to its theoretical optimum.
When the potential response distribution of each unit conditional on its baseline covariates is known, one can compute the theoretically optimal allocation probability for each unit according to a predefined optimization criterion that balances ethical and efficiency considerations.
In practice, however, this distribution is typically unknown.
Therefore, before each allocation, a CARA procedure estimates the optimal allocation probability based on the accumulated data.
The estimate represents the current allocation target and is therefore referred to as the targeted allocation ratio at the current allocation step.

To achieve the targeted allocation ratio, most of the CARA procedures directly set the allocation probability to the targeted allocation ratio \cite{biswasClassCovariateAdjustedResponseAdaptive2018,biswasClassOptimalCovariateadjusted2016,bhattacharyaClassOptimalType2015,zhangAsymptoticPropertiesCovariateadjusted2007,rosenbergerCovariateAdjustedResponseAdaptiveDesigns2001,bandyopadhyayAdaptiveDesignsNormal2001,chambazTargetedSequentialDesign2017,zhangOnlineMetaLevelAdaptiveDesign2025}.
They can be analyzed within the framework of Zhang et al. \cite{zhangAsymptoticPropertiesCovariateadjusted2007}.
However, these procedures usually suffer from covariate imbalance across treatment groups.
In addition to efficiency and ethics, a third operating characteristic of adaptive designs is balance, which can also potentially enhance statistical efficiency \cite{rosenbergerHandlingCovariatesDesign2008}.
Specifically, under equal allocation, in linear models with homoscedastic normal errors, achieving marginal covariate balance is equivalent to minimizing the variance of the estimated treatment effect, thereby improving the statistical power of hypothesis tests \cite{rosenbergerHandlingCovariatesDesign2008, baldiantogniniCovariateadaptiveBiasedCoin2011, maTestingHypothesesCovariateAdaptive2015, maStatisticalInferenceCovariateAdaptive2020}.
Therefore, under the CARA procedure, achieving covariate balance may also improve the efficiency of estimation.

Furthermore, covariate-adaptive randomization (CAR) procedures, also referred to as covariate-balanced randomization in some literature \cite{yuanBayesianResponseadaptiveCovariatebalanced2011,ballouResponseadaptiveCovariatebalancedRandomization2015,meurerSimulationVariousRandomization2016}, are primarily designed to achieve covariate balance across treatment groups under a fixed targeted allocation ratio.
They take the current covariate imbalance into account when determining the allocation probabilities.

In earlier studies, most CAR procedures focused primarily on balancing discrete covariates \cite{zelenRandomizationStratificationPatients1974,tavesMinimizationNewMethod1974,pocockSequentialTreatmentAssignment1975,weiApplicationUrnModel1978,huAsymptoticPropertiesCovariateadaptive2012}.
As a generalization of the CAR procedure proposed by Hu and Hu \cite{huAsymptoticPropertiesCovariateadaptive2012}, Zhao et al. proposed a CARA procedure for discrete covariates that balances prognostic covariates \cite{zhaoIncorporatingCovariatesInformation2022}.
The procedure uses discrete predictive covariates to define strata and to determine the targeted allocation ratio within each stratum, while balancing prognostic covariates within strata.
Hence, its application is confined to scenarios with discrete covariates, and it only balances these covariates independently within predefined strata.
However, theoretical guarantees for its asymptotic properties and inferential validity are not yet available.

Existing CAR studies have shown that achieving balance in covariates, whether continuous or discrete, improves the efficiency of average treatment effect (ATE) estimation under linear models \cite{shaoTheoryTestingHypotheses2010,maStatisticalInferenceCovariateAdaptive2020}.
Correspondingly, recent developments in CAR procedures have focused on achieving marginal balance for continuous covariates \cite{maNewUnifiedFamily2024,zhangAsymptoticPropertiesMultitreatment2023,baldiantogniniEfficientCovariateAdaptiveDesign2024,liuPropertiesCovariateadaptiveRandomization2025,fangGeneralNonMarkovianFramework2026}.
Thus, it is conceptually possible to achieve marginal balance of continuous covariates within a CARA procedure.
However, without any stratification structure, combining CAR and CARA into a unified adaptive design still poses significant challenges in both methodology and theory.

 \subsection{Contributions}

In this article, we propose the covariate-balanced-and-adjusted response-adaptive randomization (CBARA).
It addresses the methodological and theoretical challenges of balancing continuous covariates under a varying targeted allocation ratio across allocation steps, providing rigorous theoretical guarantees.
As a result of achieving balance, it also improves estimation efficiency, both theoretically and practically.

For clarity, we assume that the optimal allocation ratio is determined by a finite-dimensional model parameter that characterizes the joint distribution of covariates and responses.
However, in the CBARA procedure, the targeted allocation ratio is not exactly equal to the estimated optimal allocation ratio at each step.
This is because the sequence of estimates needs to be smoothed to reduce variation in the targeted allocation ratios across steps.
Accordingly, we refer to the parameter that determines the targeted allocation ratio as the allocation parameter, which serves as a stabilized proxy for the model parameter estimate.

In Subsubsections \ref{subsubsec_introduction_covariate_imbalance_allocation_function}--\ref{subsubsec_introduction_parameter_update}, we introduce the three core methodological components of the CBARA procedure proposed in this paper: covariate imbalance and allocation function, model parameter estimation, and allocation parameter update mechanism.
In Subsubsection \ref{subsubsec_introduction_theory}, we present the challenges and introduce a new theoretical framework for analyzing the CBARA procedure.
It clarifies why the variation of ${\theta_n}$ needs to be restricted in Subsubsection \ref{subsubsec_introduction_parameter_update}, and motivates the introduction of two allocation parameter update mechanisms designed to satisfy these restrictions.

\subsubsection{Covariate Imbalance and Allocation Function}
\label{subsubsec_introduction_covariate_imbalance_allocation_function}

The CARA procedure proposed by Zhao et al. balances discrete prognostic covariates within each stratum defined by predictive covariates, while the targeted allocation ratio among units within the same stratum is equal \cite{zhaoIncorporatingCovariatesInformation2022}.
When the covariates are continuous, the situation becomes more complex because marginal covariate imbalance is difficult to characterize when the targeted allocation ratio varies across covariate values and over allocation steps.
Even if covariate imbalance can be characterized, it remains difficult to construct a procedure that both controls it and prevents severe imbalance in potentially unobserved covariates.
The imbalance of the unobserved covariates has attracted increasing attention in the literature \cite{liuBalancingUnobservedCovariates2022,liuImpactsUnobservedCovariates2023}.
Existing studies have shown that, in the context of unequal targeted allocation ratios, such imbalance may be severe under certain CAR procedures \cite{liuPropertiesCovariateadaptiveRandomization2025,fangGeneralNonMarkovianFramework2026}.
This may compromise the consistency of subsequent estimators and the validity of inference.

In this article, we consider only two treatment groups.
If we are only interested in quantifying the covariate imbalance induced by the randomness between each treatment assignment and the targeted allocation ratio, we may define the following covariate imbalance vector:
\begin{equation*}
    \Lambda_n = \sum_{i=1}^n 
    \frac{(T_i - \rho_{i-1}(X_i)) \phi(X_i)}
    {\rho_{i-1}(X_i)(1-\rho_{i-1}(X_i))} ,
\end{equation*}
where $\phi$ is a prespecified feature mapping, $X_i$ denotes the covariate vector of the $i$th unit used for randomization, $T_i$ denotes the treatment assignment indicator and $\rho_{i-1}$ is the targeted allocation ratio for the $i$th allocation.
This definition is motivated by the notion of covariate balance in \cite{imaiCovariateBalancingPropensity2014}.
Similarly, we define the imbalance of the additional covariate $Z$ as
\begin{equation*}
    \Psi_n = \sum_{i=1}^n \frac{(T_i-\rho_{i-1}(X_i)) Z_i}{\rho_{i-1}(X_i)(1-\rho_{i-1}(X_i))} .
\end{equation*}
For the definition of the imbalance vector, the connection between the form above and the existing CAR literature is explained in Remark \ref{remark_imbalance_vector_CAR_connection}.
An obvious advantage of defining the imbalance vector in this way is that each allocation only adds an incremental term, $\left[(T_i - \rho_{i-1}(X_i)) \phi(X_i)\right]/\left[\rho_{i-1}(X_i)(1-\rho_{i-1}(X_i))\right]$, to the existing imbalance vector.

Accordingly, by generalizing existing solutions to the shift problem \cite{fangGeneralNonMarkovianFramework2026}, we set the allocation probability for the $n$th unit to be $g_{\theta_{n-1}}(\Lambda_{n-1}, X_n)$.
Here, the allocation function
\begin{equation*}
    g_\theta(\Lambda, X) = \rho_\theta(X) - p_\theta \rho_\theta(X)(1-\rho_\theta(X))
    \mathcal{S}_\phi \left( \frac{\phi(X)}{\rho_\theta(X)(1-\rho_\theta(X))} \right)^T
    \mathcal{S}_\Lambda (\Lambda),
\end{equation*}
where $p_\theta$ and $\mathcal{S}_\cdot(\cdot)$ are a scaling constant and scaling functions, respectively, designed to ensure that $g_\theta \in [0,1]$.
For further details, see \eqref{eq_CARA_allocation_function}.
This construction controls $\Lambda_n = O_P(1)$ while ensuring that $\Psi_n$ is asymptotically normal with mean zero, and we further derive its asymptotic variance.
Moreover, under certain conditions, the variance is guaranteed to be no larger than that under simple randomization with the oracle targeted allocation ratio.

\subsubsection{Model Parameter Estimation}
\label{subsubsec_introduction_estimation}

Except for stratified designs, the asymptotic properties of parameter estimators in existing CAR and CARA procedures typically rely on correct model specification.
For the CARA procedure, the estimation approach in \cite{zhangAsymptoticPropertiesCovariateadjusted2007}, based on the maximum likelihood estimator (MLE), may lack robustness under model misspecification in the CARA procedure.
When the parametric model is misspecified, the oracle model parameter $\theta^*$, defined as the maximizer of the expected log-likelihood, may not be unique, because the expected log-likelihood depends on the targeted allocation ratio.
This arises from the fact that the log-likelihood may have a nonzero conditional expectation and vary across covariates and treatment groups, and the targeted allocation ratio determines the joint distribution of covariates and treatment groups.
The nonzero conditional expectation also violates a key condition of Theorem 2.1 in \cite{zhangAsymptoticPropertiesCovariateadjusted2007}, and consequently, the theory in \cite{zhangAsymptoticPropertiesCovariateadjusted2007} cannot accommodate model misspecification.
Therefore, it is crucial to develop a universal and robust estimation method for the oracle model parameters in the CARA procedure that remains valid even when the model is not correctly specified.

For targeted sequential inference, Chambaz et al. adopted an inverse probability weighted version of the M-estimator in their design, where the propensity score is taken to be the allocation probability \cite{chambazInferenceTargetedGroupSequential2014,chambazTargetedSequentialDesign2017}.
This approach provides robustness against model misspecification.
In CAR procedures and in our CBARA procedure, the allocation probability must fluctuate around the targeted allocation ratio to maintain covariate balance, and such fluctuations may lead to high variance and instability for this weighted M-estimator.
Therefore, in this context, we use the targeted allocation ratio instead of the actual allocation probability as the propensity score, since it provides greater stability.
Although this breaks the martingale structure of the estimator, under our CBARA procedure it still maintains robustness and does not rely on correct model specification.

\subsubsection{Allocation Parameter Update Mechanism}
\label{subsubsec_introduction_parameter_update}

In this article, the targeted allocation ratio can be parameterized as $\rho_n = \rho_{\theta_n}$ for the allocation parameter $\theta_n$.
Built upon the theoretical framework in \cite{fangGeneralNonMarkovianFramework2026}, we develop a general theoretical framework for the CBARA procedure. 
Although many technical aspects differ, the theoretical analysis of the CBARA procedure exhibits a similar dependence on the variation of the allocation parameter sequence $\{\theta_n\}$ as in \cite{fangGeneralNonMarkovianFramework2026}.
We summarize these dependencies in the following two assumptions, corresponding to the law of large numbers and the central limit theorem, respectively.

The weaker assumption on the allocation parameter sequence $\{\theta_n\}$, used for establishing the law of large numbers, is stated as follows.
\begin{assumption}[Weak Diminishing Adaptation]
    \label{assumption_theta_weak_diminishing_adaptation}
    There exists some $q \in (0,1]$ such that, for any $\delta > 0$, there exists a constant $c_\delta > 0$ satisfying
    \begin{equation*}
        P\left( \frac{1}{N} \sum_{n=0}^{N-1} \| \theta_n - \theta_{n+1} \| > \delta
        \right) < c_\delta N^{-q} .
    \end{equation*}
\end{assumption}
The stronger assumption on the allocation parameter sequence $\{\theta_n\}$, used for establishing the central limit theorem, is given as follows.
\begin{assumption}[Convergence and Strong Diminishing Adaptation]
    \label{assumption_theta_convergence_strong_diminishing_adaptation}
    The allocation parameter sequence $\{\theta_n\}$ converges to $\theta^*$ in probability.
    Moreover, for some $p \in (0,1/2)$, it holds that
    \begin{equation*}
        \sum_{n=0}^{N-1} \| \theta_n - \theta_{n+1} \|
        = o_P(N^p) .
    \end{equation*}
\end{assumption}
From Assumptions \ref{assumption_theta_weak_diminishing_adaptation} and \ref{assumption_theta_convergence_strong_diminishing_adaptation}, it can be seen that the properties of the CBARA procedure and the model parameter estimation in our article do not rely on the allocation parameter having a particular asymptotic form, as is required in \cite{zhangAsymptoticPropertiesCovariateadjusted2007,alettiNonparametricCovariateadjustedResponseadaptive2018}.
They only require that the changes between successive elements of $\{\theta_n\}$ remain sufficiently small so that the variation of $\{\theta_n\}$ does not substantially disturb the analysis.

However, using the estimate of the model parameter directly as the allocation parameter may fail to satisfy this variation constraint.
To mitigate this issue, we propose two allocation parameter update mechanisms in Subsection \ref{subsec_parameter_update}, which maintain the allocation parameter in close proximity to the estimate while ensuring controlled, moderate changes at each allocation step.
Therefore, in this article, we distinguish between the estimate of the model parameter $\eta_n$ and the allocation parameter $\theta_n$.
These update mechanisms ensure that Assumptions \ref{assumption_theta_weak_diminishing_adaptation} and \ref{assumption_theta_convergence_strong_diminishing_adaptation} are automatically satisfied.
Moreover, they further guarantee the stability of the allocation parameter sequence, even under potential practical irregularities, such as the presence of response delays or the use of surrogate responses for estimation, as considered in some adaptive designs \cite{huDoublyAdaptiveBiased2008,gaoResponseAdaptiveRandomizationProcedure2024,zhangOnlineMetaLevelAdaptiveDesign2025}.

We summarize the relationships between Assumptions \ref{assumption_theta_weak_diminishing_adaptation}--\ref{assumption_theta_convergence_strong_diminishing_adaptation} and the properties of the CBARA procedure and the estimation established in Section \ref{sec_properties} as follows.
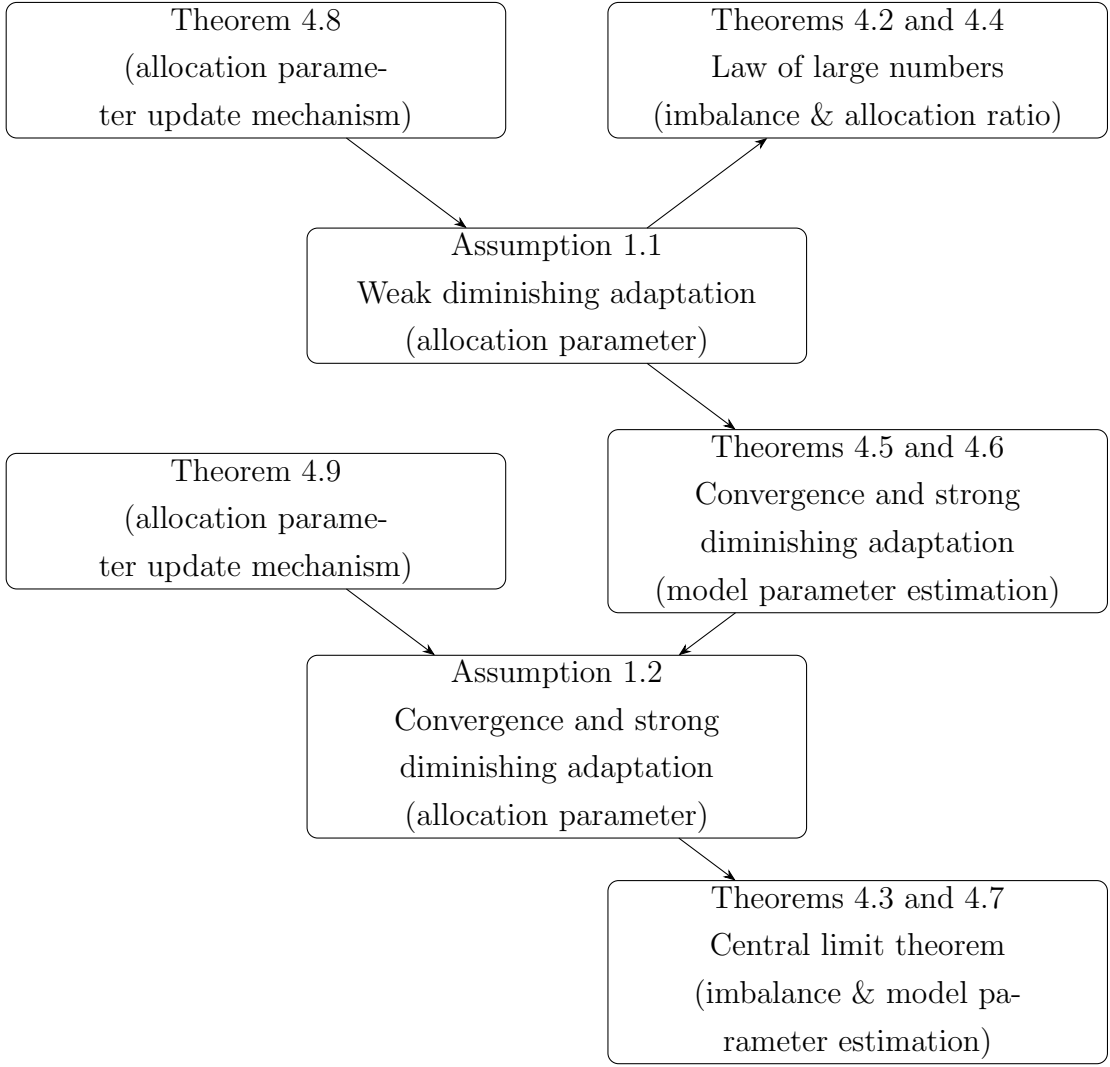
\begin{figure}[H]
    \centering
    \begin{tikzpicture}[
            every node/.style={
                draw,
                rectangle,
                rounded corners,
                align=center,
                minimum height=1.2cm,
                text width=6.5cm,
                inner sep=2pt,
                outer sep=0pt
            },
            arrow/.style={->, thin}
        ]

        \node (theorem_LLN) at (8,12) {Theorems \ref{theorem_additional_covariate_LLN} and \ref{theorem_conditional_allocation_ratio} \\ Law of large numbers \\ (imbalance \& allocation ratio)};
        \node (theorem_CLT) at (8,0) {Theorems \ref{theorem_additional_covariate_CLT} and \ref{theorem_estimation_asymptotic_normality} \\ Central limit theorem \\ (imbalance \& model parameter estimation)};
        \node (theorem_connection) at (8,6) {Theorems \ref{theorem_estimation_consistency} and \ref{theorem_estimation_difference} \\ Convergence and strong diminishing adaptation \\ (model parameter estimation)};
        \node (theorem_allocation_parameter_update_weak) at (0,12) {Theorem \ref{theorem_allocation_parameter_update_weak} \\ (allocation parameter update mechanism)};
        \node (theorem_allocation_parameter_update_strong) at (0,6) {Theorem \ref{theorem_allocation_parameter_update_strong} \\ (allocation parameter update mechanism)};
        
        \node (assumption_theta_weak_diminishing_adaptation) at (4,9) {Assumption \ref{assumption_theta_weak_diminishing_adaptation} \\ Weak diminishing adaptation \\ (allocation parameter)};
        \node (assumption_theta_convergence_strong_diminishing_adaptation) at (4,3) {Assumption \ref{assumption_theta_convergence_strong_diminishing_adaptation} \\ Convergence and strong diminishing adaptation \\ (allocation parameter)};

        \draw[-{Stealth}] (assumption_theta_weak_diminishing_adaptation) -- (theorem_LLN);
        \draw[-{Stealth}] (assumption_theta_weak_diminishing_adaptation) -- (theorem_connection);
        \draw[-{Stealth}] (assumption_theta_convergence_strong_diminishing_adaptation) -- (theorem_CLT);
        
        \draw[-{Stealth}] (theorem_allocation_parameter_update_weak) -- (assumption_theta_weak_diminishing_adaptation);
        \draw[-{Stealth}] (theorem_connection) -- (assumption_theta_convergence_strong_diminishing_adaptation);
        \draw[-{Stealth}] (theorem_allocation_parameter_update_strong) -- (assumption_theta_convergence_strong_diminishing_adaptation);
    \end{tikzpicture}
    \caption{Dependency structure among assumptions on the allocation parameter sequence and theorems.
    An arrow from node A to node B indicates that the condition on the variation of $\{\theta_n\}$ in Theorem(s) B is Assumption A, or Assumption B can be derived from Theorem(s) A.}
    \label{figure_theorem_dependency}
\end{figure}
Therefore, the conditions of Theorems \ref{theorem_estimation_consistency}, \ref{theorem_estimation_difference}, \ref{theorem_allocation_parameter_update_weak}, and \ref{theorem_allocation_parameter_update_strong} are sufficient to ensure that Assumptions \ref{assumption_theta_weak_diminishing_adaptation} and \ref{assumption_theta_convergence_strong_diminishing_adaptation} hold.
We summarize these conditions in the following proposition.
\begin{proposition}
    Suppose that Assumptions \ref{assumption_theta_compact}, \ref{assumption_lipschitz_continuity_functions_kernels}, \ref{assumption_x_sub_exponential_bound}, \ref{assumption_for_simultaneous_small_set_condition} and \ref{assumption_estimation_asymptotic} hold.
    If the allocation parameter $\theta_n$ is updated according to \eqref{eq_allocation_parameter_update_mechanism_1}, or according to \eqref{eq_allocation_parameter_update_mechanism_2} when the parameter space $\Theta$ is a convex subset of a Euclidean space and $\sum_{n=1}^\infty C_{\mathrm{clip}, n} = \infty$ for \eqref{eq_allocation_parameter_update_mechanism_2}, then both Assumptions \ref{assumption_theta_weak_diminishing_adaptation} and \ref{assumption_theta_convergence_strong_diminishing_adaptation} hold.
\end{proposition}
In the next subsection, we briefly explain why Assumptions \ref{assumption_theta_weak_diminishing_adaptation} and \ref{assumption_theta_convergence_strong_diminishing_adaptation} are required as conditions for Theorems \ref{theorem_additional_covariate_LLN}--\ref{theorem_estimation_asymptotic_normality}, through the theory of the CBARA procedure.

\subsubsection{Theory}
\label{subsubsec_introduction_theory}

The CBARA procedure exhibits features of both the CARA procedure and the CAR procedure.
Consequently, some properties that are used to analyze the CARA procedure or the CAR procedure are no longer preserved.
For example, the analysis of the CARA procedure typically relies on the allocation probability having an asymptotic linear representation in expectation, whereas in the CAR procedure, the dependence of the allocation probability on imbalance invalidates such a representation \cite{zhangAsymptoticPropertiesCovariateadjusted2007,zhangNewFamilyCovariateadjusted2009,huUnifiedFamilyCovariateAdjusted2015}.
As another example, the CARA procedure continuously updates the targeted allocation ratio, which destroys the Markovian property that is commonly used in the analysis of most CAR procedures \cite{huAsymptoticPropertiesCovariateadaptive2012,huTheoryCovariateadaptiveDesigns2020,zhangAsymptoticPropertiesMultitreatment2023,maNewUnifiedFamily2024,liuPropertiesCovariateadaptiveRandomization2025}.

In this article, we develop a new theoretical framework for adaptive designs, which allows the use of tools from Markov chain theory for a class of non-Markovian stochastic processes.
In this setting, even though the stochastic process itself does not possess the Markovian property, it can still be described using transition kernels.
For the CBARA procedure, the transition kernel corresponding to each step is different, random, and history-dependent.
Accordingly, when the object of analysis is a sum of dependent terms, we apply the Poisson equation together with a rearrangement of the terms in the sum to decompose each term into a martingale difference.
This analytical approach may appear similar to that used for the traditional CAR procedure.
However, it is important to note that in the CBARA procedure, the parameter updates introduce an additional remainder term in each summand, which captures the disturbance caused by the variation of parameters across successive steps.
Controlling this remainder term so that it does not affect the law of large numbers and the central limit theorem for the entire sum is a key analytical challenge addressed in this article.

To achieve this, we need to relate the size of the remainder term in each summand to the magnitude of the corresponding variation in the allocation parameter sequence $\{\theta_n\}$.
The size of the remainder term can be decomposed into contributions from the differences in transition kernels, the corresponding invariant probabilities and allocation functions across successive steps.
Among these differences, the main difficulty lies in the differences of the transition kernels.
Specifically, it is necessary to identify an appropriate discrepancy for the transition kernels such that, when $\theta$ is sufficiently close, the corresponding discrepancy between the kernels is also sufficiently small.
To address this issue, Subsection \ref{subsec_proof_strategy_negligibility_remaining_term} introduces a new discrepancy between transition kernels.
This discrepancy not only captures the variation in allocation probabilities induced by changes in $\theta$, but also quantifies the change in the distribution of the imbalance vector increment $\left[(T_i - \rho_{i-1}(X_i)) \phi(X_i)\right]/\left[\rho_{i-1}(X_i)(1-\rho_{i-1}(X_i))\right]$.
Using this new discrepancy, we can bound the remainder term in each summand by the corresponding variation of $\{\theta_n\}$.
Therefore, the conditions required to ensure the negligibility of the remainder term for the law of large numbers and the central limit theorem, respectively, can be summarized as the variation conditions on $\{\theta_n\}$, as stated in Assumptions \ref{assumption_theta_weak_diminishing_adaptation} and \ref{assumption_theta_convergence_strong_diminishing_adaptation}.

\subsection{Organization of the Article}

The rest of the article is organized as follows.
We first introduce the setup in Section \ref{sec_framework_setup}, and then present the CBARA procedure in Section \ref{sec_general_CARA_procedure}.
Section \ref{sec_properties} develops theoretical results.
In Section \ref{sec_proof_strategy}, we state our proof strategy for the theoretical results.
Finally, we conclude our article and provide directions for future work in Section \ref{sec_discussion}.
All technical lemmas, proofs, and experiments are provided in \cite{fangSupplementCBARACovariateBalancedandAdjusted2026}.

\section{Framework and Setup}
\label{sec_framework_setup}

Consider a trial in which patients are sequentially and randomly assigned to two groups.
Let $T_n$ be the treatment assignment of the $n$th unit, such that $T_n=1$ for the treatment and $T_n=0$ for the control.
For the $n$th unit, let $X_n \in \mathbb{R}^{d_x}$ denote the covariate vector used in the randomization procedure.
To evaluate balance between treatment groups, we consider the balance of the transformed version $\phi(X_n)$, where $\phi: \mathbb{R}^{d_x} \to \mathbb{R}^d$ is a prespecified feature map.
Let the potential outcome vector be $\bm{Y}_n = (Y_n(1), Y_n(0))$, representing the responses under treatment and control, respectively.

Denote by $Z_n$ an additional covariate vector that is not incorporated into the randomization procedure and may even be unobserved by the experimenter.
The covariate $Z_n$ is introduced solely for analytical purposes and is not required to be used in the CBARA procedure.
The additional covariate $Z$ may partially overlap with $X$, be defined as a function of $X$ and $(Y(1), Y(0))$, or represent additional covariate information collected after treatment assignment.
We suppose that the vectors $\{(X_n, \bm{Y}_n, Z_n)\}_{n \in \mathbb{N}^*}$ are independent and identically distributed (i.i.d.) random vectors.
Let $\Gamma$ denote the distribution of the covariates $X_n$, and let $\Gamma_{X,\bm{Y}}$ denote the distribution of $(X_n, \bm{Y}_n)$.
Moreover, let $\Gamma_{X,Y(t)}$ denote the distribution of $(X_n, Y_n(t))$, for $t \in \{0,1\}$.
In subsequent sections, we will similarly use $\Gamma_U$ to denote the distribution of an arbitrary random vector $U$.
\begin{assumption}
    \label{assumption_iid}

    The vectors $\{(X_n, \bm{Y}_n, Z_n)\}_{n \in \mathbb{N}^*}$ are i.i.d. random vectors.
\end{assumption}

Before the $(n+1)$th allocation, we compute the allocation parameter $\theta_n \in \Theta$ to be used for this allocation based on all information observed and available up to that point, excluding any unobserved responses or covariates from future units.
Thus, at each step, the allocation parameter may be updated dynamically based on the information accumulated so far, and in turn influence the allocation of the next unit.
Define the $\sigma$-field $\mathcal{F}_n$ by
\begin{equation*}
    \mathcal{F}_n = \sigma\left( \theta_0, X_1, \bm{Y}_1, Z_1, T_1, \theta_1, \dots, X_n, \bm{Y}_n, Z_n, T_n, \theta_n \right) .
\end{equation*}
We next formally state an assumption that, through the filtration $\{\mathcal{F}_n\}_{n \in \mathbb{N}}$, restricts the treatment assignment and the allocation parameter update at each step to depend only on the information collected up to that point, excluding any influence from future units.
\begin{assumption}
    \label{assumption_filtration}
    
The information of the future units $\{(X_{n+k}, \bm{Y}_{n+k}, Z_{n+k})\}_{k \in \mathbb{N}^*}$ is independent of $\mathcal{F}_n$ for any $n \in \mathbb{N}$.
\end{assumption}
\begin{remark}
    The most common setting for the allocation parameter is that, $\theta_0$ is fixed and for any $n \in \mathbb{N}^*$, $\theta_n$ is computed based only on the information and treatment assignments of the first $n$ units, and the CBARA procedure described in Section \ref{sec_general_CARA_procedure} satisfies this condition.
    In this case,
    \begin{equation*}
        \theta_n \in \sigma(X_1, Y_1(T_1), T_1, \dots, X_n, Y_n(T_n), T_n) .
    \end{equation*}
    Moreover, if $T_n$ depends only on the information of the first $n$ units and the treatment assignments of the first $n-1$ units, then Assumption \ref{assumption_filtration} is automatically satisfied under Assumption \ref{assumption_iid}.
\end{remark}
\begin{remark}
    The flexible specification of $\theta_n$ allows the framework to accommodate delayed responses, rare parameter updates, and other practical adaptations in constructing the allocation parameter sequence $\{\theta_n\}$.
    These capabilities are achieved through the allocation parameter update mechanism and the model parameter estimation method in Section \ref{sec_general_CARA_procedure}.
\end{remark}

Before each allocation, the desirable ratio of treatment assignments for units with covariate value $x$ is specified by the targeted allocation ratio $\rho(x)$.
Formally, $\rho(x)$ is a function of $x$ that represents the desired ratio of assigning such units to the treatment and control groups, namely $\rho(x):(1-\rho(x))$.
For brevity, we refer to ``the targeted allocation ratio for the treatment group'' as ``the targeted allocation ratio''.
The allocation probability specifies the actual conditional probability that the unit is assigned to the treatment group, given the past history $\mathcal{F}$ and its covariate $X$.

Under the CBARA procedure in this article, for the $(n+1)$th allocation, the targeted allocation ratio and the allocation probability depend on the current allocation parameter $\theta_n$.
Accordingly, at the $(n+1)$th step, the targeted allocation ratio and the allocation probability can be denoted by $\rho_{\theta_n}$ and $g_{\theta_n}$, respectively.
For any $\theta \in \Theta$, $\rho_\theta$ is a function of the covariate value $x$ alone, while $g_\theta$ is a function of both $x$ and the imbalance vector defined in \eqref{eq_definition_imbalance_vector}.

Now we start to formally develop our framework.
After $n$th allocation, we define the imbalance vector as
\begin{equation}
    \Lambda_n = \sum_{i=1}^n \frac{(T_i-\rho_{\theta_{i-1}}(X_i))\phi(X_i)}{\rho_{\theta_{i-1}}(X_i)(1-\rho_{\theta_{i-1}}(X_i))}
    \label{eq_definition_imbalance_vector} .
\end{equation}
Similarly, we can define the imbalance of the additional covariate $Z$ as
\begin{equation}
    \Psi_n = \sum_{i=1}^n \frac{(T_i-\rho_{\theta_{i-1}}(X_i)) Z_i}{\rho_{\theta_{i-1}}(X_i)(1-\rho_{\theta_{i-1}}(X_i))}
    \label{eq_definition_imbalance_vector_additional} .
\end{equation}
When the additional covariate $Z$ is taken as $\psi(X)$ for some other feature map $\psi$, $\Psi_n$ reduces to the covariate imbalance of $X$ with respect to $\psi$, in contrast to $\Lambda_n$, which is defined with respect to $\phi$.

\begin{remark}
    \label{remark_imbalance_vector_observational_study}
    Note that the definition of $\Lambda_n$ in \eqref{eq_definition_imbalance_vector} is equivalent to
    \begin{equation*}
        \Lambda_n
        = \sum_{i=1}^n \frac{T_i \phi(X_i)}{\rho_{\theta_{i-1}}(X_i)}
        - \sum_{i=1}^n \frac{(1-T_i) \phi(X_i)}{1-\rho_{\theta_{i-1}}(X_i)} .
    \end{equation*}
    Therefore, controlling $\Lambda_n$ implies that covariate means weighted by the inverse propensity score are close between treatment groups.
    This aligns with the goal in Imai and Ratkovic of achieving mean independence between the treatment and covariates after inverse propensity score weighting \cite{imaiCovariateBalancingPropensity2014}.
    The key difference between their approach and ours is that they control $\Lambda_n$ by computing appropriate propensity scores, whereas we control it at the design stage through adaptive treatment allocation.
\end{remark}

\begin{remark}
    \label{remark_imbalance_vector_CAR_connection}
    If the targeted allocation ratio $\rho_\theta(x)$ is fixed at a constant $\rho$ for all $x$ and all allocation parameter $\theta$, the imbalance vector $\Lambda_n$ reduces to a constant multiple of $\sum_{i=1}^n (T_i-\rho)\phi(X_i)$, which is a commonly used imbalance vector in the CAR literature \cite{maNewUnifiedFamily2024,bugniInferenceCovariateadaptiveRandomization2018}.
    However, when the targeted allocation ratio is allowed to vary across allocation steps and covariate values, the imbalance vector form $\sum_{i=1}^n (T_i-\rho)\phi(X_i)$ and the form in \eqref{eq_definition_imbalance_vector} are no longer equivalent.
\end{remark}

In the framework of the CBARA procedure, we consider the data-adaptive allocation mechanisms such that the allocation of $(n+1)$th unit depends on the history $\mathcal{F}_n$ and the information of the unit through the allocation parameter $\theta_n$, the imbalance vector $\Lambda_n$ and the covariate vector $X_{n+1}$.
Accordingly, for $(n+1)$th unit, the conditional probability of treatment assignment $T_{n+1}$ is given by
\begin{align}
    P(T_{n+1} = 1 \mid \mathcal{F}_n, X_{n+1}, \bm{Y}_{n+1}, Z_{n+1})
    &= P(T_{n+1} = 1 \mid \mathcal{F}_n, X_{n+1}) = g_{\theta_n}(\Lambda_n, X_{n+1}) \label{eq_T_conditional_1} \\
    P(T_{n+1} = 0 \mid \mathcal{F}_n, X_{n+1}, \bm{Y}_{n+1}, Z_{n+1})
    &= 1-P(T_{n+1} = 1 \mid \mathcal{F}_n, X_{n+1}) . \label{eq_T_conditional_0}
\end{align}

Occasionally, we adopt the following shorthand notation for clarity:
\begin{align*}
    g_\theta(1 \mid \Lambda, X) &:= g_\theta(\Lambda, X) , &
    g_\theta(0 \mid \Lambda, X) &:= 1 - g_\theta(\Lambda, X) , \\
    \rho_\theta(1 \mid X) &:= \rho_\theta(X) , &
    \rho_\theta(0 \mid X) &:= 1 - \rho_\theta(X) .
\end{align*}

In \eqref{eq_T_conditional_1} and \eqref{eq_T_conditional_0}, the function $g_\theta(\Lambda, X)$ is also referred to as the allocation function.
It fluctuates around the targeted allocation ratio $\rho_\theta(X)$ to control the imbalance vector $\Lambda_n$.
We now state the assumption regarding the treatment assignment and the ranges of the allocation function and targeted allocation ratio.
\begin{assumption}
    \label{assumption_sampling}
    The conditional probability of the treatment assignment $T_n$ satisfies \eqref{eq_T_conditional_1} and \eqref{eq_T_conditional_0}.
    In addition, there exists some $\iota > 0$ such that the allocation function $g_\theta(\Lambda, X) \in [\iota, 1 - \iota]$, and similarly $\rho_\theta(X) \in [\iota_\rho, 1 - \iota_\rho]$ for some $\iota_\rho \in (\iota, \frac{1}{2})$.
\end{assumption}
Throughout this article, we always assume that Assumptions \ref{assumption_iid}--\ref{assumption_sampling} hold.

\section{General CBARA Procedure}
\label{sec_general_CARA_procedure}

We next provide a detailed description of the CBARA procedure.
It consists of three components: an allocation function for computing the allocation probabilities of units (see Subsection \ref{subsec_allocation_function_procedure}), an estimation method of the model parameter (see Subsection \ref{subsec_estimation_parameter_procedure}), and one of two alternative allocation parameter update mechanisms (see Subsection \ref{subsec_parameter_update_procedure}).
In the absence of response delay, the CBARA procedure first obtains a sequential estimate $\eta_n$ of the model parameter $\eta^*$.
It then transforms this estimate into the allocation parameter $\theta_n$.
Rather than $\eta_n$, the allocation parameter $\theta_n$ determines both the target allocation ratio $\rho_{\theta_n}$ and the allocation function $g_{\theta_n}$ for the $(n+1)$th allocation.

\begin{remark}
    \label{remark_model_parameter}
    
    In existing CARA literature for continuous covariates, it is typically assumed that the model parameter is an element of the parameter space $\Theta$ that indexes the statistical model and characterizes the distribution of $(X, \bm{Y})$ \cite{zhangAsymptoticPropertiesCovariateadjusted2007,cheungCovariateadjustedResponseadaptiveDesigns2014,zhuCovariateadjustedResponseAdaptive2015}.
    The model is often assumed to be a generalized linear model.
    More generally, we define the model parameter $\eta^*$ as the unique maximizer of the population-level criterion in \eqref{eq_population_maximizing_point}.
\end{remark}

If the responses in each treatment group are available without delay, the CBARA procedure proceeds as follows:
\begin{algorithm}[H]
\caption{CBARA Procedure}
\begin{algorithmic}[1]
    \Require Initial allocation parameter $\theta_0$ and imbalance vector $\Lambda_0 = 0$.
    \Ensure Allocation assignments for all experimental units.

    \State \textbf{Initialization:} Set $\theta_0$ and $\Lambda_0 \gets 0$.
    Choose an update mechanism according to \eqref{eq_allocation_parameter_update_mechanism_1} or \eqref{eq_allocation_parameter_update_mechanism_2}.
    
    \For{$n = 0, 1, \dots, N-1$} \Comment{$N$ is the total number of units}
        \State Observe the $(n+1)$th unit with covariates $X_{n+1}$.
        
        \If{$n > 0$}
            \State Collect the responses $\{Y_i(T_i)\}_{i=1:n}$ of the first $n$ units.

            \State Estimate the model parameter $\eta^*$ using \eqref{eq_eta_estimation} based on the data $\{(X_i, Y_i(T_i), T_i)\}_{i=1:n}$, and denote the estimate by $\eta_n$.
            
            \State Update the allocation parameter $\theta_n$ using the pre-selected update mechanism based on the estimate sequence $\{\eta_i\}_{i=1:n}$ and the past allocation parameter sequence $\{\theta_i\}_{i=0:(n-1)}$.
        \EndIf
        
        \State Assign the $(n+1)$th unit to the treatment group with probability $g_{\theta_n}(\Lambda_n, X_{n+1})$ (see \eqref{eq_CARA_allocation_function}).
        Denote the assignment indicator by $T_{n+1} \in \{0,1\}$.
        
        \State Calculate the imbalance vector $\Lambda_{n+1} \gets \Lambda_n + \frac{(T_{n+1}-\rho_{\theta_n}(X_{n+1}))\phi(X_{n+1})}{\rho_{\theta_n}(X_{n+1})(1-\rho_{\theta_n}(X_{n+1}))}$.
    \EndFor
    
    \State \Return allocation assignments $\{T_i\}_{i=1}^N$
\end{algorithmic}
\end{algorithm}

In the presence of a delay, the only modification required in the CBARA procedure is in the estimation of $\eta_n$.
The estimation can be performed using only the units with observed responses.

\subsection{Allocation Function}

\label{subsec_allocation_function_procedure}

The allocation function, which represents the allocation mechanism in the CBARA procedure, is defined by
\begin{equation}
    g_\theta(\Lambda, X) = \rho_\theta(X) - \frac{p_\theta \phi(X)^T \Lambda}{\max\{\|\phi(X)\|/[\rho_\theta(X)(1-\rho_\theta(X))],C_\theta\} \max\{\|\Lambda\|,C_\Lambda\}} \label{eq_CARA_allocation_function} ,
\end{equation}
where $\| \cdot \|$ is the $\ell_2$-norm of the vector, $\rho_\theta: \mathbb{R}^{d_x} \to [\iota_\rho, 1-\iota_\rho]$ denotes the targeted allocation ratio, $\Lambda$ is the current imbalance vector, $\phi: \mathbb{R}^{d_x} \to \mathbb{R}^d$ is a feature mapping of the covariates, $p_\theta$, $C_\Lambda$ and $C_\theta>0$ are tuning constants.
Moreover, we assume that $p_\theta$ and $1/C_\theta$ are uniformly bounded away from zero over $\theta \in \Theta$.

The construction of the allocation function is somewhat involved.
To provide intuition, note that
\begin{align*}
    & \quad \frac{\phi(X)^T \Lambda}{\max\{\|\phi(X)\|/[\rho_\theta(X)(1-\rho_\theta(X))],C_\theta\} \max\{\|\Lambda\|,C_\Lambda\}} \\
    &= \rho_\theta(X)(1-\rho_\theta(X))
    \left[ \frac{\phi(X)/[\rho_\theta(X)(1-\rho_\theta(X))]}{\max\{\|\phi(X)\|/[\rho_\theta(X)(1-\rho_\theta(X))],C_\theta\}} \right]^T
    \frac{\Lambda}{\max\{\|\Lambda\|,C_\Lambda\}} .
\end{align*}
Therefore, the allocation function can be interpreted as adjusting the targeted allocation ratio, where the adjustment magnitude is given by the inner product between the normalized vector $\phi(X)/[\rho_\theta(X)(1-\rho_\theta(X))]$ and the normalized $\Lambda$, scaled by $\rho_\theta(X)(1-\rho_\theta(X))$ and the constant $p_\theta$.
The vector
\begin{equation*}
    \frac{\phi(X)}{\rho_\theta(X)(1-\rho_\theta(X))}
    = \frac{(1-\rho_\theta(X))\phi(X)}{\rho_\theta(X)(1-\rho_\theta(X))}
    - \frac{(0-\rho_\theta(X))\phi(X)}{\rho_\theta(X)(1-\rho_\theta(X))}
\end{equation*}
represents the difference between the increment vectors of $\Lambda$ under $T=1$ and $T=0$ at each allocation.
This inner product formulation ensures that each allocation tends to produce an increment of $\Lambda$ in a direction opposite to $\Lambda$ itself.
The quantity $\rho_\theta(X)(1-\rho_\theta(X))$ represents the adjustment weight associated with different values of $X$.
Its inclusion is necessary, because it influences the asymptotic behavior, which Remark \ref{remark_allocation_function_unequal_form} will demonstrate.
The constant $p_\theta$ controls the magnitude of this adjustment uniformly across all $X$.
Setting $p_\theta = 1$ is already sufficient for any $\theta$.
Larger values of $p_\theta$ are also permissible, provided that $g_\theta$ remains within the interval $[0,1]$.
The constants $C_\Lambda$ and $C_\theta > 0$ are introduced to avoid division by very small values and to ensure continuity during normalization.
The former is independent of $\theta$.

\begin{remark}
    When the targeted allocation ratio $\rho_\theta$ is fixed at a constant value, the CBARA procedure reduces to a CAR procedure.
    Under the conditions $\rho_\theta \equiv \rho$ for some constant $\rho \in (0,1)$, and $\|\phi(X)\| \leq C \rho(1-\rho)$ for some constant $C > 0$, taking $C_\theta = C$ and $p_\theta = 1$ yields
    \begin{equation*}
        g_\theta(\Lambda,X)
        = \rho - \frac{\phi(X)^T \Lambda}{C\max\{\|\Lambda\|,C_\Lambda\}}.
    \end{equation*}
    Hence, the allocation reduces to a constant targeted allocation ratio $\rho$ plus a bounded linear adjustment.
\end{remark}

\begin{remark}
    \label{remark_allocation_function_unequal_form}

This allocation function originates from the solution to the shift problem in the CAR setting.
    The shift problem is that, with a fixed unequal allocation ratio $\rho$, the imbalance of the additional covariate may no longer be centered at $0$ \cite{liuPropertiesCovariateadaptiveRandomization2025,fangGeneralNonMarkovianFramework2026}.
    To address this issue, Fang and Ma propose a new form of allocation function \cite{fangGeneralNonMarkovianFramework2026}
    \begin{equation*}
        g_\theta(\Lambda, x) = \rho + \alpha(x)^T \beta(\Lambda) ,
    \end{equation*}
    where $\alpha: \mathbb{R}^{d_x} \to \mathbb{R}^d$ and $\beta: \mathbb{R}^{d_x} \to \mathbb{R}^d$ are two functions that satisfy certain conditions to balance $\phi(X) = X$.
    Similarly, for the CBARA case with a flexible arbitrary targeted allocation ratio, the allocation function should have the form
    \begin{equation*}
        g_\theta(\Lambda, x) = \rho_\theta(x) + \alpha_\theta(x)^T \beta_\theta(\Lambda) .
    \end{equation*}
    The choice of the function $\alpha_\theta$ affects the asymptotic variance, as shown in Theorems \ref{theorem_additional_covariate_CLT} and \ref{theorem_estimation_asymptotic_normality}.
The details of how the function $\alpha_\theta$ influences the variance are provided in the proof of Lemma F.2 in \cite{fangSupplementCBARACovariateBalancedandAdjusted2026}.
    The choice of the function $\beta_\theta$ affects the boundedness in probability of the imbalance vector $\Lambda_n$.
    In light of these considerations, we adopt \eqref{eq_CARA_allocation_function} as the allocation function in our CBARA procedure.
\end{remark}

\subsection{Model Parameter Estimation}

\label{subsec_estimation_parameter_procedure}

Suppose that the allocation parameter $\theta$ and the estimate of the model parameter $\eta$ share a common parameter space $\Theta$.
Let $\eta^* \in \Theta$ denote the oracle model parameter associated with the joint distribution of $(X, \bm{Y})$.
Specifically, we assume that
\begin{assumption}
    \label{assumption_population_maximizing_point}
    The model parameter $\eta^*$ is defined as the unique maximizer of the expected criterion,
    \begin{align}
        \eta^*
        &= \argmax_{\eta \in \Theta}
        \mathbb{E}_{(X,Y(1),Y(0)) \sim \Gamma_{X,\bm{Y}}} \label{eq_population_maximizing_point} \\
        & \quad\quad \left[ \rho^{\mathrm{ref}}(X) m_\eta(X, Y(1), 1)
        + [1 - \rho^{\mathrm{ref}}(X)] m_\eta(X, Y(0), 0) \right] \notag ,
    \end{align}
    where $\rho^{\mathrm{ref}}: \mathbb{R}^{d_x} \to [0,1]$ is a function specifying the reference allocation ratio for the treatment group.
\end{assumption}
Here, $m_\eta$ can be any known function used as the criterion in an M-estimator \cite{vaartAsymptoticStatistics2007,zhangStatisticalInferenceMEstimators2021}.
The definition in \eqref{eq_population_maximizing_point} allows the parameter $\eta^*$ to be defined in a manner that accounts for differing contributions from the treatment and control groups, as weighted by $\rho^{\mathrm{ref}}$.
A common choice for $\rho^{\mathrm{ref}}$ is $\rho^{\mathrm{ref}} \equiv 1/2$, in which case the treatment and control groups are weighted equally \cite{chambazTargetedSequentialDesign2017}.
Throughout this article, we always assume that Assumption \ref{assumption_population_maximizing_point} holds.
\begin{example}
    Under the generalized linear model, for a given covariate $x$, the response $Y(t)$ under treatment $T = t$ is assumed to follow a distribution in the exponential family.
    Specifically, for any $t \in \{0,1\}$, the conditional density is given by
    \begin{equation*}
        f_t \left( y_t \mid x, \beta_t, \beta \right)
        = \exp \left\{ \left( y_t \mu_t - a_t\left(\mu_t\right)\right) / \phi_t
        + b_t\left(y_t, \phi_t\right) \right\}
    \end{equation*}
    with an inverse link function $\mu_t=h_t\left( x^T \left(\beta_t^T, \beta^T\right)^T \right)$, where $\beta_t$ are group-specific coefficients and $\beta$ is a common coefficient shared across two groups.
    Here, $x^T \left(\beta_t^T, \beta^T\right)^T$ is the inner product of the covariate vector $x$ and the coefficient vector $(\beta_t, \beta)$.
    Assuming that the scale parameter $\phi_t$ is fixed, and defining $\eta = (\beta_t, \beta)$, the function $m_\eta$ can be written as the log-likelihood function
    \begin{equation*}
        m_\eta(x, y, t)
        = \ln \left[ f_t \left( y \mid x, \beta_t, \beta \right) \right]
        = \left( y_t \mu_t - a_t\left(\mu_t\right)\right) / \phi_t
        + b_t\left(y_t, \phi_t\right) .
    \end{equation*}
    Accordingly, under the generalized linear model, $\eta$ includes both the common parameters $\beta$ shared across treatment groups and the treatment-specific parameters $\beta_t$ capturing group-level differences.
    This generalized linear model setup is similar to that in \cite{zhuCovariateadjustedResponseAdaptive2015,cheungCovariateadjustedResponseadaptiveDesigns2014}.
    However, it differs from \cite{zhangAsymptoticPropertiesCovariateadjusted2007,zhangNewFamilyCovariateadjusted2009}, which do not allow for common parameters shared across treatment groups.
\end{example}

Once the responses of the first $n$ units have been observed, the corresponding estimator of $\eta^*$ is defined as
\begin{equation}
    \eta_n \in \argmax_{\eta \in \Theta}
    \sum_{i=1}^n
    \left[ \frac{\rho^{\mathrm{ref}}(T_i \mid X_i)}{\rho_{\theta_{i-1}}(T_i \mid X_i)} m_\eta(X_i, Y_i(T_i), T_i)
    \right] \label{eq_eta_estimation} ,
\end{equation}
which serves as the sample analog of \eqref{eq_population_maximizing_point}.
Here, the reference allocation ratio $\rho^{\mathrm{ref}}$ is the same as that in Assumption \ref{assumption_population_maximizing_point}.
This approach follows the inverse propensity score weighted (IPW) methodology as described in \cite{chambazTargetedSequentialDesign2017}.
However, the denominator, which represents the propensity score, is taken as the targeted allocation ratio rather than the actual allocation probability.

\subsection{Allocation Parameter Update Mechanism}

\label{subsec_parameter_update_procedure}

Let the oracle allocation parameter be $\theta^* = \eta^*$.
Instead of directly setting $\theta_n = \eta_n$, we consider two mechanisms for updating the allocation parameter $\theta_n$ before the $(n+1)$th allocation, based on the estimate sequence $\{ \eta_i \}_{i=1:n}$ and the historical allocation parameter sequence $\{\theta_i\}_{i=0:(n-1)}$.
The initial value $\theta_0$ can be chosen arbitrarily, and in particular it may be chosen such that $\rho_{\theta_0} \equiv 1/2$.
\begin{enumerate}
    \item \textbf{Increasingly Rare Update Mechanism.}
    The allocation parameter $\theta_n$ is updated as
    \begin{equation}
        \theta_n =
        \begin{cases}
            \theta_{n-1} , & \text{if } n \notin S , \\
            \eta_n , & \text{if } n \in S ,
        \end{cases}
        \label{eq_allocation_parameter_update_mechanism_1}
    \end{equation}
    where the infinite set $S \subset \mathbb{N}^*$ satisfies $\frac{\#(S \cap \{1, \dots, n\})}{n} \rightarrow 0$ as $n \rightarrow \infty$.
    This mechanism is inspired by Adapted Increasingly Rarely Markov chain Monte Carlo (AirMCMC) \cite{chimisovAirMarkovChain2018}.
    This mechanism guarantees that the allocation parameter updates occur increasingly rarely as the CBARA procedure proceeds.
    \item \textbf{Clipped Update Mechanism.}
    The allocation parameter $\theta_n$ is updated as
    \begin{equation}
        \theta_n =
        \theta_{n-1} + \min \left\{ \|\eta_n - \theta_{n-1}\|, C_{\mathrm{clip}, n} \right\}
        \cdot \frac{\eta_n - \theta_{n-1}}{\|\eta_n - \theta_{n-1}\|} , \label{eq_allocation_parameter_update_mechanism_2}
    \end{equation}
    where the constant sequence $\{ C_{\mathrm{clip}, n} \}$ satisfies $C_{\mathrm{clip}, n} \rightarrow 0$ as $n \rightarrow \infty$.
    This mechanism is similar to the trust-region Newton-CG method \cite{nocedalNumericalOptimization2006}, as it restricts the magnitude of each allocation parameter update to be small and vanishing.
\end{enumerate}

\section{Properties}
\label{sec_properties}

Let the parameter space $\Theta$ be endowed with the $\ell_2$ norm $\|\cdot\|$.
We impose the following assumption on $\Theta$:
\begin{assumption}
    \label{assumption_theta_compact}

    The parameter space $\Theta$ is a compact subset of a Euclidean space equipped with the $\ell_2$ norm.
\end{assumption}
The imbalance vector is also measured under the $\ell_2$ norm.
In addition, we assume that the targeted allocation ratio $\rho_\theta(\cdot)$ and the allocation function $g_\theta(\Lambda, \cdot)$ are Lipschitz continuous with respect to the parameter $\theta \in \Theta$ and $\Lambda \in \mathbb{R}^d$.
\begin{assumption}
    \label{assumption_lipschitz_continuity_functions_kernels}

    There exist constants $L_g, L_\rho > 0$ such that for any $\theta, \theta^\prime \in \Theta$ and any $\Lambda, \Lambda^\prime$,
    \begin{equation*}
        \left\| g_{\theta}(\Lambda,\cdot) - g_{\theta^\prime}(\Lambda^\prime,\cdot) \right\|_{L^2(\Gamma)}
        \leq L_g \|\theta - \theta^\prime\| + L_g \|\Lambda - \Lambda^\prime\|
        \quad \text{and} \quad
        \left\| \rho_{\theta} - \rho_{\theta^\prime} \right\|_{L^2(\Gamma)}
        \leq L_\rho \|\theta - \theta^\prime\| .
    \end{equation*}
\end{assumption}
Under the allocation function form in \eqref{eq_CARA_allocation_function}, a sufficient condition for Assumption \ref{assumption_lipschitz_continuity_functions_kernels} is that $\phi(X)$ has a finite first moment and there exist constants $L_C, L_\rho > 0$ such that for any $\theta, \theta^\prime \in \Theta$,
\begin{equation*}
    \max \left\{ \left| C_\theta^{-1} - C_{\theta^\prime}^{-1} \right|, \left| p_\theta - p_{\theta^\prime} \right| \right\}
    \leq L_C \|\theta - \theta^\prime\|
    \quad \text{and} \quad
    \left\| \rho_{\theta} - \rho_{\theta^\prime} \right\|_{L^2(\Gamma)}
    \leq L_\rho \|\theta - \theta^\prime\| .
\end{equation*}

\subsection{Allocation Function}

\label{subsec_allocation_function}

Let $W_\phi$ denote the linear subspace spanned by the support of the distribution of $\phi(X)$.
This subspace serves as the state space of $\{\Lambda_n\}_{n \in \mathbb{N}}$.
Note that $W_\phi$ may be a proper subspace of $\mathbb{R}^d$ when the components of the vector $\phi(X)$ are not linearly independent.
Under the following assumption alone, the CBARA procedure achieves covariate balance in the sense that $\Lambda_n = O_P(1)$.
\begin{assumption}
    \label{assumption_x_sub_exponential_bound}
    For some $\lambda > 0$, $\mathbb{E} \left[ \exp(\lambda \|\phi(X)\|) \right]  = C< \infty$.
    Equivalently, $\phi(X)$ is sub-exponential.
\end{assumption}

\begin{theorem}
    \label{theorem_boundedness}
    Suppose that Assumption \ref{assumption_x_sub_exponential_bound} holds.
    Then the stochastic process $\{\Lambda_n\}$ is bounded in probability, that is, $\Lambda_n = O_P(1)$.
\end{theorem}

\begin{remark}
Existing analyses that establish $\Lambda_n = O_P(1)$ for CAR procedures generally rely on the Markov property and ergodicity of the process $\{\Lambda_n\}$ \cite{huAsymptoticPropertiesCovariateadaptive2012,huTheoryCovariateadaptiveDesigns2020,maNewUnifiedFamily2024,zhangAsymptoticPropertiesMultitreatment2023,huMultiArmCovariateAdaptiveRandomization2023,yangSequentialCovariateadjustedRandomization2024,liuPropertiesCovariateadaptiveRandomization2025}.
    To the best of our knowledge, theoretical results on covariate balance for the CARA procedure have not yet been established.
    This is partly because the Markov property and the ergodicity condition are violated, rendering the classical Markov chain techniques inapplicable.
    Instead, we establish the inequality $\mathbb{E} \left[ e^{\lambda_1 \|\Lambda_{n+1}\|} \mid \mathcal{F}_n \right] \leq \beta e^{\lambda_1 \|\Lambda_n\|} + b$ for some positive constants $\beta < 1$, $b$, and $\lambda_1$; see Subsection \ref{subsec_proof_strategy_boundedness} for the proof strategy.
    This inequality implies that $\sup_{n \geq 0} \mathbb{E} \left[ e^{\lambda_1 \|\Lambda_n\|} \right] < \infty$, and hence $\Lambda_n = O_P(1)$.
\end{remark}

Next, we consider the balance of the additional covariate, which we formalize by deriving the asymptotic distribution of $\Psi_n$.
To achieve this, we introduce an assumption that ensures the small set condition of the transition kernel $P_\theta$ on the state space $W_\phi$, where $P_\theta$ denotes the transition kernel of the Markov chain $\{\Lambda_n\}$ when the allocation parameter $\theta$ is fixed.
The definition of the small set condition can be found in \cite{meynMarkovChainsStochastic2009}.
A formal definition of $P_\theta$ is provided in Subsection \ref{subsec_connection}.
\begin{assumption}
    \label{assumption_for_simultaneous_small_set_condition}

For any possible parameter value $\theta_* \in \Theta$, there exist some neighborhood $B_{\theta_*}$, $s$ points $x^{(1)}, \dots, x^{(s)} \in \mathbb{R}^{d_x}$, and $\epsilon, r_x>0$ such that
    \begin{enumerate}
        \item \label{item_assumption_for_simultaneous_small_set_condition_1} (Mass condition) $\Gamma \geq \epsilon \mu_{\mathrm{leb}, B(x^{(i)}, r_x)}$ for any $i \in \{1, \dots, s\}$, where the collection of balls $\{ B(x^{(i)}, r_x) \}_{i=1:s}$ are pairwise disjoint.
        \item \label{item_assumption_for_simultaneous_small_set_condition_2} (Regularity) For any $\theta \in B_{\theta_*}$ and $i \in \{1, \dots, s\}$, the functions $\phi$ and $\rho_\theta$ are continuously differentiable on $B(x^{(i)}, r_x)$, and there exists some $L>0$ (independent of $i,\theta$) such that $\phi/(1-\rho_\theta)$, $\phi/\rho_\theta$, $D [\phi/(1-\rho_\theta)]$ and $D (\phi/\rho_\theta)$ are $L$-Lipschitz continuous on $B(x^{(i)}, r_x)$.
        \item \label{item_assumption_for_simultaneous_small_set_condition_3} (Non-degeneracy) There exists some $C_\sigma > 0$ such that for any $\theta \in B_{\theta_*}$ and $i \in \{1, \dots, s\}$, the $d$th singular value satisfies
        \begin{equation*}
            \sigma_d(D\Phi_\theta(x^{(1)}, \dots, x^{(s)}, x^{(1)}, \dots, x^{(s)})) \geq C_\sigma ,
        \end{equation*}
        where
        \begin{align*}
            & \quad D \Phi_\theta(x_1, \dots, x_{2s}) \\
            &= \left(
                D(\phi/\rho_\theta)(x_1), \dots, D(\phi/\rho_\theta)(x_s),
                -D[\phi/(1-\rho_\theta)](x_{s+1}), \dots, -D[\phi/(1-\rho_\theta)](x_{2s})
            \right)
        \end{align*}
        is the differentiation of the function
        \begin{align*}
            \Phi_\theta :
            & \mathbb{R}^{2s d_x} \to \mathbb{R}^d , \\
            & (x_1, \dots, x_{2s}) \mapsto 
            \Phi_\theta(x_1, \dots, x_{2s})
            = \sum_{i=1}^s \left\{ \phi(x_i)/\rho_\theta(x_i)
            - \phi(x_{i+s})/\left[ 1-\rho_\theta(x_{i+s}) \right]
            \right\} .
        \end{align*}
    \end{enumerate}
\end{assumption}
\begin{remark}
    This assumption applies only when $W_\phi = \mathbb{R}^d$.
    After an appropriate linear transformation, the theoretical result in this section also extends to the case where $W_\phi$ is a proper subspace of $\mathbb{R}^d$, through the representation $\phi = A \phi_A$,
    where $\phi_A$ has full column rank equal to the dimension of $W_\phi$.
    However, the assumption does not hold when both $\rho_\theta$ and $\phi$ are discrete valued functions.
    In such cases, regardless of how a linear transformation is applied, the small set condition may fail, and consequently, classical Markov chain theory is no longer applicable.
\end{remark}

Under Assumption \ref{assumption_for_simultaneous_small_set_condition} on the density of $\Gamma$ and the smoothness of the functions $\phi$ and $\rho_\theta$, the following two theorems establish the law of large numbers and the central limit theorem for $\Psi_n$.
They respectively characterize weaker and stronger balance properties of the additional covariate.
Although the following two theorems only consider the case where $Z$ is one-dimensional, the extension to the multivariate case is straightforward.
\begin{theorem}
    \label{theorem_additional_covariate_LLN}

    Let the additional covariate $Z$ be one-dimensional.
    Suppose that Assumptions \ref{assumption_theta_compact}, \ref{assumption_lipschitz_continuity_functions_kernels}, \ref{assumption_x_sub_exponential_bound} and \ref{assumption_for_simultaneous_small_set_condition} hold.
    If the step sizes of the allocation parameter sequence $\{\theta_n\}$ satisfy Assumption \ref{assumption_theta_weak_diminishing_adaptation} and $\mathbb{E} \left[ Z^2 \right] < \infty$, then
    \begin{equation*}
        \frac{\Psi_N}{N}
        = \frac{1}{N} \sum_{n=1}^N \frac{(T_n-\rho_{\theta_{n-1}}(X_n)) Z_n}{\rho_{\theta_{n-1}}(X_n)(1-\rho_{\theta_{n-1}}(X_n))}
        \xrightarrow{\mathbb{P}} 0 .
    \end{equation*}
\end{theorem}

Let $\Gamma_{X,Z}$ denote the joint distribution of $(X, Z)$.
\begin{theorem}
    \label{theorem_additional_covariate_CLT}

    Let the additional covariate $Z$ be one-dimensional.
    Suppose that Assumptions \ref{assumption_theta_compact}, \ref{assumption_lipschitz_continuity_functions_kernels}, \ref{assumption_x_sub_exponential_bound} and \ref{assumption_for_simultaneous_small_set_condition} hold.
    If the allocation parameter sequence $\{\theta_n\}$ satisfies Assumption \ref{assumption_theta_convergence_strong_diminishing_adaptation} and $\mathbb{E} \left[ |Z|^{4+\epsilon} \right] < \infty$ for some $\epsilon > 0$, then
    \begin{equation*}
        \frac{\Psi_N}{\sqrt{N}}
        = \frac{1}{\sqrt{N}} \sum_{n=1}^N \frac{(T_n-\rho_{\theta_{n-1}}(X_n)) Z_n}{\rho_{\theta_{n-1}}(X_n)(1-\rho_{\theta_{n-1}}(X_n))}
        \xrightarrow{d} \mathcal{N}(0, {\sigma^*_{(Z)}}^2) ,
    \end{equation*}
    where the asymptotic variance is
    \begin{equation*}
        {\sigma^*_{(Z)}}^2
        = \mathbb{E}_{(X, Z) \sim \Gamma_{X,Z}} \left[
        \left[ \rho_{\theta^*}(X)(1-\rho_{\theta^*}(X)) \right]^{-1}
        \left\{ Z - a^T \phi(X) \right\}^2 \right] ,
    \end{equation*}
    and the vector $a$ satisfies
    \begin{align*}
        & a^T \mathbb{E} \left[ [\rho_{\theta^*}(X)(1-\rho_{\theta^*}(X))]^{-1} \phi(X) \phi(X)^T
        / \max\{ \|\phi(X)\|/[\rho_{\theta^*}(X)(1-\rho_{\theta^*}(X))], C_{\theta^*} \} \right] \\
        & \quad = \mathbb{E} \left[ [\rho_{\theta^*}(X)(1-\rho_{\theta^*}(X))]^{-1} Z \phi(X)^T
        / \max\{ \|\phi(X)\|/[\rho_{\theta^*}(X)(1-\rho_{\theta^*}(X))], C_{\theta^*} \} \right] .
    \end{align*}
\end{theorem}

\begin{remark}
    The law of large numbers and the central limit theorem rely on different restrictions on the magnitude of changes between consecutive elements of the allocation parameter sequence, namely Assumptions \ref{assumption_theta_weak_diminishing_adaptation} and \ref{assumption_theta_convergence_strong_diminishing_adaptation}.
    Once these conditions are satisfied, the specific values of $\{\theta_n\}$ do not affect the fundamental asymptotic behavior of the estimators.
    In particular, the law of large numbers does not require the convergence of the parameter sequence, and the asymptotic variance in the central limit theorem depends only on the limiting value $\theta^*$ of the parameter sequence.
\end{remark}

\begin{remark}
    The asymptotic variance in the central limit theorem admits an explicit analytical form, which corresponds to a linear regression adjustment.
    The analytical form is equivalent to the variance of the imbalance of $Z - a^T \phi(X)$ under simple randomization with the oracle targeted allocation ratio.
    The variable $Z - a^T \phi(X)$ is a linearly adjusted version of $Z$, and the coefficient vector $a$ depends on the tuning constants of the CBARA procedure.
    Specifically, if the tuning constant $C_{\theta^*} \geq \|\phi(X)\|/[\rho_{\theta^*}(X)(1-\rho_{\theta^*}(X))]$ almost surely, then the coefficient vector $a$ in Theorem \ref{theorem_additional_covariate_CLT} minimizes the variance $\mathbb{E}_{(X, Z) \sim \Gamma_{X,Z}} \left[ \left[ \rho_{\theta^*}(X)(1-\rho_{\theta^*}(X)) \right]^{-1} \left\{ Z - a^T \phi(X) \right\}^2 \right]$, and the variance is guaranteed to be no larger than the variance under simple randomization with the oracle targeted allocation ratio.
    However, if $C_{\theta^*}$ is too small, the CBARA procedure still balances $\Lambda_n$, but may yield a suboptimal coefficient vector $a$, resulting in larger imbalance in $Z$ than simple randomization.
\end{remark}

Given a covariate value $x$ such that $P_\Gamma(X = x) > 0$, the conditional allocation ratio among units with the same covariate value $x$ can be defined as
\begin{equation*}
    \frac{N_{n,1}(x)}{N_n(x)} ,
    \quad \text{where }
    N_{n,1}(x) := \sum_{i=1}^n T_i \mathbb{I}(X_i = x)
    \text{ and }
    N_n(x) := \sum_{i=1}^n \mathbb{I}(X_i = x) .
\end{equation*}
The consistency of the conditional allocation ratio follows from the following theorem, which is analogous to Theorem \ref{theorem_additional_covariate_LLN} when setting $Z = \mathbb{I}(X = x)$.
\begin{theorem}
    \label{theorem_conditional_allocation_ratio}
    Given a covariate x, suppose that $P_\Gamma(X = x) > 0$.
    Suppose that Assumptions \ref{assumption_theta_compact}, \ref{assumption_lipschitz_continuity_functions_kernels}, \ref{assumption_x_sub_exponential_bound} and \ref{assumption_for_simultaneous_small_set_condition} hold.
    If the step sizes of the allocation parameter sequence $\{\theta_n\}$ satisfy Assumption \ref{assumption_theta_weak_diminishing_adaptation} and $\theta_n \xrightarrow{\mathbb{P}} \theta^*$, then
    \begin{equation*}
        \frac{N_{n,1}(x)}{N_n(x)}
        \xrightarrow{\mathbb{P}} \rho_{\theta^*}(x) .
    \end{equation*}
\end{theorem}

\subsection{Model Parameter Estimation}

\label{subsec_estimation_parameter}

The model parameter estimation in \eqref{eq_eta_estimation} requires certain regularity conditions on the function $m_\eta$.
In order to ensure the validity of the consistency and the asymptotic normality, we introduce two assumptions: a weaker one, Assumption \ref{assumption_estimation_convergence}, and a strictly stronger one, Assumption \ref{assumption_estimation_asymptotic}.
Both are standard in the theory of M-estimators \cite{vaartAsymptoticStatistics2007,zhangStatisticalInferenceMEstimators2021}.

\begin{assumption}
    \label{assumption_estimation_convergence}

    The function $m_\eta$ satisfies the following conditions.
    \begin{enumerate}
        \item (Lower-semicontinuous) For each $t \in \{0, 1\}$, there exists a measurable set $\mathcal{N}_t$ such that $m_\eta(x, y, t)$ is lower-semicontinuous in $\eta$ for all $(x, y) \in \mathcal{N}_t$, and $\Gamma_{X,Y(t)}(\mathcal{N}_t) = 1$.
        \item (Local) For each $t \in \{0, 1\}$ and $\eta \in \Theta$, there exists a neighborhood $U$ of $\eta$ such that
        \begin{equation*}
            \mathbb{E}_{(X,Y(t)) \sim \Gamma_{X,Y(t)}} \left| \sup_{\eta \in U} m_\eta(X, Y(t), t) \right|^2 < \infty .
        \end{equation*}
    \end{enumerate}
\end{assumption}

\begin{theorem}
    \label{theorem_estimation_consistency}

    Suppose that Assumptions \ref{assumption_theta_compact}, \ref{assumption_lipschitz_continuity_functions_kernels}, \ref{assumption_x_sub_exponential_bound}, \ref{assumption_for_simultaneous_small_set_condition} and \ref{assumption_estimation_convergence} hold.
    If the step sizes of the allocation parameter sequence $\{\theta_n\}$ satisfy Assumption \ref{assumption_theta_weak_diminishing_adaptation}, then $\eta_n \xrightarrow{\mathbb{P}} \eta^*$.
\end{theorem}

\begin{assumption}
    \label{assumption_estimation_asymptotic}

    The function $m_\eta$ and the parameter space $\Theta$ satisfy the following conditions.
    \begin{enumerate}
        \item The parameter space $\Theta$ is a subset of a Euclidean space, and $\eta^*$ is an interior point of $\Theta$.
        \item For each $t \in \{0, 1\}$, and for all $(x, y)$ in the support of $\Gamma_{X,Y(t)}$ and $\eta \in \Theta$, the first and second derivatives of $m_\eta(x, y, t)$ with respect to $\eta$ exist and are denoted by $\dot{m}_\eta(x, y, t)$ and $\ddot{m}_\eta(x, y, t)$, respectively.
        \item For each $t \in \{0, 1\}$,        
        \begin{align*}
            \mathbb{E}_{(X,Y(t)) \sim \Gamma_{X,Y(t)}} \left\| \dot{m}_{\eta^*}(X,Y(t),t) \right\|^2
            &< \infty, \quad \text{and} \\
            \mathbb{E}_{(X,Y(t)) \sim \Gamma_{X,Y(t)}} \left\| \ddot{m}_{\eta^*}(X,Y(t),t) \right\|^2
            &< \infty .
        \end{align*}
        \item There exists a positive function $s(x, y, t)$ such that
        \begin{itemize}
            \item for any $\eta$, $\eta^\prime \in \Theta$,
            \begin{equation*}
                \left\| \ddot{m}_\eta(x,y,t) - \ddot{m}_{\eta^\prime}(x,y,t) \right\|
                \leq s(x,y,t) \| \eta - \eta^\prime \| ,
            \end{equation*}
            where $\|\cdot\|$ denotes the $\ell_2$ norm of a matrix,
            \item for each $t \in \{0, 1\}$,
            \begin{equation*}
                \mathbb{E}_{(X,Y(t)) \sim \Gamma_{X,Y(t)}} \left[ s(X,Y(t),t) \right]^2 < \infty .
            \end{equation*}
        \end{itemize}
        \item The matrix
        \begin{equation*}
            \ddot{M}_{\eta^*}
            = \mathbb{E}_{(X,Y(1),Y(0)) \sim \Gamma_{X,\bm{Y}}} \left[ \rho^{\mathrm{ref}}(X) \ddot{m}_{\eta^*}(X, Y(1), 1)
            + [1 - \rho^{\mathrm{ref}}(X)] \ddot{m}_{\eta^*}(X, Y(0), 0) \right]
        \end{equation*}
        is invertible.
    \end{enumerate}
\end{assumption}

Before stating the central limit theorem, we first present a bound on the magnitude of the variations of $\{\eta_n\}$.
\begin{theorem}
    \label{theorem_estimation_difference}

    Suppose that Assumptions \ref{assumption_theta_compact}, \ref{assumption_lipschitz_continuity_functions_kernels}, \ref{assumption_x_sub_exponential_bound}, \ref{assumption_for_simultaneous_small_set_condition} and \ref{assumption_estimation_asymptotic} hold.
    If the step sizes of the allocation parameter sequence $\{\theta_n\}$ satisfy Assumption \ref{assumption_theta_weak_diminishing_adaptation}, then
    \begin{equation*}
        \mathbb{E} \left[ \| \eta_n - \eta_{n-1} \| \right] = O(n^{-q}) ,
    \end{equation*}
    and
    \begin{equation*}
        \sum_{n=1}^{N} \| \eta_n - \eta_{n-1} \| = o_P(N^p) ,
    \end{equation*}
    for any $p \in (1-q, 1)$, where $q \in (0,1]$ is defined in Assumption \ref{assumption_theta_weak_diminishing_adaptation}.
\end{theorem}

\begin{remark}
    When we directly set $\theta_n \equiv \eta_n$, Assumption \ref{assumption_theta_convergence_strong_diminishing_adaptation} holds provided that Assumption \ref{assumption_theta_weak_diminishing_adaptation} with $q > 1/2$ is satisfied.
    The result in Theorem \ref{theorem_estimation_difference},
    \begin{equation*}
        \sum_{n=1}^{N} \| \eta_n - \eta_{n-1} \| = o_P(N^p) ,
    \end{equation*}
    parallels the condition in Assumption \ref{assumption_theta_convergence_strong_diminishing_adaptation}, $\sum_{n=0}^{N-1} \| \theta_n - \theta_{n+1} \| = o_P(N^p)$.
    Moreover, the result in Theorem \ref{theorem_estimation_consistency}, $\eta_n \xrightarrow{\mathbb{P}} \eta^*$, similarly parallels the corresponding convergence condition in Assumption \ref{assumption_theta_convergence_strong_diminishing_adaptation}, $\theta_n \xrightarrow{\mathbb{P}} \theta^*$.
    These would imply that if $\theta_n \equiv \eta_n$, Assumption \ref{assumption_theta_weak_diminishing_adaptation} with $q > 1/2$ would be sufficient to establish Assumption \ref{assumption_theta_convergence_strong_diminishing_adaptation}.
However, if we directly use the estimate sequence as the allocation parameter sequence, Assumption \ref{assumption_theta_weak_diminishing_adaptation} may not hold.
    Therefore, a suitable parameter update mechanism is still required, as discussed in Subsection \ref{subsec_parameter_update}.
\end{remark}

Under the stronger assumption, Assumption \ref{assumption_theta_convergence_strong_diminishing_adaptation}, the following central limit theorem can be established.
\begin{theorem}
    \label{theorem_estimation_asymptotic_normality}

    Suppose that Assumptions \ref{assumption_theta_compact}, \ref{assumption_lipschitz_continuity_functions_kernels}, \ref{assumption_x_sub_exponential_bound}, \ref{assumption_for_simultaneous_small_set_condition} and \ref{assumption_estimation_asymptotic} hold.
    If the allocation parameter sequence $\{\theta_n\}$ satisfies Assumption \ref{assumption_theta_convergence_strong_diminishing_adaptation}, and for each $t \in \{0, 1\}$, there exists some $\epsilon > 0$ such that
    \begin{equation*}
        \mathbb{E}_{(X,Y(t)) \sim \Gamma_{X,Y(t)}} \left\| \dot{m}_{\eta^*}(X,Y(t),t) \right\|^{4+\epsilon}
        < \infty ,
    \end{equation*}
    then
    \begin{equation*}
        \sqrt{n} (\eta_n - \eta^*)
        = - \ddot{M}_{\eta^*}^{-1} \frac{1}{\sqrt{n}} \dot{M}_{\eta^*, n}
        + o_P(1) .
    \end{equation*}
    Here,
    \begin{equation*}
        \dot{M}_{\eta^*, n}
        = \sum_{i=1}^n
        \left[ \frac{\rho^{\mathrm{ref}}(T_i \mid X_i)}{\rho_{\theta_{i-1}}(T_i \mid X_i)}
        \dot{m}_{\eta^*}(X_i, Y_i(T_i), T_i) \right]
    \end{equation*}
    is asymptotically normal with mean zero and covariance matrix
    \begin{align*}
        \Sigma_{(Z)}
        &= \Cov(Z^{(u)}) + \mathbb{E}_{(X, Z) \sim \Gamma_{X,Z}} \left[ \rho_{\theta^*}(X)(1-\rho_{\theta^*}(X)) \right. \\
        & \quad \left. \left\{ Z^{(c)} - \frac{A \phi(X)}{\rho_{\theta^*}(X)(1-\rho_{\theta^*}(X))} \right\}
        \left\{ Z^{(c)} - \frac{A \phi(X)}{\rho_{\theta^*}(X)(1-\rho_{\theta^*}(X))} \right\}^T \right] ,
    \end{align*}
    where
    \begin{align*}
        Z^{(c)}
        &= \frac{\rho^{\mathrm{ref}}(X)}{\rho_{\theta^*}(X)}
        \ddot{M}_{\eta^*}^{-1} \dot{m}_{\eta^*}(X, Y(1), 1)
        -
        \frac{1-\rho^{\mathrm{ref}}(X)}{1-\rho_{\theta^*}(X)}
        \ddot{M}_{\eta^*}^{-1} \dot{m}_{\eta^*}(X, Y(0), 0) , \\
        Z^{(u)}
        &= 
        \rho^{\mathrm{ref}}(X) \ddot{M}_{\eta^*}^{-1} \dot{m}_{\eta^*}(X, Y(1), 1)
        + (1-\rho^{\mathrm{ref}}(X)) \ddot{M}_{\eta^*}^{-1} \dot{m}_{\eta^*}(X, Y(0), 0) ,
    \end{align*}
    and the matrix $A$ satisfies
    \begin{align}
        & A \mathbb{E} \left[ [\rho_{\theta^*}(X)(1-\rho_{\theta^*}(X))]^{-1} \phi(X) \phi(X)^T
        / \max\{ \|\phi(X)\|/[\rho_{\theta^*}(X)(1-\rho_{\theta^*}(X))], C_{\theta^*} \} \right] \notag \\
        & \quad = \mathbb{E} \left[ Z^{(c)} \phi(X)^T
        / \max\{ \|\phi(X)\|/[\rho_{\theta^*}(X)(1-\rho_{\theta^*}(X))], C_{\theta^*} \} \right] 
        \label{eq_theorem_estimation_asymptotic_normality_matrix_A} .
    \end{align}
\end{theorem}

\begin{remark}
    The CBARA procedure affects only the second term in the asymptotic covariance matrix $\Sigma_{(Z)}$, through the targeted allocation ratio $\rho_{\theta^*}$ and the linear adjustment matrix $A$ arising from covariate balance.
    In particular, when $A = 0$, the asymptotic covariance matrix reduces to that of the estimator under simple randomization with the oracle targeted allocation ratio $\rho_{\theta^*}$.
    Furthermore, the optimal linear adjustment matrix can be derived as follows.
    For any direction vector $x$, the derivative of the quadratic form $x^T \Sigma_{(Z)} x$ for the covariance matrix $\Sigma_{(Z)}$ with respect to the matrix $A$ is given by
    \begin{equation}
        x^T \mathbb{E}_{(X, Z) \sim \Gamma_{X,Z}} \left[
        \left\{ Z^{(c)} - \frac{A \phi(X)}{\rho_{\theta^*}(X)(1-\rho_{\theta^*}(X))} \right\} \phi(X)^T \right] x
        \label{eq_derivative_covariance_matrix} .
    \end{equation}
    Therefore, the optimal linear adjustment matrix $A$ must satisfy that the derivative in \eqref{eq_derivative_covariance_matrix} is equal to zero for any $x$, namely
    \begin{equation}
        A \mathbb{E} \left[ [\rho_{\theta^*}(X)(1-\rho_{\theta^*}(X))]^{-1} \phi(X) \phi(X)^T \right]
        = \mathbb{E} \left[ Z^{(c)} \phi(X)^T \right]
        \label{eq_derivative_covariance_matrix_optimal} .
    \end{equation}
    If the normalization constant $C_{\theta^*} \geq \|\phi(X)\|/[\rho_{\theta^*}(X)(1-\rho_{\theta^*}(X))]$ almost surely, then the condition \eqref{eq_theorem_estimation_asymptotic_normality_matrix_A} reduces to \eqref{eq_derivative_covariance_matrix_optimal}.
    Therefore, we conclude that, under the CBARA procedure, the optimal linearly adjusted covariance matrix is achievable, and this covariance matrix is guaranteed to be no larger than that under simple randomization with the oracle targeted allocation ratio.
\end{remark}

\subsection{Allocation Parameter Update Mechanism}

\label{subsec_parameter_update}

The theorems in this subsection provide a rigorous justification of the two central assumptions introduced in Subsection \ref{subsubsec_introduction_parameter_update}, namely Assumptions \ref{assumption_theta_weak_diminishing_adaptation} and \ref{assumption_theta_convergence_strong_diminishing_adaptation}.
Theorem \ref{theorem_allocation_parameter_update_weak} is used to verify Assumption \ref{assumption_theta_weak_diminishing_adaptation}, while Theorem \ref{theorem_allocation_parameter_update_strong} is used to establish Assumption \ref{assumption_theta_convergence_strong_diminishing_adaptation}.
\begin{theorem}
    \label{theorem_allocation_parameter_update_weak}

    If the allocation parameter $\theta_n$ is updated according to \eqref{eq_allocation_parameter_update_mechanism_1} under Assumption \ref{assumption_theta_compact}, or according to \eqref{eq_allocation_parameter_update_mechanism_2} when the parameter space $\Theta$ is a convex subset of a Euclidean space, then Assumption \ref{assumption_theta_weak_diminishing_adaptation} holds for any $q \in (0,1]$.
\end{theorem}

\begin{theorem}
    \label{theorem_allocation_parameter_update_strong}
    Suppose that Assumption \ref{assumption_theta_compact} holds, the estimate sequence $\{\eta_n\}$ converges to $\eta^*$ in probability and there exists some $p \in (0,1/2)$ such that
    \begin{equation*}
        \sum_{n=0}^{N-1} \| \eta_n - \eta_{n+1} \|
        = o_P(N^p) .
    \end{equation*}
    If the allocation parameter $\theta_n$ is updated according to \eqref{eq_allocation_parameter_update_mechanism_1}, or according to \eqref{eq_allocation_parameter_update_mechanism_2} when the parameter space $\Theta$ is a convex subset of a Euclidean space and $\sum_{n=1}^\infty C_{\mathrm{clip}, n} = \infty$ for \eqref{eq_allocation_parameter_update_mechanism_2}, then Assumption \ref{assumption_theta_convergence_strong_diminishing_adaptation} is satisfied with $\theta^* = \eta^*$ and the same exponent $p$.
\end{theorem}

The assumptions of Theorem \ref{theorem_allocation_parameter_update_strong} on the estimate sequence $\{\eta_n\}$ are satisfied if the allocation parameter sequence $\{\theta_n\}$ fulfills Assumption \ref{assumption_theta_weak_diminishing_adaptation} with $q > 1/2$, as guaranteed by Theorems \ref{theorem_estimation_consistency} and \ref{theorem_estimation_difference}.
Thus, with Theorems \ref{theorem_allocation_parameter_update_weak} and \ref{theorem_allocation_parameter_update_strong} established, we complete the development of the theory in this section.
The resulting logical dependencies among the assumptions on the allocation parameter sequence and the theorems are summarized in Figure \ref{figure_theorem_dependency}.

\section{Proof Strategy}
\label{sec_proof_strategy}

In this section, we outline the proof strategy.
A key theoretical contribution of this article lies in the new discrepancy for transition kernels introduced in Subsection \ref{subsec_proof_strategy_negligibility_remaining_term}, which is well aligned with the setting considered here.
Subsections \ref{subsec_connection} and \ref{subsec_proof_strategy_LLN_CLT} present the theoretical background of the CBARA procedure.
In particular, the procedure can be described via transition kernels, which allows us to employ tools from Markov chain theory.
The discrepancy introduced in Subsection \ref{subsec_proof_strategy_negligibility_remaining_term} is then used to handle the disturbance caused by the variation of $\{\theta_n\}$.
Subsection \ref{subsec_proof_strategy_boundedness} outlines the approach for establishing that $\Lambda = O_P(1)$, with some technical details omitted.
Subsection \ref{subsec_explicit_form} provides several properties of the allocation function design, which are used to derive the explicit form of the variance.

\subsection{Connection of the CBARA Procedure and Transition Kernels}

\label{subsec_connection}

The transition kernel is commonly used to characterize the dynamics of a Markov chain.
Although the CBARA procedure considered in this article is not Markovian in general, it can become Markovian when the allocation parameter $\theta_n$ is forcibly fixed at a constant value.
For each $\theta \in \Theta$, define a transition kernel $P_\theta$ on the state space $\mathrm{X}$ by
\begin{equation*}
    P_\theta(\Lambda,h) = \int \left[ g_\theta(\Lambda, X)h(\Lambda+\phi(X)/\rho_\theta(X)) + [1-g_\theta(\Lambda, X)]h(\Lambda-\phi(X)/(1-\rho_\theta(X))) \right] \Gamma( \mathrm{d} X) ,
\end{equation*}
for any integrable function $h$.
Then under the forcibly fixed parameter sequence $\{\theta_n = \theta\}_{n \in \mathbb{N}}$, the stochastic process $\{\Lambda_n\}_{n \in \mathbb{N}}$ becomes a Markov chain with transition kernel $P_\theta$.
For the original CBARA procedure, the following lemma establishes its connection with the family of transition kernels $\{P_\theta\}_{\theta \in \Theta}$.
\begin{lemma}
    \label{lemma_transition_kernel_main}
    
    Let the function $h: \mathrm{X} \to \mathbb{R}$ be any integrable function.
    Then, under the CBARA procedure, for any $n \in \mathbb{N}$,
    \begin{equation*}
        \mathbb{E} \left[ h(\Lambda_{n+1}) \mid \mathcal{F}_n \right]
        = P_{\theta_n}(\Lambda_n,h) .
    \end{equation*}
\end{lemma}
In this article, we refer to a stochastic process $\{((\Lambda_n, \theta_n),\mathcal{F}_n)\}_{n \in \mathbb{N}}$ that satisfies the conclusions of Lemma \ref{lemma_transition_kernel_main} for any integrable function $h$ as a pseudo-Markov chain.
Although $\{\Lambda_n\}$ is not a Markov chain, we can still employ tools from Markov chain theory, such as the Poisson equation, as outlined in Subsection \ref{subsec_proof_strategy_LLN_CLT}.

\subsection{Proof Strategy of Theorem \ref{theorem_boundedness}}

\label{subsec_proof_strategy_boundedness}

Theorem \ref{theorem_boundedness} relies on the inequality
\begin{equation*}
    \mathbb{E} \left[ e^{\lambda_1 \|\Lambda_{n+1}\|} \mid \mathcal{F}_n \right] \leq \beta e^{\lambda_1 \|\Lambda_n\|} + b
\end{equation*}
for some positive constants $\beta<1$, $b$ and $\lambda_1$.
This inequality can be obtained by showing that under any fixed parameter sequence $\{\theta_n = \theta\}_{n \in \mathbb{N}}$, $\mathbb{E}_\theta \left[ e^{\lambda_1 (\|\Lambda_{n+1}\| - \|\Lambda_n\|)} \mid \mathcal{F}_n \right]$ is bounded by a constant less than $1$ when $\lambda_1$ is sufficiently small and $\|\Lambda_n\|$ is sufficiently large, with these bounds chosen uniformly over $\theta$.
To achieve this, we linearize both the exponential function and the norm function $x \mapsto |x|$ to prove
\begin{equation*}
    \mathbb{E}_\theta \left[ e^{\lambda_1 (\|\Lambda_{n+1}\| - \|\Lambda_n\|)} \right]
    \approx
    1 + \lambda_1 \mathbb{E}_\theta \left[ \|\Lambda_{n+1}\| - \|\Lambda_n\| \right]
    \approx 1 + \lambda_1 \mathbb{E}_\theta \left[ (\Lambda_{n+1} - \Lambda_n)^T \frac{\Lambda_n}{\|\Lambda_n\|} \right] ,
\end{equation*}
and establish that for any $\Lambda \neq 0$,
\begin{equation*}
    \mathbb{E}_\theta \left[ (\Lambda_{n+1} - \Lambda_n)^T \frac{\Lambda_n}{\|\Lambda_n\|} \mid \Lambda_n = \Lambda \right]
    = \mathbb{E}_{X \sim \Gamma} \left[
    \frac{g_\theta(\Lambda,X) - \rho_\theta(X)}{\rho_\theta(X)(1-\rho_\theta(X))} \cdot
    \frac{\phi(X)^T \Lambda}{\|\Lambda\|} \right]
    \leq -\Delta .
\end{equation*}

\subsection{Proof Strategy of Theorems \ref{theorem_additional_covariate_LLN}--\ref{theorem_estimation_asymptotic_normality}}
\label{subsec_proof_strategy_LLN_CLT}

Theorems \ref{theorem_additional_covariate_LLN}-\ref{theorem_estimation_asymptotic_normality} rely on the asymptotic behavior of the sum
\begin{equation}
    \sum_{n=1}^N \text{Term}_n
    = \underbrace{ \sum_{n=1}^N \left[ \text{Term}_n - \mathbb{E} \left[ \text{Term}_n \mid \mathcal{F}_{n-1} \right] \right] }_{\text{martingale term}}
    + \underbrace{ \sum_{n=1}^N \mathbb{E} \left[ \text{Term}_n \mid \mathcal{F}_{n-1} \right] }_{\text{dependent term}} 
    \label{eq_decomposition_martingale_dependent} ,
\end{equation}
where $\text{Term}_n$ denotes a generic term in the summation, and the conditional expectation $\mathbb{E} \left[ \text{Term}_n \mid \mathcal{F}_{n-1} \right]$ can be expressed in terms of $\theta_{n-1}$, $\Lambda_{n-1}$.
An example of such a generic term $\text{Term}_n$ is the summand in the definition of $\Psi_n$ in \eqref{eq_definition_imbalance_vector_additional},
\begin{equation*}
    \text{Term}_n = \frac{(T_n-\rho_{\theta_{n-1}}(X_n)) Z_n}{\rho_{\theta_{n-1}}(X_n)(1-\rho_{\theta_{n-1}}(X_n))} .
\end{equation*}
Its conditional expectation satisfies
\begin{equation*}
    \mathbb{E} \left[ \text{Term}_n \mid \mathcal{F}_{n-1} \right]
    = \mathbb{E}_{(X, Z) \sim \Gamma_{X,Z}} \left[ \frac{(g_{\theta_{n-1}}(\Lambda_{n-1}, X)-\rho_{\theta_{n-1}}(X)) Z}{\rho_{\theta_{n-1}}(X)(1-\rho_{\theta_{n-1}}(X))} \right] ,
\end{equation*}
which is a function of $\theta_{n-1}$ and $\Lambda_{n-1}$.
Therefore, to prove Theorems \ref{theorem_additional_covariate_LLN} and \ref{theorem_additional_covariate_CLT}, we establish the law of large numbers and the central limit theorem for $\Psi_N = \sum_{n=1}^N \text{Term}_n$.

Denote the conditional expectation function
\begin{equation*}
    h_{\theta_{n-1}}(\Lambda_{n-1}) = \mathbb{E} \left[ \text{Term}_n \mid \mathcal{F}_{n-1} \right] .
\end{equation*}
Then the martingale term in \eqref{eq_decomposition_martingale_dependent} can be analyzed using martingale techniques, while the dependent term can be analyzed by leveraging the properties of the pseudo-Markov chain $\{\Lambda_n\}$ and the transition kernel $P_\theta$.
Specifically, under geometric ergodicity of the Markov chain with transition kernel $P_\theta$ and invariant probability $\pi_\theta$, the Poisson equation associated with the function $h_\theta$,
\begin{equation*}
    \hat{h}_\theta - P_\theta \hat{h}_\theta = h_\theta - \pi_\theta h_\theta ,
\end{equation*}
admits a solution given by
\begin{equation*}
    \hat{h}_\theta = \sum_{n=0}^\infty (P_\theta^n - \pi_\theta) (h_\theta) .
\end{equation*}

The Poisson equation involves the term $h_\theta$, which corresponds to the dependent term in \eqref{eq_decomposition_martingale_dependent}.
It also involves the term $\pi_\theta h_\theta$, which does not appear in \eqref{eq_decomposition_martingale_dependent}.
In the central limit theorem considered in this article, the corresponding function $h_\theta$ satisfies $\pi_\theta h_\theta \equiv 0$, so the term $\pi_\theta h_\theta$ in the above expression can be omitted.
In contrast, in the law of large numbers required in this article, the quantity $\pi_{\theta_n} h_{\theta_n}$ converges to $\pi_{\theta^*} h_{\theta^*}$ as $\theta_n \xrightarrow{\mathbb{P}} \theta^*$.
Hence, the term $\pi_\theta h_\theta$ can be treated as a convergent centering term and can be omitted.

Therefore, the centered dependent term in \eqref{eq_decomposition_martingale_dependent} can be rewritten as
\begin{align*}
    & \quad \sum_{n=1}^N \mathbb{E} \left[ \text{Term}_n \mid \mathcal{F}_{n-1} \right]
    - \sum_{n=1}^N \left[ \pi_{\theta_{n-1}} h_{\theta_{n-1}} \right]
    = \sum_{n=1}^N \left[ h_{\theta_{n-1}}(\Lambda_{n-1}) - \pi_{\theta_{n-1}} h_{\theta_{n-1}} \right] \\
    &= \sum_{n=1}^N \left[ \hat{h}_{\theta_{n-1}}(\Lambda_{n-1}) - P_{\theta_{n-1}} \hat{h}_{\theta_{n-1}}(\Lambda_{n-1}) \right]
    = \sum_{n=0}^{N-1} \left[ \hat{h}_{\theta_n}(\Lambda_n) - P_{\theta_n} \hat{h}_{\theta_n}(\Lambda_n) \right] \\
    &= \sum_{n=0}^{N-1} \left[ \hat{h}_{\theta_n}(\Lambda_{n+1}) - P_{\theta_n} \hat{h}_{\theta_n}(\Lambda_n) \right] + \sum_{n=0}^{N-1} \left[ \hat{h}_{\theta_n}(\Lambda_n) - \hat{h}_{\theta_n}(\Lambda_{n+1}) \right] \\
    &= \underbrace{ \sum_{n=0}^{N-1} \left[ \hat{h}_{\theta_n}(\Lambda_{n+1}) - P_{\theta_n} \hat{h}_{\theta_n}(\Lambda_n) \right] }_{\text{martingale term}} \\
    & \quad + \underbrace{ \sum_{n=0}^{N-1} \left[ \hat{h}_{\theta_{n+1}}(\Lambda_{n+1}) - \hat{h}_{\theta_n}(\Lambda_{n+1}) \right]
    + \left[ \hat{h}_{\theta_0}(\Lambda_0) - \hat{h}_{\theta_N}(\Lambda_N) \right] }_{\text{remaining term}} .
\end{align*}
The term $\hat{h}_{\theta_n}(\Lambda_{n+1}) - P_{\theta_n} \hat{h}_{\theta_n}(\Lambda_n)$ forms a martingale difference because, by Lemma \ref{lemma_transition_kernel_main}, the conditional distribution of $\Lambda_{n+1}$ given $\mathcal{F}_n$ is $P_{\theta_n}(\Lambda_n, \cdot)$.
This term can therefore be analyzed using martingale techniques.
Therefore, the law of large numbers and the central limit theorem for the summation in \eqref{eq_decomposition_martingale_dependent} can be established by analyzing the two martingale terms and verifying that the remaining term is negligible.
As a remark, we note that the degree of negligibility required differs for the law of large numbers and the central limit theorem.

\subsection{Negligibility of the Remaining Term}

\label{subsec_proof_strategy_negligibility_remaining_term}

Let the distance on the parameter space be defined by $d(\theta, \theta^\prime) := \|\theta - \theta^\prime\|$.
With Assumptions \ref{assumption_theta_weak_diminishing_adaptation} or \ref{assumption_theta_convergence_strong_diminishing_adaptation}, the negligibility of the remaining term follows from the bound on $\hat{h}_{\theta_{n+1}} - \hat{h}_{\theta_n}$ in terms of $d(\theta_n, \theta_{n+1})$ together with the boundedness in probability of $\Lambda_n$.
The two parts of the remaining term introduced in Subsection \ref{subsec_proof_strategy_LLN_CLT}, $\sum_{n=0}^{N-1} \left[ \hat{h}_{\theta_{n+1}}(\Lambda_{n+1}) - \hat{h}_{\theta_n}(\Lambda_{n+1}) \right]$ and $\left[ \hat{h}_{\theta_0}(\Lambda_0) - \hat{h}_{\theta_N}(\Lambda_N) \right]$, have different levels of analytical difficulty.
The latter term can be handled straightforwardly using the bound of $\hat{h}_\theta$ together with the boundedness in probability of $\Lambda_n$.
However, the former term should be analyzed in the following way.

Note that $\hat{h}_\theta = \sum_{n=0}^\infty (P_\theta^n - \pi_\theta) (h_\theta)$, and the difference $\hat{h}_\theta(\Lambda) - \hat{h}_{\theta^\prime}(\Lambda)$ for different parameters $\theta$ and $\theta^\prime$ can be expressed as
\begin{equation}
    \hat{h}_\theta(\Lambda) - \hat{h}_{\theta^\prime}(\Lambda)
    = \sum_{n=0}^\infty \left[ (P_\theta^n h_\theta)(\Lambda)
    - (P_{\theta^\prime}^n h_{\theta^\prime})(\Lambda)
    - \pi_\theta h_\theta + \pi_{\theta^\prime} h_{\theta^\prime} \right] \label{eq_difference_Poisson_equation_summation} .
\end{equation}
To control each term in terms of the distance $d(\theta, \theta^\prime)$, we need to relate a discrepancy between the $n$-step transition kernels $P_\theta^n$ and $P_{\theta^\prime}^n$ to a discrepancy between the one-step kernels $P_\theta$ and $P_{\theta^\prime}$, and then relate the latter to $d(\theta, \theta^\prime)$.

\subsubsection{Lipschitz Continuity of the Family $\{P_\theta\}_{\theta \in \Theta}$ with Respect to $\theta$}

To characterize the continuity of the transition kernels $\{P_\theta\}_{\theta \in \Theta}$, we first define a suitable discrepancy measure between kernels.
A standard choice is the $V$-norm, as used in the adaptive MCMC literature \cite{fortConvergenceAdaptiveInteracting2011,fortCentralLimitTheorem2014}.
The $V$-norm between the transition kernels $P_\theta$ and $P_{\theta^\prime}$ is measured by
\begin{equation*}
    \sup_{\Lambda} V^{-1}(\Lambda)\|P_\theta(\Lambda, \cdot) - P_{\theta^\prime}(\Lambda, \cdot)\|_V
\end{equation*}
where $V \geq 1$ is a Lyapunov function and $\|\mu\|_V := \sup_{|f| \leq V} |\mu(f)|$ defines a norm on the measure $\mu$.
This discrepancy is also adopted in the framework of the CAR procedure \cite{fangGeneralNonMarkovianFramework2026}.
However, this discrepancy may not be suitable for the CBARA procedure, as we cannot control the $V$-norm between the transition kernels $P_\theta$ and $P_{\theta^\prime}$ using $d(\theta, \theta^\prime)$.
The underlying reason why the $V$-norm is unsuitable is discussed in detail in Subsubsection \ref{subsubsec_limitations_V_norm}.

To address this issue, we develop an alternative way to characterize the discrepancy between the transition kernels $P_\theta$ and $P_{\theta^\prime}$ using $d(\theta, \theta^\prime)$.
\begin{definition}[Coupled Robust Lipschitz Continuity of Transition Kernels]
    \label{definition_coupled_robust_lipschitz_continuity_kernels_main}

    Let $P$ and $Q$ be transition probability kernels on $(\mathrm{X}, \mathcal{X})$. 
    We say that $P$ and $Q$ are \emph{$(L_P,\tau,\epsilon)$-coupled robustly Lipschitz continuous} if there exists a coupling kernel
    \begin{equation*}
        K: \mathrm{X}^2 \times \mathcal{X}^{\otimes 2} \to [0,1]
    \end{equation*}
    such that, for all $\Lambda, \Lambda^\prime \in \mathrm{X}$,
    \begin{align*}
        K(\Lambda, \Lambda^\prime; A \times \mathrm{X})
        & \leq P(\Lambda, A), 
        && A \in \mathcal{X}, \\
        K(\Lambda, \Lambda^\prime; \mathrm{X} \times B)
        & \leq Q(\Lambda^\prime, B), 
        && B \in \mathcal{X},
    \end{align*}
    and the following bounds hold:
    \begin{align*}
        0 \leq 1 - K(\Lambda, \Lambda^\prime; \mathrm{X} \times \mathrm{X})
        & \leq L_P d(\Lambda, \Lambda^\prime) + \tau , \\
        \int d(u,v) K(\Lambda, \Lambda^\prime; \mathrm{d} u \times \mathrm{d} v)
        &\leq d(\Lambda, \Lambda^\prime) + \epsilon .
    \end{align*}
\end{definition}
Definition \ref{definition_coupled_robust_lipschitz_continuity_kernels_main} shows that, for any $\Lambda$ and $\Lambda^\prime$, and allowing for a mass deficiency of size $\epsilon$, the generalized Wasserstein distance between the measures $P(\Lambda, \cdot)$ and $Q(\Lambda^\prime, \cdot)$ is upper bounded by $L_P d(\Lambda, \Lambda^\prime) + \tau$.
In particular, if $\Lambda = \Lambda^\prime$, then the generalized Wasserstein distance between $P(\Lambda, \cdot)$ and $Q(\Lambda, \cdot)$ is bounded by $\tau$.
This implies that when the mass deficiency $\epsilon$ and the bound $\tau$ are zero, we can conclude that $P$ and $Q$ coincide.
Therefore, $(\tau, \epsilon)$ can be regarded as measuring the discrepancy between the transition kernels $P$ and $Q$.

Under Assumption \ref{assumption_lipschitz_continuity_functions_kernels}, we establish the following lemma characterizing the coupled robust Lipschitz continuity of the transition kernels $P_\theta$ and $P_{\theta^\prime}$ in terms of $d(\theta, \theta^\prime)$.
\begin{lemma}
    \label{lemma_continuity_transition_kernels_main}

    Suppose that Assumption \ref{assumption_lipschitz_continuity_functions_kernels} holds and that $\phi(X)$ has a finite second moment.
    Then the family of transition probability kernels $\{P_\theta\}_{\theta \in \Theta}$ is robustly Lipschitz continuous with a Lipschitz constant $L_P \geq 0$, that is, for any parameters $\theta$, $\theta^\prime \in \Theta$, the kernels $P_\theta$ and $P_{\theta^\prime}$ are $(L_P, L_P d(\theta, \theta^\prime), L_P d(\theta, \theta^\prime))$-coupled robustly Lipschitz continuous.
\end{lemma}
The discrepancy result in Lemma \ref{lemma_continuity_transition_kernels_main} can be interpreted as the Lipschitz continuity of the family $\{P_\theta\}_{\theta \in \Theta}$ with respect to $\theta$ in the sense of Definition \ref{definition_coupled_robust_lipschitz_continuity_kernels_main}, and the non-expansivity of the transition kernels $P_\theta$ with respect to the state variable $\Lambda$.
This result can be extended to the family of $n$-step transition kernels $\{P_\theta^n\}_{\theta \in \Theta}$ via the following corollary.
\begin{corollary}
    \label{corollary_n_step_family_robustly_Lipschitz_continuous_main}
    Let $\{P_\theta\}_{\theta \in \Theta}$ be a family of transition probability kernels that is robustly Lipschitz continuous with a Lipschitz constant $L_P \geq 0$.
    Then for any $\theta,\theta^\prime\in\Theta$ and any $n \in \mathbb{N}$, the probability kernels $P_\theta^n$ and $P_{\theta^\prime}^n$ are $(n L_P, n L_P d(\theta,\theta^\prime) + \frac{n(n-1) L_P^2}{2} d(\theta,\theta^\prime), n L_P d(\theta,\theta^\prime))$-coupled robustly Lipschitz continuous.
\end{corollary}
Therefore, when the family $\{P_\theta\}_{\theta \in \Theta}$ is robustly Lipschitz continuous, the family $\{P_\theta^n\}_{\theta \in \Theta}$ is also Lipschitz continuous, with a Lipschitz constant inflated at a rate of $n^2$.

\subsubsection{Limitations of the $V$-Norm for the CBARA Procedure}

\label{subsubsec_limitations_V_norm}

To understand why the $V$-norm works in the context of the CAR procedure but not in the CBARA procedure, consider how the transition kernel $P_\theta$ changes with respect to $\theta$ under the CAR and CBARA procedures.
More specifically, we examine the change in the imbalance vector $\Lambda$ when a unit with covariate value $x$ is assigned to the treatment group.

Under the CAR procedure, when the imbalance vector $\Lambda$ encounters a unit with covariate $x$ and this unit is assigned to the treatment group, the transition corresponds to moving from $\Lambda$ to $\Lambda + (1-\rho)\phi(x)$ with probability $P(X=x) g_\theta(\Lambda,x)$.
In this case, changing $\theta$ only affects the allocation probability $g_\theta(\Lambda,x)$, while the destination state remains unchanged.
Thus, the $V$-norm discrepancy between the corresponding transition kernels can be controlled by the difference between the allocation probabilities, and the $V$-norm is suitable for the CAR procedure.

In contrast, under the CBARA procedure, when the imbalance vector $\Lambda$ encounters a unit with covariate $x$ and this unit is assigned to the treatment group, the transition corresponds to moving from $\Lambda$ to $\Lambda + \phi(x)/\rho_\theta(x)$ with probability $P(X=x) g_\theta(\Lambda,x)$.
Consequently, changing $\theta$ not only alters the probability but also changes the destination state itself.
Therefore, when the $V$-norm is used to measure the discrepancy between the transition kernels, the change in the destination state induced by $\theta$ under the CBARA procedure can lead to a non-negligible discrepancy between $P_\theta(\Lambda,\cdot)$ and $P_{\theta^\prime}(\Lambda,\cdot)$.
In particular, even when $\theta$ and $\theta^\prime$ are close, the corresponding transitions may place probability mass on different states, namely $\Lambda + \phi(x)/\rho_\theta(x)$ and $\Lambda + \phi(x)/\rho_{\theta^\prime}(x)$. 
As a result, the quantity $\|P_\theta(\Lambda,\cdot)-P_{\theta^\prime}(\Lambda,\cdot)\|_V$ may remain bounded away from zero, since the two measures are supported on different locations in the state space.
Unfortunately, the $V$-norm is insensitive to the geometry of the state space and may assign maximal discrepancy to probability measures with nearby but non-identical supports.
These observations suggest that the $V$-norm used in the adaptive MCMC literature is suitable for the CAR procedure but may not be well suited for analyzing the CBARA procedure.

\subsubsection{Results for the New Discrepancy Measure}

Let the family of functions $\{h_\theta\}_{\theta \in \Theta}$ possess certain $\alpha$-Hölder continuity for any $\alpha \in (0,1)$ within any sufficiently small region in both $\Lambda$ and $\theta$.
With continuity properties on $\{h_\theta\}_{\theta \in \Theta}$ and $\{P_\theta\}_{\theta \in \Theta}$, together with some properties of $P_\theta$, it can be shown that the difference $\left| (P_\theta^n h_\theta)(\Lambda) - (P_{\theta^\prime}^n h_{\theta^\prime})(\Lambda) \right|$ can be controlled by a multiple of $(n^2 + 1) \left[ d(\theta, \theta^\prime) \right]^\alpha$ for any $\alpha \in (0,1)$ and any $\Lambda$, when $d(\theta,\theta^\prime)$ is bounded by a positive constant.

For the $\pi_\theta h_\theta - \pi_{\theta^\prime} h_{\theta^\prime}$ term in \eqref{eq_difference_Poisson_equation_summation}, by using the result above and the geometric ergodicity $\left| (P_\theta^n h_\theta)(\Lambda) - \pi_\theta h_\theta \right| = O(\rho^{-n} \exp(\lambda^\prime \|\Lambda\|))$, we can also establish a similar bound by decomposing that $\pi_\theta h_\theta - \pi_{\theta^\prime} h_{\theta^\prime}$ into
\begin{equation*}
    \left[ \pi_\theta h_\theta - P_\theta^n h_\theta \right]
    + \left[ P_\theta^n h_\theta - P_{\theta^\prime}^n h_{\theta^\prime} \right]
    + \left[ P_{\theta^\prime}^n h_{\theta^\prime} - \pi_{\theta^\prime} h_{\theta^\prime} \right] ,
\end{equation*}
and choosing an appropriate $n$ to balance the bounds for each term.
Therefore, it holds that $\pi_\theta h_\theta - \pi_{\theta^\prime} h_{\theta^\prime} = O(\left[ d(\theta,\theta^\prime) \right]^\alpha)$ for any $\alpha \in (0,1)$, when $d(\theta,\theta^\prime)$ is bounded by a positive constant.

In conclusion, each term in \eqref{eq_difference_Poisson_equation_summation} can be controlled by a multiple of $(n^2 + 1) \left[ d(\theta, \theta^\prime) \right]^\alpha$ for any $\alpha \in (0,1)$, when $d(\theta,\theta^\prime)$ is bounded by a positive constant.
By the geometric ergodicity of the Markov chain with transition kernel $P_\theta$, we have $\left| (P_\theta^n h_\theta)(\Lambda) - \pi_\theta h_\theta \right| = O(\rho^{-n} \exp(\lambda^\prime \|\Lambda\|))$ for some $\rho \in (0,1)$ and $\lambda^\prime > 0$.
Therefore, we use the polynomial bound $O(n^2 \left[ d(\theta, \theta^\prime) \right]^\alpha)$ to control the first finite terms of the series in \eqref{eq_difference_Poisson_equation_summation}, and use the exponential bound $O(\rho^{-n} \exp(\lambda^\prime \|\Lambda\|))$ to control the tail of the series in \eqref{eq_difference_Poisson_equation_summation}.
It results in a theoretical bound for $\hat{h}_\theta(\Lambda) - \hat{h}_{\theta^\prime}(\Lambda)$, which is $O(\left[ d(\theta, \theta^\prime) \right]^\alpha)$ for any $\alpha \in (0,1)$, when $d(\theta,\theta^\prime)$ is bounded by a positive constant.

\subsection{Explicit Form of Variance}

\label{subsec_explicit_form}

The previous discussion explains how the remaining term can be eliminated.
Therefore, the asymptotic behavior of \eqref{eq_decomposition_martingale_dependent} is determined by the martingale terms
\begin{equation*}
    \sum_{n=1}^N \left[ \text{Term}_n - \mathbb{E} \left[ \text{Term}_n \mid \mathcal{F}_{n-1} \right] \right]
    + \sum_{n=0}^{N-1} \left[ \hat{h}_{\theta_n}(\Lambda_{n+1}) - P_{\theta_n} \hat{h}_{\theta_n}(\Lambda_n) \right] .
\end{equation*}
By applying the martingale central limit theorem and establishing a law of large numbers for the conditional variance, using the same arguments as those used to eliminate the remaining term in Subsection \ref{subsec_proof_strategy_negligibility_remaining_term}, we obtain that the asymptotic variance coincides with the variance associated with the limiting parameter.
Therefore, it suffices to analyze the asymptotic variance when $\theta$ is fixed at the limiting parameter.

For central limit theorems considered in this article, the corresponding function $h_\theta$ takes the form
\begin{equation}
    h_\theta(\Lambda)
    = \mathbb{E}_{X \sim \Gamma} \left[ \left(g_\theta(\Lambda, X) - \rho_\theta(X)\right) f_\theta(X) \right]
    \label{eq_form_special_h_theta} ,
\end{equation}
for some function $f_\theta$ depending on the specific setting.
It can be shown that the associated solution $\hat{h}_\theta$ to the Poisson equation admits an explicit expression, up to an unknown additive constant.
Therefore, the asymptotic variance can be computed explicitly.
It is worth noting that not all functions $h_\theta$ used in the proof admit the special form in \eqref{eq_form_special_h_theta}.
For instance, when applying the martingale central limit theorem, we need to establish a law of large numbers for the conditional variance, and the corresponding function $h_\theta$ associated with the conditional variance does not take the form in \eqref{eq_form_special_h_theta}.
It is therefore necessary to use the analysis developed in Subsection \ref{subsec_proof_strategy_negligibility_remaining_term} to establish the continuity of the Poisson equation solution $\hat{h}_\theta$, and thereby prove the negligibility of the remaining term.

\section{Discussion}
\label{sec_discussion}

In this article, we propose the CBARA procedure.
It combines the advantages of the CARA and CAR procedures, allowing for adjustment of the targeted allocation ratio while simultaneously balancing covariates.
We show that it improves both the balance of additional covariates and the performance of weighted M-estimation.

\paragraph*{Valid Inference}
Building on these theoretical properties, we can further conduct statistical inference.
For the IPW estimator
\begin{equation*}
    \hat{\tau}_{\mathrm{IPW}}
    = \frac{1}{N} \sum_{n=1}^N \left[ \frac{T_n Y_n(1)}{\rho_{\theta_{n-1}}(X_n)}
    - \frac{(1-T_n) Y_n(0)}{1-\rho_{\theta_{n-1}}(X_n)} \right]
\end{equation*}
we can establish the following asymptotic properties.
\begin{theorem}[IPW Estimator]
    \label{theorem_IPW_estimator}

    Suppose that Assumptions \ref{assumption_theta_compact}, \ref{assumption_lipschitz_continuity_functions_kernels}, \ref{assumption_x_sub_exponential_bound} and \ref{assumption_for_simultaneous_small_set_condition} hold.
    If the allocation parameter sequence $\{\theta_n\}$ satisfies Assumption \ref{assumption_theta_convergence_strong_diminishing_adaptation} and there exists $\epsilon > 0$ such that $\mathbb{E} \left[ Y(t)^{4+\epsilon} \right] < \infty$ for any $t \in \{0,1\}$, then the IPW estimator is asymptotically normal with mean zero and variance
    \begin{equation*}
        \var(Y(1) - Y(0))
        + \mathbb{E} \left[ \rho_{\theta^*}(X)(1-\rho_{\theta^*}(X))
        \left\{ \frac{Y(1)}{\rho_{\theta^*}(X)} + \frac{Y(0)}{1 - \rho_{\theta^*}(X)} - \frac{a^T \phi(X)}{\rho_{\theta^*}(X)(1-\rho_{\theta^*}(X))} \right\}^2 \right] ,
    \end{equation*}
    where the vector $a$ satisfies
    \begin{align*}
        & a^T \mathbb{E} \left[ [\rho_{\theta^*}(X)(1-\rho_{\theta^*}(X))]^{-1} \phi(X) \phi(X)^T
        / \max\{ \|\phi(X)\|/[\rho_{\theta^*}(X)(1-\rho_{\theta^*}(X))], C_{\theta^*} \} \right] \\
        & \quad = \mathbb{E} \left[ \left[ \frac{Y(1)}{\rho_{\theta^*}(X)} + \frac{Y(0)}{1 - \rho_{\theta^*}(X)} \right] \phi(X)^T
        / \max\{ \|\phi(X)\|/[\rho_{\theta^*}(X)(1-\rho_{\theta^*}(X))], C_{\theta^*} \} \right] .
    \end{align*}
\end{theorem}
Since the theorem provides an explicit expression for the variance, it is feasible to conduct inference by constructing a consistent estimator of the variance.
Moreover, inference can also be conducted using other estimators.

\paragraph*{Estimation Using Machine Learning}
In this article, the targeted allocation ratio is updated by estimating a finite-dimensional parameter via weighted M-estimation.
In practice, a promising extension is to incorporate machine learning methods to learn an optimal targeted allocation ratio in a data-driven manner, potentially improving flexibility and performance in complex settings.

\paragraph*{Multiple Treatment Arms}
The current CBARA procedure focuses on the case of two treatment groups.
It can be naturally extended to settings with multiple treatment arms, where the allocation mechanism and the corresponding theoretical analysis can also be adapted to the multi-arm setting.

\newpage

\appendix

\section{Structure of the Proofs}

\subsection{Notation}

Given a space $\mathrm{X}$ and a function $V : \mathrm{X} \to [1,+\infty)$, define the $V$-norm of a function $f : \mathrm{X} \to \mathbb{R}$ by
\begin{equation*}
    |f|_V := \sup_{x \in \mathrm{X}} \frac{|f|(x)}{V(x)}
\end{equation*}
When $V=1$, the $V$-norm is the supremum norm denoted by $|f|_{\infty}$.
For a measure $\mu$ on $\mathrm{X}$, we define the $L^2(\mu)$ norm of a measurable function $f: \mathrm{X} \to \mathbb{R}$ by
\begin{equation*}
    \left\| f \right\|_{L^2(\mu)}
    = \left( \int |f|^2 d\mu \right)^{1/2}.
\end{equation*}

Denote the Lebesgue measure by $\mu_{\mathrm{leb}}$.
For a measure $\mu$, let $\mu_A = \mu|_A$ denote the restriction of $\mu$ to $A$, that is, $\mu_A(B) = (\mu |_A)(B) = \mu(A \cap B)$ for any measurable set $B$.
Let $\mu f = \int f(\Lambda) \mu(d\Lambda)$ and $\mu \left[ g(\cdot, x) \right] = \int g(\Lambda,x) \mu(d\Lambda)$.
For $\mu$ a finite signed measure on the measurable space $(\mathrm{X}, \mathcal{X})$ and $V: \mathrm{X} \to [1, \infty)$ such that $|\mu|(V)<\infty$, where $|\mu|$ is the variation of $\mu$, we define $\|\mu\|_V$ the $V$-norm of $\mu$ as
\begin{equation*}
    \|\mu\|_V := \sup_{|f|_V \leq 1}|\mu(f)|
\end{equation*}
When $V \equiv 1$, the $V$-norm corresponds to the total variation norm.
Denote the state space by $\mathrm{X} = \mathbb{R}^d$.

Let $\mathbb{E}_\theta$ denote the expectation under the randomization procedure with a fixed parameter sequence $\{\theta_n = \theta\}_{n \in \mathbb{N}}$.
If it exists, let $\pi_\theta$ denote the invariant probability measure corresponding to the transition kernel $P_\theta$ defined in \eqref{eq_transition_kernel_CARA_procedure}.

Let $P$ be a finite signed kernel on $(\mathrm{X}, \mathcal{X})$ such that $|P(x,\cdot)|(V) < \infty$ for any $x \in \mathrm{X}$.
For any measurable $f:\mathrm{X} \to \mathbb{R}$, we write $(Pf)(x) = P(x,f) := \int f(y) P(x, \mathrm{d} y)$, which is the integral of $f$ with respect to the signed measure $P(x,\cdot)$.
Denote $(\mu P)(A) := \int \mu(\mathrm{d} x) P(x, A)$ and $P^n(x, A) := \int P(x, d y) P^{n-1}(y, A)$ for $n$-step transition kernel $P^n$.
The $V$-norm of $P$ is defined by $\|P\|_V := \sup_{x \in \mathrm{X}} V^{-1}(x)\|P(x, \cdot)\|_V$.

Let $B^{(d)}(x,r)$ denote the $d$-dimensional ball centered at $x$ with radius $r$.
When there is no ambiguity, we write $B(x,r)$ for simplicity.

We equip the spaces with the metrics
\begin{equation*}
    d(\Lambda, \Lambda^\prime) = \| \Lambda - \Lambda^\prime \|, \quad
    d(\theta, \theta^\prime) = \| \theta - \theta^\prime \|,
\end{equation*}
where $\|\cdot\|$ denotes the corresponding norm in each space.

\subsection{Structure of the Proofs}

The CBARA procedure considered in this article is not Markovian in general.
However, if the allocation parameter $\theta$ is fixed, the stochastic process $\{\Lambda_n\}$ becomes a Markov chain.
We denote the corresponding transition kernel by $P_\theta$, which depends on the allocation parameter $\theta$.
This simplified scenario is relevant to our CBARA procedure. 
As shown in Lemma \ref{lemma_transition_kernel}, for any $n \in \mathbb{N}$, the conditional expectation at the $(n+1)$th step of the CBARA procedure, given the history, can be expressed in terms of the transition kernel $P_{\theta_n}$ and the imbalance vector $\Lambda_n$.
In this article, we refer to a stochastic process that satisfies the conclusions of Lemma \ref{lemma_transition_kernel} in Subsection \ref{subsec_transition_kernel} as a pseudo-Markov chain.

The works \cite{fortConvergenceAdaptiveInteracting2011,fortCentralLimitTheorem2014} inspired our approach to analyzing a pseudo-Markov chain associated with the family of transition kernels $\{P_\theta\}$, by leveraging both the properties of $P_\theta$ and transition kernels.
The properties required for this analysis are those specified in Assumption \ref{assumption_simultaneous_properties}.
Lemma \ref{lemma_application_simultaneous_properties} summarizes all lemmas in Subsections \ref{subsec_continuity_transition_kernels} and \ref{subsec_technical_lemmas_CARA} and is used to verify Assumption \ref{assumption_simultaneous_properties}.
In Subsections \ref{subsec_definitions_lemmas_kernel}--\ref{subsec_definitions_assumptions_chains_functions}, building on Assumption \ref{assumption_simultaneous_properties}, we establish bounds and continuity properties of $P_\theta$ and certain related functions with respect to different values of $\theta$.
Building on this, Lemma \ref{lemma_allocation_convergence} in Subsection \ref{subsec_lemmas_LLN} establishes a weak law of large numbers (WLLN) in a straightforward manner.

The main text presents numerous results concerning asymptotic normality.
A key requirement for establishing asymptotic normality is that variations in the parameter do not affect the limiting distribution.
This asymptotic normality result is formalized in Lemma \ref{lemma_CLT} in Subsection \ref{subsec_central_limit_theorem}.
The proof of Lemma \ref{lemma_CLT} relies on bounds and continuity properties of certain functions associated with the central limit theorem.
By applying the theorems and corollaries in Subsections \ref{subsec_definitions_lemmas_kernel}--\ref{subsec_definitions_assumptions_chains_functions}, we establish these properties in Subsection \ref{subsec_Holder_continuity_functions}.
Moreover, the center and the variance of the asymptotic distribution presented in the main text correspond to the case when the parameter is fixed at its limiting value $\theta^*$, under which the procedure becomes Markovian.
The specific values of the center and variance under this fixed parameter require Lemmas \ref{lemma_expectation_pi_h} and \ref{lemma_expression_variance} in Section \ref{sec_lemmas_allocation_form}, and the proofs of these lemmas rely heavily on the form of the allocation function \eqref{eq_CARA_allocation_function}.

At this point, the majority of the technical work has been completed, and Section \ref{sec_lemmas_estimation} applies the above law of large numbers and the central limit theorem results to the model parameter estimation.
Consequently, the theorems presented in the main text are proved in Sections \ref{sec_properties_CARA_procedure} and \ref{sec_treatment_effect_estimation}.

\section{Detailed Properties of CBARA Procedure}
\label{sec_properties_CARA_procedure}

\subsection{Proof of Theorem \ref{theorem_boundedness}}

The proof of Theorem \ref{theorem_boundedness} is based on Lemmas \ref{lemma_drift_condition_inequality} and \ref{lemma_convergence_V_as}.
Lemma \ref{lemma_drift_condition_inequality} implies that
\begin{equation*}
    (P_\theta V^\alpha)(\Lambda) \leq \beta_\alpha V^\alpha(\Lambda) + b_\alpha
\end{equation*}
for any $\theta \in \Theta$ and $\alpha \in (0,1]$ with $V(\Lambda) = \exp(\lambda_1 \|\Lambda\|)$.
Then based on Lemma \ref{lemma_convergence_V_as}, we have $V(\Lambda_n) = O_P(1)$.
It implies that $\{\Lambda_n\}$ is bounded in probability.

\subsection{Proof of Theorem \ref{theorem_additional_covariate_LLN}}

By Lemma \ref{lemma_application_simultaneous_properties}, Assumption \ref{assumption_simultaneous_properties} holds.

Define
\begin{equation*}
    h_\theta(\Lambda)
    = \mathbb{E}_{X \sim \Gamma} \left[ \frac{(g_\theta(\Lambda, X) - \rho_\theta(X)) Z}{\rho_\theta(X)(1-\rho_\theta(X))} \right] .
\end{equation*}

Under Assumption \ref{assumption_lipschitz_continuity_functions_kernels} and $\mathbb{E} \left[ Z^2 \right] < \infty$, by Lemmas \ref{lemma_h_lipschitz_continuous_bounded} and \ref{lemma_lipschitz_continuity_conditional_expectation} with
\begin{equation*}
    Z^{(c)}_n = Z_n / [\rho_{\theta_{n-1}}(X_n)(1-\rho_{\theta_{n-1}}(X_n))] ,
\end{equation*}
we have the family of functions $\{h_\theta\}_{\theta \in \Theta}$ is $((1, L_{h, \kappa, \alpha} V^\kappa, \alpha), (1, \tilde{L}_{h, \tilde{\kappa}, \tilde{\alpha}} V^{\tilde{\kappa}}, \tilde{\alpha}))$-joint locally Hölder continuous, and bounded by $C_{h,\gamma} V^\gamma$ with corresponding positive constants.

By Lemma \ref{lemma_convergence_average_enter}, the limit holds that
\begin{equation*}
    \frac{1}{N} \sum_{n=0}^{N-1} h_{\theta_n}(\Lambda_n)
    -\frac{1}{N} \sum_{n=0}^{N-1} \pi_{\theta_n}h_{\theta_n}
    \xrightarrow{\mathbb{P}} 0 .
\end{equation*}
By Lemma \ref{lemma_expectation_pi_h}, each term, $\pi_\theta h_\theta$, is zero.
Thus,
\begin{equation}
    \frac{1}{N} \sum_{n=0}^{N-1} h_{\theta_n}(\Lambda_n)
    \xrightarrow{\mathbb{P}} 0
    \label{eq_proof_theorem_additional_covariate_LLN_1} .
\end{equation}

Furthermore, $\{\Delta M_n\}_{n \in \mathbb{N}^*}$ forms a martingale difference sequence with respect to the filtration $\{\mathcal{F}_n\}_{n \in \mathbb{N}}$, where
\begin{equation*}
    \Delta M_n =
    \frac{(T_n-\rho_{\theta_{n-1}}(X_n)) Z_n}{\rho_{\theta_{n-1}}(X_n)(1-\rho_{\theta_{n-1}}(X_n))}
    - h_{\theta_{n-1}}(\Lambda_{n-1}) .
\end{equation*}
Thus,
\begin{equation*}
    \frac{\Psi_N}{N} - \frac{1}{N} \sum_{n=0}^{N-1} h_{\theta_n}(\Lambda_n)
    = \frac{1}{N} \sum_{n=1}^{N}
    \left[ \frac{(T_n-\rho_{\theta_{n-1}}(X_n)) Z_n}{\rho_{\theta_{n-1}}(X_n)(1-\rho_{\theta_{n-1}}(X_n))}
    - h_{\theta_{n-1}}(\Lambda_{n-1}) \right]
    = \frac{1}{N} \sum_{n=1}^{N} \Delta M_n
\end{equation*}
To apply a law of large numbers for martingales, we bound its conditional moments.
We use the Jensen's inequality and choose appropriate $a \in [1/2,1)$, it holds that
\begin{align*}
    \mathbb{E} \left[ \left| \Delta M_n \right|^{1/a} \mid \mathcal{F}_{n-1} \right]
    &= \mathbb{E} \left[ \left| \frac{(T_n-\rho_{\theta_{n-1}}(X_n)) Z_n}{\rho_{\theta_{n-1}}(X_n)(1-\rho_{\theta_{n-1}}(X_n))}
    - h_{\theta_{n-1}}(\Lambda_{n-1}) \right|^{1/a} \mid \mathcal{F}_{n-1} \right] \\
    & \leq \frac{1}{\iota_\rho^{1/a}(1-\iota_\rho)^{1/a}} \left\{ \mathbb{E} \left[ \left| Z_n \right|^{1/a} \mid \mathcal{F}_n \right]
    + \mathbb{E} \left[ \left| h_{\theta_{n-1}}(\Lambda_{n-1}) \right|^{1/a} \mid \mathcal{F}_{n-1} \right] \right\} \\
    & \leq \frac{2}{\iota_\rho^{1/a}(1-\iota_\rho)^{1/a}} \mathbb{E} \left| Z \right|^{1/a}
    < \infty .
\end{align*}
Thus,
\begin{equation*}
    \sum_{n=1}^\infty n^{-1/a}
    \mathbb{E} \left[ \left| \Delta M_n \right|^{1/a} \mid \mathcal{F}_{n-1} \right]
    < \infty .
\end{equation*}
By Theorem 2.18 in \cite{hallMartingaleLimitTheory1980}, it holds that
\begin{equation}
    \frac{\Psi_N}{N} - \frac{1}{N} \sum_{n=0}^{N-1} h_{\theta_n}(\Lambda_n)
    = \frac{1}{N} \sum_{n=1}^{N} \Delta M_n
    \rightarrow 0 \quad \text{a.s.}
    \label{eq_proof_theorem_additional_covariate_LLN_2}
\end{equation}

Combining \eqref{eq_proof_theorem_additional_covariate_LLN_1} and \eqref{eq_proof_theorem_additional_covariate_LLN_2}, we obtain
\begin{equation*}
    \frac{\Psi_N}{N}
    = \frac{1}{N} \sum_{n=1}^N \frac{(T_n-\rho_{\theta_{n-1}}(X_n)) Z_n}{\rho_{\theta_{n-1}}(X_n)(1-\rho_{\theta_{n-1}}(X_n))}
    \xrightarrow{\mathbb{P}} 0 .
\end{equation*}

\subsection{Proof of Theorem \ref{theorem_additional_covariate_CLT}}

By Lemma \ref{lemma_application_simultaneous_properties}, Assumption \ref{assumption_simultaneous_properties} holds.

By Lemma \ref{lemma_CLT} with $Z^{(c)}_n := Z_n / [\rho_{\theta_{n-1}}(X_n)(1-\rho_{\theta_{n-1}}(X_n))]$, the limit holds that
\begin{equation*}
    \frac{1}{\sqrt{N}} \sum_{n=1}^{N}
    \left[ \frac{(T_n-\rho_{\theta_{n-1}}(X_n)) Z_n}{\rho_{\theta_{n-1}}(X_n)(1-\rho_{\theta_{n-1}}(X_n))}
    - \pi_{\theta_{n-1}} h_{\theta_{n-1}} \right]
    \xrightarrow{d} \mathcal{N}(0, {\sigma^*_{(Z)}}^2) ,
\end{equation*}
where $\pi_\theta$ is the invariant probability of the transition kernel $P_\theta$, the function $h_\theta$ is defined by
\begin{equation*}
    h_\theta(\Lambda)
    = \mathbb{E}_{X \sim \Gamma} \left[ \frac{(g_\theta(\Lambda, X) - \rho_\theta(X)) Z}{\rho_\theta(X)(1-\rho_\theta(X))} \right] ,
\end{equation*}
and the asymptotic variance ${\sigma^*_{(Z)}}^2$ equals that under the fixed parameter sequence $\{\theta_n = \theta^*\}_{n \in \mathbb{N}}$.
Note that one requirement in Lemma \ref{lemma_CLT}, Assumption \ref{assumption_lipschitz_continuity_conditional_expectation}, is satisfied as a direct consequence of Assumption \ref{assumption_lipschitz_continuity_functions_kernels} together with the identity
\begin{equation*}
    \mathbb{E} \left[ \frac{Z}{\rho_\theta(X)(1-\rho_\theta(X))} \mid X \right]
    = \frac{\mathbb{E} \left[ Z \mid X \right]}{\rho_\theta(X)(1-\rho_\theta(X))} .
\end{equation*}

Finally, Lemma \ref{lemma_expectation_pi_h} implies that $\pi_\theta h_\theta = 0$ for any $\theta \in \Theta$. 
Moreover, Lemma \ref{lemma_expression_variance} implies that under the fixed parameter sequence $\{\theta_n = \theta^*\}_{n \in \mathbb{N}}$, the asymptotic variance is
\begin{equation*}
    {\sigma^*_{(Z)}}^2
    = \mathbb{E}_{(X, Z) \sim \Gamma_{X,Z}} \left[
    \left[ \rho_{\theta^*}(X)(1-\rho_{\theta^*}(X)) \right]^{-1}
    \left\{ Z - a^T \phi(X) \right\}^2 \right] ,
\end{equation*}
where the vector $a$ satisfies
\begin{align*}
    & a^T \mathbb{E} \left[ [\rho_{\theta^*}(X)(1-\rho_{\theta^*}(X))]^{-1} \phi(X) \phi(X)^T
    / \max\{ \|\phi(X)\|/[\rho_{\theta^*}(X)(1-\rho_{\theta^*}(X))], C_{\theta^*} \} \right] \\
    & \quad = \mathbb{E} \left[ [\rho_{\theta^*}(X)(1-\rho_{\theta^*}(X))]^{-1} Z \phi(X)^T
    / \max\{ \|\phi(X)\|/[\rho_{\theta^*}(X)(1-\rho_{\theta^*}(X))], C_{\theta^*} \} \right] .
\end{align*}

\subsection{Proof of Theorem \ref{theorem_conditional_allocation_ratio}}

Assumption \ref{assumption_lipschitz_continuity_functions_kernels} implies that $\rho_\theta(x) \rightarrow \rho_{\theta^*}(x)$ as $\theta \rightarrow \theta^*$.

Following a similar argument as in the proof of Theorem \ref{theorem_additional_covariate_LLN} with
\begin{equation*}
    h_\theta(\Lambda)
    = \mathbb{E}_{X \sim \Gamma} \left[ (g_\theta(\Lambda, X) - \rho_\theta(X)) \mathbb{I}(X = x) \right] .
\end{equation*}
It then follows that
\begin{equation*}
    \frac{N_{n,1}(x)}{N}
    - \frac{1}{N} \sum_{n=1}^N \left[ \rho_{\theta_{n-1}}(X_n) \mathbb{I}(X_n = x) \right]
    = \frac{1}{N} \sum_{n=1}^N (T_n-\rho_{\theta_{n-1}}(X_n)) \mathbb{I}(X_n = x)
    \xrightarrow{\mathbb{P}} 0 .
\end{equation*}
Moreover,
\begin{equation*}
    \sum_{n=1}^N \left[ \rho_{\theta_{n-1}}(X_n) \mathbb{I}(X_n = x)
    - \rho_{\theta_{n-1}}(x) P_\Gamma(X = x) \right]
\end{equation*}
forms a martingale.
Because the conditional variance of each term
\begin{equation*}
    \mathbb{E}_{X_n \sim \Gamma} \left[ \rho_{\theta_{n-1}}(X_n) \mathbb{I}(X_n = x)
    - \rho_{\theta_{n-1}}(x) P_\Gamma(X = x) \right]^2
    \leq 1 < \infty ,
\end{equation*}
by Theorem 2.18 in \cite{hallMartingaleLimitTheory1980}, it holds that
\begin{equation*}
    \frac{1}{N} \sum_{n=1}^N \left[ \rho_{\theta_{n-1}}(X_n) \mathbb{I}(X_n = x)
    - \rho_{\theta_{n-1}}(x) P_\Gamma(X = x) \right]
    \rightarrow 0 \quad \text{a.s.}
\end{equation*}

Combining the limits above, it holds that
\begin{equation*}
    \frac{N_{n,1}(x)}{N}
    - \frac{\sum_{n=1}^N \rho_{\theta_{n-1}}(x) }{N} P_\Gamma(X = x)
    \xrightarrow{\mathbb{P}} 0 .
\end{equation*}
Thus, by $\theta_n \xrightarrow{\mathbb{P}} \theta^*$ and $\rho_\theta(x) \rightarrow \rho_{\theta^*}(x)$ as $\theta \rightarrow \theta^*$, it follows that
\begin{equation*}
    \frac{N_{n,1}(x)}{N}
    \xrightarrow{\mathbb{P}} \rho_{\theta^*}(x) P_\Gamma(X = x) .
\end{equation*}

By the law of large numbers for the i.i.d. variables $\{\mathbb{I}(X_n = x)\}$, we have
\begin{equation*}
    \frac{N_n(x)}{N}
    = \frac{\sum_{n=1}^N \mathbb{I}(X_n = x)}{N}
    \xrightarrow{\mathbb{P}} P_\Gamma(X = x) .
\end{equation*}

In conclusion,
\begin{equation*}
    \frac{N_{n,1}(x)}{N_n(x)}
    = \left[ \frac{N_n(x)}{N} \right]^{-1} \frac{N_{n,1}(x)}{N}
    \xrightarrow{\mathbb{P}} \rho_{\theta^*}(x) .
\end{equation*}

\subsection{Proof of Theorem \ref{theorem_estimation_consistency}}

By Lemma \ref{lemma_application_simultaneous_properties}, Assumption \ref{assumption_simultaneous_properties} is satisfied.
Moreover, by Lemma \ref{lemma_expectation_pi_h}, we have $\pi_\theta \left[ g_\theta(\cdot, x) \right] = \rho_\theta(x)$ for $\Gamma$-a.e. $x$.
Therefore, Lemma \ref{lemma_estimation_convergence} directly implies the desired conclusion.

\subsection{Proof of Theorem \ref{theorem_estimation_difference}}

By Lemma \ref{lemma_application_simultaneous_properties}, Assumption \ref{assumption_simultaneous_properties} is satisfied.
Moreover, by Lemma \ref{lemma_expectation_pi_h}, we have $\pi_\theta \left[ g_\theta(\cdot, x) \right] = \rho_\theta(x)$ for $\Gamma$-a.e. $x$.
Therefore, Lemma \ref{lemma_estimation_difference_bounded} implies that
\begin{equation*}
    \mathbb{E} \left[ \| \eta_n - \eta_{n-1} \| \right] < c n^{-q} ,
\end{equation*}
with $q \in (0,1]$ defined in Assumption \ref{assumption_theta_weak_diminishing_adaptation}.
Thus,
\begin{equation*}
    \sum_{n=0}^{N-1} \mathbb{E} \| \eta_n - \eta_{n+1} \|
    = \begin{cases}
        O(N^{1-q}) ,  & 0 < q < 1, \\
        O(\ln N)   ,  & q = 1 .
    \end{cases}
\end{equation*}
In conclusion,
\begin{equation*}
    \sum_{n=0}^{N-1} \| \eta_n - \eta_{n+1} \| = o_P(N^p) ,
\end{equation*}
for any $p \in (1-q, 1)$.

\subsection{Proof of Theorem \ref{theorem_estimation_asymptotic_normality}}

By Lemma \ref{lemma_application_simultaneous_properties}, Assumption \ref{assumption_simultaneous_properties} is satisfied.
Moreover, by Lemma \ref{lemma_expectation_pi_h}, we have $\pi_\theta \left[ g_\theta(\cdot, x) \right] = \rho_\theta(x)$ for $\Gamma$-a.e. $x$.
Therefore, Lemma \ref{lemma_estimation_asymptotic_normality} directly implies the desired conclusion.

\subsection{Proof of Theorem \ref{theorem_allocation_parameter_update_weak}}

When the allocation parameter $\theta_{n+1}$ is updated according to \eqref{eq_allocation_parameter_update_mechanism_1} under Assumption \ref{assumption_theta_compact},
\begin{equation*}
    \frac{1}{N} \sum_{n=0}^{N-1} d(\theta_n,\theta_{n+1})
    = \frac{1}{N} \sum_{n=1}^{N} [ d(\theta_n,\theta_{n-1}) \mathbb{I}(n \in S) ]
    \leq \diam(\Theta) \frac{\#(S \cap \{1, \dots, N\})}{N}
    \rightarrow 0
\end{equation*}

When the allocation parameter $\theta_{n+1}$ is updated according to \eqref{eq_allocation_parameter_update_mechanism_2},
\begin{equation*}
    \frac{1}{N} \sum_{n=0}^{N-1} d(\theta_n,\theta_{n+1})
    \leq \frac{1}{N} \sum_{n=1}^{N} C_{\mathrm{clip}, n}
    \rightarrow 0 .
\end{equation*}

Because these bounds are deterministic and converge to zero, the proof of Theorem \ref{theorem_allocation_parameter_update_weak} is complete.

\subsection{Proof of Theorem \ref{theorem_allocation_parameter_update_strong}}

By Lemma \ref{lemma_allocation_parameter_update}, we obtain $\theta_n \rightarrow \theta^*$ in probability, $\theta^* = \eta^*$ and
\begin{equation*}
    \sum_{n=0}^{N-1} d(\theta_n,\theta_{n+1})
    \leq \sum_{n=0}^{N-1} d(\eta_n,\eta_{n+1})
    = o_P(N^p) .
\end{equation*}

\section{Treatment Effect Estimation}
\label{sec_treatment_effect_estimation}

\subsection{Proof of Theorem \ref{theorem_IPW_estimator}}

The asymptotic distribution of $\hat{\tau}_{\mathrm{IPW}}$ follows from Lemmas \ref{lemma_expectation_pi_h}, \ref{lemma_expression_variance}, \ref{lemma_CLT} and \ref{lemma_lipschitz_continuity_conditional_expectation} with $Z^{(c)}_n := Y_n(1)/\rho_{\theta_{n-1}}(X_n) + Y_n(0)/\left(1-\rho_{\theta_{n-1}}(X_n)\right)$ by the decomposition
\begin{align*}
    \hat{\tau}_{\mathrm{IPW}}
    &= \frac{1}{N} \sum_{n=1}^N \left[ \frac{T_n Y_n(1)}{\rho_{\theta_{n-1}}(X_n)}
    - \frac{(1-T_n) Y_n(0)}{1-\rho_{\theta_{n-1}}(X_n)} \right] \\
    &= \frac{1}{N} \sum_{n=1}^N \left[ \left( T_n - \rho_{\theta_{n-1}}(X_n) \right) \left[ \frac{Y_n(1)}{\rho_{\theta_{n-1}}(X_n)} + \frac{Y_n(0)}{1-\rho_{\theta_{n-1}}(X_n)} \right]
    + Y_n(1) - Y_n(0) \right] .
\end{align*}

\section{Lemmas for the CBARA Procedure}
\label{sec_application_process}

\subsection{Definition of the Transition Kernel}

\label{subsec_transition_kernel}

Let the state space be $\mathrm{X} = \mathbb{R}^d$ and the parameter space be $\Theta$.
For each $\theta \in \Theta$, define a transition kernel $P_\theta$ on $\mathrm{X}$ by
\begin{equation}
    P_\theta(\Lambda,h) = \int \left[ g_\theta(\Lambda, X)h(\Lambda+\phi(X)/\rho_\theta(X)) + [1-g_\theta(\Lambda, X)]h(\Lambda-\phi(X)/(1-\rho_\theta(X))) \right] \Gamma( \mathrm{d} X)
    \label{eq_transition_kernel_CARA_procedure} .
\end{equation}
In this subsection, we show, via the following lemma, that the evolution of the process $\{\Lambda_n\}$ under the CBARA procedure can be fully characterized by the transition kernel $P_\theta$ defined above.
\begin{lemma}[Lemma \ref{lemma_transition_kernel_main}]
    \label{lemma_transition_kernel}
    
    Let the function $h: \mathrm{X} \to \mathbb{R}$ be any integrable function.
    Then, for any $n \in \mathbb{N}$,
    \begin{equation*}
        \mathbb{E} \left[ h(\Lambda_{n+1}) \mid \mathcal{F}_n \right]
        = P_{\theta_n}(\Lambda_n,h) .
    \end{equation*}
\end{lemma}
\begin{proof}
    By definition of $\Lambda_{n+1}$, we have
    \begin{align*}
        & \quad \mathbb{E} \left[ h(\Lambda_{n+1}) \mid \mathcal{F}_n \right] \\
        &= \mathbb{E} \left[ h(\Lambda_n + (T_{n+1}-\rho_{\theta_n}(X_{n+1})) \phi(X_{n+1}) / [\rho_{\theta_n}(X_{n+1})(1-\rho_{\theta_n}(X_{n+1}))]) \mid \mathcal{F}_n \right] \\
        &= \mathbb{E} \left[ \mathbb{E} \left[ h(\Lambda_n + (T_{n+1}-\rho_{\theta_n}(X_{n+1})) \phi(X_{n+1}) / [\rho_{\theta_n}(X_{n+1})(1-\rho_{\theta_n}(X_{n+1}))]) \mid \mathcal{F}_n, X_{n+1} \right] \mid \mathcal{F}_n \right] \\
        &= \mathbb{E} \left[ \left[ h(\Lambda_n +  \phi(X_{n+1}) / \rho_{\theta_n}(X_{n+1}))g_{\theta_n}(\Lambda_n, X_{n+1}) \right. \right. \\
        & \quad \left. \left. + h(\Lambda_n - \phi(X_{n+1}) / (1-\rho_{\theta_n}(X_{n+1})))(1-g_{\theta_n}(\Lambda_n, X_{n+1})) \right]
        \mid \mathcal{F}_n \right] \\
        &= \mathbb{E}_{X \sim \Gamma} \left[ h(\Lambda_n + \phi(X)/\rho_{\theta_n}(X))g_{\theta_n}(\Lambda_n, X)
        + h(\Lambda_n - \phi(X)/(1-\rho_{\theta_n}(X)))(1-g_{\theta_n}(\Lambda_n, X)) \right] \\
        &= P_{\theta_n}(\Lambda_n,h) ,
    \end{align*}
    where the third equality is from Assumption \ref{assumption_sampling} and the fourth equality is from the independence between $X_{n+1}$ and $\mathcal{F}_n$ in Assumption \ref{assumption_filtration}.
\end{proof}

Although many of the subsequent proofs involve not a fixed function $h$ but a family of functions $\{h_\theta\}_{\theta \in \Theta}$, Lemma \ref{lemma_transition_kernel} still implies that
\begin{equation*}
    \mathbb{E} \left[ h_{\theta_n}(\Lambda_{n+1}) \mid \mathcal{F}_n \right] = P_{\theta_n}(\Lambda_n,h_{\theta_n}) ,
\end{equation*}
since $h_{\theta_n}$ is $\mathcal{F}_n$-measurable.

\subsection{Continuity of Transition Kernels}

\label{subsec_continuity_transition_kernels}

In this subsection, based on Assumption \ref{assumption_lipschitz_continuity_functions_kernels}, for the transition kernels \eqref{eq_transition_kernel_CARA_procedure}, we establish a particular form of continuity (Definition \ref{definition_robust_lipschitz_continuity_family_kernels}). 
The key step is provided by the following lemma.
\begin{lemma}[Lemma \ref{lemma_continuity_transition_kernels_main}]
    \label{lemma_continuity_transition_kernels}

    Suppose that Assumption \ref{assumption_lipschitz_continuity_functions_kernels} holds and that $\phi(X)$ has a finite second moment.
    Then the family of transition probability kernels $\{P_\theta\}_{\theta \in \Theta}$ is robustly Lipschitz continuous with a Lipschitz constant $L_P \geq 0$.
\end{lemma}
\begin{proof}
    For any given $\theta$, $\theta^\prime \in \Theta$, define the coupling kernel 
    \begin{equation*}
        K_{\theta, \theta^\prime}: \mathrm{X}^2 \times \mathcal{X}^{\otimes 2} \to [0,1]
    \end{equation*}
    such that, for all $\Lambda, \Lambda^\prime \in \mathrm{X}$, for any function $\Phi: \mathrm{X}^2 \to \mathbb{R}$,
    \begin{align*}
        & \quad K_{\theta, \theta^\prime}(\Lambda, \Lambda^\prime; \Phi) \\
        &= \int \left[ \min \left\{ g_\theta(\Lambda, x), g_{\theta^\prime}(\Lambda^\prime, x) \right\}
        \Phi(\Lambda + \phi(x)/\rho_\theta(x), \Lambda^\prime + \phi(x)/\rho_{\theta^\prime}(x))
        \right] \Gamma(\mathrm{d} x) \\
        & \quad + \int \left[ \min \left\{ 1 - g_\theta(\Lambda, x), 1 - g_{\theta^\prime}(\Lambda^\prime, x) \right\}
        \Phi(\Lambda - \phi(x)/(1-\rho_\theta(x)), \Lambda^\prime - \phi(x)/(1-\rho_{\theta^\prime}(x)))
        \right] \Gamma(\mathrm{d} x) .
    \end{align*}

    Consider an indicator function $\chi_A$ on $A \in \mathcal{X}$.  
    Since
    \begin{align*}
        P_\theta(\Lambda, \chi_A)
        &= \int \left[ g_\theta(\Lambda, x)
        \chi_A(\Lambda + \phi(x)/\rho_\theta(x))
        \right] \Gamma(\mathrm{d} x) \\
        & \quad + \int \left[ \left\{ 1 - g_\theta(\Lambda, x) \right\}
        \chi_A(\Lambda - \phi(x)/(1-\rho_\theta(x)))
        \right] \Gamma(\mathrm{d} x) ,
    \end{align*}
    taking $\Phi(x_1, x_2) = \chi_A(x_1)$ yields
    \begin{align*}
        & \min \left\{ g_\theta(\Lambda, x), g_{\theta^\prime}(\Lambda^\prime, x) \right\}
        \Phi(\Lambda + \phi(x)/\rho_\theta(x), \Lambda^\prime + \phi(x)/\rho_{\theta^\prime}(x)) \\
        & \quad \leq g_\theta(\Lambda, x) \chi_A(\Lambda + \phi(x)/\rho_\theta(x)) , \\
        & \min \left\{ 1 - g_\theta(\Lambda, x), 1 - g_{\theta^\prime}(\Lambda^\prime, x) \right\}
        \Phi(\Lambda - \phi(x)/(1-\rho_\theta(x)), \Lambda^\prime - \phi(x)/(1-\rho_{\theta^\prime}(x))) \\
        & \quad \leq \left\{ 1 - g_\theta(\Lambda, x) \right\}
        \chi_A(\Lambda - \phi(x)/(1-\rho_\theta(x))) .
    \end{align*}
    Therefore, for any $A \in \mathcal{X}$,
    \begin{equation*}
        K_{\theta,\theta^\prime}(\Lambda, \Lambda^\prime; A \times \mathrm{X})
        \leq P_\theta(\Lambda,A) .
    \end{equation*}
    Similarly, for any $B \in \mathcal{X}$,
    \begin{equation*}
        K_{\theta,\theta^\prime}(\Lambda, \Lambda^\prime; \mathrm{X} \times B)
        \leq P_{\theta^\prime}(\Lambda^\prime,B) .
    \end{equation*}

    Moreover,
    \begin{align*}
        & \quad K_{\theta, \theta^\prime}(\Lambda, \Lambda^\prime; \mathrm{X} \times \mathrm{X}) \\
        &= \int \left[ \min \left\{ g_\theta(\Lambda, x), g_{\theta^\prime}(\Lambda^\prime, x) \right\}
        + \min \left\{ 1 - g_\theta(\Lambda, x), 1 - g_{\theta^\prime}(\Lambda^\prime, x) \right\} \right] \Gamma(\mathrm{d} x) \\
        &= 1 - \int \left| g_\theta(\Lambda, x) - g_{\theta^\prime}(\Lambda^\prime, x) \right| \Gamma(\mathrm{d} x) \\
        & \geq 1 - \left\| g_{\theta}(\Lambda,\cdot) - g_{\theta^\prime}(\Lambda^\prime,\cdot) \right\|_{L^2(\Gamma)} ,
    \end{align*}
    and
    \begin{align*}
        & \quad \int d(u,v) K(\Lambda, \Lambda^\prime; \mathrm{d} u \times \mathrm{d} v)
        = \int \| u - v \| K(\Lambda, \Lambda^\prime; \mathrm{d} u \times \mathrm{d} v) \\
        &= \int \left[ \min \left\{ g_\theta(\Lambda, x), g_{\theta^\prime}(\Lambda^\prime, x) \right\}
        \left\| (\Lambda + \phi(x)/\rho_\theta(x))
        - (\Lambda^\prime + \phi(x)/\rho_{\theta^\prime}(x)) \right\|
        \right] \Gamma(\mathrm{d} x) \\
        & \quad + \int \left[ \min \left\{ 1 - g_\theta(\Lambda, x), 1 - g_{\theta^\prime}(\Lambda^\prime, x) \right\}
        \left\| (\Lambda - \phi(x)/(1-\rho_\theta(x))) - (\Lambda^\prime - \phi(x)/(1-\rho_{\theta^\prime}(x))) \right\|
        \right] \Gamma(\mathrm{d} x) \\
        & \leq \int \left[ \min \left\{ g_\theta(\Lambda, x), g_{\theta^\prime}(\Lambda^\prime, x) \right\}
        \left\{ \left\| \Lambda - \Lambda^\prime \right\|
        + \left\| \phi(x)/\rho_\theta(x) - \phi(x)/\rho_{\theta^\prime}(x) \right\| \right\}
        \right] \Gamma(\mathrm{d} x) \\
        & \quad + \int \left[ \min \left\{ 1 - g_\theta(\Lambda, x), 1 - g_{\theta^\prime}(\Lambda^\prime, x) \right\}
        \left\{ \left\| \Lambda - \Lambda^\prime \right\|
        + \left\| \phi(x)/(1-\rho_\theta(x)) - \phi(x)/(1-\rho_{\theta^\prime}(x)) \right\| \right\}
        \right] \Gamma(\mathrm{d} x) \\
        & \leq \left\| \Lambda - \Lambda^\prime \right\|
        + \frac{1}{\iota_\rho^2(1-\iota_\rho)^2} \int \left| \rho_\theta(x) - \rho_{\theta^\prime}(x) \right|
        \left\| \phi(x) \right\| \Gamma(\mathrm{d} x) \\
        & \leq \left\| \Lambda - \Lambda^\prime \right\|
        + \frac{1}{\iota_\rho^2(1-\iota_\rho)^2}
        \left\| \rho_{\theta} - \rho_{\theta^\prime} \right\|_{L^2(\Gamma)}
        \left\| \phi \right\|_{L^2(\Gamma)} .
    \end{align*}

    Thus, we conclude that
    \begin{align*}
        0 \leq 1 - K(\Lambda, \Lambda^\prime; \mathrm{X} \times \mathrm{X})
        & \leq L_P d(\Lambda, \Lambda^\prime) + L_P d(\theta, \theta^\prime) , \\
        \int d(u,v) K(\Lambda, \Lambda^\prime; \mathrm{d} u \times \mathrm{d} v)
        &\leq d(\Lambda, \Lambda^\prime) + L_P d(\theta, \theta^\prime) .
    \end{align*}
    This completes the proof.
\end{proof}

\subsection{Technical Lemmas under the CBARA Procedure}

\label{subsec_technical_lemmas_CARA}

The lemmas used in our analysis are primarily aimed at establishing Lemma \ref{lemma_for_simultaneous_small_set_condition}.
The conclusions of these lemmas constitute the components of Assumption \ref{assumption_simultaneous_properties}, which plays a central role in Section \ref{sec_lemmas_Markov_chain}.

The following two lemmas are adapted from \cite{fangGeneralNonMarkovianFramework2026}.
\begin{lemma}
    \label{lemma_all_negative_feedback}

    If $\mathbb{E}_{X \sim \Gamma} \left[ \|\phi(X)\| \right] < \infty$, then there exists some $M>0$ and $\Delta>0$ such that for any $\theta \in \Theta$ and $\Lambda \in W_\phi$ with $\|\Lambda\| \geq M$,
    \begin{equation*}
        \mathbb{E}_\theta \left[ (\Lambda_1 - \Lambda_0)^T \frac{\Lambda_0}{\|\Lambda_0\|} \mid \Lambda_0 = \Lambda \right]
        = \mathbb{E}_{X \sim \Gamma} \left[
        \frac{g_\theta(\Lambda,X) - \rho_\theta(X)}{\rho_\theta(X)(1-\rho_\theta(X))} \cdot
        \frac{\phi(X)^T \Lambda}{\|\Lambda\|} \right]
        \leq -\Delta .
    \end{equation*}
\end{lemma}

\begin{lemma}
    \label{lemma_drift_condition_inequality}
    If Assumption \ref{assumption_x_sub_exponential_bound} holds, then there exists positive constants $\beta<1$, $b$ and $\lambda_1$ that only depend on $M$, $\Delta$, $\lambda$ and $C$ such that for any $\theta \in \Theta$ and $\Lambda \in W_\phi$,
    \begin{equation*}
        \mathbb{E}_\theta \left[ e^{\lambda_1 \|\Lambda_1\|} \mid \Lambda_0 = \Lambda \right] \leq \beta e^{\lambda_1 \|\Lambda\|} + b .
    \end{equation*}

    Denote the Lyapunov function $V(\Lambda) = e^{\lambda_1 \|\Lambda\|}$. For any $\alpha \in (0,1]$, there exists positive constants $\beta_\alpha<1$ and $b_\alpha=b$ such that the inequality
    \begin{equation*}
        (P_\theta V^\alpha)(\Lambda) \leq \beta_\alpha V^\alpha(\Lambda) + b_\alpha
    \end{equation*}
    holds.
\end{lemma}

\begin{lemma}
    \label{lemma_lipschitz_continuity_conditional_expectation}

    Suppose that the distribution $\Gamma_{X,Z,\theta}$ is the distribution of
    \begin{equation*}
        \left( X, \frac{Z^{(c)}}{\rho_\theta(X)^{m_{c,1}}(1-\rho_\theta(X))^{m_{c,0}}},
        \frac{Z^{(u)}}{\rho_\theta(X)^{m_{u,1}}(1-\rho_\theta(X))^{m_{u,0}}} \right) ,
    \end{equation*}
    where $X$ and $Z = (Z^{(c)}, Z^{(u)})$ are random variables following the joint distribution $\Gamma_{X,Z}$, and $m_{c,1}, m_{c,0}, m_{u,1}, m_{u,0} \in \mathbb{N}$.
    If Assumption \ref{assumption_lipschitz_continuity_functions_kernels} holds and $\mathbb{E} \left[ \|Z\|^2 \right] < \infty$, then Assumption \ref{assumption_lipschitz_continuity_conditional_expectation_relaxed} holds.
    If Assumption \ref{assumption_lipschitz_continuity_functions_kernels} holds and $\mathbb{E} \left[ \|Z\|^{4+\epsilon} \right] < \infty$ for some $\epsilon > 0$, then Assumption \ref{assumption_lipschitz_continuity_conditional_expectation} holds.
\end{lemma}
\begin{proof}
    For brevity, we only present the proof for the case of Assumption \ref{assumption_lipschitz_continuity_conditional_expectation}, and the proof for Assumption \ref{assumption_lipschitz_continuity_conditional_expectation_relaxed} is similar.
    We only consider
    \begin{equation*}
        f_\theta = f_{\left( Z^{(c)} \right)^2,\theta} = \mathbb{E}_\theta \left[ \left( \frac{Z^{(c)}}{\rho_\theta(X)^{m_{c,1}}(1-\rho_\theta(X))^{m_{c,0}}} \right)^2 \middle| X = x \right]
    \end{equation*}
    here.
    Other cases can be proved similarly.
    Define
    \begin{equation*}
        f_{\left( Z^{(c)} \right)^2} = \mathbb{E} \left[ \left( Z^{(c)} \right)^2 \mid X = x \right] .
    \end{equation*}
    Then $f_\theta$ can be rewritten as
    \begin{equation*}
        f_\theta(x) = \frac{f_{\left( Z^{(c)} \right)^2}(x)}{\rho_\theta(x)^{2 m_{c,1}} (1-\rho_\theta(x))^{2 m_{c,0}}} .
    \end{equation*}

    Then
    \begin{align*}
        | f_\theta(x) - f_{\theta^\prime}(x) |
        & \leq \left| f_{\left( Z^{(c)} \right)^2}(x) \right|
        \left| \frac{1}{\rho_\theta(x)^{2 m_{c,1}} (1-\rho_\theta(x))^{2 m_{c,0}}}
        - \frac{1}{\rho_{\theta^\prime}(x)^{2 m_{c,1}} (1-\rho_{\theta^\prime}(x))^{2 m_{c,0}}} \right| \\
        & \leq \left| f_{\left( Z^{(c)} \right)^2}(x) \right|
        \frac{ 2 (m_{c,1} + m_{c,0})
        \left| \rho_{\theta^\prime}(x) - \rho_\theta(x) \right|}
        {\iota_\rho^{2(m_{c,1}+m_{c,0})+1} (1-\iota_\rho)^{2(m_{c,1}+m_{c,0})+1}} .
    \end{align*}

    For any $\theta, \theta^\prime \in \Theta$, it holds that
    \begin{align*}
        \left\| e^{\lambda_f \|\phi(X)\|} \left| f_{\theta} - f_{\theta^\prime} \right| \right\|_{L^1(\Gamma)}
        & \lesssim \left\| e^{\lambda_f \|\phi(X)\|} \left| f_{\left( Z^{(c)} \right)^2} \right|
        \left| \rho_{\theta^\prime} - \rho_\theta \right|
        \right\|_{L^2(\Gamma)} \\
        & \lesssim
        \left\| e^{\lambda_f \|\phi(X)\|} \right\|_{L^{M_\epsilon}(\Gamma)}
        \left\| f_{\left( Z^{(c)} \right)^2} \right\|_{L^{2+\epsilon/2}(\Gamma)}
        \left\| \rho_{\theta^\prime} - \rho_\theta \right\|_{L^2(\Gamma)} \\
        & \lesssim \left\| \rho_{\theta^\prime} - \rho_\theta \right\|_{L^2(\Gamma)} ,
    \end{align*}
    where $M_\epsilon$ satisfies $\frac{1}{M_\epsilon} + \frac{1}{2+\epsilon/2} + \frac{1}{2} = 1$ and the last inequality holds by $\mathbb{E} \left[ \|Z\|^{4+\epsilon} \right] < \infty$ and Assumption \ref{assumption_x_sub_exponential_bound} when $\lambda_f$ is sufficiently small such that $M_\epsilon \lambda_f < \lambda$.

    By Assumption \ref{assumption_lipschitz_continuity_functions_kernels}, Assumption \ref{assumption_lipschitz_continuity_conditional_expectation} holds.
\end{proof}

\begin{lemma}
    \label{lemma_for_simultaneous_small_set_condition}

    Under Assumptions \ref{assumption_theta_compact} and \ref{assumption_for_simultaneous_small_set_condition}, for any $R > 0$, there exists some $d_R \in \mathbb{N}^*$ and $\delta_{R,P} > 0$ such that for any $\theta \in \Theta$,
    \begin{equation*}
        P_\theta^{d_R}(\Lambda, \cdot)
        \geq \delta_{R,P} \mu_{\mathrm{leb}, B(\Lambda,R)} .
    \end{equation*}
\end{lemma}
The proof of Lemma \ref{lemma_for_simultaneous_small_set_condition} is presented in Subsection \ref{subsec_proof_lemma_for_simultaneous_small_set_condition}. 
The argument is somewhat technical and relies on three main observations.
First, Assumption \ref{assumption_sampling} implies that the transition kernel $P_\theta$ controls a constant multiple of a random walk.
Second, Assumption \ref{assumption_for_simultaneous_small_set_condition} shows that this random walk, starting from any point, admits a density on an open ball with a possibly nonzero center after $2s$ steps.
Therefore, using the increment distribution of the random walk, the density on this open ball can be concentrated into a mass at the center.
As $n$ increases, after $n$ steps, the region where the original random walk and the concentrated random walk admit densities grows on the order of $n$. 
Third, as in the original random walk, the concentrated random walk starting from the original point has zero mean and is therefore concentrated around that point. 
As a result, its typical range grows on the order of $O(\sqrt{n})$, which is much slower than $n$.
Therefore, when $n$ is sufficiently large, the open ball on which the original random walk admits a density necessarily covers the original point. 
Finally, by the translation invariance of the random walk, this conclusion extends to any starting point.

\begin{lemma}
    \label{lemma_application_simultaneous_geometric_ergodicity}
    
    Let the Lyapunov function $V(\Lambda) = e^{\lambda_1 \|\Lambda\|}$ be as defined in Lemma \ref{lemma_drift_condition_inequality}.
    If Assumptions \ref{assumption_theta_compact}, \ref{assumption_x_sub_exponential_bound} and \ref{assumption_for_simultaneous_small_set_condition} hold, then for any $\theta \in \Theta$, $P_\theta$ is positive Harris recurrent with the unique invariant probability $\pi_\theta$ and $\pi_\theta V \leq \frac{b}{1-\beta}$.
    Moreover, there exists a positive number $L>1$ such that for any $\theta \in \Theta$,
    \begin{equation*}
        \| P_\theta^{n}(\Lambda, \cdot) - \pi_\theta \|_V \leq L(1-L^{-1})^n V(\Lambda) .
    \end{equation*}
    
    If $V$ is replaced by $V^\alpha$ for some $\alpha \in (0,1]$, the corresponding constant is denoted by $L_\alpha$.
\end{lemma}
The proof of Lemma \ref{lemma_application_simultaneous_geometric_ergodicity} is presented in Subsection \ref{subsec_proof_lemma_application_simultaneous_geometric_ergodicity}.
The proof follows as a simple corollary of Lemma \ref{lemma_simultaneous_geometric_ergodicity}.

\begin{lemma}
    \label{lemma_application_simultaneous_properties}

    Let the Lyapunov function $V(\Lambda) = e^{\lambda_1 \|\Lambda\|}$ be as defined in Lemma \ref{lemma_drift_condition_inequality}.
    If Assumptions \ref{assumption_theta_compact}, \ref{assumption_lipschitz_continuity_functions_kernels}, \ref{assumption_x_sub_exponential_bound} and \ref{assumption_for_simultaneous_small_set_condition} hold, then Assumption \ref{assumption_simultaneous_properties} holds.
\end{lemma}
\begin{proof}
    Item \ref{item_assumption_simultaneous_properties_1} of Assumption \ref{assumption_simultaneous_properties} can be deduced from Lemma \ref{lemma_application_simultaneous_geometric_ergodicity}.
    Item \ref{item_assumption_simultaneous_properties_2} of Assumption \ref{assumption_simultaneous_properties} can be deduced from Lemma \ref{lemma_drift_condition_inequality}.
    Item \ref{item_assumption_simultaneous_properties_3} of Assumption \ref{assumption_simultaneous_properties} can be deduced from Lemma \ref{lemma_continuity_transition_kernels}.
    Item \ref{item_assumption_simultaneous_properties_4} of Assumption \ref{assumption_simultaneous_properties} can be deduced from the definition of the Lyapunov function $V(\Lambda) = e^{\lambda_1 \|\Lambda\|}$.
\end{proof}

\begin{lemma}
    \label{lemma_allocation_parameter_update}

    If the allocation parameter $\theta_n$ is updated according to \eqref{eq_allocation_parameter_update_mechanism_1}, or according to \eqref{eq_allocation_parameter_update_mechanism_2} when the parameter space $\Theta$ is a convex subset of a Euclidean space, it holds that
    \begin{equation*}
        \sum_{n=0}^{N-1} d(\theta_n,\theta_{n+1})
        \leq \sum_{n=0}^{N-1} d(\eta_n,\eta_{n+1}) .
    \end{equation*}
    In addition, if $\eta_n \rightarrow \eta^*$ in probability, Assumption \ref{assumption_theta_compact} holds and $\sum_{n=1}^\infty C_{\mathrm{clip}, n} = \infty$ for \eqref{eq_allocation_parameter_update_mechanism_2}, then $\theta_n \rightarrow \theta^*$ in probability, and $\theta^* = \eta^*$.
\end{lemma}

\subsection{Proof of Lemma \ref{lemma_all_negative_feedback}}

\label{subsec_proof_lemma_all_negative_feedback}

The quantity
\begin{align*}
    & \quad \mathbb{E}_{X \sim \Gamma} \left[
    \frac{g_\theta(\Lambda,X) - \rho_\theta(X)}{\rho_\theta(X)(1-\rho_\theta(X))} \cdot
    \frac{\phi(X)^T \Lambda}{\|\Lambda\|} \right] \\
    &= p_\theta \mathbb{E}_{X \sim \Gamma} \left[
    \frac{- [\rho_\theta(X)(1-\rho_\theta(X))]^{-1} \phi(X)^T \Lambda}{\max\{\|\phi(X)\|/[\rho_\theta(X)(1-\rho_\theta(X))],C_\theta\}
    \max\{\|\Lambda\|,C_\Lambda\}} \frac{\phi(X)^T \Lambda}{\|\Lambda\|} \right] \\
    &= -\frac{p_\theta}{\|\Lambda\|\max\{\|\Lambda\|,C_\Lambda\}}
    \mathbb{E}_{X \sim \Gamma} \left[
        \frac{[\rho_\theta(X)(1-\rho_\theta(X))]^{-1} [\phi(X)^T \Lambda]^2}{\max\{\|\phi(X)\|/[\rho_\theta(X)(1-\rho_\theta(X))],C_\theta\}}
    \right] .
\end{align*}

Because of the definition of the subspace $W_\phi$, for any $\Lambda \in W_\phi$ satisfying $\Lambda \neq 0$, it holds that $P(\phi(X)^T \Lambda \neq 0) > 0$.
Since $\rho_\theta(X) \in [\iota_\rho, 1 - \iota_\rho]$ in Assumption \ref{assumption_sampling}, for any $\Lambda \in W_\phi$ with $\|\Lambda\| = 1$,
\begin{align}
    & \quad \mathbb{E}_{X \sim \Gamma} \left[
        \frac{[\rho_\theta(X)(1-\rho_\theta(X))]^{-1} [\phi(X)^T \Lambda]^2}{\max\{\|\phi(X)\|/[\rho_\theta(X)(1-\rho_\theta(X))],C_\theta\}}
    \right] \notag \\
    & \geq \mathbb{E}_{X \sim \Gamma} \left[
        \frac{ [\phi(X)^T \Lambda]^2}{\max\{\|\phi(X)\|,C_\theta\rho_\theta(X)(1-\rho_\theta(X))\}}
    \right]
    > 0 \label{eq_negative_feedback_quadratic_form} .
\end{align}
Furthermore, by $\inf_{\theta \in \Theta} 1/C_\theta > 0$, the compactness of the unit sphere $\{ \Lambda \in W_\phi \mid \|\Lambda\| = 1\}$ and the continuity of \eqref{eq_negative_feedback_quadratic_form} in $\Lambda$, there exists a positive number $c$, independent of $\theta \in \Theta$, such that for any $\Lambda \in W_\phi$ with $\|\Lambda\| = 1$,
\begin{equation*}
    \mathbb{E}_{X \sim \Gamma} \left[
        \frac{ [\phi(X)^T \Lambda]^2}{\max\{\|\phi(X)\|,C_\theta \rho_\theta(X)(1-\rho_\theta(X))\}}
    \right] .
\end{equation*}
Therefore, for any $\Lambda \in W_\phi$ with $\|\Lambda\| \geq C_\Lambda$,
\begin{align*}
    & \mathbb{E}_{X \sim \Gamma} \left[
    \frac{g_\theta(\Lambda,X) - \rho_\theta(X)}{\rho_\theta(X)(1-\rho_\theta(X))} \cdot
    \frac{\phi(X)^T \Lambda}{\|\Lambda\|} \right]
    = -\frac{p_\theta}{\|\Lambda\|^2}
    \mathbb{E}_{X \sim \Gamma} \left[
        \frac{[\phi(X)^T \Lambda]^2}{\max\{\|\phi(X)\|,C_\theta \rho_\theta(X)(1-\rho_\theta(X))\}}
    \right] \\
    & \quad = - p_\theta \mathbb{E}_{X \sim \Gamma} \left[
        \frac{[\phi(X)^T \Lambda/\|\Lambda\|]^2}{\max\{\|\phi(X)\|,C_\theta \rho_\theta(X)(1-\rho_\theta(X))\}}
    \right]
    \leq - c \inf_{\theta \in \Theta} p_\theta .
\end{align*}
The proof is thus complete by choosing $\Delta = c \inf_{\theta \in \Theta} p_\theta$ and $M = C_\Lambda$.

\subsection{Proof of Lemma \ref{lemma_drift_condition_inequality}}

Under Assumption \ref{assumption_x_sub_exponential_bound}, Lemma \ref{lemma_all_negative_feedback} guarantees the existence of constants $M > 0$ and $\Delta > 0$ such that for any $\theta \in \Theta$ and $\Lambda \in W_\phi$ with $\|\Lambda\| \geq M$,
\begin{equation*}
    \mathbb{E}_\theta \left[ (\Lambda_1 - \Lambda_0)^T \frac{\Lambda_0}{\|\Lambda_0\|} \mid \Lambda_0 = \Lambda \right]
    \leq -\Delta .
\end{equation*}

The primary objective of this proof is to establish an upper bound for the expression
\begin{equation}
    \mathbb{E}_\theta \left[ e^{\lambda_1 (\|\Lambda_1\| - \|\Lambda_0\|) } \mid \Lambda_0 = \Lambda \right] \label{eq_target_to_bound}
\end{equation}
where $\Lambda \in W_\phi$ is sufficiently large and $\lambda_1 > 0$ is a constant determined by $M$, $\Delta$, $\lambda$, and $C$ in Assumption \ref{assumption_x_sub_exponential_bound}.
This bound is universal for all $\theta \in \Theta$.
Without loss of generality, we assume that $\phi$ is the identity map, such that $\phi(X_1) = X_1$.

For any parameter $\theta \in \Theta$ and initial state $\Lambda_0$, let the random variable $\Lambda_1$ be distributed according to $P_\theta(\Lambda_0, \cdot)$.
Given that $\mathbb{E} \|X_1\| < \infty$, there exists a sufficiently large $M_1 > 0$ such that
\begin{align*}
    & \quad \mathbb{E}_\theta \left[ \| \Lambda_1 - \Lambda_0 \| \mathbb{I} \left( \|\Lambda_1 - \Lambda_0\| > M_1 \right) \mid \Lambda_0=\Lambda \right] \\
    & \leq \mathbb{E}_\theta \left[ \|X_1/[\iota_\rho(1-\iota_\rho)]\| \mathbb{I} \left( \|X_1/[\iota_\rho(1-\iota_\rho)]\| > M_1 \right) \mid \Lambda_0=\Lambda \right]
    \leq \frac{\Delta}{2} .
\end{align*}

We now proceed to the main body of the proof.
First, we decompose the expression (\ref{eq_target_to_bound}) as follows:
\begin{equation*}
    \mathbb{E}_\theta \left[ e^{\lambda_1 (\|\Lambda_1\| - \|\Lambda_0\|) } \mathbb{I} \left( \|\Lambda_1 - \Lambda_0\| > M_2 \right) \mid \Lambda_0 = \Lambda \right]
\end{equation*}
and
\begin{equation*}
    \mathbb{E}_\theta \left[ e^{\lambda_1 (\|\Lambda_1\| - \|\Lambda_0\|) } \mathbb{I} \left( \|\Lambda_1 - \Lambda_0\| \leq M_2 \right) \mid \Lambda_0 = \Lambda \right]
\end{equation*}
for some $M_2 > 0$.

Regarding the former expectation, let $C = \mathbb{E} e^{\lambda \|X_1\|} < \infty$ for some $\lambda > 0$.
For any $M_2 > 0$ and $\lambda_1 \in (0, \lambda \iota_\rho(1-\iota_\rho))$, it can be shown that
\begin{align*}
    & \quad \mathbb{E}_\theta \left[ e^{\lambda_1 (\|\Lambda_1\| - \|\Lambda_0\|) } \mathbb{I} \left( \|\Lambda_1 - \Lambda_0\| > M_2 \right) \mid \Lambda_0 = \Lambda \right] \\
    & \leq \mathbb{E}_\theta \left[ e^{\lambda_1 (\|\Lambda_1 - \Lambda_0\|) } \mathbb{I} \left( \|\Lambda_1 - \Lambda_0\| > M_2 \right) \mid \Lambda_0 = \Lambda \right] \\
    & \leq \mathbb{E}_\theta \left[ e^{\lambda\iota_\rho(1-\iota_\rho) (\|\Lambda_1 - \Lambda_0\|) } e^{(\lambda_1-\lambda\iota_\rho(1-\iota_\rho)) M_2} \mathbb{I} \left( \|\Lambda_1 - \Lambda_0\| > M_2 \right) \mid \Lambda_0 = \Lambda \right] \\
    & \leq \mathbb{E}_\theta \left[ e^{\lambda\iota_\rho(1-\iota_\rho) (\|\Lambda_1 - \Lambda_0\|) } e^{(\lambda_1-\lambda\iota_\rho(1-\iota_\rho)) M_2} \mid \Lambda_0 = \Lambda \right] \\
    & \leq \mathbb{E}_\theta \left[ e^{\lambda \|X_1\| } e^{(\lambda_1-\lambda\iota_\rho(1-\iota_\rho)) M_2} \mid \Lambda_0 = \Lambda \right] \\
    & \leq C e^{(\lambda_1-\lambda\iota_\rho(1-\iota_\rho)) M_2} .
\end{align*}

For the latter expectation, we linearize the exponential term.
We first control the error between the exponent and its linear approximation under the condition $\|\Lambda_1 - \Lambda_0\| \leq M_2$ such that
\begin{align*}
    \left| e^{\lambda_1 (\|\Lambda_1\| - \|\Lambda_0\|)}
    - [ 1 + \lambda_1 (\|\Lambda_1\| - \|\Lambda_0\|) ] \right|
    & \leq \left[ \lambda_1 (\|\Lambda_1\| - \|\Lambda_0\|) \right]^2
    \leq \left( \lambda_1 \|\Lambda_1 - \Lambda_0\| \right)^2 \\
    & \leq \left( \lambda_1 M_2 \right)^2, 
\end{align*}
provided that $\lambda_1 M_2 \leq \inf_{x \in \mathbb{R}} \{|x| \mid |e^x-1-x|>x^2 \}$.

By treating $\| \Lambda_1 \| - \| \Lambda_0 \| = \|\Lambda_1 - \Lambda_0 + \Lambda_0\| - \| \Lambda_0 \|$ as a function of the perturbation $\Lambda_1 - \Lambda_0$ and the initial state $\Lambda_0$, we apply a first-order approximation via the projection $(\Lambda_1 - \Lambda_0)^T \frac{\Lambda_0}{\|\Lambda_0\|}$.
The resulting approximation error under the condition $\|\Lambda_1 - \Lambda_0\| \leq M_2$ is bounded as follows:
\begin{align*}
    & \quad \left| \|\Lambda_1\| - \|\Lambda_0\| - (\Lambda_1 - \Lambda_0)^T \frac{\Lambda_0}{\|\Lambda_0\|} \right| \\
    &= \left| \frac{2 (\Lambda_1-\Lambda_0)^T \Lambda_0 + \|\Lambda_1-\Lambda_0\|^2}{\|\Lambda_1\| + \|\Lambda_0\|} - (\Lambda_1 - \Lambda_0)^T \frac{\Lambda_0}{\|\Lambda_0\|} \right| \\
    &= \left| (\Lambda_1 - \Lambda_0)^T \Lambda_0 \left( \frac{2}{\|\Lambda_1\| + \|\Lambda_0\|} - \frac{1}{\|\Lambda_0\|} \right) + \frac{\|\Lambda_1-\Lambda_0\|^2}{\|\Lambda_1\| + \|\Lambda_0\|} \right| \\
    & \leq \left| \| \Lambda_1 - \Lambda_0 \| \| \Lambda_0 \| \frac{\|\Lambda_0\| - \|\Lambda_1\|}{(\|\Lambda_1\| + \|\Lambda_0\|) \|\Lambda_0\|} \right| + \left| \frac{\|\Lambda_1-\Lambda_0\|^2}{\|\Lambda_1\| + \|\Lambda_0\|} \right| \\
    & \leq \frac{\| \Lambda_1 - \Lambda_0 \|^2}{\|\Lambda_1\| + \|\Lambda_0\|} + \frac{\|\Lambda_1-\Lambda_0\|^2}{\|\Lambda_1\| + \|\Lambda_0\|} \\
    & \leq \frac{2 M_2^2}{\|\Lambda_0\|} .
\end{align*}

Combining the above two inequalities, we obtain
\begin{align*}
    e^{\lambda_1 (\|\Lambda_1\| - \|\Lambda_0\|)}
    & \leq 1 + \lambda_1 (\|\Lambda_1\| - \|\Lambda_0\|) + \left( \lambda_1 M_2 \right)^2 \\
    & \leq 1 + \lambda_1
    \left[ (\Lambda_1 - \Lambda_0)^T \frac{\Lambda_0}{\|\Lambda_0\|} + \frac{2 M_2^2}{\|\Lambda_0\|} \right]
    + \left( \lambda_1 M_2 \right)^2 ,
\end{align*}
subject to $\|\Lambda_1 - \Lambda_0\| \leq M_2$ and $\lambda_1 M_2 \leq \inf_{x \in \mathbb{R}} \{|x| \mid |e^x-1-x|>x^2 \}$.

Consequently, if $\lambda_1 M_2 \leq \inf_{x \in \mathbb{R}} \{|x| \mid |e^x-1-x|>x^2 \}$, then
\begin{align*}
    & \quad \mathbb{E}_\theta \left[ e^{\lambda_1 (\|\Lambda_1\| - \|\Lambda_0\|) } \mid \Lambda_0 = \Lambda \right] \\
    & \leq C e^{(\lambda_1-\lambda\iota_\rho(1-\iota_\rho)) M_2} + \mathbb{E}_\theta \left[ e^{\lambda_1 (\|\Lambda_1\| - \|\Lambda_0\|) } \mathbb{I} \left( \|\Lambda_1 - \Lambda_0\| \leq M_2 \right) \mid \Lambda_0 = \Lambda \right] \\
    & \leq C e^{(\lambda_1-\lambda\iota_\rho(1-\iota_\rho)) M_2} \\
    & \quad + \mathbb{E}_\theta \left[
    \left[ 1 + \lambda_1(\Lambda_1 - \Lambda_0)^T \frac{\Lambda_0}{\|\Lambda_0\|} + \lambda_1^2 M_2^2 + \frac{2 \lambda_1 M_2^2}{\|\Lambda_0\|} \right]
    \mathbb{I} \left( \|\Lambda_1 - \Lambda_0\| \leq M_2 \right)
    \mid \Lambda_0 = \Lambda \right] \\
    & \leq \lambda_1 \mathbb{E}_\theta \left[ (\Lambda_1 - \Lambda_0)^T \frac{\Lambda_0}{\|\Lambda_0\|} \mid \Lambda_0 = \Lambda \right] + \lambda_1 \mathbb{E}_\theta \left[ \|\Lambda_1 - \Lambda_0\| \mathbb{I} \left( \|\Lambda_1 - \Lambda_0\| > M_2 \right) \mid \Lambda_0 = \Lambda \right] \\
    & \quad + C e^{(\lambda_1-\lambda\iota_\rho(1-\iota_\rho)) M_2} + 1 + \lambda_1^2 M_2^2 + \frac{2 \lambda_1 M_2^2}{\|\Lambda\|} .
\end{align*}
For $\theta \in \Theta$, $\Lambda \in W_\phi$, $\|\Lambda\| \geq M$ and $M_2 \geq M_1$, the expression is further bounded by
\begin{align}
    & \quad \lambda_1 \cdot (-\Delta) + \lambda_1 \frac{\Delta}{2} + C e^{(\lambda_1-\lambda\iota_\rho(1-\iota_\rho)) M_2} + 1 + \lambda_1^2 M_2^2 + \frac{2 \lambda_1 M_2^2}{\|\Lambda\|} \label{eq_target_bound} \\
    &= -\frac{\lambda_1 \Delta}{2} + C e^{(\lambda_1-\lambda\iota_\rho(1-\iota_\rho)) M_2} + 1 + \lambda_1^2 M_2^2 + \frac{2 \lambda_1 M_2^2}{\|\Lambda\|} \notag .
\end{align}

In summary, we have assumed
\begin{equation}
    \label{eq_inequalities_drift_condition_1_rev}
    \begin{cases}
        \lambda_1 M_2 & \leq \inf_{x \in \mathbb{R}} \{|x| \mid |e^x-1-x|>x^2 \} , \\
        M_2 & \geq M_1 , \\
        \|\Lambda\| & \geq M .
    \end{cases}
\end{equation}

Furthermore, suppose
\begin{equation}
    \label{eq_inequalities_drift_condition_2_rev}
    \begin{cases}
        \lambda_1 &= \frac{1}{M_2^3} , \\
        \lambda_1 & \leq \frac{\lambda\iota_\rho(1-\iota_\rho)}{2} , \\
        \|\Lambda\| & \geq M_2^3 , \\
        M_2^3 & \geq M .
    \end{cases}
\end{equation}

Under conditions (\ref{eq_inequalities_drift_condition_1_rev}) and (\ref{eq_inequalities_drift_condition_2_rev}), an upper bound for (\ref{eq_target_bound}) is
\begin{equation*}
    -\frac{\Delta}{2M_2^3} + C e^{-\frac{\lambda\iota_\rho(1-\iota_\rho) M_2}{2}} + 1 + \frac{1}{M_2^4} + \frac{2}{M_2^4}
    = 1 + \frac{1}{M_2^3} \left[ -\frac{\Delta}{2} + C M_2^3 e^{-\frac{\lambda\iota_\rho(1-\iota_\rho) M_2}{2}} + \frac{3}{M_2} \right] .
\end{equation*}

Conditions (\ref{eq_inequalities_drift_condition_1_rev}) and (\ref{eq_inequalities_drift_condition_2_rev}) are equivalent to
\begin{equation}
    M_2 \geq \max \left\{ M^{\frac{1}{3}}, \left(\frac{2}{\lambda\iota_\rho(1-\iota_\rho)}\right)^{\frac{1}{3}}, M_1,
    \left[ \inf_{x \in \mathbb{R}} \{|x| \mid |e^x-1-x|>x^2 \} \right]^{-\frac{1}{2}} \right\} \label{eq_M2_lower_bound}
\end{equation}
with $\lambda_1 = \frac{1}{M_2^3}$ and $\|\Lambda\| \geq M_2^3$.
Since $C M_2^3 e^{-\frac{\lambda\iota_\rho(1-\iota_\rho) M_2}{2}} + \frac{3}{M_2} \rightarrow 0$ as $M_2 \rightarrow +\infty$, there exists a sufficiently large $M_2^*$ satisfying (\ref{eq_M2_lower_bound}) such that for any $M_2 \geq M_2^*$,
\begin{equation*}
    -\frac{\Delta}{2} + C M_2^3 e^{-\frac{\lambda\iota_\rho(1-\iota_\rho) M_2}{2}} + \frac{3}{M_2}
    < -\frac{\Delta}{4} .
\end{equation*}

We have shown that if $M_2 \geq M_2^*$, then for any $\Lambda \in W_\phi$ satisfying $\|\Lambda\| \geq M_2^3$, it holds that
\begin{equation*}
    \mathbb{E}_\theta \left[ e^{\frac{1}{M_2^3} (\|\Lambda_1\| - \|\Lambda_0\|) } \mid \Lambda_0 = \Lambda \right]
    \leq 1 + \frac{1}{M_2^3} \left[ -\frac{\Delta}{2} + C M_2^3 e^{-\frac{\lambda\iota_\rho(1-\iota_\rho) M_2}{2}} + \frac{3}{M_2} \right]
    < 1 .
\end{equation*}
Conversely, for any $\Lambda$ such that $\|\Lambda\| < M_2^3$, it holds that
\begin{equation*}
    \mathbb{E}_\theta \left[ e^{\frac{1}{M_2^3} \|\Lambda_1\| } \mid \Lambda_0 = \Lambda \right]
    \leq \mathbb{E}_\theta \left[ e^{\frac{1}{M_2^3} \|\Lambda_1 - \Lambda_0\| } \mid \Lambda_0 = \Lambda \right]
    e^{\frac{1}{M_2^3} \|\Lambda\|}
    \leq \mathbb{E} \left[ e^{\frac{1}{M_2^3\iota_\rho(1-\iota_\rho)} \|X\| } \right] \cdot e
    \leq C e ,
\end{equation*}
since $\frac{1}{M_2^3} \leq \frac{\lambda\iota_\rho(1-\iota_\rho)}{2} < \lambda$.
These results lead to the inequality
\begin{equation}
    \label{eq_origin_drift_condition_rev}
    \mathbb{E}_\theta \left[ e^{\frac{1}{M_2^3} \|\Lambda_1\|} \mid \Lambda_0 = \Lambda \right]
    \leq \left\{ 1 + \frac{1}{M_2^3} \left[ -\frac{\Delta}{2} + C M_2^3 e^{-\frac{\lambda\iota_\rho(1-\iota_\rho) M_2}{2}} + \frac{3}{M_2} \right] \right\}
    e^{\frac{1}{M_2^3} \|\Lambda\|} + C e .
\end{equation}

By setting $\lambda_1=\frac{1}{(M_2^*)^3}$, $b = C e$, and
\begin{equation*}
    \beta
    = 1 + \frac{1}{(M_2^*)^3} \left[ -\frac{\Delta}{2} + C (M_2^*)^3 e^{-\frac{\lambda\iota_\rho(1-\iota_\rho) M_2^*}{2}} + \frac{3}{M_2^*} \right]
    < 1-\frac{\Delta}{4(M_2^*)^3}
    < 1 ,
\end{equation*}
the drift condition
\begin{equation*}
    \mathbb{E}_\theta \left[ e^{\lambda_1 \|\Lambda_1\|} \mid \Lambda_0 = \Lambda \right] \leq \beta e^{\lambda_1 \|\Lambda\|} + b
\end{equation*}
holds for all $\theta \in \Theta$.
Let $V(\Lambda) = \exp(\lambda_1\|\Lambda\|)$.

Finally, for any $\alpha \in (0,1]$ and $M_2 = \frac{M_2^*}{\alpha^{1/3}}$, the left-hand side of (\ref{eq_origin_drift_condition_rev}) becomes
\begin{equation*}
    \mathbb{E}_\theta \left[ e^{\frac{1}{M_2^3} \|\Lambda_1\|} \mid \Lambda_0 = \Lambda \right]
    = \mathbb{E}_\theta \left[ e^{\lambda_1 \alpha \|\Lambda_1\|} \mid \Lambda_0 = \Lambda \right]
    = \mathbb{E}_\theta \left[ V^\alpha(\Lambda_1) \mid \Lambda_0 = \Lambda \right]
    = (P_\theta V^\alpha)(\Lambda) .
\end{equation*}
The inequality (\ref{eq_origin_drift_condition_rev}) with $M_2 = \frac{M_2^*}{\alpha^{1/3}}$ implies that $P_\theta V^{\alpha} \leq \beta_{\alpha} V^{\alpha} + b_{\alpha}$ for some $\beta_{\alpha} \in (0,1)$ and $b_{\alpha}>0$, where
\begin{align*}
    \beta_{\alpha}
    &= 1 + \frac{\alpha}{(M_2^*)^3} \left[ -\frac{\Delta}{2} + \frac{C (M_2^*)^3 e^{-\frac{\lambda\iota_\rho(1-\iota_\rho) M_2^*}{2\alpha^{1/3}}}}{\alpha} + \frac{3\alpha^{1/3}}{M_2^*} \right] , \\
    b_{\alpha} &= C e .
\end{align*}
Thus, $P_\theta V^\alpha \leq \beta_\alpha V^\alpha + b$ holds with a constant $b_\alpha=b$ and a parameter $\beta_\alpha < 1$, such that $\lim_{\alpha \rightarrow 0} \beta_\alpha = 1$.

\subsection{Proof of Lemma \ref{lemma_for_simultaneous_small_set_condition}}

\label{subsec_proof_lemma_for_simultaneous_small_set_condition}

\subsubsection{Proof Sketch and Lemma}

Before stating the proof sketch and the lemma needed for the proof, we introduce the notion of the pushforward of a measure.
Given a measure $\nu$ on the $\mathbb{R}^p$, the pushforward of $\nu$ by $\Phi$, denoted as $\Phi^* \nu$, is defined as the measure on $\mathbb{R}^d$ given by
\begin{equation*}
    \Phi^*\nu(A) := \nu\left( \Phi^{-1}(A) \right) .
\end{equation*}

By the compactness of $\Theta$, it suffices to work locally on an arbitrary neighborhood $B_{\theta_*} \subset \Theta$.
Denote $X^* = (x^{(1)}, \dots, x^{(s)}, x^{(1)}, \dots, x^{(s)})$, and let $S$ be any sufficiently small neighborhood of $X^*$.
Under Assumption \ref{assumption_for_simultaneous_small_set_condition}, the Jacobian
$D\Phi_\theta(x^*)$ is uniformly nondegenerate on $B_{\theta_*}$, i.e.,
its smallest singular value is bounded below by a positive constant.
After an orthogonal change of coordinates, $\Phi_\theta$ is locally Lipschitz with full-rank derivative at $x^*$, and Lemma \ref{lemma_density_lowerbound} yields the uniform density lower bound
\begin{equation*}
    \Phi_\theta^* \mu_{\mathrm{leb}, S}
    \geq
    \epsilon^\prime \mu_{\mathrm{leb}, V^*} ,
\end{equation*}
for some ball $V^*$ centered at $\Phi_\theta(x^*)$ and constant $\epsilon^\prime > 0$.
It is established rigorously in \eqref{eq_proof_lemma_for_simultaneous_small_set_condition_Phi_density}, and Lemma \ref{lemma_density_lowerbound} will be proved in Subsection \ref{subsec_proof_lemma_density_lowerbound}.
This provides a nontrivial absolutely continuous component in the distribution of the aggregated update map $\Phi_\theta$.

We then transfer this density lower bound to the transition kernel $P_\theta$.
For each treatment assignment and any state $\Lambda$, there exists a
centered random increment $\Gamma_{\mathrm{Bin}, i}$, whose distribution corresponds to the case where the allocation probability exactly equals the targeted allocation ratio, and which is uniformly dominated by the increment of $\Lambda$ under the transition kernel $P_\theta$.
By the convolution and the inequality above, the distribution $\Gamma_{\mathrm{Bin}, 1}^{2*} * \cdots * \Gamma_{\mathrm{Bin}, s}^{2*}$ dominates a multiple of a Lebesgue measure on some ball $V^*$.
Thus, introducing some auxiliary notations and constants, this convolution can be written as the sum of a mean-zero drift term and an i.i.d. continuous perturbations.
A second-moment argument shows that, for sufficiently large $n$, the cumulative drift $\sum_{i=1}^n Z_1^{(i)}$ remains of order $o_P(n)$, whereas the support of the perturbations $\sum_{k=1}^{T_n} Z_3^{(k)}$ grows linearly in $n$ and retains positive Lebesgue density on its interior.
This implies that the resulting centered random walk, with increments 
$\Gamma_{\mathrm{Bin}, 1}^{2*} * \cdots * \Gamma_{\mathrm{Bin}, s}^{2*}$, covers a centered Euclidean ball of radius proportional to $n$ with uniformly positive probability.
Because of the connection between $P_\theta$ and $\Gamma_{\mathrm{Bin}, i}$, it yields that, for all $\theta \in B_{\theta_*}$ and all $n$ large enough,
\begin{equation*}
    P_\theta^{2ns}(\Lambda, \cdot)
    \geq
    c_n \mu_{\mathrm{leb}} \left( B(\Lambda, c^\prime n) \right) ,
\end{equation*}
which establishes a simultaneous small set condition.

\begin{lemma}
    \label{lemma_density_lowerbound}

    Suppose the function $\Phi: \mathbb{R}^p \to \mathbb{R}^d$ and $p > d$.
    If the function $\Phi$ is differentiable on the ball $B(x^*,r_x)$, the continuous differentiation $D\Phi = \left( D_1\Phi, D_2\Phi \right)$ is with $d \times (p-d)$ matrix $D_1\Phi$ and $d \times d$ matrix $D_2\Phi$, where the latter is full rank at $x^* = (x_1^*, x_2^*)$, where $x_1^*$ is $(p-d)$-dimensional and $x_2^*$ is $d$-dimensional.

    Let $\mu_{\mathrm{leb}, B(x^*,r_x)}$ denote the measure on $\mathbb{R}^d$ defined by
    \begin{equation*}
        \mu_{\mathrm{leb}, B(x^*,r_x)}(A) := \mu_{\mathrm{leb}} ( A \cap B(x^*,r_x) ) ,
        \quad A \subset \mathbb{R}^d .
    \end{equation*}
    Let $L_\Phi > 0$ be the Lipschitz constant of $\Phi$ and $L_{D\Phi}$ be the Lipschitz constant of $D\Phi$ on the ball $B(x^*,r_x)$.

    Then $\Phi^* \mu_{\mathrm{leb}, B(x^*,r_x)} \geq \epsilon \mu_{\mathrm{leb}, V^*}$, where
    \begin{equation*}
        V^* = B \left( \Phi(x^*), C_{\mathrm{dl}} \right) ,
        \quad
        \epsilon = \frac{\mu_{\mathrm{leb}}(B^{(d)}(0, 1)) C_{\mathrm{dl}}^d}{L_\Phi^p} ,
    \end{equation*}
    and
    \begin{equation*}
        C_{\mathrm{dl}}
        = \frac{\min \left\{ \left[ 2 L_{D\Phi} \| \{ D_2\Phi(x^*) \}^{-1} \| \right]^{-1} , r_x \right\}}
        {128 L_\Phi \| \{ D_2\Phi(x^*) \}^{-1} \|^2} .
    \end{equation*}
\end{lemma}

\subsubsection{Formal Proof}

Due to the compactness assumed in Assumption \ref{assumption_theta_compact}, it suffices to consider each neighborhood $B_{\theta_*}$ around a possible parameter value $\theta_* \in \Theta$.

Denote
\begin{equation*}
    S = B^{(d_x)}(x^{(1)}, r_x) \times \dots \times B^{(d_x)}(x^{(s)}, r_x)
    \times B^{(d_x)}(x^{(1)}, r_x) \times \dots \times B^{(d_x)}(x^{(s)}, r_x) .
\end{equation*}
We first show that $\Phi_\theta^* \mu_{\mathrm{leb}, S} \geq \epsilon^\prime \mu_{\mathrm{leb}, V^*}$ for some ball $V^*$ and some constant $\epsilon^\prime$.

Suppose $x^* = (x^{(1)}, \dots, x^{(s)}, x^{(1)}, \dots, x^{(s)})$, which is a $2sd_x$-dimensional vector.
Then $B^{(2sd_x)}(x^*, r_x) \subset S$.
For any $\theta \in B_{\theta_*}$, denote the $d \times (2sd_x)$ matrix $M_\theta = D\Phi_\theta(x^*)$.
Then Assumption \ref{assumption_for_simultaneous_small_set_condition} ensures that $\sigma_d(M_\theta) \geq C_\sigma$.
By SVD decomposition and Item \ref{item_assumption_for_simultaneous_small_set_condition_3} of Assumption \ref{assumption_for_simultaneous_small_set_condition}, it holds that
\begin{equation*}
    M_\theta = U_\theta \Sigma_\theta V_\theta^T ,
\end{equation*}
where $U_\theta \in \mathbb{R}^{d \times d}$ and $V_\theta \in \mathbb{R}^{(2sd_x) \times (2sd_x)}$ are orthogonal matrices, and $\Sigma_\theta \in \mathbb{R}^{m\times n}$ is a diagonal (rectangular) matrix of the form
\begin{equation*}
    \Sigma_\theta =
    \begin{bmatrix}
        0 & \cdots & 0 & \sigma_1(M_\theta) & 0 & \cdots & 0 \\
        0 & \cdots & 0 & 0 & \sigma_2(M_\theta) & \cdots & 0 \\
        0 & \cdots & 0 & \vdots & \vdots & \ddots & \vdots \\
        0 & \cdots & 0 & 0 & 0 & \cdots & \sigma_d(M_\theta) \\
    \end{bmatrix} ,
\end{equation*}
with singular values ordered as
\begin{equation*}
    \sigma_1(M_\theta) \geq \sigma_2(M_\theta) \geq \cdots \geq \sigma_d(M_\theta) \geq C_\sigma .
\end{equation*}

Let $\tilde{\Phi}_\theta(x) = \Phi_\theta (V_\theta x)$.
Then we can obtain that
\begin{equation*}
    D\tilde{\Phi}_\theta (V_\theta^T x^*) = D\Phi_\theta(x^*) V_\theta = U_\theta \Sigma_\theta .
\end{equation*}
Thus, by the form of $\Sigma_\theta$, following the statement of Lemma \ref{lemma_density_lowerbound}, it holds that $D_2 \tilde{\Phi}_\theta$ is full rank at $V_\theta^T x^*$.
Moreover, $\tilde{\Phi}_\theta$ is differentiable on the ball $B^{(2sd_x)}(V_\theta^T x^*,r_x)$.

Then the Lipschitz constant in Item \ref{item_assumption_for_simultaneous_small_set_condition_2} of Assumption \ref{assumption_for_simultaneous_small_set_condition} implies that $\Phi_\theta$ is $\sqrt{2s}L$-Lipschitz continuous on $S$
Thus, by Lemma \ref{lemma_density_lowerbound} and the equality
\begin{equation*}
    \| \{ D_2 \tilde{\Phi}_\theta(V_\theta^T x^*) \}^{-1} \|
    = \| \{ U_\theta \diag \left( \sigma_1(M_\theta), \sigma_2(M_\theta), \dots, \sigma_d(M_\theta) \right) \}^{-1} \|
    = \left[ \sigma_d(M_\theta) \right]^{-1}
    \leq C_\sigma^{-1} ,
\end{equation*}
we can conclude that
\begin{equation}
    \Phi_\theta^* \mu_{\mathrm{leb}, S}
    \geq \Phi_\theta^* \mu_{\mathrm{leb}, B^{(2sd_x)}(x^*, r_x)}
    = \tilde{\Phi}_\theta^* \mu_{\mathrm{leb}, B^{(2sd_x)}(V_\theta^T x^*, r_x)}
    \geq \epsilon^\prime \mu_{\mathrm{leb}, V^*} \label{eq_proof_lemma_for_simultaneous_small_set_condition_Phi_density} ,
\end{equation}
where
\begin{equation*}
    V^* = B \left( \Phi_\theta(x^*), C_{\mathrm{dl}} \right) ,
    \quad
    \epsilon^\prime = \frac{\mu_{\mathrm{leb}}(B^{(d)}(0, 1)) C_{\mathrm{dl}}^d}{(\sqrt{2s}L)^{2s d_x}} ,
\end{equation*}
and
\begin{equation*}
    C_{\mathrm{dl}}
    = \frac{\min \left\{ \left[ 2 \sqrt{2s}L C_\sigma^{-1} \right]^{-1} , r_x \right\}}
    {128 \sqrt{2s}L C_\sigma^{-2}} .
\end{equation*}

For any $T_i \in \{ 0,1 \}$, any $i \in \{ 1, \dots, 2s \}$, let the function
\begin{equation*}
    \Phi_{\theta, (T_i)_{i=1:(2s)}}
    = \sum_{i=1}^{2s} \frac{\left[ T_i-\rho_\theta(x_i) \right] \phi(x_i)}{\rho_\theta(x_i)\left[ 1-\rho_\theta(x_i) \right]} .
\end{equation*}
When $T_i = 0$ for $i \in \{1, \dots, s\}$ and $T_i = 1$ for $i \in \{s+1, \dots, 2s\}$, $\Phi_{\theta, (T_i)_{i=1:(2s)}} = \Phi_\theta$.

Define
\begin{equation*}
    \Gamma_{\mathrm{leb}, i}
    = \frac{\mu_{\mathrm{leb}, B^{(d_x)}(x^{(i)}, r_x)}}
    {\mu_{\mathrm{leb}}(B^{(d_x)}(0, r_x))}
\end{equation*}
as the uniform probability distribution over the ball $B(x^{(i)}, r_x)$.

Let $X \sim \Gamma_{\mathrm{leb}, i}$. Conditionally on $X$, the random variable $\Delta$ is distributed as  
\begin{equation*}
    P (\Delta \mid X) =
    \begin{cases}
        \phi(X)/\rho_\theta(X), & \text{with probability } \rho_\theta(X), \\[6pt]
        -\phi(X)/(1-\rho_\theta(X)), & \text{with probability } 1-\rho_\theta(X) .
    \end{cases}
\end{equation*}
Denote the distribution of $\Delta$ in this case by $\Gamma_{\mathrm{Bin}, i}$.

It then follows that the first moment of $\Gamma_{\mathrm{Bin}, i}$ is $0$,
\begin{equation*}
    \Gamma_{\mathrm{Bin}, i}
    \geq \iota \{-\phi/(1-\rho_\theta)\}^* \Gamma_{\mathrm{leb}, i}
\end{equation*}
and
\begin{equation*}
    \Gamma_{\mathrm{Bin}, i}
    \geq \iota \{\phi/\rho_\theta\}^* \Gamma_{\mathrm{leb}, i}
\end{equation*}
by $\rho_\theta \in [\iota_\rho, 1-\iota_\rho]$ in Assumption \ref{assumption_sampling}.

Moreover, from the definition of $\Phi_{\theta, (T_i)_{i=1:(2s)}}$, we have
\begin{align*}
    & \quad \Phi_{\theta, (T_i)_{i=1:(2s)}}^*
    (\Gamma_{\mathrm{leb}, 1} \otimes \dots \otimes \Gamma_{\mathrm{leb}, s}
    \otimes \Gamma_{\mathrm{leb}, 1} \otimes \dots \otimes \Gamma_{\mathrm{leb}, s}) \\
    &= \{ [(T_1-\rho_\theta) \phi / (\rho_\theta(1-\rho_\theta))]^* \Gamma_{\mathrm{leb}, 1} \} * \dots * \{ [(T_s-\rho_\theta) \phi / (\rho_\theta(1-\rho_\theta))]^* \Gamma_{\mathrm{leb}, s} \} \\
    & \quad * \{ [(T_{s+1}-\rho_\theta) \phi / (\rho_\theta(1-\rho_\theta))]^* \Gamma_{\mathrm{leb}, 1} \} * \dots * \{ [(T_{2s}-\rho_\theta) \phi / (\rho_\theta(1-\rho_\theta))]^* \Gamma_{\mathrm{leb}, s} \} ,
\end{align*}
where $*$ denotes the coupling of measures on $\mathbb{R}^d$, which also corresponds to the distribution of the sum of independent random variables associated with the respective probability distributions.

Therefore,
\begin{align*}
    \Gamma_{\mathrm{Bin}, 1}^{2*} * \dots * \Gamma_{\mathrm{Bin}, s}^{2*}
    & \geq \iota^{2s}
    \Phi_{\theta, (T_i)_{i=1:(2s)}}^*
    (\Gamma_{\mathrm{leb}, 1} \otimes \dots \otimes \Gamma_{\mathrm{leb}, s}
    \otimes \Gamma_{\mathrm{leb}, 1} \otimes \dots \otimes \Gamma_{\mathrm{leb}, s}) \\
    &= \iota^{2s} \mu_{\mathrm{leb}}(B^{(d_x)}(0, r_x))^{-2s} \Phi_{\theta, (T_i)_{i=1:(2s)}}^* \mu_{\mathrm{leb}, S} ,
\end{align*}
for any $T_i \in \{0, 1\}$.

Furthermore, from \eqref{eq_proof_lemma_for_simultaneous_small_set_condition_Phi_density}, we have
\begin{equation*}
    \Gamma_{\mathrm{Bin}, 1}^{2*} * \dots * \Gamma_{\mathrm{Bin}, s}^{2*}
    \geq \iota^{2s} \mu_{\mathrm{leb}}(B^{(d_x)}(0, r_x))^{-2s} \Phi_\theta^* \mu_{\mathrm{leb}, S}
    \geq \iota^{2s} \mu_{\mathrm{leb}}(B^{(d_x)}(0, r_x))^{-2s} \epsilon^\prime \mu_{\mathrm{leb}, V^*} .
\end{equation*}

Let $C_{\mathrm{gap}} = \iota^{2s} \mu_{\mathrm{leb}}(B^{(d_x)}(0, r_x))^{-2s} \epsilon^\prime$, and define the distribution
\begin{equation*}
    \Gamma^{(\mathrm{gap})}
    = \left[ \Gamma_{\mathrm{Bin}, 1}^{2*} * \dots * \Gamma_{\mathrm{Bin}, s}^{2*}
    - C_{\mathrm{gap}} \mu_{\mathrm{leb}, V^*} \right] \otimes \delta_0
    + \left[ C_{\mathrm{gap}} \mu_{\mathrm{leb}}(V^*) \delta_{\Phi_\theta(x^*)} \right] \otimes \delta_1 ,
\end{equation*}
where $\delta_{\Phi_\theta(x^*)}$ denotes the Dirac measure at $\Phi_\theta(x^*)$, and $\otimes$ represents the product measure.

Since the center of the ball $V^*$ is $\Phi_\theta(x^*)$, the first $d$ components of $\Gamma^{(\mathrm{gap})}$ have first moments equal to those of $\Gamma_{\mathrm{Bin}, 1}^{2*} * \dots * \Gamma_{\mathrm{Bin}, s}^{2*}$, which in turn are the sums of the first moments of $\Gamma_{\mathrm{Bin}, 1}^{2*}$, and hence equal to $0$.
The last component of $\Gamma^{(\mathrm{gap})}$ follows a $\mathrm{Bin}(1, C_{\mathrm{gap}} \mu_{\mathrm{leb}}(V^*))$ distribution, with expected value $C_{\mathrm{gap}} \mu_{\mathrm{leb}}(V^*)$.

Let $\{ Z^{(i)} = (Z_1^{(i)}, Z_2^{(i)}) \}_{i \in \mathbb{N}^*}$ be an i.i.d. sequence of random vectors with distribution $\Gamma^{(\mathrm{gap})}$, where $Z_1^{(i)} \in \mathbb{R}^d$ and $Z_2^{(i)} \in \mathbb{R}$.
Next, let $\{ Z_3^{(i)} \}_{i \in \mathbb{N}^*}$ be an i.i.d. sequence of random vectors, independent of $\{ Z^{(i)} \}$, with uniform distribution over the ball $B^{(d)}(0, C_{\mathrm{dl}})$.
Then
\begin{equation}
    Z_1^{(i)} + Z_2^{(i)}Z_3^{(i)} \sim \Gamma_{\mathrm{Bin}, 1}^{2*} * \dots * \Gamma_{\mathrm{Bin}, s}^{2*} \label{eq_proof_lemma_for_simultaneous_small_set_condition_equivalent_1} .
\end{equation}
Therefore, by independence,
\begin{equation}
    \sum_{i=1}^n Z_1^{(i)}
    + \sum_{k=1}^{\sum_{i=1}^n Z_2^{(i)}} Z_3^{(k)}
    \stackrel{d}{=}
    \sum_{i=1}^n \left[ Z_1^{(i)} + Z_2^{(i)}Z_3^{(i)} \right]
    \sim \left\{ \Gamma_{\mathrm{Bin}, 1}^{2*} * \dots * \Gamma_{\mathrm{Bin}, s}^{2*} \right\}^{n*} \label{eq_proof_lemma_for_simultaneous_small_set_condition_equivalent_2} .
\end{equation}
Moreover, from $P_\theta(\Lambda, \cdot) \geq \Lambda + \iota \epsilon \mu_{\mathrm{leb}}(B^{(d_x)}(0, r_x)) \Gamma_{\mathrm{Bin}, i}$ for any $i \in \{1, \dots, s\}$, where ``$+$'' denotes the measure addition, the relationship between $P_\theta$ and $\Gamma_{\mathrm{Bin}, 1}^{2*} * \dots * \Gamma_{\mathrm{Bin}, s}^{2*}$ also satisfies
\begin{equation}
    P_\theta^{2ns}(\Lambda, \cdot) - \Lambda
    \geq \left\{ \iota \epsilon \mu_{\mathrm{leb}}(B^{(d_x)}(0, r_x)) \right\}^{2ns}
    \left\{ \Gamma_{\mathrm{Bin}, 1}^{2*} * \dots * \Gamma_{\mathrm{Bin}, s}^{2*} \right\}^{n} \label{eq_proof_lemma_for_simultaneous_small_set_condition_P_random_walk} .
\end{equation}

Hence, it suffices to control a lower bound of $\left\{ \Gamma_{\mathrm{Bin}, 1}^{2*} * \dots * \Gamma_{\mathrm{Bin}, s}^{2*} \right\}^{n}$.
Since the random variable $\sum_{i=1}^n Z_1^{(i)} + \sum_{k=1}^{\sum_{i=1}^n Z_2^{(i)}} Z_3^{(k)}$ follows this distribution, it is enough to focus on this random variable.

Denote the indicator function $\chi_m(x) = \mathbb{I} \left( x \in B^{(d)}(0,(m+1)C_{\mathrm{dl}}/2) \right)$, then
\begin{equation*}
    (\chi_{m_1} * \chi_{m_2})(\Lambda)
    \geq \mu_{\mathrm{leb}}(B^{(d)}(0,C_{\mathrm{dl}}/4)) \chi_{m_1+m_2}(\Lambda) .
\end{equation*}
A detailed proof can be found in the proof of Lemma C.1 in \cite{fangGeneralNonMarkovianFramework2026}.
Thus, the distribution of $\sum_{k=1}^K Z_3^{(k)}$ is $\frac{\mu_{\mathrm{leb}, B^{(d)}(0, C_{\mathrm{dl}})}^{K*}}{\mu_{\mathrm{leb}}^K(B^{(d)}(0, C_{\mathrm{dl}}))}$, and
\begin{align}
    \frac{\mu_{\mathrm{leb}, B^{(d)}(0, C_{\mathrm{dl}})}^{K*}}{\mu_{\mathrm{leb}}^K(B^{(d)}(0, C_{\mathrm{dl}}))}
    & \geq \frac{\mu_{\mathrm{leb}}^{K-1}(B^{(d)}(0, C_{\mathrm{dl}}/4)) \mu_{\mathrm{leb}, B^{(d)}(0, (K+1)C_{\mathrm{dl}}/2)}}{\mu_{\mathrm{leb}}^K(B^{(d)}(0, C_{\mathrm{dl}}))} \label{eq_proof_lemma_for_simultaneous_small_set_condition_n_step} \\
    &= \frac{\mu_{\mathrm{leb}, B^{(d)}(0, (K+1)C_{\mathrm{dl}}/2)}}{4^{(K-1)d} \mu_{\mathrm{leb}}(B^{(d)}(0, C_{\mathrm{dl}}))} \notag .
\end{align}

Denote
\begin{equation*}
    p = C_{\mathrm{gap}} \mu_{\mathrm{leb}}(V^*),
    \quad
    S_n = \sum_{i=1}^n Z_1^{(i)},
    \quad
    T_n = \sum_{i=1}^n Z_2^{(i)} .
\end{equation*}
Let $S_{n,j}$ and $Z_{1,j}$ denote the $j$-th components of $S_n$ and $Z_1$, respectively.
We aim to find an $n$ such that
\begin{equation}
    P \left( |T_n - pn| > \frac{pn}{3} \right)
    + \sum_{j=1}^d P \left( |S_{n,j}| > \frac{pn}{6\sqrt{d}}C_{\mathrm{dl}} \right)
    \leq \frac{1}{2} \label{eq_proof_lemma_for_simultaneous_small_set_condition_n_condition_1}.
\end{equation}
Moreover,
\begin{equation}
    P \left( |T_n - pn| > \frac{pn}{3} \right)
    \leq \mathbb{E} \left[ (T_n - pn)^2 \right] \frac{9}{p^2n^2}
    = \frac{9(1-p)}{np} \label{eq_proof_lemma_for_simultaneous_small_set_condition_n_condition_2} ,
\end{equation}
\begin{equation}
    P \left( |S_{n,j}| > \frac{pn}{6\sqrt{d}}C_{\mathrm{dl}} \right)
    \leq \mathbb{E} \left[ S_{n,j}^2 \right] \frac{36d}{p^2n^2 C_{\mathrm{dl}}^2}
    \leq \frac{36d \mathbb{E} \left[ Z_{1,j}^2 \right] }{p^2n C_{\mathrm{dl}}^2} \label{eq_proof_lemma_for_simultaneous_small_set_condition_n_condition_3} ,
\end{equation}
and
\begin{align*}
    | Z_{1,j} |
    & \leq 2 \sum_{i=1}^s \max\{\|\phi/\rho_\theta\|_{L^\infty(\Gamma_{\mathrm{leb}, i})}, \|\phi/(1-\rho_\theta)\|_{L^\infty(\Gamma_{\mathrm{leb}, i})}\} \\
    & \leq 2 \sum_{i=1}^s \left[ \| \phi(x^{(i)}) \|_{L^\infty(\Gamma_{\mathrm{leb}, i})}/\iota + L r_x \right] ,
\end{align*}
almost surely, where $Z_{1,j}$ denotes the $j$th component of the random vector $Z_1$, and $Z_1$ follows the first $d$ components of the distribution $\Gamma^{(\mathrm{gap})}$.
The last inequality is by Item \ref{item_assumption_for_simultaneous_small_set_condition_2} in Assumption \ref{assumption_for_simultaneous_small_set_condition}.

Therefore, from \eqref{eq_proof_lemma_for_simultaneous_small_set_condition_n_condition_1}, \eqref{eq_proof_lemma_for_simultaneous_small_set_condition_n_condition_2}, and \eqref{eq_proof_lemma_for_simultaneous_small_set_condition_n_condition_3}, it suffices to find an $n$ such that
\begin{equation*}
    \frac{9(1-p)}{np} + \frac{36d \left\{ 2 \sum_{i=1}^s \left[ \| \phi(x^{(i)}) \|_\infty/\iota + L r_x \right] \right\}^2 }{p^2n C_{\mathrm{dl}}^2}
    \leq \frac{1}{2} .
\end{equation*}
Thus, we set
\begin{equation*}
    n_l
    = \frac{18(1-p)}{p}
    + \frac{72d \left\{ 2 \sum_{i=1}^s \left[ \| \phi(x^{(i)}) \|_\infty/\iota + L r_x \right] \right\}^2 }{p^2 C_{\mathrm{dl}}^2}
    + 1 .
\end{equation*}

For $n \geq n_l$, the intersection of the events
\begin{equation*}
    \| S_n \|
    = \left\| \sum_{i=1}^n Z_1^{(i)} \right\|
    \leq \frac{pn}{6}C_{\mathrm{dl}}
\end{equation*}
and
\begin{equation*}
    T_n = \sum_{i=1}^n Z_2^{(i)} \geq \frac{2pn}{3}
\end{equation*}
occurs with probability greater than $\frac{1}{2}$.
For any $A \subset B^{(d)} \left( 0, \frac{pnC_{\mathrm{dl}}}{3} \right)$, $S \in \mathbb{R}^d$ with $\| S \| \leq \frac{pn}{6}C_{\mathrm{dl}}$ and $T \in [\frac{2pn}{3}, \infty)$, by \eqref{eq_proof_lemma_for_simultaneous_small_set_condition_n_step},
\begin{align*}
    & \quad P \left(
        \sum_{k=1}^{T_n} Z_3^{(k)} \in A \mid S_n = S, T_n = T
    \right)
    = \frac{\left\{ \mu_{\mathrm{leb}, B^{(d)}(0, C_{\mathrm{dl}})}^{T_n*} \right\}(A)}{\mu_{\mathrm{leb}}^{T_n}(B^{(d)}(0, C_{\mathrm{dl}}))} \\
    & \geq \frac{\mu_{\mathrm{leb}, B^{(d)}(0, (T_n+1)C_{\mathrm{dl}}/2)}(A)}{4^{(T_n-1)d} \mu_{\mathrm{leb}}(B^{(d)}(0, C_{\mathrm{dl}}))}
    \geq \frac{\mu_{\mathrm{leb}, B^{(d)}(0, pnC_{\mathrm{dl}}/3)}(A)}{4^{(n-1)d} \mu_{\mathrm{leb}}(B^{(d)}(0, C_{\mathrm{dl}}))} .
\end{align*}
Let $A - S_n := \{\, x - y \mid x \in A, y \in S_n \}$ denotes the translation of the set $A$ by $-S_n$.
Then for any $A \subset B^{(d)} \left( 0, \frac{pnC_{\mathrm{dl}}}{6} \right)$, $S \in \mathbb{R}^d$, by $A - S_n \subset B^{(d)} \left( 0, \frac{pnC_{\mathrm{dl}}}{3} \right)$, we have
\begin{align*}
    & \quad P \left(
        \sum_{i=1}^n Z_1^{(i)} + \sum_{k=1}^{\sum_{i=1}^n Z_2^{(i)}} Z_3^{(k)} \in A \mid S_n = S, T_n = T
    \right) \\
    &= P \left(
        \sum_{k=1}^{T_n} Z_3^{(k)} \in A - S_n \mid S_n = S, T_n = T
    \right)
    \geq \frac{\mu_{\mathrm{leb}, B^{(d)}(0, pnC_{\mathrm{dl}}/3)}(A)}{4^{(n-1)d} \mu_{\mathrm{leb}}(B^{(d)}(0, C_{\mathrm{dl}}))} .
\end{align*}

Thus, we can conclude that for any $n \geq n_l$, for any $A \subset B^{(d)} \left( 0, \frac{pnC_{\mathrm{dl}}}{6} \right)$
\begin{align*}
    & \quad P \left(
        \sum_{i=1}^n Z_1^{(i)} + \sum_{k=1}^{\sum_{i=1}^n Z_2^{(i)}} Z_3^{(k)} \in A
    \right) \\
    & \geq \int
    \mathbb{I} \left( \| S \| \leq \frac{pn}{6}C_{\mathrm{dl}} \right)
    \mathbb{I} \left( T \in [\frac{2pn}{3}, \infty) \right)
    P \left(
        \sum_{i=1}^n Z_1^{(i)} + \sum_{k=1}^{\sum_{i=1}^n Z_2^{(i)}} Z_3^{(k)} \in A \mid S_n = S, T_n = T
    \right) \\
    & \geq \int
    \mathbb{I} \left( \forall j \in \{1, \dots, d\}, |S_{n,j}| \leq \frac{pn}{6\sqrt{d}}C_{\mathrm{dl}} \right)
    \mathbb{I} \left( T \in [\frac{2pn}{3}, \infty) \right)
    \frac{\mu_{\mathrm{leb}, B^{(d)}(0, pnC_{\mathrm{dl}}/3)}(A)}{4^{(n-1)d} \mu_{\mathrm{leb}}(B^{(d)}(0, C_{\mathrm{dl}}))} \\
    & \geq \frac{\mu_{\mathrm{leb}, B^{(d)}(0, pnC_{\mathrm{dl}}/3)}(A)}{2^{2(n-1)d+1} \mu_{\mathrm{leb}}(B^{(d)}(0, C_{\mathrm{dl}}))} .
\end{align*}
Because of \eqref{eq_proof_lemma_for_simultaneous_small_set_condition_equivalent_1} and \eqref{eq_proof_lemma_for_simultaneous_small_set_condition_equivalent_2}, the inequality above implies that
\begin{equation*}
    \left\{ \Gamma_{\mathrm{Bin}, 1}^{2*} * \dots * \Gamma_{\mathrm{Bin}, s}^{2*} \right\}^{n*}
    \geq \frac{\mu_{\mathrm{leb}, B^{(d)}(0, pnC_{\mathrm{dl}}/3)}}{2^{2(n-1)d+1} \mu_{\mathrm{leb}}(B^{(d)}(0, C_{\mathrm{dl}}))} .
\end{equation*}
Combining \eqref{eq_proof_lemma_for_simultaneous_small_set_condition_P_random_walk}, we obtain that, for all $n \geq n_l$,
\begin{equation*}
    P_\theta^{2ns}(\Lambda, \cdot)
    \geq \left\{ \iota \epsilon \mu_{\mathrm{leb}}(B^{(d_x)}(0, r_x)) \right\}^{2ns}
    \frac{\mu_{\mathrm{leb}, B^{(d)}(\Lambda, pnC_{\mathrm{dl}}/3)}}{2^{2(n-1)d+1} \mu_{\mathrm{leb}}(B^{(d)}(0, C_{\mathrm{dl}}))} .
\end{equation*}

\subsubsection{Proof of Lemma \ref{lemma_density_lowerbound}}

\label{subsec_proof_lemma_density_lowerbound}

Let $\|\cdot\|$ be the $\ell_2$ norm of the linear map and the vector.
In the proof, we will use both $\Phi(x^*)$ and $\Phi(x_1^*, x_2^*)$ interchangeably.
We first show that there exist neighborhoods $W^*$, $U^*$, and $V^*$ of the points $x_1^*$, $x_2^*$, and $\Phi(x^*)$, respectively, such that
\begin{enumerate}
    \item $V^* \subset \Phi_{x_1}(U^*)$ and $\Phi_{x_1}$ is injective on the domain $U^*$ for any $x_1 \in W^*$,
    \item The inverse function $\Psi_{x_1}$ of the function $\Phi_{x_1}$ restricted to $U^*$ is Lipschitz continuous on the domain $V^*$.
\end{enumerate}

Before the formal proof, we first demonstrate that the inverse of the differential $D_2 \Phi(x^*)$ is bounded within a certain neighborhood.
First, assume that
\begin{equation*}
    \delta = \frac{1}{2\|\{ D_2\Phi(x^*) \}^{-1}\|}
    \leq \frac{\| D_2\Phi(x^*) \|}{2}
    \leq \frac{\| D\Phi(x^*) \|}{2}
    \leq \frac{L_\Phi}{2}.
\end{equation*}
Let the constant $r = \min \left\{ L_{D\Phi}^{-1} \delta, r_x \right\}$.
As a consequence of the definition of $r$, we have
\begin{equation*}
    \| D_2\Phi(x) - D_2\Phi(x^*) \|
    \leq \| D\Phi(x) - D\Phi(x^*) \|
    \leq L_{D\Phi} \| x - x^* \|
    < L_{D\Phi} r
    \leq \delta .
\end{equation*}
Thus,
\begin{align*}
    \|\{ D_2\Phi(x^*) \}^{-1} \{ D_2\Phi(x) - D_2\Phi(x^*) \}\|
    & \leq \|\{ D_2\Phi(x^*) \}^{-1}\| \|D_2\Phi(x) - D_2\Phi(x^*)\| \\
    & \leq \|\{ D_2\Phi(x^*) \}^{-1}\| \delta
    = \frac{1}{2} < 1 .
\end{align*}
Furthermore, by $\|\{ D_2\Phi(x^*) \}^{-1} \{ D_2\Phi(x) - D_2\Phi(x^*) \}\| \leq \frac{1}{2} < 1$, we have
\begin{equation*}
    \{ I + \{ D_2\Phi(x^*) \}^{-1} \{ D_2\Phi(x) - D_2\Phi(x^*) \} \}^{-1}
    = \sum_{j=0}^\infty \left[ -\{ D_2\Phi(x^*) \}^{-1} \{ D_2\Phi(x) - D_2\Phi(x^*) \} \right]^j
\end{equation*}
Thus,
\begin{align*}
    & \quad \| \{ I + \{ D_2\Phi(x^*) \}^{-1} \{ D_2\Phi(x) - D_2\Phi(x^*) \} \}^{-1} \| \\
    & \leq \sum_{j=0}^\infty \| -\{ D_2\Phi(x^*) \}^{-1} \{ D_2\Phi(x) - D_2\Phi(x^*) \} \|^j
    \leq 2 .
\end{align*}
In conclusion, for any $x \in B(x^*, r)$,
\begin{align}
    \|\{ D_2\Phi(x) \}^{-1}\|
    &= \|\{ I + \{ D_2\Phi(x^*) \}^{-1} \{ D_2\Phi(x) - D_2\Phi(x^*) \} \}^{-1}
    \{ D_2\Phi(x^*) \}^{-1}\| \label{eq_proof_lemma_density_1} \\
    & \leq \| \{ I + \{ D_2\Phi(x^*) \}^{-1} \{ D_2\Phi(x) - D_2\Phi(x^*) \} \}^{-1} \|
    \| \{ D_2\Phi(x^*) \}^{-1} \| \notag \\
    & \leq 2\|\{ D_2\Phi(x^*) \}^{-1}\| \notag .
\end{align}

We next provide explicit constructions of the neighborhoods $W^*$, $V^*$, and $U^*$.
Denote the Lipschizt continuous constant $L_\Psi = \frac{2}{\delta}$.
Then
\begin{equation}
    L_\Psi L_\Phi = \frac{2 L_\Phi}{\delta} \geq 4 \label{eq_proof_lemma_density_L_L_lower_bound} .
\end{equation}
Let the set $W^* = B(x_1^*, r_w)$, where $r_w = \frac{r \delta}{16L_\Phi^2 L_\Psi}$.
For any $x_1 \in W^*$, let the function $\Phi_{x_1}(\cdot) = \Phi(x_1, \cdot)$ denote the restriction of $\Phi$, and let the corresponding reference point be $x^{**} = (x_1, x_2)$.
By Theorems 3.1 in \cite{phienQuantitativeResultsLipschitz2012}, denote
\begin{equation*}
    \delta_{x_1}
    = \frac{1}{2\|\{ D_2\Phi(x^{**}) \}^{-1}\|} .
\end{equation*}
Then from the definition of the constant $\delta_{x_1}$, we have the bounds
\begin{equation*}
    \delta_{x_1}
    \leq \frac{\| D_2\Phi(x^{**}) \|}{2}
    \leq \frac{\| D\Phi(x^{**}) \|}{2}
    \leq \frac{L_\Phi}{2} ,
\end{equation*}
and
\begin{equation}
    \label{eq_proof_lemma_density_delta_x}
    \delta_{x_1}
    \geq \frac{1}{4\|\{ D_2\Phi(x^*) \}^{-1}\|}
    = \frac{\delta}{2} ,
\end{equation}
where the last inequality follows from $x^{**} \in B(x^*, r)$ and \eqref{eq_proof_lemma_density_1}.
Let the constant $r_{x_1} = \min \left\{ L_{D\Phi}^{-1} \delta_{x_1}, \frac{r_x}{2} \right\}$.
Then
\begin{equation}
    \label{eq_proof_lemma_density_r_x}
    \frac{r}{2}
    \leq \min \left\{ L_{D\Phi}^{-1} \frac{\delta}{2}, \frac{r_x}{2} \right\}
    \leq r_{x_1}
    \leq \frac{r_x}{2} .
\end{equation}
For any $x \in B(x^{**}, r_{x_1})$, it holds that
\begin{equation*}
    \| D\Phi(x) - D\Phi(x^{**}) \|
    \leq L_{D\Phi} \| x - x^{**} \|
    < L_{D\Phi} r_{x_1}
    \leq \delta_{x_1} ,
\end{equation*}
and $B(x^{**}, r_{x_1})$ is a subset of $B(x^*, r_x)$ by \eqref{eq_proof_lemma_density_r_x} and the inequality
\begin{equation}
    \|x^{**} - x^*\|
    \leq r_w
    = \frac{r \delta}{16L_\Phi^2 L_\Psi}
    \leq \frac{r}{128}
    \leq \frac{r_x}{128} \label{eq_proof_lemma_density_r_w} .
\end{equation}
Thus, the definition of $r_{x_1}$ satisfies the requirement of Theorem 3.1 for the constant $r$ in \cite{phienQuantitativeResultsLipschitz2012}.

Therefore, by Theorem 3.1 in \cite{phienQuantitativeResultsLipschitz2012}, there exists a Lipschitz mapping $\Psi_{x_1}: V_{x_1} \to \mathbb{R}^d$ such that
\begin{enumerate}
    \item $\Psi_{x_1}(\Phi_{x_1}(u)) = u$ for every $u \in U_{x_1}$,
    \item $\Phi_{x_1}(\Psi_{x_1}(v)) = v$ for every $v \in V_{x_1}$,
\end{enumerate}
where $U_{x_1} = B(x_2^*, \frac{r_{x_1}\delta_{x_1}}{2L_\Phi})$, $V_{x_1} = B(\Phi_{x_1}(x_2^*), \frac{r_{x_1}\delta_{x_1}}{2})$, and the Lipschitz constant of the function $\Psi_{x_1}$, $L_{\Psi,x_1} = \frac{1}{\delta_{x_1}}$.
Note that $L_{\Psi,x_1} \leq L_\Psi$ by \eqref{eq_proof_lemma_density_delta_x}.
According to \eqref{eq_proof_lemma_density_r_w} and \eqref{eq_proof_lemma_density_r_x}, the inclusion $\{x_1\} \times U_{x_1} \subset \{x_1\} \times \Psi_{x_1}(V_{x_1}) \subset B(x^*,r_x)$ follows immediately, and we omit the details.

Let the set $U^* = B(x_2^*, r_u)$, where $r_u = \frac{r \delta}{8 L_\Phi}$.
Then by \eqref{eq_proof_lemma_density_delta_x} and \eqref{eq_proof_lemma_density_r_x}, we have $r_u \leq \frac{r_{x_1}\delta_{x_1}}{2L_\Phi}$ for any $x_1 \in W^*$.
Thus, for any $x_1 \in W^*$, it holds that $U^* \subset U_{x_1}$.
Therefore, by $\Psi_{x_1}(\Phi_{x_1}(u)) = u$ for every $u \in U_{x_1}$, it follows that $\Phi_{x_1}$ is injective on the domain $U^*$ for any $x_1 \in W^*$.

Moreover, denote the set $V^* = B(\Phi_{x_1^*}(x_2^*), r_v)$, where $r_v = \frac{r \delta}{16 L_\Phi L_\Psi}$.
Then by \eqref{eq_proof_lemma_density_L_L_lower_bound}, \eqref{eq_proof_lemma_density_delta_x} and \eqref{eq_proof_lemma_density_r_x}, we have
\begin{align*}
    \| y - \Phi_{x_1}(x_2^*) \|
    & \leq \| y - \Phi_{x_1^*}(x_2^*) \| + \| \Phi_{x_1^*}(x_2^*) - \Phi_{x_1}(x_2^*) \| \\
    & \leq r_v + L_\Phi \| x_1^*-x_1 \|
    \leq r_v + L_\Phi r_w
    \leq \frac{r \delta}{8 L_\Phi L_\Psi}
    \leq \frac{r_{x_1}\delta_{x_1}}{8} ,
\end{align*}
and it implies that $V^* \subset V_{x_1}$ for any $x_1 \in W^*$.
Thus, because the function $\Psi_{x_1}$ is Lipschitz continuous on the domain $V_{x_1}$, the function $\Psi_{x_1}$ is Lipschitz continuous on $V^*$.

In addition, for any $y \in V^*$ and any $x_1 \in W^*$,
\begin{align*}
    \| \Psi_{x_1}(y) - x_2^* \|
    &= \| \Psi_{x_1}(y) -  \Psi_{x_1}(\Phi_{x_1}(x_2^*)) \|
    \leq L_\Psi \| y - \Phi_{x_1}(x_2^*) \| \\
    & \leq L_\Psi \| y - \Phi_{x_1^*}(x_2^*) \|
    + L_\Psi \| \Phi_{x_1^*}(x_2^*) - \Phi_{x_1}(x_2^*) \| \\
    & \leq L_\Psi (r_v + L_\Phi r_w)
    \leq \frac{r \delta}{8 L_\Phi}
    = r_u .
\end{align*}
Thus, $\Psi_{x_1}(V^*) \subset U^*$, and it implies that $V^* \subset \Phi_{x_1}(U^*)$.
This concludes the first part of the proof.

Similarly to $\mu_{\mathrm{leb}, B(x^*,r_x)}$, let $\mu_{\mathrm{leb}, W^* \times U^*}$ and $\mu_{\mathrm{leb}, V^*}$ denote the Lebesgue measure restricted to $W^* \times U^*$ and $V^*$, respectively.
The remaining part of the proof is to establish the following two assertions:
\begin{enumerate}
    \item $\left. \left\{ \Phi^* \mu_{\mathrm{leb}, W^* \times U^*} \right\} \right|_{V^*}$, the measure $\Phi^* \mu_{\mathrm{leb}, W^* \times U^*}$ restricted to $V^*$, is absolutely continuous with respect to $\mu_{\mathrm{leb}}$,
    \item There exists some $\epsilon>0$ such that for any $B(x,r) \subset V^*$, $(\Phi^* \mu_{\mathrm{leb}, W^* \times U^*})(B(x,r)) \geq \epsilon \mu_{\mathrm{leb}, V^*}(B(x,r))$.
\end{enumerate}
According to Lebesgue-Radon-Nikodym Theorem in \cite{follandRealAnalysisModern1999}, there exists a Lebesgue integrable function $f: V^* \to [0, \infty)$ such that for any Lebesgue measurable set $A \subset V^*$, we have
\begin{equation*}
    \left( \left. \left\{ \Phi^* \mu_{\mathrm{leb}, W^* \times U^*} \right\} \right|_{V^*} \right)(A)
    = \int_A f(y)  \mu_{\mathrm{leb}}(\mathrm{d} y).
\end{equation*}
Thus, for any $B(x,r) \subset V^*$, $(\Phi^* \mu_{\mathrm{leb}, W^* \times U^*})(B(x,r)) \geq \epsilon \mu_{\mathrm{leb}}(B(x,r))$ is equivalent to
\begin{equation*}
    \int_{B(x,r)} f(y) \mu_{\mathrm{leb}}(\mathrm{d} y)
    \geq \epsilon \mu_{\mathrm{leb}}(B(x,r)) .
\end{equation*}
By Lebesgue Differentiation Theorem in \cite{follandRealAnalysisModern1999}, we have for almost every $x \in V^*$, it holds that
\begin{equation*}
    f(x)
    = \lim_{r \rightarrow 0^+} \frac{1}{\mu_{\mathrm{leb}}(B(x,r))}
    \int_{B(x,r)} f(y) \mu_{\mathrm{leb}}(\mathrm{d} y)
    \geq \epsilon .
\end{equation*}
In conclusion, $\Phi^* \mu_{\mathrm{leb}, B(x^*,r_x)} \geq \Phi^* \mu_{\mathrm{leb}, W^* \times U^*} \geq \left. \left\{ \Phi^* \mu_{\mathrm{leb}, W^* \times U^*} \right\} \right|_{V^*} \geq \epsilon \mu_{\mathrm{leb}, V^*}$.

We now proceed to prove the two assertions stated above.
First, because for any $x_1 \in W^*$, the restriction of $\Phi_{x_1}$ to $U^*$ is injective, $V^* \subset \Phi_{x_1}(U^*)$, and its inverse $\Psi_{x_1}$ is Lipschitz continuous on the domain $V^*$, we obtain that for any Lebesgue null set $A \subset V^*$,
\begin{align*}
    & \quad \mu_{\mathrm{leb}, W^* \times U^*}(\Phi^{-1}(A))
    = \mu_{\mathrm{leb}}(\Phi^{-1}(A) \cap (W^* \times U^*))
    = \int_{W^*} \mu_{\mathrm{leb}}(\mathrm{d} x_1) \mu_{\mathrm{leb}}(\Phi_{x_1}^{-1}(A) \cap U^*) \\
    &= \int_{W^*} \mu_{\mathrm{leb}}(\mathrm{d} x_1) \mu_{\mathrm{leb}}(\Psi_{x_1}(A) \cap U^*)
    \leq \int_{W^*} \mu_{\mathrm{leb}}(\mathrm{d} x_1) \mu_{\mathrm{leb}}(\Psi_{x_1}(A))
    = 0 .
\end{align*}
The last equality follows from the fact that $A$ is of Lebesgue measure zero and $\Psi_{x_1}$, being Lipschitz continuous, maps Lebesgue null sets to Lebesgue null sets.

Then, for any ball $B(y,r_y) \subset V^*$, we have $B \left( \Psi_{x_1}(y), \frac{r_y}{L_\Phi} \right) \subset U^*$, since
\begin{align*}
    & \quad \| \Psi_{x_1}(y) - x_2^* \| + \frac{r_y}{L_\Phi}
    = \| \Psi_{x_1}(y) - \Psi_{x_1}(\Phi_{x_1}(x_2^*)) \| + \frac{r_y}{L_\Phi}
    \leq L_\Psi \| y - \Phi_{x_1}(x_2^*) \| + \frac{r_y}{L_\Phi} \\
    & \leq L_\Psi \left[ \| y - \Phi_{x_1^*}(x_2^*) \| + \| \Phi_{x_1^*}(x_2^*) - \Phi_{x_1}(x_2^*) \| \right]
    + \frac{r_y}{L_\Phi} \\
    & \leq L_\Psi \left[ \| y - \Phi_{x_1^*}(x_2^*) \| + L_\Phi \| x_1 - x_1^* \| \right]
    + \frac{r_y}{L_\Phi} \\
    & \leq L_\Psi \left[ r_v - r_y + L_\Phi \| x_1 - x_1^* \| \right]
    + \frac{L_\Psi r_y}{4}
    \leq L_\Psi r_v + L_\Psi L_\Phi r_w
    \leq r_u ,
\end{align*}
where the last three inequalities follow from \eqref{eq_proof_lemma_density_L_L_lower_bound} and the definitions of $r_u$, $r_v$, and $r_w$.

Therefore, by the Lipschitz continuity of the function $\Phi$, we can conclude that
\begin{align*}
    (\Phi^* \mu_{\mathrm{leb}, W^* \times U^*})(B(y,r_y))
    &= \int \mathbb{I} \left( \Phi(x) \in B(y,r_y) \right) \mu_{\mathrm{leb}, W^* \times U^*}(\mathrm{d} x) \\
    &= \int_{W^*} \int_{U^*} \mathbb{I} \left( \Phi_{x_1}(x_2) \in B(y,r_y) \right) \mu_{\mathrm{leb}}(\mathrm{d} x_2) \mu_{\mathrm{leb}}(\mathrm{d} x_1) \\
    & \geq \int_{W^*} \int_{B \left( \Psi_{x_1}(y), \frac{r_y}{L_\Phi} \right) } \mathbb{I} \left( \Phi_{x_1}(x_2) \in B(y,r_y) \right) \mu_{\mathrm{leb}}(\mathrm{d} x_2) \mu_{\mathrm{leb}}(\mathrm{d} x_1) \\
    &= \int_{W^*} \int_{B \left( \Psi_{x_1}(y), \frac{r_y}{L_\Phi} \right) } \mu_{\mathrm{leb}}(\mathrm{d} x_2) \mu_{\mathrm{leb}}(\mathrm{d} x_1) \\
    &= \mu_{\mathrm{leb}} \left( B \left( \Psi_{x_1}(y), \frac{r_y}{L_\Phi} \right) \right) \mu_{\mathrm{leb}}(W^*) \\
    &= \frac{1}{L_\Phi^d}
    \mu_{\mathrm{leb}} \left( B(y,r_y) \right)
    \mu_{\mathrm{leb}}(W^*)
\end{align*}
Thus, letting $\epsilon = \frac{\mu_{\mathrm{leb}}(W^*)}{L_\Phi^d}$, we obtain that for any ball $B(y,r_y) \subset V^*$,
\begin{equation*}
    (\Phi^* \mu_{\mathrm{leb}, W^* \times U^*})(B(y,r_y)) \geq \epsilon \mu_{\mathrm{leb}, V^*}(B(y,r_y)) .
\end{equation*}

Let
\begin{equation*}
    C_{\mathrm{dl}}
    = \frac{\min \left\{ \left[ 2 L_{D\Phi} \| \{ D_2\Phi(x^*) \}^{-1} \| \right]^{-1} , r_x \right\}}
    {128 L_\Phi \| \{ D_2\Phi(x^*) \}^{-1} \|^2} .
\end{equation*}
Noting that
\begin{equation*}
    V^*
    = B \left( \Phi(x^*), \frac{r \delta}{16L_\Psi L_\Phi} \right)
    = B \left( \Phi(x^*), \frac{r \delta^2}{32L_\Phi} \right)
    = B \left( \Phi(x^*), C_{\mathrm{dl}} \right) ,
\end{equation*}
and
\begin{equation*}
    \epsilon
    = \frac{\mu_{\mathrm{leb}}(W^*)}{L_\Phi^d}
    = \frac{\mu_{\mathrm{leb}}(B(x_1^*, \frac{r \delta}{16 L_\Phi^2 L_\Psi}))}{L_\Phi^d}
    = \frac{\mu_{\mathrm{leb}}(B(x_1^*, \frac{C_{\mathrm{dl}}}{L_\Phi}))}{L_\Phi^d}
    = \frac{\mu_{\mathrm{leb}}(B^{(d)}(0, 1)) C_{\mathrm{dl}}^d}{L_\Phi^p} ,
\end{equation*}
the proof is complete.

\subsection{Proof of Lemma \ref{lemma_application_simultaneous_geometric_ergodicity}}
\label{subsec_proof_lemma_application_simultaneous_geometric_ergodicity}

We note that, although the proofs in this subsection are carried out for the function $V$, they extend naturally to $V^\alpha$ for any $\alpha \in (0,1)$, because Lemma \ref{lemma_drift_condition_inequality} remains valid with $V^\alpha$ in place of $V$.

\subsubsection{Positive Recurrence}

With Assumption \ref{assumption_x_sub_exponential_bound}, we can derive the following inequality from Lemma \ref{lemma_drift_condition_inequality}:
\begin{equation*}
    \mathbb{E}_\theta \left[ e^{\lambda_1 \|\Lambda_1\|} \mid \Lambda_0 = \Lambda \right] \leq \beta e^{\lambda_1 \|\Lambda_0\|} + b ,
\end{equation*}
for some positive numbers $\beta<1$, $b$ and $\lambda_1$.

For positive recurrence, we need Assumption \ref{assumption_for_simultaneous_small_set_condition} to ensure the validity of Lemma \ref{lemma_for_simultaneous_small_set_condition}.
Based on Lemma \ref{lemma_for_simultaneous_small_set_condition} and the inequality above, we can prove $\{\Lambda_n\}$ is positive recurrent by Theorem 11.3.4 in \cite{meynMarkovChainsStochastic2009}.
Theorem 10.4.9 in \cite{meynMarkovChainsStochastic2009} implies that the invariant probability measure $\pi_\theta$ for $P_\theta$ is unique and equivalent to the maximal irreducibility measure $\mu_{\mathrm{leb}}$.
This result implies that $P_\theta$ is $\pi_\theta$-irreducible.

\subsubsection{Simultaneous Geometric Ergodicity for $\theta$}

Lemma \ref{lemma_drift_condition_inequality} implies that there exist positive constants $\lambda_1$, $\beta<1$ and $b$ such that for any $\theta \in \Theta$,
\begin{equation*}
    P_\theta V \leq \beta V + b ,
\end{equation*}
where the Lyapunov function $V(\Lambda) = \exp(\lambda_1 \|\Lambda\|)$.

Since $V(\Lambda) \leq 2b(1-\beta)^{-1}$ when $\|\Lambda\| \leq \ln\left[2b(1-\beta)^{-1}\right]/\lambda_1$, denote the constant
\begin{equation*}
    c_\Lambda = \ln\left[2b(1-\beta)^{-1}\right]/\lambda_1 > 0 .
\end{equation*}
By Lemma \ref{lemma_for_simultaneous_small_set_condition} with $d = d_{2c_\Lambda}$, it holds that for any $\theta \in \Theta$,
\begin{equation}
    P_\theta^d(\Lambda, \cdot) \geq \delta_{R, P} \mu_{\mathrm{leb}}(\cdot \cap B(\Lambda, 2c_\Lambda))
    \label{eq_simultaneous_ergodicity_small_set_condition} ,
\end{equation}
and
\begin{equation}
    P_\theta^d V \leq \beta^d V + b\frac{1-\beta^d}{1-\beta}
    \label{eq_simultaneous_ergodicity_inequality_condition} .
\end{equation}

Combining the identity
\begin{equation*}
    \ln\left[2b\frac{1-\beta^d}{1-\beta}(1-\beta^d)^{-1}\right]/\lambda_1 = \ln\left[2b(1-\beta)^{-1}\right]/\lambda_1 = c_\Lambda ,
\end{equation*}
with the two inequalities (\ref{eq_simultaneous_ergodicity_small_set_condition}) and (\ref{eq_simultaneous_ergodicity_inequality_condition}), the condition in Lemma \ref{lemma_simultaneous_geometric_ergodicity} is satisfied with the $d$-step transition probability kernel $P_\theta^d$ for $\theta \in \Theta$ and $\nu$ defined as the uniform distribution on $B(0, c_\Lambda)$.
Therefore, we can conclude that there exists some constant $L_{\mathrm{small}}>1$, as defined in Lemma \ref{lemma_simultaneous_geometric_ergodicity}, such that for any $\theta \in \Theta$,
\begin{equation*}
    \| P_\theta^{d n}(\Lambda, \cdot) - \pi_\theta \|_V \leq L_{\mathrm{small}}(1-L_{\mathrm{small}}^{-1})^n V(\Lambda) .
\end{equation*}

Furthermore, we extend the inequality for $P_\theta^{d n}$ to arbitrary powers $P_\theta^n$.
For any nonnegative integer $l$, the following inequality holds:
\begin{align*}
\| P_\theta^{d n + l}(\Lambda, \cdot) - \pi_\theta \|_V
&= \| (P_\theta^{d n}(\Lambda, \cdot) - \pi_\theta ) P_\theta^{l} \|_V \\
& \leq \| P_\theta^{d n}(\Lambda, \cdot) - \pi_\theta \|_{\beta^l V + b\frac{1-\beta^l}{1-\beta}} \\
& \leq \| P_\theta^{d n}(\Lambda, \cdot) - \pi_\theta \|_{ \left( \beta^l + b\frac{1-\beta^l}{1-\beta} \right) V } \\
& \leq \left( \beta^l + b\frac{1-\beta^l}{1-\beta} \right) \| P_\theta^{d n}(\Lambda, \cdot) - \pi_\theta \|_V \\
& \leq \left( 1+\frac{b}{1-\beta} \right) \| P_\theta^{d n}(\Lambda, \cdot) - \pi_\theta \|_V ,
\end{align*}
where the first inequality can be derived by
\begin{equation*}
    P_\theta^l V
    \leq P_\theta^{l-1} (\beta V + b)
    \leq \dots
    \leq \beta^l V + b\frac{1-\beta^l}{1-\beta} .
\end{equation*}

Let $n=n^\prime d+l$ with $l \in \{0,\dots,d-1\}$. Using the inequality $n^\prime \geq \frac{n}{d}-1$, it holds that
\begin{equation*}
    \| P_\theta^{n}(\Lambda, \cdot) - \pi_\theta \|_V \leq \left( 1+\frac{b}{1-\beta} \right)L_{\mathrm{small}}(1-L_{\mathrm{small}}^{-1})^{\frac{n}{d}-1} V(\Lambda) .
\end{equation*}

Thus, for the constant $L$ in Lemma \ref{lemma_application_simultaneous_geometric_ergodicity}, we set $L = \max\{C,(1-\rho)^{-1}\}$, where $C = \left( 1+\frac{b}{1-\beta} \right)L_{\mathrm{small}}(1-L_{\mathrm{small}}^{-1})^{-1}$ and $\rho = (1-L_{\mathrm{small}}^{-1})^{\frac{1}{d}}$.

\subsubsection{Boundedness of $\pi_\theta V$}

For any $\theta \in \Theta$,
\begin{equation*}
    P_\theta V \leq \beta V + b
\end{equation*}
implies
\begin{equation*}
    \mathbb{E} V(\Lambda_n)
    \leq \beta \mathbb{E} V(\Lambda_{n-1}) + b
    \leq \dots
    \leq \beta^n \mathbb{E} V(\Lambda_0) + b\frac{1-\beta^n}{1-\beta}
\end{equation*}
and
\begin{equation*}
    \pi_\theta V
    = (\pi_\theta P_\theta) V
    = \pi_\theta (P_\theta V)
    \leq \beta \pi_\theta V + b .
\end{equation*}
Thus, $\mathbb{E} V(\Lambda_n) \leq \max \left\{ \frac{b}{1-\beta}, \mathbb{E} V(\Lambda_0) \right\}$ and $\pi_\theta V \leq \frac{b}{1-\beta}$.

\subsection{Proof of Lemma \ref{lemma_allocation_parameter_update}}

When the allocation parameter $\theta_{n+1}$ is updated according to \eqref{eq_allocation_parameter_update_mechanism_1}, suppose that the elements of $S$ can be ordered increasingly as $\{a_n\}_{n \in \mathbb{N}^*}$, and define $a_0 = 0$. Then
\begin{align*}
    \sum_{n=0}^{N-1} d(\theta_n,\theta_{n+1})
    &= \sum_{n \in \mathbb{N}^*, a_n \leq N} d(\theta_{a_n},\theta_{a_n-1})
    = \sum_{n \in \mathbb{N}^*, a_n \leq N} d(\eta_{a_n},\eta_{a_{n-1}}) \\
    &= \sum_{n \in \mathbb{N}^*, a_n \leq N} \sum_{i=a_{n-1}+1}^{a_n} d(\eta_i,\eta_{i-1})
    \leq \sum_{n=0}^{N-1} d(\eta_n,\eta_{n+1}) ,
\end{align*}
and the limit of $\{\theta_n\}$ coincides with that of $\{\eta_n\}$, denoted by $\eta^*$.

When the allocation parameter $\theta_{n+1}$ is updated according to \eqref{eq_allocation_parameter_update_mechanism_2} and the parameter space $\Theta$ is a convex subset of a Euclidean space, each $\theta_i$ lies on the line segment connecting $\theta_{i-1}$ and $\eta_i$.
Hence, for any $i$,
\begin{equation*}
    d(\theta_{i-1}, \eta_i) - d(\theta_{i-1}, \theta_i)
    = d(\eta_i, \theta_i) .
\end{equation*}
By the triangle inequality, it follows that
\begin{equation*}
    d(\theta_{i-1}, \theta_i) + d(\theta_i, \eta_{i+1})
    \leq d(\theta_{i-1}, \eta_i) + d(\eta_i, \eta_{i+1}) .
\end{equation*}
Therefore, setting $\eta_0 = \theta_0$, we have
\begin{align*}
    \sum_{n=0}^{N-1} d(\theta_n,\theta_{n+1})
    & \leq \sum_{n=0}^{N-2} d(\theta_n,\theta_{n+1}) + d(\theta_{N-1},\eta_N) \\
    & \leq \sum_{n=0}^{N-3} d(\theta_n,\theta_{n+1}) + d(\theta_{N-2}, \eta_{N-1})
    + d(\eta_{N-1},\eta_N) \\
    & \leq \dots
    \leq d(\theta_0, \eta_1)
    + \sum_{n=1}^{N-1} d(\eta_n,\eta_{n+1})
    = \sum_{n=0}^{N-1} d(\eta_n,\eta_{n+1}) .
\end{align*}

For any $\epsilon > 0$, denote $B_1 = B(\eta^*, \epsilon/3)$ and $B_2 = B(\eta^*, \epsilon)$.
Thus, for any $n$ such that $C_{\mathrm{clip}, n} < \epsilon/3$, if $\theta_{n-1} \in B_2^c$ and $\eta_n \in B_1$, then since $\| \eta_n - \theta_{n-1} \| > 2\epsilon/3 > 2C_{\mathrm{clip}, n}$, it follows that
\begin{align*}
    \| \eta^* - \theta_n \|^2 - \| \eta^* - \theta_{n-1} \|^2
    &= \| \eta^* - \theta_{n-1} - C (\eta_n - \theta_{n-1}) \|^2 - \| \eta^* - \theta_{n-1} \|^2 \\
    &= C^2 \| \eta_n - \theta_{n-1} \|^2 - 2C(\eta^* - \theta_{n-1})^T(\eta_n - \theta_{n-1}) \\
    &= C_{\mathrm{clip}, n}^2 - 2C \| \eta_n - \theta_{n-1} \|^2
    + 2C(\eta_n - \eta^*)^T(\eta_n - \theta_{n-1}) \\
    & \leq C_{\mathrm{clip}, n}^2 - 2 C_{\mathrm{clip}, n} \| \eta_n - \theta_{n-1} \| + 2 C_{\mathrm{clip}, n} \epsilon/3 \\
    & \leq C_{\mathrm{clip}, n}^2 - C_{\mathrm{clip}, n} \| \eta_n - \theta_{n-1} \| \\
    & \leq - C_{\mathrm{clip}, n} \| \eta_n - \theta_{n-1} \| / 2 < 0 ,
\end{align*}
where
\begin{equation*}
    C = \frac{\min \left\{ \|\eta_n - \theta_{n-1}\|, C_{\mathrm{clip}, n} \right\}}{\|\eta_n - \theta_{n-1}\|}
    = \frac{C_{\mathrm{clip}, n}}{\|\eta_n - \theta_{n-1}\|} .
\end{equation*}
Hence, by $\| \eta_n - \theta_{n-1} \| \geq \| \eta_n - \theta_n \|$, it holds that
\begin{align*}
    & \| \eta^* - \theta_n \| - \| \eta^* - \theta_{n-1} \|
    \leq \frac{\| \eta^* - \theta_n \|^2 - \| \eta^* - \theta_{n-1} \|^2}{\| \eta^* - \theta_n \| + \| \eta^* - \theta_{n-1} \|}
    \leq - \frac{C_{\mathrm{clip}, n} \| \eta_n - \theta_{n-1} \| / 2}{\| \eta^* - \theta_n \| + \| \eta^* - \theta_{n-1} \|} \\
    & \quad \leq - \frac{C_{\mathrm{clip}, n} \| \eta_n - \theta_{n-1} \| / 2}{2\epsilon/3 + 2\| \eta_n - \theta_{n-1} \|}
    \leq - \frac{C_{\mathrm{clip}, n} \| \eta_n - \theta_{n-1} \| / 2}{3\| \eta_n - \theta_{n-1} \|}
    \leq - \frac{C_{\mathrm{clip}, n}}{6} .
\end{align*}

Let $N_0 \in \mathbb{N}^*$ be such that $C_{\mathrm{clip}, n} < \epsilon/3$ for all $n \geq N_0$, and define $a_n = \| \eta^* - \theta_n \|$.
Then
\begin{equation*}
    a_n \leq
    \begin{cases}
        a_{n-1} - C_{\mathrm{clip}, n}/6 , & \text{if } a_{n-1} \geq \epsilon \text{ and } \eta_n \in B_1 , \\
        a_{n-1} + C_{\mathrm{clip}, n} , & \text{otherwise} .
    \end{cases}
\end{equation*}
Let $b_n = \max\{a_n, 4\epsilon/3\}$.
Because
\begin{equation*}
    a_n \leq
    \begin{cases}
        a_{n-1} - C_{\mathrm{clip}, n}/6 < a_{n-1} \leq b_{n-1} ,
        & \text{if } a_{n-1} \geq \epsilon \text{ and } \eta_n \in B_1 , \\
        a_{n-1} + C_{\mathrm{clip}, n} \leq 4\epsilon/3 \leq b_{n-1} ,
        & \text{if } a_{n-1} < \epsilon \text{ and } \eta_n \in B_1 , \\
    \end{cases}
\end{equation*}
it holds that $b_n \leq b_{n-1}$ when $\eta_n \in B_1$.
If $b_{n-1} \geq 3\epsilon/2$, it holds that $a_{n-1} - C_{\mathrm{clip}, n}/6 > 3\epsilon/2 - C_{\mathrm{clip}, n}/6 > 4\epsilon/3$, and then
\begin{equation*}
    b_n
    = \max\{a_n, 4\epsilon/3\}
    \leq \max\{a_{n-1} - C_{\mathrm{clip}, n}/6, 4\epsilon/3\}
    = a_{n-1} - C_{\mathrm{clip}, n}/6
    \leq b_{n-1} - C_{\mathrm{clip}, n}/6 .
\end{equation*}
Thus, we can obtain
\begin{equation*}
    b_n \leq
    \begin{cases}
        b_{n-1} - C_{\mathrm{clip}, n}/6 , & \text{if } b_{n-1} \geq 3\epsilon/2 \text{ and } \eta_n \in B_1 , \\
        b_{n-1} + C_{\mathrm{clip}, n} , & \text{if } \eta_n \notin B_1 , \\
        b_{n-1} , & \text{otherwise} . \\
    \end{cases}
\end{equation*}

Therefore,
\begin{align}
    \mathbb{E} [ b_n ]
    & \leq \mathbb{E} [ b_{n-1} ]
    + P(\eta_n \notin B_1) \cdot C_{\mathrm{clip}, n}
    - [P( b_{n-1} \geq 3\epsilon/2) - P(\eta_n \notin B_1)] \cdot C_{\mathrm{clip}, n}/6 \label{eq_proof_lemma_allocation_parameter_update} \\
    & \leq \mathbb{E} [ b_{n-1} ]
    + P(\eta_n \notin B_1) \cdot 7C_{\mathrm{clip}, n}/6
    - P( b_{n-1} \geq 3\epsilon/2) \cdot C_{\mathrm{clip}, n}/6 \notag .
\end{align}
By Assumption \ref{assumption_theta_compact}, $\diam(\Theta) < \infty$, so $b_n$ is bounded.
Let $C_b = \diam(\Theta) + 2\epsilon$.
Then $b_{n-1} \leq C_b$ almost surely, and
\begin{equation*}
    C_b P( b_{n-1} \geq 3\epsilon/2) + 3\epsilon/2 \cdot (1 - P( b_{n-1} \geq 3\epsilon/2))
    \geq \mathbb{E} \left[ b_{n-1} \right] ,
\end{equation*}
which implies
\begin{equation}
    P( b_{n-1} \geq 3\epsilon/2 )
    \geq \frac{\mathbb{E} \left[ b_{n-1} \right] - 3\epsilon/2}{C_b - 3\epsilon/2} \label{eq_proof_lemma_allocation_parameter_update_lower_bound} .
\end{equation}
Substituting \eqref{eq_proof_lemma_allocation_parameter_update_lower_bound} into \eqref{eq_proof_lemma_allocation_parameter_update} gives
\begin{equation}
    \mathbb{E} [ b_n ]
    \leq \mathbb{E} [ b_{n-1} ]
    + P(\eta_n \notin B_1) \cdot 7C_{\mathrm{clip}, n}/6
    - \frac{\mathbb{E} \left[ b_{n-1} \right] - 3\epsilon/2}{\diam(\Theta) + \epsilon/2} \cdot C_{\mathrm{clip}, n}/6 \label{eq_proof_lemma_allocation_parameter_update_final} .
\end{equation}

Denote
\begin{align*}
    B_n &= \mathbb{E} [ b_n ] , \\
    c_n &= \frac{C_{\mathrm{clip}, n}}{6\diam(\Theta)+3\epsilon} < 1 , \\
    D &= \frac{3\epsilon}{2} , \\
    p_n &= \left[ 7 \diam(\Theta) + 7\epsilon/2 \right] P(\eta_n \notin B_1) .
\end{align*}
Thus, for any $N \geq N_0$ and $n \geq N$, \eqref{eq_proof_lemma_allocation_parameter_update_final} can be rewritten as
\begin{equation*}
    B_n
    \leq \left( 1 - c_n \right) B_{n-1} + p_n c_n + D c_n
    \leq \left( 1 - c_n \right) B_{n-1} + \left\{ \sup_{m \geq N} p_m \right\} c_n + D c_n ,
\end{equation*}
which implies that
\begin{equation*}
    B_n - \left\{ \sup_{m \geq N} p_m \right\} - D
    \leq \left( 1 - c_n \right) \left[ B_{n-1} - \left\{ \sup_{m \geq N} p_m \right\} - D \right] .
\end{equation*}

Because $\sum_{n=1}^\infty c_n = +\infty$,
\begin{align*}
    B_{N+k} - \left\{ \sup_{m \geq N} p_m \right\} - D
    & \leq \prod_{i=0}^k (1 - c_{N+i}) \left[ B_{N-1} - \left\{ \sup_{m \geq N} p_m \right\} - D \right] \\
    & \leq \exp(- \sum_{i=0}^k c_{N+i}) \left[ B_{N-1} - \left\{ \sup_{m \geq N} p_m \right\} - D \right] \\
    & \rightarrow 0 ,
\end{align*}
as $k \rightarrow \infty$.
Because $p_n \rightarrow 0$ as $n \rightarrow \infty$,
\begin{equation*}
    \limsup_{k \rightarrow \infty} B_{N+k} \leq \left\{ \sup_{m \geq N} p_m \right\} + D
    \rightarrow D = \frac{3\epsilon}{2} ,
\end{equation*}
as $N \rightarrow \infty$.

Therefore,
\begin{equation*}
    \limsup_{n \rightarrow \infty}
    \mathbb{E} \left[ \| \eta^* - \theta_n\| \right]
    \leq \limsup_{n \rightarrow \infty}
    \mathbb{E} \left[ \max\{ \| \eta^* - \theta_n\|, 3\epsilon/2 \} \right]
    \leq 3\epsilon/2 .
\end{equation*}
Due to arbitrariness of $\epsilon > 0$, we can conclude that
\begin{equation*}
    \mathbb{E} \left[ \| \eta^* - \theta_n\| \right] \rightarrow 0 ,
\end{equation*}
and
\begin{equation*}
    \theta_n \xrightarrow{\mathbb{P}} \eta^* .
\end{equation*}

\section{Lemmas for the Estimation}
\label{sec_lemmas_estimation}

\begin{lemma}
    \label{lemma_estimation_convergence}

    Suppose that the allocation function $g_\theta$ satisfies $\pi_\theta \left[ g_\theta(\cdot, x) \right] = \rho_\theta(x)$ for $\Gamma$-a.e. $x$.
    Under Assumptions \ref{assumption_theta_compact}, \ref{assumption_lipschitz_continuity_functions_kernels}, \ref{assumption_estimation_convergence} and \ref{assumption_simultaneous_properties}, if the step sizes of the allocation parameter sequence $\{\theta_n\}$ satisfy Assumption \ref{assumption_theta_difference_average_convergence}, then $\eta_n \rightarrow \eta^*$ in probability.

    In addition, under Assumption \ref{assumption_theta_difference_average_convergence_large_deviation} on the step sizes of the allocation parameter sequence $\{\theta_n\}$, for any $\delta>0$, there exists a constant $c_\delta > 0$ such that
    \begin{equation*}
        P \left( d(\eta_n, \eta^*)
        > \delta \right) < c_\delta n^{-q} ,
    \end{equation*}
    with $q \in (0,1]$ defined in Assumption \ref{assumption_theta_difference_average_convergence_large_deviation}.
\end{lemma}
\begin{proof}
    Let $U \subset \Theta$ be either an open subset or a singleton,
    \begin{equation}
        M_U = \mathbb{E}_{(X,Y(1),Y(0)) \sim \Gamma_{X,\bm{Y}}} \left[ \rho^{\mathrm{ref}}(X) m_U(X, Y(1), 1)
        + [1 - \rho^{\mathrm{ref}}(X)] m_U(X, Y(0), 0) \right]
        \label{eq_proof_lemma_estimation_convergence_sup_expectation} ,
    \end{equation}
    and
    \begin{equation*}
        M_{U,n}
        = \sum_{i=1}^n
        \left[ \frac{\rho^{\mathrm{ref}}(T_i \mid X_i)}{\rho_{\theta_{i-1}}(T_i \mid X_i)} m_U(X_i, Y_i(T_i), T_i) \right]
        \label{eq_proof_lemma_estimation_convergence_sup_empirical} ,
    \end{equation*}
    where
    \begin{equation*}
        m_U(x, y, t) = \sup_{\eta \in U} m_\eta(x, y, t) .
    \end{equation*}
    Denote $M_\eta = M_{\{\eta\}}$ and $M_{\eta,n} = M_{\{\eta\},n}$, where $\{\eta\}$ is the singleton set containing $\eta$, and $M_{\{\eta\}}$ and $M_{\{\eta\},n}$ are defined in \eqref{eq_proof_lemma_estimation_convergence_sup_expectation} and \eqref{eq_proof_lemma_estimation_convergence_sup_empirical}, respectively.

    By applying Lemma \ref{lemma_allocation_convergence}, we have
    \begin{equation*}
        \frac{M_{U,n}}{n} \xrightarrow{\mathbb{P}} M_U
        \quad \text{and} \quad
        \frac{M_{\eta^*,n}}{n} \xrightarrow{\mathbb{P}} M_{\eta^*} .
    \end{equation*}

    Similar to the argument in the proof of Theorem 5.14 in \cite{vaartAsymptoticStatistics2007}, fix some $\eta \in \Theta$ and let $U_l \downarrow \eta$ be a decreasing sequence of open balls around $\eta$ of diameter converging to zero.
    The sequence $m_{U_l}$ is decreasing and greater than $m_\eta$ for every $l$.
    Lower semicontinuity of $m_\eta(x,y,t)$ yields that $m_{U_l} \downarrow m_\eta$ almost surely.
    We can apply the monotone convergence theorem and obtain that
    \begin{align*}
        & M_{U_l} = \mathbb{E}_{(X,Y(1),Y(0)) \sim \Gamma_{X,\bm{Y}}} \left[ \rho^{\mathrm{ref}}(X) m_{U_l}(X, Y(1), 1) + [1 - \rho^{\mathrm{ref}}(X)] m_{U_l}(X, Y(0), 0) \right] \\
        & \quad \rightarrow M_\eta = \mathbb{E}_{(X,Y(1),Y(0)) \sim \Gamma_{X,\bm{Y}}} \left[ \rho^{\mathrm{ref}}(X) m_\eta(X, Y(1), 1) + [1 - \rho^{\mathrm{ref}}(X)] m_\eta(X, Y(0), 0) \right] .
    \end{align*}
    
    For any $\eta \neq \eta^*$, we have $M_\eta < M_{\eta^*}$ since $\eta^*$ is the unique maximizer of $M_\eta$.
    Combine this with the preceding paragraph to see that for every $\eta \neq \eta^*$, there exists an open ball $U_\eta$ around $\eta$ with $M_{U_\eta}<M_{\eta^*}$.
    The set $B=\left\{\eta \in \Theta \mid d\left(\eta, \eta^*\right) \geq \delta\right\}$ is compact and is covered by the balls $\left\{U_\eta: \eta \in B\right\}$, where $\delta>0$. Let $U_{\eta_1}, \ldots, U_{\eta_q}$ be a finite subcover.

    By $\frac{M_{U_{\eta_j},n}}{n} \xrightarrow{\mathbb{P}} M_{U_{\eta_j}}$, $\frac{M_{\eta^*,n}}{n} \xrightarrow{\mathbb{P}} M_{\eta^*}$, $M_{U_{\eta_j}} < M_{\eta^*}$ and the definition of
    \begin{equation*}
        \eta_n \in \argmax_{\eta \in \Theta}
        \sum_{i=1}^n
        \left[ \frac{\rho^{\mathrm{ref}}(T_i \mid X_i)}{\rho_{\theta_{i-1}}(T_i \mid X_i)} m_\eta(X_i, Y_i(T_i), T_i)
        \right] ,
    \end{equation*}
    we can conclude that
    \begin{align*}
        P\left( \eta_n \in B \right)
        & \leq P\left( \sup_{\eta \in B} M_{\eta,n} \geq M_{\eta^*,n} \right)
        \leq \sum_{j=1}^q P\left( M_{U_{\eta_j},n} \geq M_{\eta^*,n} \right) \\
        &= \sum_{j=1}^q P\left( M_{U_{\eta_j}} + o_P(1) \geq M_{\eta^*} + o_P(1) \right)
        \rightarrow 0 .
    \end{align*}

    By the arbitrariness of $\delta$ in the definition of $B$, we have $\eta_n \xrightarrow{\mathbb{P}} \eta^*$.

    Under Assumption \ref{assumption_theta_difference_average_convergence_large_deviation}, for any $\delta>0$, by the second part of Lemma \ref{lemma_allocation_convergence}, we have
    \begin{align*}
        & P \left( d(\eta_n, \eta^*) > \delta \right)
        = P\left( \eta_n \in B \right)
        \leq P\left( \sup_{\eta \in B} M_{\eta,n} \geq M_{\eta^*,n} \right)
        \leq \sum_{j=1}^q P\left( M_{U_{\eta_j},n} \geq M_{\eta^*,n} \right) \\
        & \quad \leq \sum_{j=1}^q \left[ P\left( \left| M_{U_{\eta_j},n} - M_{U_{\eta_j}} \right|
        > \left| M_{\eta^*} - M_{U_{\eta_j}} \right|/2 \right) \right. \\
        & \quad\quad \left. + P\left( \left| M_{\eta^*,n} - M_{\eta^*} \right|
        > \left| M_{\eta^*} - M_{U_{\eta_j}} \right|/2 \right)
        \right]
        \\
        & \quad = O(n^{-q}) .
    \end{align*}
    with $q \in (0,1]$ defined in Assumption \ref{assumption_theta_difference_average_convergence_large_deviation}.
\end{proof}

\begin{lemma}
    \label{lemma_estimation_difference_bounded}
    
    Suppose that the allocation function $g_\theta$ satisfies $\pi_\theta \left[ g_\theta(\cdot, x) \right] = \rho_\theta(x)$ for $\Gamma$-a.e. $x$.
    Under Assumptions \ref{assumption_theta_compact}, \ref{assumption_lipschitz_continuity_functions_kernels}, \ref{assumption_estimation_asymptotic} and \ref{assumption_simultaneous_properties}, if the step sizes of the allocation parameter sequence $\{\theta_n\}$ satisfy Assumption \ref{assumption_theta_difference_average_convergence}, then
    \begin{equation*}
        \mathbb{E} \left[ \| \eta_n - \eta_{n-1} \| \right] \rightarrow 0 .
    \end{equation*}

    In addition, under Assumption \ref{assumption_theta_difference_average_convergence_large_deviation} on the step sizes of the allocation parameter sequence $\{\theta_n\}$, there exists a constant $c > 0$ such that
    \begin{equation*}
        \mathbb{E} \left[ \| \eta_n - \eta_{n-1} \| \right] \leq c n^{-q} ,
    \end{equation*}
    with $q \in (0,1]$ defined in Assumption \ref{assumption_theta_difference_average_convergence_large_deviation}.
\end{lemma}
\begin{proof}
    Let $\sigma_{-1}(\ddot{M}_{\eta^*})$ denote the smallest singular value of $\ddot{M}_{\eta^*}$.
    Define
    \begin{equation*}
        M_{s,n}
        = \sum_{i=1}^n
        \left[ \frac{\rho^{\mathrm{ref}}(T_i \mid X_i)}{\rho_{\theta_{i-1}}(T_i \mid X_i)} s(X_i, Y_i(T_i), T_i) \right]
    \end{equation*}
    and
    \begin{equation*}
        M_s = \mathbb{E}_{(X,Y(1),Y(0)) \sim \Gamma_{X,\bm{Y}}} \left[ \rho^{\mathrm{ref}}(X) s(X, Y(1), 1) + [1 - \rho^{\mathrm{ref}}(X)] s(X, Y(0), 0) \right] .
    \end{equation*}
    Similarly, define $M_{\ddot{m}_{\eta^*},n}$ and $M_{\ddot{m}_{\eta^*}}$ in an analogous manner.

    For any $n \in \mathbb{N}^*$, the proof is divided into two cases. First, consider the case where
    \begin{enumerate}
        \item $\left\| \eta_n - \eta^* \right\| < \delta_{\mathrm{max}}$,
        \item $\left\| \eta_{n-1} - \eta^* \right\| < \delta_{\mathrm{max}}$,
        \item $\left\| M_{\ddot{m}_{\eta^*},n}/n - M_{\ddot{m}_{\eta^*}} \right\| \leq \sigma_{-1}(\ddot{M}_{\eta^*}) / 4$,
        \item $|M_{s,n}/n - M_s| \leq M_s$,
    \end{enumerate}
    where $\delta_{\mathrm{max}} \in \left. \left( 0, \frac{\sigma_{-1}(\ddot{M}_{\eta^*})}{8M_s} \right. \right]$ is chosen such that $B(\eta^*, \delta_{\mathrm{max}}) \subset \mathring{\Theta}$.
    Denote this event by $A_{\mathrm{good}, n}$.

    Because for any $\eta \in B(\eta^*, \delta_{\mathrm{max}}) \subset \mathring{\Theta}$,
    \begin{equation*}
        \dot{m}_\eta(x, y, t)
        = \dot{m}_{\eta^*}(x, y, t)
        + \int_{\eta^*}^\eta \ddot{m}_\eta(x, y, t) d\eta ,
    \end{equation*}
    it follows that
    \begin{align*}
        & \quad \| \dot{m}_\eta(x, y, t) - \dot{m}_{\eta^*}(x, y, t) \| \\
        & \leq \left\| \int_0^1 \left[ \ddot{m}_{\eta^* + a(\eta - \eta^*)}(x, y, t) \right]
        [ \eta - \eta^* ] \mathrm{d} a \right\| \\
        & \leq \int_0^1 \left\| \ddot{m}_{\eta^* + a(\eta - \eta^*)}(x, y, t) \right\|
        \| \eta - \eta^* \| \mathrm{d} a \\
        & \leq \int_0^1 \left\| \ddot{m}_{\eta^* + a(\eta - \eta^*)}(x, y, t)
        - \ddot{m}_{\eta^*}(x, y, t) \right\|
        \| \eta - \eta^* \| \mathrm{d} a
        + \int_0^1 \left\| \ddot{m}_{\eta^*}(x, y, t) \right\|
        \| \eta - \eta^* \| \mathrm{d} a \\
        & \leq s(x, y, t) \frac{\| \eta - \eta^* \|^2}{2}
        + \left\| \ddot{m}_{\eta^*}(x, y, t) \right\| \| \eta - \eta^* \| .
    \end{align*}

    Consequently,
    \begin{align*}
        & \quad \mathbb{E}_{(X,Y(t)) \sim \Gamma_{X,Y(t)}}
        \left\| \dot{m}_{\eta}(X, Y(t), t) \right\| \\
        & \leq \mathbb{E}_{(X,Y(t)) \sim \Gamma_{X,Y(t)}}
        \left[ \left\| \dot{m}_{\eta^*}(X, Y(t), t) \right\|
        + s(X, Y(t), t) \frac{\| \eta - \eta^* \|^2}{2}
        + \left\| \ddot{m}_{\eta^*}(X, Y(t), t) \right\| \| \eta - \eta^* \| \right] ,
    \end{align*}
    which is uniformly bounded for $\eta \in B(\eta^*, \delta_{\mathrm{max}})$.
    Therefore,
    \begin{align}
        & \quad \sup_{\eta \in B(\eta^*, \delta_{\mathrm{max}})}
        \mathbb{E} \left\| \dot{m}_\eta(X_n, Y_n(T_n), T_n) \right\| \label{eq_proof_lemma_estimation_difference_bounded_1} \\
        & \leq \sup_{\eta \in B(\eta^*, \delta_{\mathrm{max}})} \mathbb{E}_{(X,Y(1)) \sim \Gamma_{X,Y(1)}}
        \left\| \dot{m}_{\eta}(X, Y(1), 1) \right\|
        + \sup_{\eta \in B(\eta^*, \delta_{\mathrm{max}})} \mathbb{E}_{(X,Y(0)) \sim \Gamma_{X,Y(0)}}
        \left\| \dot{m}_{\eta}(X, Y(0), 0) \right\|
        < \infty \notag ,
    \end{align}
    and
    \begin{equation}
        \max\{ \| \eta_n - \eta^* \|, \| \eta_{n-1} - \eta^* \| \}
        \frac{M_{s,n}}{n}
        \leq 2 \delta_{\mathrm{max}} M_s
        \leq \sigma_{-1}(\ddot{M}_{\eta^*}) / 4 \label{eq_proof_lemma_estimation_difference_bounded_2} .
    \end{equation}
    These inequalities play a crucial role in establishing the bound in \eqref{eq_proof_lemma_estimation_difference_bounded_7}.

    Since $\eta^*$, $\eta_n$, and $\eta_{n-1}$ are interior points of $\Theta$, it holds that
    \begin{align}
        0
        &= \dot{M}_{\eta^*}
        = \mathbb{E}_{(X,Y(1),Y(0)) \sim \Gamma_{X,\bm{Y}}} \left[ \rho^{\mathrm{ref}}(X) \dot{m}_{\eta^*}(X, Y(1), 1)
        + [1 - \rho^{\mathrm{ref}}(X)] \dot{m}_{\eta^*}(X, Y(0), 0) \right] \notag , \\
        0
        &= \dot{M}_{\eta_n, n}
        = \sum_{i=1}^n
        \left[ \frac{\rho^{\mathrm{ref}}(T_i \mid X_i)}{\rho_{\theta_{i-1}}(T_i \mid X_i)}
        \dot{m}_{\eta_n}(X_i, Y_i(T_i), T_i) \right] \label{eq_proof_lemma_estimation_difference_bounded_3} , \\
        0
        &= \dot{M}_{\eta_{n-1}, n-1}
        = \sum_{i=1}^{n-1}
        \left[ \frac{\rho^{\mathrm{ref}}(T_i \mid X_i)}{\rho_{\theta_{i-1}}(T_i \mid X_i)}
        \dot{m}_{\eta_{n-1}}(X_i, Y_i(T_i), T_i) \right] \label{eq_proof_lemma_estimation_difference_bounded_4} .
    \end{align}

    \eqref{eq_proof_lemma_estimation_difference_bounded_3} and \eqref{eq_proof_lemma_estimation_difference_bounded_4} imply that
    \begin{align}
        & \quad - \frac{\rho^{\mathrm{ref}}(T_n \mid X_n)}{\rho_{\theta_{n-1}}(T_n \mid X_n)}
        \dot{m}_{\eta_n}(X_n, Y_n(T_n), T_n) \label{eq_proof_lemma_estimation_difference_bounded_5} \\
        &= \sum_{i=1}^n
        \left[ \frac{\rho^{\mathrm{ref}}(T_i \mid X_i)}{\rho_{\theta_{i-1}}(T_i \mid X_i)}
        \left[ \dot{m}_{\eta_n}(X_i, Y_i(T_i), T_i) - \dot{m}_{\eta_{n-1}}(X_i, Y_i(T_i), T_i) \right] \right] \notag \\
        &= \sum_{i=1}^n
        \left[ \frac{\rho^{\mathrm{ref}}(T_i \mid X_i)}{\rho_{\theta_{i-1}}(T_i \mid X_i)}
        \left[ \int_{\eta_{n-1}}^{\eta_n} \ddot{m}_\eta(X_i, Y_i(T_i), T_i) d\eta \right] \right] \notag \\
        &= \sum_{i=1}^n
        \left[ \frac{\rho^{\mathrm{ref}}(T_i \mid X_i)}{\rho_{\theta_{i-1}}(T_i \mid X_i)}
        \left[ \int_{\eta_{n-1}}^{\eta_n} \ddot{m}_{\eta^*}(X_i, Y_i(T_i), T_i) d\eta \right] \right] \notag \\
        & \quad + \sum_{i=1}^n
        \left[ \frac{\rho^{\mathrm{ref}}(T_i \mid X_i)}{\rho_{\theta_{i-1}}(T_i \mid X_i)}
        \left[ \int_{\eta_{n-1}}^{\eta_n} \left[ \ddot{m}_\eta(X_i, Y_i(T_i), T_i)
        - \ddot{m}_{\eta^*}(X_i, Y_i(T_i), T_i) \right] d\eta \right] \right] \notag \\
        &= \left[ \sum_{i=1}^n
        \left[ \frac{\rho^{\mathrm{ref}}(T_i \mid X_i)}{\rho_{\theta_{i-1}}(T_i \mid X_i)}
        \ddot{m}_{\eta^*}(X_i, Y_i(T_i), T_i) \right] \right] (\eta_n - \eta_{n-1}) \notag \\
        & \quad + \sum_{i=1}^n
        \left[ \frac{\rho^{\mathrm{ref}}(T_i \mid X_i)}{\rho_{\theta_{i-1}}(T_i \mid X_i)}
        \left[ \int_{\eta_{n-1}}^{\eta_n} \left[ \ddot{m}_\eta(X_i, Y_i(T_i), T_i)
        - \ddot{m}_{\eta^*}(X_i, Y_i(T_i), T_i) \right] d\eta \right] \right] \notag .
    \end{align}

    The last term of the \eqref{eq_proof_lemma_estimation_difference_bounded_5} can be bounded as
    \begin{align}
        & \quad \left\| \sum_{i=1}^n
        \left[ \frac{\rho^{\mathrm{ref}}(T_i \mid X_i)}{\rho_{\theta_{i-1}}(T_i \mid X_i)}
        \left[ \int_{\eta_{n-1}}^{\eta_n} \left[ \ddot{m}_\eta(X_i, Y_i(T_i), T_i)
        - \ddot{m}_{\eta^*}(X_i, Y_i(T_i), T_i) \right] d\eta \right] \right] \right\| \label{eq_proof_lemma_estimation_difference_bounded_6} \\
        & \leq \sum_{i=1}^n
        \frac{\rho^{\mathrm{ref}}(T_i \mid X_i)}{\rho_{\theta_{i-1}}(T_i \mid X_i)}
        \left\| \int_0^1 \left[ \ddot{m}_{\eta_{n-1} + t(\eta_n - \eta_{n-1})}(X_i, Y_i(T_i), T_i)
        - \ddot{m}_{\eta^*}(X_i, Y_i(T_i), T_i) \right]
        \left[ \eta_n - \eta_{n-1} \right] \mathrm{d} t \right\| \notag \\
        & \leq \sum_{i=1}^n
        \frac{\rho^{\mathrm{ref}}(T_i \mid X_i)}{\rho_{\theta_{i-1}}(T_i \mid X_i)}
        \int_0^1 \left\| \ddot{m}_{\eta_{n-1} + t(\eta_n - \eta_{n-1})}(X_i, Y_i(T_i), T_i)
        - \ddot{m}_{\eta^*}(X_i, Y_i(T_i), T_i) \right\|
        \| \eta_n - \eta_{n-1} \| \mathrm{d} t \notag \\
        & \leq \sum_{i=1}^n
        \frac{\rho^{\mathrm{ref}}(T_i \mid X_i)}{\rho_{\theta_{i-1}}(T_i \mid X_i)}
        \int_0^1 \left[ s(X_i, Y_i(T_i), T_i)
        \max\{ \| \eta_n - \eta^* \|, \| \eta_{n-1} - \eta^* \| \}
        \| \eta_n - \eta_{n-1} \| \right] \mathrm{d} t \notag \\
        &= \max\{ \| \eta_n - \eta^* \|, \| \eta_{n-1} - \eta^* \| \}
        M_{s,n} \| \eta_n - \eta_{n-1} \| \notag \\
        & \leq \frac{ n \sigma_{-1}(\ddot{M}_{\eta^*}) \| \eta_n - \eta_{n-1} \| }{4} \notag ,
    \end{align}
    where the last inequality follows from \eqref{eq_proof_lemma_estimation_difference_bounded_2}.

    It then follows from \eqref{eq_proof_lemma_estimation_difference_bounded_5} and \eqref{eq_proof_lemma_estimation_difference_bounded_6} that
    \begin{align*}
        & \left\| - \frac{1}{n} \frac{\rho^{\mathrm{ref}}(T_n \mid X_n)}{\rho_{\theta_{n-1}}(T_n \mid X_n)}
        \dot{m}_{\eta_n}(X_n, Y_n(T_n), T_n) \right\| \\
        & \geq \sigma_{-1} \left( \frac{1}{n} \sum_{i=1}^n
        \left[ \frac{\rho^{\mathrm{ref}}(T_i \mid X_i)}{\rho_{\theta_{i-1}}(T_i \mid X_i)}
        \ddot{m}_{\eta^*}(X_i, Y_i(T_i), T_i) \right] \right) \left\| \eta_n - \eta_{n-1} \right\|
        - \frac{ \sigma_{-1}(\ddot{M}_{\eta^*}) \| \eta_n - \eta_{n-1} \| }{4} \\
        & \geq \left[ \sigma_{-1} \left( M_{\ddot{m}_{\eta^*}} \right)
        - \left\| \frac{M_{\ddot{m}_{\eta^*},n}}{n} - M_{\ddot{m}_{\eta^*}} \right\| \right]
        \left\| \eta_n - \eta_{n-1} \right\|
        - \frac{ \sigma_{-1}(\ddot{M}_{\eta^*}) \| \eta_n - \eta_{n-1} \| }{4} \\
        & \geq \left[ \sigma_{-1} \left( M_{\ddot{m}_{\eta^*}} \right)
        - \frac{ \sigma_{-1}(\ddot{M}_{\eta^*}) }{4} \right]
        \left\| \eta_n - \eta_{n-1} \right\|
        - \frac{ \sigma_{-1}(\ddot{M}_{\eta^*}) \| \eta_n - \eta_{n-1} \| }{4} \\
        &= \frac{ \sigma_{-1}(\ddot{M}_{\eta^*}) \| \eta_n - \eta_{n-1} \| }{2} .
    \end{align*}

    Hence,
    \begin{align*}
        \left\| \eta_n - \eta_{n-1} \right\|
        & \leq \frac{2}{\sigma_{-1}(\ddot{M}_{\eta^*})}
        \left\| - \frac{1}{n} \frac{\rho^{\mathrm{ref}}(T_n \mid X_n)}{\rho_{\theta_{n-1}}(T_n \mid X_n)}
        \dot{m}_{\eta_n}(X_n, Y_n(T_n), T_n) \right\| \\
        & \leq \frac{1}{n} \frac{2}{\iota \sigma_{-1}(\ddot{M}_{\eta^*})}
        \left\| \dot{m}_{\eta_n}(X_n, Y_n(T_n), T_n) \right\|
    \end{align*}
    on the event $A_{\mathrm{good}, n}$. Combining this with \eqref{eq_proof_lemma_estimation_difference_bounded_1}, it follows that
    \begin{align}
        & \quad \mathbb{E} \left[ \left\| \eta_n - \eta_{n-1} \right\| \mathbb{I}_{A_{\mathrm{good}, n}} \right]
        \label{eq_proof_lemma_estimation_difference_bounded_7} \\
        & \leq \frac{1}{n} \frac{2}{\iota \sigma_{-1}(\ddot{M}_{\eta^*})}
        \left[ \sup_{\eta \in B(\eta^*, \delta_{\mathrm{max}})} \mathbb{E}_{(X,Y(1)) \sim \Gamma_{X,Y(1)}} \right.
        \left\| \dot{m}_{\eta}(X, Y(1), 1) \right\| \notag \\
        & \quad + \left. \sup_{\eta \in B(\eta^*, \delta_{\mathrm{max}})} \mathbb{E}_{(X,Y(0)) \sim \Gamma_{X,Y(0)}}
        \left\| \dot{m}_{\eta}(X, Y(0), 0) \right\| \right] \notag \\
        & = O \left( \frac{1}{n} \right) \notag .
    \end{align}
    
    Next, consider the complement event $A_{\mathrm{good}, n}^c$.
    Because Assumption \ref{assumption_estimation_asymptotic} implies Assumption \ref{assumption_estimation_convergence}, by Lemmas \ref{lemma_estimation_convergence} and \ref{lemma_allocation_convergence},
    \begin{equation*}
        P(A_{\mathrm{good}, n}^c) \rightarrow 0 .
    \end{equation*}
    Moreover, since Assumption \ref{assumption_estimation_asymptotic} implies Assumption \ref{assumption_estimation_convergence}, under Assumption \ref{assumption_theta_difference_average_convergence_large_deviation}, Lemmas \ref{lemma_estimation_convergence} and \ref{lemma_allocation_convergence} guarantee the existence of a constant $c_A$ such that
    \begin{equation*}
        P(A_{\mathrm{good}, n}^c) < c_A n^{-q} ,
    \end{equation*}
    with $q \in (0,1]$ defined in Assumption \ref{assumption_theta_difference_average_convergence_large_deviation}.

    Due to the compactness of $\Theta$, the distance $d(\eta_n, \eta_{n-1})$ is uniformly bounded. 
    Consequently, $P(A_{\mathrm{good}, n}^c) \rightarrow 0$ implies $\mathbb{E} \left[ \left\| \eta_n - \eta_{n-1} \right\| \mathbb{I}_{A_{\mathrm{good}, n}^c} \right] = o(1)$ and the bound $P(A_{\mathrm{good}, n}^c) < c_A n^{-q}$ further implies $\mathbb{E} \left[ \left\| \eta_n - \eta_{n-1} \right\| \mathbb{I}_{A_{\mathrm{good}, n}^c} \right] = O(n^{-q})$.

    In conclusion, combining the analysis on the event $A_{\mathrm{good}, n}$ and its complement $A_{\mathrm{good}, n}^c$, we obtain that
    \begin{equation*}
        \mathbb{E} \left[ \| \eta_n - \eta_{n-1} \| \right] = o(1) ,
    \end{equation*}
    and, under Assumption \ref{assumption_theta_difference_average_convergence_large_deviation},
    \begin{equation*}
        \mathbb{E} \left[ \| \eta_n - \eta_{n-1} \| \right] = O\left(n^{-q}\right) ,
    \end{equation*}
    with $q \in (0,1]$. This establishes the desired asymptotic bound on the difference between successive estimators.
\end{proof}

\begin{lemma}
    \label{lemma_estimation_asymptotic_normality}

    Suppose that the allocation function $g_\theta$ satisfies $\pi_\theta \left[ g_\theta(\cdot, x) \right] = \rho_\theta(x)$ for $\Gamma$-a.e. $x$.
    Under Assumptions \ref{assumption_theta_compact}, \ref{assumption_lipschitz_continuity_functions_kernels}, \ref{assumption_estimation_asymptotic} and \ref{assumption_simultaneous_properties}, if the allocation parameter sequence $\{\theta_n\}$ satisfies Assumption \ref{assumption_theta_convergence_strong_diminishing_adaptation},
    \begin{equation*}
        \ddot{M}_{\eta^*}
        = \mathbb{E}_{(X,Y(1),Y(0)) \sim \Gamma_{X,\bm{Y}}} \left[ \rho^{\mathrm{ref}}(X) \ddot{m}_{\eta^*}(X, Y(1), 1)
        + [1 - \rho^{\mathrm{ref}}(X)] \ddot{m}_{\eta^*}(X, Y(0), 0) \right]
    \end{equation*}
    is an invertible matrix and for each $t \in \{0, 1\}$, there exists some $\epsilon > 0$ such that
    \begin{equation*}
        \mathbb{E}_{(X,Y(t)) \sim \Gamma_{X,Y(t)}} \left\| \dot{m}_{\eta^*}(X,Y(t),t) \right\|^{4+\epsilon}
        < \infty ,
    \end{equation*}
    then
    \begin{equation*}
        \sqrt{n} (\eta_n - \eta^*)
        = - \ddot{M}_{\eta^*}^{-1} \frac{1}{\sqrt{n}} \dot{M}_{\eta^*, n}
        + o_P(1) .
    \end{equation*}
    Here,
    \begin{equation*}
        \dot{M}_{\eta^*, n}
        = \sum_{i=1}^n
        \left[ \frac{\rho^{\mathrm{ref}}(T_i \mid X_i)}{\rho_{\theta_{i-1}}(T_i \mid X_i)}
        \dot{m}_{\eta^*}(X_i, Y_i(T_i), T_i) \right]
    \end{equation*}
    is asymptotically normal with mean zero and covariance matrix
    \begin{align*}
        \Sigma_{(Z)}
        &= \Cov(Z^{(u)}) + \mathbb{E}_{(X, Z) \sim \Gamma_{X,Z}} \left[ \rho_{\theta^*}(X)(1-\rho_{\theta^*}(X)) \right. \\
        & \quad \left. \left\{ Z^{(c)} - \frac{A \phi(X)}{\rho_{\theta^*}(X)(1-\rho_{\theta^*}(X))} \right\}
        \left\{ Z^{(c)} - \frac{A \phi(X)}{\rho_{\theta^*}(X)(1-\rho_{\theta^*}(X))} \right\}^T \right] ,
    \end{align*}
    where
    \begin{align*}
        Z^{(c)}
        &= \frac{\rho^{\mathrm{ref}}(X)}{\rho_{\theta^*}(X)}
        \ddot{M}_{\eta^*}^{-1} \dot{m}_{\eta^*}(X, Y(1), 1)
        -
        \frac{1-\rho^{\mathrm{ref}}(X)}{1-\rho_{\theta^*}(X)}
        \ddot{M}_{\eta^*}^{-1} \dot{m}_{\eta^*}(X, Y(0), 0) , \\
        Z^{(u)}
        &= 
        \rho^{\mathrm{ref}}(X) \ddot{M}_{\eta^*}^{-1} \dot{m}_{\eta^*}(X, Y(1), 1)
        + (1-\rho^{\mathrm{ref}}(X)) \ddot{M}_{\eta^*}^{-1} \dot{m}_{\eta^*}(X, Y(0), 0) ,
    \end{align*}
    and the matrix $A$ satisfies
    \begin{align*}
        & A \mathbb{E} \left[ [\rho_{\theta^*}(X)(1-\rho_{\theta^*}(X))]^{-1} \phi(X) \phi(X)^T
        / \max\{ \|\phi(X)\|/[\rho_{\theta^*}(X)(1-\rho_{\theta^*}(X))], C_{\theta^*} \} \right] \\
        & \quad = \mathbb{E} \left[ Z^{(c)} \phi(X)^T
        / \max\{ \|\phi(X)\|/[\rho_{\theta^*}(X)(1-\rho_{\theta^*}(X))], C_{\theta^*} \} \right] .
    \end{align*}
\end{lemma}
\begin{proof}
    Because Assumption \ref{assumption_estimation_asymptotic} implies Assumption \ref{assumption_estimation_convergence}, by Lemma \ref{lemma_estimation_convergence}, $\eta_n \xrightarrow{\mathbb{P}} \eta^*$.
    Thus, $P(\eta_n \in \mathring{\Theta}) \rightarrow 1$.

    Because $\eta^*$ is an interior point of $\Theta$, it holds that
    \begin{equation*}
        0
        = \dot{M}_{\eta^*}
        = \mathbb{E}_{(X,Y(1),Y(0)) \sim \Gamma_{X,\bm{Y}}} \left[ \rho^{\mathrm{ref}}(X) \dot{m}_{\eta^*}(X, Y(1), 1)
        + [1 - \rho^{\mathrm{ref}}(X)] \dot{m}_{\eta^*}(X, Y(0), 0) \right] .
    \end{equation*}

    When $\eta_n \in \mathring{\Theta}$, it holds that
    \begin{align}
        0
        &= \dot{M}_{\eta_n, n}
        = \sum_{i=1}^n
        \left[ \frac{\rho^{\mathrm{ref}}(T_i \mid X_i)}{\rho_{\theta_{i-1}}(T_i \mid X_i)}
        \dot{m}_{\eta_n}(X_i, Y_i(T_i), T_i) \right] \label{eq_proof_lemma_estimation_asymptotic_normality_fff} \\
        &= \dot{M}_{\eta^*, n}
        + \int_{\eta^*}^{\eta_n} \left[ \sum_{i=1}^n
        \left[ \frac{\rho^{\mathrm{ref}}(T_i \mid X_i)}{\rho_{\theta_{i-1}}(T_i \mid X_i)}
        \ddot{m}_\eta(X_i, Y_i(T_i), T_i) \right] \right] d\eta \notag \\
        &= \dot{M}_{\eta^*, n}
        + \left[ \sum_{i=1}^n
        \left[ \frac{\rho^{\mathrm{ref}}(T_i \mid X_i)}{\rho_{\theta_{i-1}}(T_i \mid X_i)}
        \ddot{m}_{\eta^*}(X_i, Y_i(T_i), T_i) \right] \right] (\eta_n - \eta^*) \notag \\
        & \quad + \sum_{i=1}^n
        \left[ \frac{\rho^{\mathrm{ref}}(T_i \mid X_i)}{\rho_{\theta_{i-1}}(T_i \mid X_i)}
        \left[ \int_{\eta^*}^{\eta_n} \left[ \ddot{m}_\eta(X_i, Y_i(T_i), T_i)
        - \ddot{m}_{\eta^*}(X_i, Y_i(T_i), T_i) \right] d\eta \right] \right] \notag .
    \end{align}
    The last term of the \eqref{eq_proof_lemma_estimation_asymptotic_normality_fff} can be bounded as
    \begin{align}
        & \quad \left\| \sum_{i=1}^n
        \left[ \frac{\rho^{\mathrm{ref}}(T_i \mid X_i)}{\rho_{\theta_{i-1}}(T_i \mid X_i)}
        \left[ \int_{\eta^*}^{\eta_n} \left[ \ddot{m}_\eta(X_i, Y_i(T_i), T_i)
        - \ddot{m}_{\eta^*}(X_i, Y_i(T_i), T_i) \right] d\eta \right] \right] \right\| \label{eq_proof_lemma_estimation_asymptotic_normality_eee} \\
        & \leq \sum_{i=1}^n
        \frac{\rho^{\mathrm{ref}}(T_i \mid X_i)}{\rho_{\theta_{i-1}}(T_i \mid X_i)}
        \left\| \int_0^1 \left[ \ddot{m}_{\eta^* + t(\eta_n - \eta^*)}(X_i, Y_i(T_i), T_i)
        - \ddot{m}_{\eta^*}(X_i, Y_i(T_i), T_i) \right]
        \left[ \eta_n - \eta^* \right] \mathrm{d} t \right\| \notag \\
        & \leq \sum_{i=1}^n
        \frac{\rho^{\mathrm{ref}}(T_i \mid X_i)}{\rho_{\theta_{i-1}}(T_i \mid X_i)}
        \int_0^1 \left\| \ddot{m}_{\eta^* + t(\eta_n - \eta^*)}(X_i, Y_i(T_i), T_i)
        - \ddot{m}_{\eta^*}(X_i, Y_i(T_i), T_i) \right\|
        \| \eta_n - \eta^* \| \mathrm{d} t \notag \\
        & \leq \sum_{i=1}^n
        \frac{\rho^{\mathrm{ref}}(T_i \mid X_i)}{\rho_{\theta_{i-1}}(T_i \mid X_i)}
        \int_0^1 \left[ t s(X_i, Y_i(T_i), T_i) \| \eta_n - \eta^* \|^2 \right] \mathrm{d} t \notag \\
        & \leq \frac{\| \eta_n - \eta^* \|^2}{2} \sum_{i=1}^n
        \left[ \frac{\rho^{\mathrm{ref}}(T_i \mid X_i)}{\rho_{\theta_{i-1}}(T_i \mid X_i)} s(X_i, Y_i(T_i), T_i) \right] \notag .
    \end{align}
    By applying Lemma \ref{lemma_allocation_convergence}, we have
    \begin{align}
        & \frac{1}{n} \sum_{i=1}^n
        \left[ \frac{\rho^{\mathrm{ref}}(T_i \mid X_i)}{\rho_{\theta_{i-1}}(T_i \mid X_i)} s(X_i, Y_i(T_i), T_i) \right] \label{eq_proof_lemma_estimation_asymptotic_normality_3} \\
        & \quad \xrightarrow{\mathbb{P}}
        \mathbb{E}_{(X,Y(1),Y(0)) \sim \Gamma_{X,\bm{Y}}} \left[ \rho^{\mathrm{ref}}(X) s(X, Y(1), 1) + [1 - \rho^{\mathrm{ref}}(X)] s(X, Y(0), 0) \right] \notag .
    \end{align}
    Combining \eqref{eq_proof_lemma_estimation_asymptotic_normality_eee} with \eqref{eq_proof_lemma_estimation_asymptotic_normality_3} and $\eta_n \xrightarrow{\mathbb{P}} \eta^*$, we obtain
    \begin{align*}
        & \frac{1}{n} \sum_{i=1}^n
        \left[ \frac{\rho^{\mathrm{ref}}(T_i \mid X_i)}{\rho_{\theta_{i-1}}(T_i \mid X_i)}
        \left[ \int_{\eta^*}^{\eta_n} \left[ \ddot{m}_\eta(X_i, Y_i(T_i), T_i)
        - \ddot{m}_{\eta^*}(X_i, Y_i(T_i), T_i) \right] d\eta \right] \right] \\
        & \quad = O_P(\| \eta_n - \eta^* \|^2)
        = o_P(\| \eta_n - \eta^* \|) .
    \end{align*}

    Applying Lemma \ref{lemma_allocation_convergence} once again yields
    \begin{align*}
        & \frac{1}{n} \sum_{i=1}^n
        \left[ \frac{\rho^{\mathrm{ref}}(T_i \mid X_i)}{\rho_{\theta_{i-1}}(T_i \mid X_i)} \ddot{m}_{\eta^*}(X_i, Y_i(T_i), T_i) \right] \\
        & \quad \xrightarrow{\mathbb{P}}
        \ddot{M}_{\eta^*} =
        \mathbb{E}_{(X,Y(1),Y(0)) \sim \Gamma_{X,\bm{Y}}} \left[ \rho^{\mathrm{ref}}(X) \ddot{m}_{\eta^*}(X, Y(1), 1)
        + [1 - \rho^{\mathrm{ref}}(X)] \ddot{m}_{\eta^*}(X, Y(0), 0) \right] .
    \end{align*}
    According to these above results, \eqref{eq_proof_lemma_estimation_asymptotic_normality_fff} reduces to
    \begin{align*}
        0 &= \frac{\dot{M}_{\eta^*, n}}{n}
        + \left[ \frac{1}{n} \sum_{i=1}^n
        \left[ \frac{\rho^{\mathrm{ref}}(T_i \mid X_i)}{\rho_{\theta_{i-1}}(T_i \mid X_i)}
        \ddot{m}_{\eta^*}(X_i, Y_i(T_i), T_i) \right] \right] (\eta_n - \eta^*)
        + o_P(\| \eta_n - \eta^* \|) \\
        &= \frac{\dot{M}_{\eta^*, n}}{n}
        + \left[ \ddot{M}_{\eta^*} + o_P(1) \right] (\eta_n - \eta^*) .
    \end{align*}
    Therefore,
    \begin{equation}
        \sqrt{n} (\eta_n - \eta^*)
        = - \left[ \ddot{M}_{\eta^*} + o_P(1) \right]^{-1}
        \frac{\dot{M}_{\eta^*, n}}{\sqrt{n}}
        = - \left[ 1 + o_P(1) \right]
        \ddot{M}_{\eta^*}^{-1}
        \frac{\dot{M}_{\eta^*, n}}{\sqrt{n}} \label{eq_proof_lemma_estimation_asymptotic_normality_4} .
    \end{equation}

    Because each term in the sum
    \begin{equation*}
        \ddot{M}_{\eta^*}^{-1} \dot{M}_{\eta^*, n}
        = \sum_{i=1}^n
        \left[ \frac{\rho^{\mathrm{ref}}(T_i \mid X_i)}{\rho_{\theta_{i-1}}(T_i \mid X_i)}
        \ddot{M}_{\eta^*}^{-1} \dot{m}_{\eta^*}(X_i, Y_i(T_i), T_i) \right]
    \end{equation*}
    can be transformed into
    \begin{align*}
        & \quad \frac{\rho^{\mathrm{ref}}(T_i \mid X_i)}{\rho_{\theta_{i-1}}(T_i \mid X_i)}
        \ddot{M}_{\eta^*}^{-1} \dot{m}_{\eta^*}(X_i, Y_i(T_i), T_i) \\
        &= T_i \frac{\rho^{\mathrm{ref}}(X_i)}{\rho_{\theta_{i-1}}(X_i)}
        \ddot{M}_{\eta^*}^{-1} \dot{m}_{\eta^*}(X_i, Y_i(1), 1)
        + (1-T_i) \frac{1-\rho^{\mathrm{ref}}(X_i)}{1-\rho_{\theta_{i-1}}(X_i)}
        \ddot{M}_{\eta^*}^{-1} \dot{m}_{\eta^*}(X_i, Y_i(0), 0) \\
        &= (T_i - \rho_{\theta_{i-1}}(X_i)) Z^{(c)}_i + Z^{(u)}_i ,
    \end{align*}
    where
    \begin{align*}
        Z^{(c)}_i
        &= 
        \frac{\rho^{\mathrm{ref}}(X_i)}{\rho_{\theta_{i-1}}(X_i)}
        \ddot{M}_{\eta^*}^{-1} \dot{m}_{\eta^*}(X_i, Y_i(1), 1)
        -
        \frac{1-\rho^{\mathrm{ref}}(X_i)}{1-\rho_{\theta_{i-1}}(X_i)}
        \ddot{M}_{\eta^*}^{-1} \dot{m}_{\eta^*}(X_i, Y_i(0), 0) , \\
        Z^{(u)}_i
        &= 
        \rho^{\mathrm{ref}}(X_i) \ddot{M}_{\eta^*}^{-1} \dot{m}_{\eta^*}(X_i, Y_i(1), 1)
        + (1-\rho^{\mathrm{ref}}(X_i)) \ddot{M}_{\eta^*}^{-1} \dot{m}_{\eta^*}(X_i, Y_i(0), 0) .
    \end{align*}

    The requirement for the distribution of $(X_i, Z^{(c)}_i, Z^{(u)}_i)$ in Assumption \ref{assumption_lipschitz_continuity_conditional_expectation} can be proved by Assumption \ref{assumption_lipschitz_continuity_functions_kernels} and Lemma \ref{lemma_h_lipschitz_continuous_bounded}.

    Therefore, by applying Lemmas \ref{lemma_lipschitz_continuity_conditional_expectation}, \ref{lemma_CLT} and \ref{lemma_expression_variance} with the Cramér--Wold device, it holds that
    \begin{align*}
        \frac{\ddot{M}_{\eta^*}^{-1} \dot{M}_{\eta^*, n}}{\sqrt{n}}
        &= \frac{1}{\sqrt{n}} \sum_{i=1}^n
        \left[ \frac{\rho^{\mathrm{ref}}(T_i \mid X_i)}{\rho_{\theta_{i-1}}(T_i \mid X_i)}
        \ddot{M}_{\eta^*}^{-1} \dot{m}_{\eta^*}(X_i, Y_i(T_i), T_i) \right] \\
        &= \frac{1}{\sqrt{n}} \sum_{i=1}^n \left[ (T_i - \rho_{\theta_{i-1}}(X_i)) Z^{(c)}_i + Z^{(u)}_i \right]
    \end{align*}
    is asymptotically normal with mean zero and covariance matrix
    \begin{align*}
        \Sigma_{(Z)}
        &= \Cov(Z^{(u)}) + \mathbb{E}_{(X, Z) \sim \Gamma_{X,Z}} \left[ \rho_{\theta^*}(X)(1-\rho_{\theta^*}(X)) \right. \\
        & \quad \left. \left\{ Z^{(c)} - \frac{A \phi(X)}{\rho_{\theta^*}(X)(1-\rho_{\theta^*}(X))} \right\}
        \left\{ Z^{(c)} - \frac{A \phi(X)}{\rho_{\theta^*}(X)(1-\rho_{\theta^*}(X))} \right\}^T \right] ,
    \end{align*}
    where
    \begin{align*}
        Z^{(c)}_i
        &= \frac{\rho^{\mathrm{ref}}(X_i)}{\rho_{\theta_{i-1}}(X_i)}
        \ddot{M}_{\eta^*}^{-1} \dot{m}_{\eta^*}(X_i, Y_i(1), 1)
        -
        \frac{1-\rho^{\mathrm{ref}}(X_i)}{1-\rho_{\theta_{i-1}}(X_i)}
        \ddot{M}_{\eta^*}^{-1} \dot{m}_{\eta^*}(X_i, Y_i(0), 0) , \\
        Z^{(u)}_i
        &= 
        \rho^{\mathrm{ref}}(X_i) \ddot{M}_{\eta^*}^{-1} \dot{m}_{\eta^*}(X_i, Y_i(1), 1)
        + (1-\rho^{\mathrm{ref}}(X_i)) \ddot{M}_{\eta^*}^{-1} \dot{m}_{\eta^*}(X_i, Y_i(0), 0) , \\
        Z^{(c)}
        &= \frac{\rho^{\mathrm{ref}}(X)}{\rho_{\theta^*}(X)}
        \ddot{M}_{\eta^*}^{-1} \dot{m}_{\eta^*}(X, Y(1), 1)
        -
        \frac{1-\rho^{\mathrm{ref}}(X)}{1-\rho_{\theta^*}(X)}
        \ddot{M}_{\eta^*}^{-1} \dot{m}_{\eta^*}(X, Y(0), 0) , \\
        Z^{(u)}
        &= 
        \rho^{\mathrm{ref}}(X) \ddot{M}_{\eta^*}^{-1} \dot{m}_{\eta^*}(X, Y(1), 1)
        + (1-\rho^{\mathrm{ref}}(X)) \ddot{M}_{\eta^*}^{-1} \dot{m}_{\eta^*}(X, Y(0), 0) ,
    \end{align*}
    and the matrix $A$ satisfies
    \begin{align*}
        & A \mathbb{E} \left[ [\rho_{\theta^*}(X)(1-\rho_{\theta^*}(X))]^{-1} \phi(X) \phi(X)^T
        / \max\{ \|\phi(X)\|/[\rho_{\theta^*}(X)(1-\rho_{\theta^*}(X))], C_{\theta^*} \} \right] \\
        & \quad = \mathbb{E} \left[ Z^{(c)} \phi(X)^T
        / \max\{ \|\phi(X)\|/[\rho_{\theta^*}(X)(1-\rho_{\theta^*}(X))], C_{\theta^*} \} \right] .
    \end{align*}
    
    Thus, $\frac{\ddot{M}_{\eta^*}^{-1} \dot{M}_{\eta^*, n}}{\sqrt{n}} = O_P(1)$.
    Consequently, from \eqref{eq_proof_lemma_estimation_asymptotic_normality_4}, we have
    \begin{equation*}
        \sqrt{n} (\eta_n - \eta^*)
        = - \frac{\ddot{M}_{\eta^*}^{-1} \dot{M}_{\eta^*, n}}{\sqrt{n}}
        + o_P(1) .
    \end{equation*}
\end{proof}

\section{Lemmas for New Allocation Form under the CBARA Procedure}
\label{sec_lemmas_allocation_form}

\subsection{Expectation under the Invariant Probability}

\begin{lemma}
    \label{lemma_expectation_pi_h}

    Suppose the additional covariate $Z$ is one-dimensional.
    If $\mathbb{E}_{X \sim \Gamma} \left[ \|\phi(X)\| \right] < \infty$, $\mathbb{E}_{\Lambda \sim \pi_\theta} \left[ \| \Lambda \| \right] < \infty$, $\mathbb{E} \left| Z \right| < \infty$ and the transition kernel $P_\theta$ is positive recurrent with an invariant probability $\pi_\theta$, then $\pi_\theta \left[ g_\theta(\cdot, x) \right] = \rho_\theta(x)$ for $\Gamma$-a.e. $x$ and
    \begin{equation*}
        \pi_\theta h_\theta = 0 ,
    \end{equation*}
    where
    \begin{equation*}
        h_\theta(\Lambda)
        = \mathbb{E}_{X \sim \Gamma} \left[ [g_\theta(\Lambda, X) - \rho_\theta(X)] Z \right] .
    \end{equation*}
\end{lemma}
\begin{proof}
    Denote
    \begin{equation*}
        \alpha_\theta(X) = \frac{\phi(X)}{\max\{\|\phi(X)\|/[\rho_\theta(X)(1-\rho_\theta(X))],C_\theta\}}
        \quad \text{and} \quad
        \beta(\Lambda) = \frac{\Lambda}{\max\{\|\Lambda\|,C_\Lambda\}} .
    \end{equation*}
    Then
    \begin{equation*}
        g_\theta(\Lambda,X) = \rho_\theta(X) - \alpha_\theta(X)^T \beta(\Lambda) ,
    \end{equation*}
    and hence
    \begin{align}
        \pi_\theta h_\theta
        &= - \mathbb{E}_{\Lambda \sim \pi_\theta, X \sim \Gamma}
        \left[ \frac{\phi(X)^T \Lambda}{\max\{\|\phi(X)\|/[\rho_\theta(X)(1-\rho_\theta(X))],C_\theta\} \max\{\|\Lambda\|,C_\Lambda\}} \cdot Z \right] \label{eq_proof_lemma_expectation_pi_h_1} \\
        &= - (\mathbb{E}_{(X,Z) \sim \Gamma_{X,Z}} \left[ Z \alpha_\theta(X)^T \right])
        (\pi_\theta \beta) \notag .
    \end{align}

    Let $\mathrm{id}_\Lambda$ denote the identity map on $\Lambda$.
    Since $\pi_\theta = \pi_\theta P_\theta$, we have
    \begin{equation*}
        (\pi_\theta P_\theta)(\mathrm{id}_\Lambda)
        = \pi_\theta \mathrm{id}_\Lambda .
    \end{equation*}
    Moreover,
    \begin{align*}
        & \quad [P_\theta(\Lambda, \mathrm{id}_\Lambda) - \mathrm{id}_\Lambda](\Lambda)
        = \mathbb{E}_\theta \left[ \Lambda_1 - \Lambda_0 \mid \Lambda_0 = \Lambda \right] \\
        &= \mathbb{E}_{X \sim \Gamma} \left[ \left[ g_\theta(\Lambda,X) - \rho_\theta(X) \right] \phi(X)/[\rho_\theta(X)(1-\rho_\theta(X))] \right] \\
        &= - (\mathbb{E}_{X \sim \Gamma} \left[ \phi(X) \alpha_\theta(X)^T/[\rho_\theta(X)(1-\rho_\theta(X))] \right])
        \beta(\Lambda) .
    \end{align*}
    It follows that
    \begin{align}
        & \quad 0
        = (\pi_\theta P_\theta - \pi_\theta)(\mathrm{id}_\Lambda)
        = - (\mathbb{E}_{X \sim \Gamma} \left[ \phi(X) \alpha_\theta(X)^T/[\rho_\theta(X)(1-\rho_\theta(X))] \right])
        (\pi_\theta \beta) \label{eq_proof_lemma_expectation_pi_h_2} \\
        &= - \left\{ \mathbb{E}_{X \sim \Gamma} \left[
            \frac{\phi(X) \phi(X)^T}{ \rho_\theta(X)(1-\rho_\theta(X)) \max\{\|\phi(X)\|/[\rho_\theta(X)(1-\rho_\theta(X))],C_\theta\} }
        \right] \right\}
        (\pi_\theta \beta) \notag .
    \end{align}
    By an argument similar to that in Subsection \ref{subsec_proof_lemma_all_negative_feedback}, the row space of
    \begin{equation*}
        \mathbb{E}_{X \sim \Gamma} \left[
            \frac{\phi(X) \phi(X)^T}{ \rho_\theta(X)(1-\rho_\theta(X)) \max\{\|\phi(X)\|/[\rho_\theta(X)(1-\rho_\theta(X))],C_\theta\} }
        \right]
    \end{equation*}
    is $W_\phi$.
    Thus, \eqref{eq_proof_lemma_expectation_pi_h_2} implies that $\pi_\theta \beta \in W_\phi^\perp$.

    Since $(\mathbb{E}_{(X,Z) \sim \Gamma_{X,Z}} \left[ Z \alpha_\theta(X)^T \right])^T \in W_\phi$, $\pi_\theta \beta \in W_\phi^\perp$ implies that \eqref{eq_proof_lemma_expectation_pi_h_1} equals zero, that is,
    \begin{equation*}
        \pi_\theta h_\theta
        = - (\mathbb{E}_{(X,Z) \sim \Gamma_{X,Z}} \left[ Z \alpha_\theta(X)^T \right])
        (\pi_\theta \beta)
        = 0 .
    \end{equation*}

    Denote the conditional expectation function $f_Z(x) = \mathbb{E} \left[ Z \mid X=x \right]$.
    Due to
    \begin{equation*}
        0 = \pi_\theta h_\theta
        = \mathbb{E}_{X \sim \Gamma} \left[ \left[ \pi_\theta \left[ g_\theta(\cdot, x) \right] - \rho_\theta(x) \right] f_Z(X) \right]
    \end{equation*}
    and the arbitrariness of the function $f_Z$, it holds that
    \begin{equation*}
        \pi_\theta \left[ g_\theta(\cdot, x) \right] - \rho_\theta(x) = 0 ,
    \end{equation*}
    for $\Gamma$-a.e. $x$.
\end{proof}

\subsection{Expression of the Variance}

\begin{lemma}
    \label{lemma_expression_variance}

    Suppose that Assumption \ref{assumption_simultaneous_properties} holds.
    If the allocation parameter sequence $\{\theta_n\}$ is the fixed parameter sequence $\{\theta_n = \theta\}_{n \in \mathbb{N}}$, $\mathbb{E} \left[ {Z^{(c)}}^4 \right] < \infty$ and $\mathbb{E} \left[ {Z^{(u)}}^4 \right] < \infty$, then
    \begin{equation*}
        \frac{1}{\sqrt{N}} \sum_{n=1}^{N}
        \left[ (T_n-\rho_\theta(X_n)) Z^{(c)}_n + Z^{(u)}_n
        - \pi_\theta h_\theta
        - \mathbb{E}Z^{(u)}_n \right]
        \xrightarrow{d} \mathcal{N}(0, \sigma_{(Z)}^2) ,
    \end{equation*}
    where $\pi_\theta$ is the invariant probability of the transition kernel $P_\theta$, the function $h_\theta$ is defined by
    \begin{equation*}
        h_\theta(\Lambda)
        = \mathbb{E}_{X \sim \Gamma} \left[ [g_\theta(\Lambda, X) - \rho_\theta(X)] f_{Z^{(c)}}(X) \right]
        = \mathbb{E}_{X \sim \Gamma} \left[ [g_\theta(\Lambda, X) - \rho_\theta(X)] \mathbb{E}[Z^{(c)} | X] \right] ,
    \end{equation*}
    and the asymptotic variance is
    \begin{equation*}
        \sigma_{(Z)}^2
        = \var(Z^{(u)})
        + \mathbb{E}_{(X, Z) \sim \Gamma_{X,Z}} \left[ \rho_\theta(X)(1-\rho_\theta(X)) \left\{ Z^{(c)} - \frac{a^T \phi(X)}{\rho_\theta(X)(1-\rho_\theta(X))} \right\}^2 \right] ,
    \end{equation*}
    where the vector $a$ satisfies
    \begin{align*}
        & a^T \mathbb{E} \left[ [\rho_\theta(X)(1-\rho_\theta(X))]^{-1} \phi(X) \phi(X)^T
        / \max\{ \|\phi(X)\|/[\rho_\theta(X)(1-\rho_\theta(X))], C_\theta \} \right] \\
        & \quad = \mathbb{E} \left[ Z^{(c)} \phi(X)^T
        / \max\{ \|\phi(X)\|/[\rho_\theta(X)(1-\rho_\theta(X))], C_\theta \} \right] .
    \end{align*}
\end{lemma}
\begin{proof}
    All aspects of this lemma, except for the explicit expression of the variance, are already stated as results in Lemma \ref{lemma_CLT}.
    As shown in the proof of Lemma \ref{lemma_CLT}, the variance is
    \begin{equation*}
        \sigma_{(Z)}^2 = \pi_\theta[G_\theta+F_\theta+2H_\theta] ,
    \end{equation*}
    and this expression equals the probabilistic representation
    \begin{equation*}
        \mathbb{E}_{(X, Z) \sim \Gamma_{X,Z,\theta}, \Lambda \sim \pi_\theta, T \sim \Bernoulli(1, g_\theta(\Lambda, X))} \left[ \left\{ \Delta M \right\}^2 \right] ,
    \end{equation*}
    where
    \begin{equation*}
        \Delta M =
        (T-\rho_\theta(X)) Z^{(c)} - h_\theta(\Lambda) + Z^{(u)}-\mathbb{E}Z^{(u)}
        + \hat{h}_\theta \left( \Lambda + \frac{(T-\rho_\theta(X)) \phi(X)}{\rho_\theta(X)(1-\rho_\theta(X))} \right)
        - (P_\theta\hat{h}_\theta)(\Lambda) .
    \end{equation*}

    Denote $g_\theta(\Lambda, \alpha) = \rho_\theta(X) - w_\theta(X) \alpha_\theta(X)^T \beta_\theta(\Lambda)$, where
    \begin{align*}
        \alpha_\theta(X) &= \frac{\phi(X)/[\rho_\theta(X)(1-\rho_\theta(X))]}{\max\{\|\phi(X)\|/[\rho_\theta(X)(1-\rho_\theta(X))],C_\theta\}} , \\
        \beta_\theta(\Lambda) &= \Lambda / \max\{\|\Lambda\|,C_\Lambda\} , \\
        w_\theta(X) &= \rho_\theta(X)(1-\rho_\theta(X)) .
    \end{align*}
    Since the parameter is fixed at $\theta$, we henceforth suppress the subscript $\theta$ in the following proof to simplify notation.

    Given the function $h$, the transition kernel $P$ and the invariant probability $\pi$, consider the following Poisson equation, which is common in the theory of Markov chain \cite{meynMarkovChainsStochastic2009}:
    \begin{equation*}
        \hat{h} - P \hat{h} = h - \pi h ,
    \end{equation*}
    where $\hat{h}$ denotes the solution to the Poisson equation.
    If $P$ is geometrically ergodic and $|h|_V < \infty$, the Poisson equation admits the following solution:
    \begin{equation*}
        \hat{h} = \sum_{n=0}^\infty (P^n - \pi) (h) .
    \end{equation*}

    An obvious fact is that this Poisson equation does not admit a unique solution.
    Suppose that
    \begin{equation*}
        S_{Z^{(c)}\alpha} = \mathbb{E} \left[ w(X) Z^{(c)} \alpha(X)^T \right]
        \quad \text{and} \quad
        S_{\phi\alpha} = \mathbb{E} \left[ w(X) \phi(X) \alpha(X)^T / [\rho(X)(1-\rho(X))] \right] .
    \end{equation*}
    Since the function $h$ is defined by
    \begin{equation*}
        h(\Lambda)
        = \mathbb{E} \left[ [g(\Lambda, X) - \rho(X)] f_{Z^{(c)}}(X) \right]
        = - \mathbb{E} \left[ w(X) \alpha(X)^T \beta(\Lambda) Z^{(c)} \right]
        = - S_{Z^{(c)}\alpha} \beta(\Lambda) ,
    \end{equation*}
    and
    \begin{equation*}
        \mathbb{E} \left[ \Lambda_1 \mid \Lambda_0 = \Lambda \right] - \Lambda
        = \mathbb{E} \left[ \frac{(T_1 - \rho(X_1)) \phi(X_1)}{\rho(X_1)(1-\rho(X_1))}
        + \Lambda_0 \middle| \Lambda_0 = \Lambda \right] - \Lambda
        = - S_{\phi\alpha} \beta(\Lambda) ,
    \end{equation*}
    both quantities are linear in $\beta(\Lambda)$.
    Hence, consider a linear candidate function
    \begin{equation*}
        \tilde{h}_a(\Lambda) = -a^T \Lambda .
    \end{equation*}
    It follows that
    \begin{align*}
        \tilde{h}_a(\Lambda) - \left(P\tilde{h}_a\right)(\Lambda)
        &= -a^T S_{\phi\alpha} \beta(\Lambda) .
    \end{align*}

    The proof of Lemma \ref{lemma_expectation_pi_h} shows that $S_{Z^{(c)}\alpha} \in \Row(S_{\phi\alpha})$.
    Thus, there exists a vector $a$ such that $a^T S_{\phi\alpha} = S_{Z^{(c)}\alpha}$.
    Consequently, $\hat{h} = \tilde{h}_a$ is a valid solution to the Poisson equation.
    
    By positive Harris recurrence in Assumption \ref{assumption_simultaneous_properties} and Proposition 17.4.1 in \cite{meynMarkovChainsStochastic2009}, the difference between
    \begin{equation*}
        \hat{h} = \tilde{h}_a
        \quad \text{and} \quad
        \hat{h} = \sum_{n=0}^\infty (P^n - \pi) (h)
    \end{equation*}
    is $\pi$-almost surely constant.
    Moreover, because both functions are continuous and Lemma \ref{lemma_for_simultaneous_small_set_condition} implies that $\pi$ is strictly positive on open sets, this difference must be constant for all $\Lambda$.

    Therefore, we can use $\hat{h} = \tilde{h}_a$ as an alternative to $\hat{h} = \sum_{n=0}^\infty (P^n - \pi) (h)$ in the expression of $\Delta M$.
    The resulting expression is
    \begin{align*}
        \Delta M
        &= (T-\rho(X)) Z^{(c)} - h(\Lambda) + Z^{(u)}-\mathbb{E}Z^{(u)}
        + \hat{h}(\Lambda + (T-\rho(X)) \phi(X)) - (P\hat{h})(\Lambda) \\
        &= (T-\rho(X)) Z^{(c)} - S_{Z^{(c)}\alpha} \beta(\Lambda)
        + Z^{(u)}-\mathbb{E}Z^{(u)} \\
        & \quad - a^T \left[ \Lambda + (T-\rho(X)) \phi(X) / [\rho(X)(1-\rho(X))] \right]
        + a^T \left[ \Lambda + S_{\phi\alpha} \beta(\Lambda) \right] \\
        &= (T-\rho(X)) \left\{ Z^{(c)} - a^T \phi(X) / [\rho(X)(1-\rho(X))] \right\} + Z^{(u)}-\mathbb{E}Z^{(u)} .
    \end{align*}
    In conclusion, the variance $\sigma_{(Z)}^2 = \pi_\theta[G_\theta+F_\theta+2H_\theta]$ equals
    \begin{align*}
        & \quad \mathbb{E}_{(X, Z) \sim \Gamma_{X,Z}, \Lambda \sim \pi, T \sim \mathrm{Ber}(1, g(\Lambda, X))} \left[ \left\{ \Delta M \right\}^2 \right] \\
        &= \mathbb{E}_{(X, Z) \sim \Gamma_{X,Z}, T \sim \mathrm{Ber}(1, \rho(X))} \left[ \left\{ \Delta M \right\}^2 \right] \\
        &= \mathbb{E}_{(X, Z) \sim \Gamma_{X,Z}} \left[ \left\{ Z^{(u)}-\mathbb{E}Z^{(u)} \right\}^2 \right]
        + \mathbb{E}_{(X, Z) \sim \Gamma_{X,Z}} \left[ \rho(X)(1-\rho(X)) \left\{ Z^{(c)} - \frac{a^T \phi(X)}{\rho(X)(1-\rho(X))} \right\}^2 \right] .
    \end{align*}
\end{proof}

\section{CLT}
\label{sec_CLT}

\subsection{Central Limit Theorem}

\label{subsec_central_limit_theorem}

In this section, we extend our analysis to a more general setting that relaxes the independence condition of Assumption \ref{assumption_iid}.
We now consider a scenario where the distribution of the additional covariate is parameterized by $\theta$, and the parameter $\theta$ is sequentially updated as the procedure progresses.

Specifically, the joint conditional distribution of $(X_n, Z_n)$ given the past information $\mathcal{F}_{n-1}$ is denoted by $\Gamma_{X,Z,\theta_{n-1}}$, which depends on a parameter $\theta_{n-1}$ that can be determined by the history.
A crucial constraint in our model is that the conditional marginal distribution of $X_n$ remains a fixed distribution $\Gamma$, regardless of the parameter $\theta_{n-1}$.
This more general setting allows the definition of $Z_n$ to depend on the targeted allocation ratio $\rho_{\theta_{n-1}}$.

Let $Z_n = (Z^{(c)}_n, Z^{(u)}_n)$ be a two-dimensional vector.
For a given parameter $\theta$ and a random variable $Z$, we define the conditional mean function $f_{Z, \theta}$ as
\begin{equation*}
    f_{Z, \theta}(x) := \mathbb{E}_\theta[Z | X = x] .
\end{equation*}
This function represents the expected value of $Z$ given $X_n = x$, under the model parameterized by $\theta$.
\begin{assumption}
    \label{assumption_lipschitz_continuity_conditional_expectation}
    
    Given $\mathcal{F}_{n-1}$, the conditional distribution of $(X_n, Z_n)$ is $\Gamma_{X,Z,\theta_{n-1}}$, and the marginal distribution of $X_n$ remains fixed as $\Gamma$.
    There exists a constant $L_f > 0$ and $\lambda_f > 0$ such that for any $\theta, \theta^\prime \in \Theta$ and for any family of functions
    \begin{equation*}
        \left\{ f_\theta \right\}_{\theta \in \Theta} \in \left\{
            \left\{ f_{Z^{(c)},\theta} \right\}_{\theta \in \Theta} ,
            \left\{ f_{Z^{(u)},\theta} \right\}_{\theta \in \Theta} ,
            \left\{ f_{\left( Z^{(c)} \right)^2,\theta} \right\}_{\theta \in \Theta} ,
            \left\{ f_{\left( Z^{(u)} \right)^2,\theta} \right\}_{\theta \in \Theta} ,
            \left\{ f_{(Z^{(u)}-\mathbb{E}[Z^{(u)}])Z^{(c)}, \theta} \right\}_{\theta \in \Theta}
        \right\} ,
    \end{equation*}
    the Lipschitz continuous condition
    \begin{equation*}
        \left\| e^{\lambda_f \|\phi(X)\|} \left| f_{\theta} - f_{\theta^\prime} \right| \right\|_{L^1(\Gamma)}
        \leq L_f d(\theta, \theta^\prime)
    \end{equation*}
    holds.
\end{assumption}
Sometimes the Lipschitz continuity condition in Assumption \ref{assumption_lipschitz_continuity_conditional_expectation} is overly restrictive and need not hold for all five function families specified therein.
In such cases, the assumption can be relaxed as follows.
\begin{assumption}
    \label{assumption_lipschitz_continuity_conditional_expectation_relaxed}
    
    Given $\mathcal{F}_{n-1}$, the conditional distribution of $(X_n, Z_n)$ is $\Gamma_{X,Z,\theta_{n-1}}$, and the marginal distribution of $X_n$ remains fixed as $\Gamma$.
    There exists a constant $L_f > 0$ and $\lambda_f > 0$ such that for any $\theta, \theta^\prime \in \Theta$ and for any family of functions
    \begin{equation*}
        \left\{ f_\theta \right\}_{\theta \in \Theta} \in \left\{
            \left\{ f_{Z^{(c)},\theta} \right\}_{\theta \in \Theta} ,
            \left\{ f_{Z^{(u)},\theta} \right\}_{\theta \in \Theta}
        \right\} ,
    \end{equation*}
    the Lipschitz continuous condition
    \begin{equation*}
        \left\| f_{\theta} - f_{\theta^\prime} \right\|_{L^1(\Gamma)}
        \leq L_f d(\theta, \theta^\prime)
    \end{equation*}
    holds.
\end{assumption}

\begin{lemma}
    \label{lemma_CLT}

    Suppose that Assumptions \ref{assumption_lipschitz_continuity_functions_kernels}, \ref{assumption_lipschitz_continuity_conditional_expectation} and \ref{assumption_simultaneous_properties} hold.
    If the allocation parameter sequence $\{\theta_n\}$ satisfies Assumption \ref{assumption_theta_convergence_strong_diminishing_adaptation}, $\sup_{\theta \in \Theta} \mathbb{E}_\theta \left[ {Z^{(c)}}^4 \right] < \infty$ and $\sup_{\theta \in \Theta} \mathbb{E}_\theta \left[ {Z^{(u)}}^4 \right] < \infty$, then
    \begin{equation*}
        \frac{1}{\sqrt{N}} \sum_{n=1}^{N}
        \left[ (T_n-\rho_{\theta_{n-1}}(X_n)) Z^{(c)}_n + Z^{(u)}_n
        - \pi_{\theta_{n-1}} h_{\theta_{n-1}}
        - \mathbb{E}Z^{(u)}_n \right]
        \xrightarrow{d} \mathcal{N}(0, {\sigma^*_{(Z)}}^2) ,
    \end{equation*}
    where $\pi_\theta$ is the invariant probability of the transition kernel $P_\theta$, the function $h_\theta$ is defined by
    \begin{equation*}
        h_\theta(\Lambda)
        = \mathbb{E}_{X \sim \Gamma} \left[ [g_\theta(\Lambda, X) - \rho_\theta(X)] f_{Z^{(c)},\theta}(X) \right] ,
    \end{equation*}
    and the asymptotic variance ${\sigma^*_{(Z)}}^2$ equals that under the fixed parameter sequence $\{\theta_n = \theta^*\}_{n \in \mathbb{N}}$.
\end{lemma}
\begin{proof}
    Recall that
    \begin{equation*}
        h_\theta(\Lambda) = \mathbb{E} \left[ [g_\theta(\Lambda, X) - \rho_\theta(X)] f_{Z^{(c)},\theta}(X) \right] .
    \end{equation*}

    Given a function $h_\theta$, the transition kernel $P_\theta$, and its invariant probability measure $\pi_\theta$, consider the following Poisson equation, which is standard in the theory of Markov chains \cite{meynMarkovChainsStochastic2009}:
    \begin{equation*}
        \hat{h}_\theta - P_\theta \hat{h}_\theta = h_\theta - \pi_\theta h_\theta ,
    \end{equation*}
    where $\hat{h}_\theta$ denotes the solution to the Poisson equation.
    If $P_\theta$ is geometrically ergodic and $|h_\theta|_V < \infty$, the Poisson equation admits the following solution:
    \begin{equation*}
        \hat{h}_\theta = \sum_{n=0}^\infty (P_\theta^n - \pi_\theta) (h_\theta) .
    \end{equation*}

    We decompose the target expression into a sum:
    \begin{equation*}
        \frac{1}{\sqrt{N}} \sum_{n=1}^{N}
        \left[ (T_n-\rho_{\theta_{n-1}}(X_n)) Z^{(c)}_n + Z^{(u)}_n
        - \pi_{\theta_{n-1}} h_{\theta_{n-1}}
        - \mathbb{E}Z^{(u)}_n \right]
        = T_{N,1} + T_{N,2} + T_{N,3} + T_{N,4} ,
    \end{equation*}
    where
    \begin{align*}
        T_{N,1} &:= \frac{1}{\sqrt{N}} \sum_{n=0}^{N-1} \left[ (T_{n+1}-\rho_{\theta_n}(X_{n+1})) Z^{(c)}_{n+1} - h_{\theta_n}(\Lambda_n) + Z^{(u)}_{n+1}-\mathbb{E}Z^{(u)}_{n+1} \right] , \\
        T_{N,2} &:= \frac{1}{\sqrt{N}} \sum_{n=0}^{N-1} \left[ \hat{h}_{\theta_n}(\Lambda_{n+1}) - (P_{\theta_n}\hat{h}_{\theta_n})(\Lambda_n) \right] , \\
        T_{N,3} &:= \frac{1}{\sqrt{N}} \sum_{n=0}^{N-1} \left[ \hat{h}_{\theta_{n+1}}(\Lambda_{n+1}) - \hat{h}_{\theta_n}(\Lambda_{n+1}) \right] , \\
        T_{N,4} &:= \frac{1}{\sqrt{N}} \left[ \hat{h}_{\theta_0}(\Lambda_0) - \hat{h}_{\theta_N}(\Lambda_N) \right] .
    \end{align*}
    We will prove that the sum of $T_{N,1}$ and $T_{N,2}$ is asymptotically normal, while the last two terms are $o_P(1)$.

    The sum of $T_{N,1}$ and $T_{N,2}$ forms a martingale sequence.
    The corresponding martingale difference sequence, denoted by $\{\Delta M_n\}_{n \in \mathbb{N}^*}$, is given by
    \begin{equation*}
        \Delta M_{n+1} =
        (T_{n+1}-\rho_{\theta_n}(X_{n+1})) Z^{(c)}_{n+1} - h_{\theta_n}(\Lambda_n) + Z^{(u)}_{n+1}-\mathbb{E}Z^{(u)}_{n+1}
        + \hat{h}_{\theta_n}(\Lambda_{n+1}) - (P_{\theta_n}\hat{h}_{\theta_n})(\Lambda_n) .
    \end{equation*}
    Moreover, the conditional variance is
    \begin{align*}
        & \quad \mathbb{E} \left[ \left\{ \Delta M_{n+1} \right\}^2
        \mid \mathscr{F}_n \right] \\
        &= \mathbb{E} \left[ \left[ (T_{n+1}-\rho_{\theta_n}(X_{n+1})) Z^{(c)}_{n+1} - h_{\theta_n}(\Lambda_n) + Z^{(u)}_{n+1}-\mathbb{E}Z^{(u)}_{n+1}
        + \hat{h}_{\theta_n}(\Lambda_{n+1}) - (P_{\theta_n}\hat{h}_{\theta_n})(\Lambda_n)
        \right]^2
        \mid \mathscr{F}_n \right] \\
        &= \mathbb{E} \left[ \left[ (T_{n+1}-\rho_{\theta_n}(X_{n+1})) Z^{(c)}_{n+1} - h_{\theta_n}(\Lambda_n) + Z^{(u)}_{n+1}-\mathbb{E}Z^{(u)}_{n+1} \right]^2 \mid \mathscr{F}_n \right] \\
        & \quad + \mathbb{E} \left[ \left[ \hat{h}_{\theta_n}(\Lambda_{n+1}) - (P_{\theta_n}\hat{h}_{\theta_n})(\Lambda_n) \right]^2 \mid \mathscr{F}_n \right] \\
        & \quad + 2 \mathbb{E} \left[ \left[ (T_{n+1}-\rho_{\theta_n}(X_{n+1})) Z^{(c)}_{n+1} - h_{\theta_n}(\Lambda_n) + Z^{(u)}_{n+1}-\mathbb{E}Z^{(u)}_{n+1} \right]
        \left[ \hat{h}_{\theta_n}(\Lambda_{n+1}) - (P_{\theta_n}\hat{h}_{\theta_n})(\Lambda_n) \right] \mid \mathscr{F}_n \right] \\
        &= \mathbb{E} \left[ \left[ (T_{n+1}-\rho_{\theta_n}(X_{n+1})) Z^{(c)}_{n+1} - h_{\theta_n}(\Lambda_n) + Z^{(u)}_{n+1}-\mathbb{E}Z^{(u)}_{n+1} \right]^2 \mid \mathscr{F}_n \right] \\
        & \quad + \mathbb{E} \left[ P_{\theta_n}(\hat{h}_{\theta_n}^2)(\Lambda_n) - (P_{\theta_n}\hat{h}_{\theta_n})^2(\Lambda_n) \mid \mathscr{F}_n \right] \\
        & \quad + 2 \mathbb{E} \left[ \left[ (T_{n+1}-\rho_{\theta_n}(X_{n+1})) Z^{(c)}_{n+1} - h_{\theta_n}(\Lambda_n) + Z^{(u)}_{n+1}-\mathbb{E}Z^{(u)}_{n+1} \right]
        \left[ \hat{h}_{\theta_n}(\Lambda_{n+1}) - (P_{\theta_n}\hat{h}_{\theta_n})(\Lambda_n) \right] \mid \mathscr{F}_n \right] \\
        &= G_{\theta_n}(\Lambda_n) + F_{\theta_n}(\Lambda_n) + 2 H_{\theta_n}(\Lambda_n) ,
    \end{align*}
    where the functions $G_\theta$, $F_\theta$, and $H_\theta$ are defined in Subsection \ref{subsec_Holder_continuity_functions}.

    We now turn to the following limit
    \begin{equation*}
        \frac{1}{N} \sum_{n=0}^{N-1} \left[ G_{\theta_n}(\Lambda_n) + F_{\theta_n}(\Lambda_n) + 2H_{\theta_n}(\Lambda_n) \right]
        \xrightarrow{\mathbb{P}} \pi_{\theta^*}[G_{\theta^*}+F_{\theta^*}+2H_{\theta^*}] .
    \end{equation*}
    By Lemma \ref{lemma_convergence_average_enter}, the limit holds provided that
    \begin{equation}
        \frac{1}{N} \sum_{n=0}^{N-1} \pi_{\theta_n} \left[ G_{\theta_n} + F_{\theta_n} + 2H_{\theta_n} \right]
        \xrightarrow{\mathbb{P}} {\sigma^*_{(Z)}}^2 = \pi_{\theta^*}[G_{\theta^*}+F_{\theta^*}+2H_{\theta^*}]
        \label{eq_proof_lemma_CLT_convergence_1} .
    \end{equation}

    The limit \eqref{eq_proof_lemma_CLT_convergence_1} can be proved as follows.
    Based on the properties in Lemmas \ref{lemma_properties_F}--\ref{lemma_properties_H}, a direct consequence of Theorem \ref{theorem_V_invariant_probability_Holder} is that for any $\alpha \in (0,1)$, there exists some constants $L_\alpha > 0$ and $\Delta > 0$ such that for all $\theta, \theta^\prime$ satisfying $d(\theta, \theta^\prime) < \Delta$,
    \begin{equation*}
        \left| \pi_\theta \left[ G_\theta + F_\theta + 2 H_\theta \right]
        - \pi_{\theta^\prime} \left[ G_{\theta^\prime} + F_{\theta^\prime} + 2 H_{\theta^\prime} \right] \right|
        \leq L_\alpha [d(\theta, \theta^\prime)]^\alpha .
    \end{equation*}
    The function $\theta \mapsto \pi_\theta \left[ G_\theta + F_\theta + 2H_\theta \right]$ is $\alpha$-Hölder continuous with a Hölder constant $L_\alpha$ when $d(\theta, \theta^\prime) < \Delta$.
    Thus, by $\theta_n \xrightarrow{\mathbb{P}} \theta^*$ in Assumption \ref{assumption_theta_convergence_strong_diminishing_adaptation} and compactness of $\Theta$ in Assumption \ref{assumption_theta_compact},
    \begin{align*}
        & \quad \left| \frac{1}{N} \sum_{n=0}^{N-1} \pi_{\theta_n} \left[ G_{\theta_n} + F_{\theta_n} + 2H_{\theta_n} \right]
        - \pi_{\theta^*}[G_{\theta^*}+F_{\theta^*}+2H_{\theta^*}] \right| \\
        & \leq \frac{1}{N} \sum_{n=0}^{N-1} \left| \pi_{\theta_n} \left[ G_{\theta_n} + F_{\theta_n} + 2H_{\theta_n} \right]
        - \pi_{\theta^*}[G_{\theta^*}+F_{\theta^*}+2H_{\theta^*}] \right| \\
        & \leq \frac{1}{N} \sum_{n=0}^{N-1} \left[ L_\alpha \mathbb{I}(d(\theta_n, \theta^*) < \Delta)
        [d(\theta_n, \theta^*)]^\alpha \right]
        + \frac{1}{N} \sum_{n=0}^{N-1} \left[ 2 \mathbb{I}(d(\theta_n, \theta^*) \geq \Delta)
        \sup_{\theta \in \Theta} \pi_\theta \left| G_\theta + F_\theta + 2 H_\theta \right| \right] \\
        & \xrightarrow{\mathbb{P}} 0 .
    \end{align*}
    The last inequality follows from the fact that
    \begin{equation*}
        \frac{1}{N} \sum_{n=1}^N d^{\alpha_0}(\theta_n, \theta_{n+1}) \mathbb{I}(d(\theta_n, \theta_{n+1}) < 1) \xrightarrow{\mathbb{P}} 0
        \quad \text{and} \quad
        \frac{1}{N} \sum_{n=1}^N
        \mathbb{I} \left( d(\theta_n,\theta_{n+1}) \geq 1 \right) \xrightarrow{\mathbb{P}} 0 .
    \end{equation*}
    These limits correspond to Assumption \ref{assumption_theta_difference_average_convergence}, which follows from Assumption \ref{assumption_theta_convergence_strong_diminishing_adaptation}.

    Thus, \eqref{eq_proof_lemma_CLT_convergence_1} has been established.
    Therefore, for the asymptotic normality, it remains to verify the Lindeberg condition.
    For any $\epsilon>0$,
    \begin{align*}
        & \quad \frac{1}{N} \sum_{n=0}^{N-1} \mathbb{E} \left[ \Delta M_{n+1}^2
        \mathbb{I}(|\Delta M_{n+1}| \geq \epsilon \sqrt{N}) \right] \\
        & \leq \frac{25}{N} \sum_{n=0}^{N-1} \mathbb{E} \left[ ((T_{n+1}-\rho_{\theta_n}(X_{n+1}))Z^{(c)}_{n+1})^2
        \mathbb{I}(|(T_{n+1}-\rho_{\theta_n}(X_{n+1}))Z^{(c)}_{n+1}| \geq \epsilon \sqrt{N}/5) \right] \\
        & \quad + \frac{25}{N} \sum_{n=0}^{N-1} \mathbb{E} \left[ \left( Z^{(u)}_{n+1}-\mathbb{E}Z^{(u)}_{n+1} \right)^2
        \mathbb{I}(|Z^{(u)}_{n+1}-\mathbb{E}Z^{(u)}_{n+1}| \geq \epsilon \sqrt{N}/5) \right] \\
        & \quad + \frac{25}{N} \sum_{n=0}^{N-1} \mathbb{E} \left[ (h_{\theta_n}(\Lambda_n))^2
        \mathbb{I}(|h_{\theta_n}(\Lambda_n)| \geq \epsilon \sqrt{N}/5) \right] \\
        & \quad + \frac{25}{N} \sum_{n=0}^{N-1} \mathbb{E} \left[ (\hat{h}_{\theta_n}(\Lambda_{n+1}))^2
        \mathbb{I}(|\hat{h}_{\theta_n}(\Lambda_{n+1})| \geq \epsilon \sqrt{N}/5) \right] \\
        & \quad + \frac{25}{N} \sum_{n=0}^{N-1} \mathbb{E} \left[ ((P_{\theta_n}\hat{h}_{\theta_n})(\Lambda_n))^2
        \mathbb{I}(|(P_{\theta_n}\hat{h}_{\theta_n})(\Lambda_n)| \geq \epsilon \sqrt{N}/5) \right] \\
        & \leq \frac{25}{N} \sum_{n=0}^{N-1} \mathbb{E} \left[ \left( Z^{(c)}_{n+1} \right)^2
        \mathbb{I}(|Z^{(c)}_{n+1}| \geq \epsilon \sqrt{N}/5) \right]
        + \frac{25}{N} \sum_{n=0}^{N-1} \mathbb{E} \left[ (\mathbb{E}|Z^{(c)}|)^2
        \mathbb{I}((\mathbb{E}|Z^{(c)}|)^2 \geq \epsilon \sqrt{N}/5) \right] \\
        & \quad + \frac{25}{N} \sum_{n=0}^{N-1} \mathbb{E} \left[ \left( Z^{(u)}_{n+1}-\mathbb{E}Z^{(u)}_{n+1} \right)^2
        \mathbb{I}(|Z^{(u)}_{n+1}-\mathbb{E}Z^{(u)}_{n+1}| \geq \epsilon \sqrt{N}/5) \right] \\
        & \quad + \frac{25}{N} \sum_{n=0}^{N-1} \mathbb{E} \left[ (C_{\hat{h},\gamma} V^{\gamma}(\Lambda_{n+1}))^2
        \mathbb{I}(|C_{\hat{h},\gamma} V^{\gamma}(\Lambda_{n+1})| \geq \epsilon \sqrt{N}/5) \right] \\
        & \quad + \frac{25}{N} \sum_{n=0}^{N-1} \mathbb{E} \left[ (C_{P\hat{h},\gamma} V^{\gamma}(\Lambda_n))^2
        \mathbb{I}(|C_{P\hat{h},\gamma} V^{\gamma}(\Lambda_n)| \geq \epsilon \sqrt{N}/5) \right] ,
    \end{align*}
    where the constant $C_{P\hat{h},\gamma} = C_{\hat{h},\gamma}(\beta_\gamma+b_\gamma)$ satisfies $|P\hat{h}| \leq C_{\hat{h},\gamma} PV^\gamma \leq C_{\hat{h},\gamma}(\beta_\gamma+b_\gamma)V^\gamma$.
    Because The second moments of the sequence $\{(Z^{(c)}_{n+1},Z^{(u)}_{n+1})\}$ are uniformly bounded, the first three terms on the right-hand side converge to zero.
    The fourth and fifth terms can be derived from the inequality that
    \begin{align*}
        & \quad \frac{1}{N} \sum_{n=0}^{N-1} \mathbb{E} \left[ V^{2\gamma}(\Lambda_n) \mathbb{I}(|V^{\gamma}(\Lambda_n)| \geq \epsilon \sqrt{N}/5) \right]
        \leq \frac{1}{N} \sum_{n=0}^{N-1} \mathbb{E} \left[ \frac{5}{\epsilon \sqrt{N}} V^{3\gamma}(\Lambda_n) \right] \\
        & \leq \frac{1}{N} \sum_{n=0}^{N-1} \frac{5}{\epsilon \sqrt{N}} \max \left\{ \frac{b_{3\gamma}}{1-\beta_{3\gamma}},\mathbb{E} V(\Lambda_0) \right\}
        \leq \frac{5}{\epsilon \sqrt{N}}
        \max \left\{ \frac{b_{3\gamma}}{1-\beta_{3\gamma}},\mathbb{E} V(\Lambda_0) \right\}
        \rightarrow 0 ,
    \end{align*}
    as $N \rightarrow \infty$, where $\gamma \in (0,1/3]$.

    Thus, we have already established the law of large numbers for the conditional variance and verified the Lindeberg condition.
    Therefore, we can conclude from Corollary 3.1 in \cite{hallMartingaleLimitTheory1980} that
    \begin{align*}
        \frac{1}{\sqrt{N}} \sum_{n=0}^{N-1} \left[ (T_{n+1}-\rho_{\theta_n}(X_{n+1})) Z^{(c)}_{n+1} - h_{\theta_n}(\Lambda_n) + Z^{(u)}_{n+1}-\mathbb{E}Z^{(u)}_{n+1}
        + \hat{h}_{\theta_n}(\Lambda_{n+1}) - (P_{\theta_n}\hat{h}_{\theta_n})(\Lambda_n) \right] \\
        \xrightarrow{d} \mathcal{N}(0, {\sigma^*_{(Z)}}^2) . &
    \end{align*}

    For $T_{N,3} = \frac{1}{\sqrt{N}} \sum_{n=0}^{N-1} \left[ \hat{h}_{\theta_{n+1}}(\Lambda_{n+1}) - \hat{h}_{\theta_n}(\Lambda_{n+1}) \right]$, by Corollary \ref{corollary_V_Poisson_solution_robust_lipschitz_continuous_bounded}, we have
    \begin{equation*}
        \left| \hat{h}_{\theta_{n+1}}(\Lambda_{n+1}) - \hat{h}_{\theta_n}(\Lambda_{n+1}) \right|
        \leq \tilde{L}_{\hat{h}, \kappa, 1-\kappa} V^\kappa(\Lambda_{n+1})[d(\theta_n,\theta_{n+1})]^{1-\kappa}
        + 2C_{\hat{h},\kappa} V^\kappa(\Lambda_{n+1}) \mathbb{I}(d(\theta_n,\theta_{n+1}) \geq \tilde{\Delta}_{\hat{h}, \kappa, 1-\kappa})
    \end{equation*}
    with any $\kappa \in (0,1)$.
    Moreover, when $\kappa + p(1-\kappa) = \frac{1}{2}$, where $p  \in (0,1/2)$ is the constant in Assumption \ref{assumption_theta_convergence_strong_diminishing_adaptation}, it holds that
    \begin{align*}
        & \quad |T_{N,3}| \\
        & \leq \frac{1}{\sqrt{N}} \sum_{n=0}^{N-1} \left[ \tilde{L}_{\hat{h}, \kappa, 1-\kappa} V^\kappa(\Lambda_{n+1})[d(\theta_n,\theta_{n+1})]^{1-\kappa} \right] \\
        & \quad + \frac{1}{\sqrt{N}} \sum_{n=0}^{N-1} \left[ 2C_{\hat{h},\kappa} V^\kappa(\Lambda_{n+1}) \mathbb{I}(d(\theta_n,\theta_{n+1}) \geq \tilde{\Delta}_{\hat{h}, \kappa, 1-\kappa}) \right] \\
        & \leq \frac{\tilde{L}_{\hat{h}, \kappa, 1-\kappa}}{\sqrt{N}} \left( \sum_{n=0}^{N-1} V(\Lambda_{n+1}) \right)^{\kappa} \left( \sum_{n=0}^{N-1} d(\theta_n,\theta_{n+1}) \right)^{1-\kappa} \\
        & \quad + \frac{2C_{\hat{h},\kappa}}{\sqrt{N}} \left( \sum_{n=0}^{N-1} V(\Lambda_{n+1}) \right)^{\kappa} \left( \sum_{n=0}^{N-1} \mathbb{I}(d(\theta_n,\theta_{n+1}) \geq \tilde{\Delta}_{\hat{h}, \kappa, 1-\kappa}) \right)^{1-\kappa} \\
        &= \tilde{L}_{\hat{h}, \kappa, 1-\kappa}
        \left( \frac{1}{N} \sum_{n=0}^{N-1} V(\Lambda_{n+1}) \right)^{\kappa}
        \left( \frac{1}{N^p} \sum_{n=0}^{N-1} d(\theta_n,\theta_{n+1}) \right)^{1-\kappa} \\
        & \quad + 2C_{\hat{h},\kappa}
        \left( \frac{1}{N} \sum_{n=0}^{N-1} V(\Lambda_{n+1}) \right)^{\kappa}
        \left( \frac{1}{N^p} \sum_{n=0}^{N-1} \mathbb{I}(d(\theta_n,\theta_{n+1}) \geq \tilde{\Delta}_{\hat{h}, \kappa, 1-\kappa}) \right)^{1-\kappa} \\
        & \xrightarrow{\mathbb{P}} 0
    \end{align*}
    by $\frac{1}{N} \sum_{n=1}^N V(\Lambda_n) = O_P(1)$ in Lemma \ref{lemma_convergence_V_as} and $\sum_{n=0}^{N-1} d(\theta_n,\theta_{n+1}) = o_P(N^p)$ in Assumption \ref{assumption_theta_convergence_strong_diminishing_adaptation}.

    Finally, $T_{N,4} = \frac{1}{\sqrt{N}} \left[ \hat{h}_{\theta_0}(\Lambda_0) - \hat{h}_{\theta_N}(\Lambda_N) \right] = o_P(1)$ can be obtained by the bound of $\hat{h}$ in Corollary \ref{corollary_V_Poisson_solution_robust_lipschitz_continuous_bounded} and $V(\Lambda_n) = O_P(1)$ in Lemma \ref{lemma_convergence_V_as} directly.

    Combining the asymptotic normality of $T_{N,1}+T_{N,2}$ with the fact that $T_{N,3} = o_P(1)$ and $T_{N,4} = o_P(1)$, we conclude that
    \begin{equation*}
        \frac{1}{\sqrt{N}} \sum_{n=1}^{N}
        \left[ (T_n-\rho_{\theta_{n-1}}(X_n)) Z^{(c)}_n + Z^{(u)}_n
        - \pi_{\theta_{n-1}} h_{\theta_{n-1}}
        -\mathbb{E}Z^{(u)}_n \right]
        \xrightarrow{d} \mathcal{N}(0, {\sigma^*_{(Z)}}^2) .
    \end{equation*}
\end{proof}

\subsection{Continuity of Functions $h$, $F$, $G$, $H$}

\label{subsec_Holder_continuity_functions}

This subsection consists mainly of tedious calculations.
Throughout the proofs, the Cauchy-Schwarz inequality and Hölder's inequality are used repeatedly.
Therefore, we do not explicitly indicate each instance in which these inequalities are applied.
The continuity results in this subsection can be combined with Lemma \ref{lemma_convergence_average_enter} to establish the law of large numbers for the conditional variance in the proof of Lemma \ref{lemma_CLT}.

\subsubsection{Properties of $h$}

Define
\begin{equation*}
    h_\theta(\Lambda)
    = \mathbb{E}_{X \sim \Gamma} \left[ [g_\theta(\Lambda, X) - \rho_\theta(X)] f_{Z^{(c)},\theta}(X) \right] .
\end{equation*}

\begin{lemma}
    \label{lemma_h_lipschitz_continuous_bounded}

    Let the Lyapunov function $V: \mathrm{X} \to [1, \infty)$.
    Suppose that Assumptions \ref{assumption_lipschitz_continuity_functions_kernels} and \ref{assumption_lipschitz_continuity_conditional_expectation_relaxed} hold.
    If $\sup_{\theta \in \Theta} \mathbb{E}_\theta \left[ {Z^{(c)}}^2 \right] < \infty$, then the family of functions $\{h_\theta\}_{\theta \in \Theta}$ is bounded by $\sup_{\theta \in \Theta} \mathbb{E}_\theta \left| Z^{(c)} \right|$ and Lipschitz continuous with respect to $(\theta,\Lambda)$ with a Lipschitz constant $L_P = L_g \sqrt{\sup_{\theta \in \Theta} \mathbb{E}_\theta \left[ {Z^{(c)}}^2 \right]} + L_\rho \sqrt{\sup_{\theta \in \Theta} \mathbb{E}_\theta \left[ {Z^{(c)}}^2 \right]} + 2L_f$.
    In particular, for any $\alpha,\tilde{\alpha},\kappa,\tilde{\kappa} \in (0,1)$ and $\gamma \in (0,1]$, the family of functions $\{h_\theta\}_{\theta \in \Theta}$ is $((1, L_{h, \kappa, \alpha} V^\kappa, \alpha), (1, \tilde{L}_{h, \tilde{\kappa}, \tilde{\alpha}} V^{\tilde{\kappa}}, \tilde{\alpha}))$-joint locally Hölder continuous, and bounded by $C_{h,\gamma} V^\gamma$ with corresponding positive constants.
\end{lemma}
\begin{proof}
    First, by Jensen's inequality,
    \begin{equation*}
        \left| h_\theta(\Lambda) \right| \leq \mathbb{E}_\theta \left| Z^{(c)} \right| .
    \end{equation*}
    Moreover,
    \begin{align*}
        & \quad \left| h_\theta(\Lambda) - h_{\theta^\prime}(\Lambda^\prime) \right| \\
        & \leq \mathbb{E}_{X \sim \Gamma} \left|
            [g_\theta(\Lambda, X) - \rho_\theta(X)] f_{Z^{(c)},\theta}(X)
            - [g_{\theta^\prime}(\Lambda^\prime, X) - \rho_{\theta^\prime}(X)] f_{Z^{(c)},\theta^\prime}(X)
        \right| \\
        & \leq \mathbb{E}_{X \sim \Gamma} \left|
            g_\theta(\Lambda, X) f_{Z^{(c)},\theta}(X)
            - g_{\theta^\prime}(\Lambda^\prime, X) f_{Z^{(c)},\theta^\prime}(X)
        \right|
        + \mathbb{E}_{X \sim \Gamma} \left|
            \rho_\theta(X) f_{Z^{(c)},\theta}(X)
            - \rho_{\theta^\prime}(X) f_{Z^{(c)},\theta^\prime}(X)
        \right| \\
        & \leq \mathbb{E}_{X \sim \Gamma} \left[ \left| g_\theta(\Lambda, X) - g_{\theta^\prime}(\Lambda^\prime, X) \right| |f_{Z^{(c)},\theta}(X)| \right]
        + \mathbb{E}_{X \sim \Gamma} \left[ g_{\theta^\prime}(\Lambda^\prime, X) \left| f_{Z^{(c)},\theta}(X) - f_{Z^{(c)},\theta^\prime}(X) \right| \right] \\
        & \quad + \mathbb{E}_{X \sim \Gamma} \left[ \left| \rho_\theta(X) - \rho_{\theta^\prime}(X) \right| |f_{Z^{(c)},\theta}(X)| \right]
        + \mathbb{E}_{X \sim \Gamma} \left[ \rho_{\theta^\prime}(X) \left| f_{Z^{(c)},\theta}(X) - f_{Z^{(c)},\theta^\prime}(X) \right| \right] \\
        & \leq \mathbb{E}_{X \sim \Gamma} \left[ \left| g_\theta(\Lambda, X) - g_{\theta^\prime}(\Lambda^\prime, X) \right| |f_{Z^{(c)},\theta}(X)| \right]
        + \mathbb{E}_{X \sim \Gamma} \left[ \left| \rho_\theta(X) - \rho_{\theta^\prime}(X) \right| |f_{Z^{(c)},\theta}(X)| \right] \\
        & \quad + 2 \mathbb{E}_{X \sim \Gamma} \left[ \left| f_{Z^{(c)},\theta}(X) - f_{Z^{(c)},\theta^\prime}(X) \right| \right] \\
        & \leq \sqrt{\sup_{\theta \in \Theta} \mathbb{E}_\theta \left[ {Z^{(c)}}^2 \right]} \left\| g_{\theta}(\Lambda,\cdot) - g_{\theta^\prime}(\Lambda^\prime,\cdot) \right\|_{L^2(\Gamma)}
        + \sqrt{\sup_{\theta \in \Theta} \mathbb{E}_\theta \left[ {Z^{(c)}}^2 \right]} \left\| \rho_{\theta} - \rho_{\theta^\prime} \right\|_{L^2(\Gamma)} \\
        & \quad + 2 \left\| f_{Z^{(c)},\theta} - f_{Z^{(c)},\theta^\prime} \right\|_{L^1(\Gamma)} \\
        & \leq \left( L_g \sqrt{\sup_{\theta \in \Theta} \mathbb{E}_\theta \left[ {Z^{(c)}}^2 \right]}
        + L_\rho \sqrt{\sup_{\theta \in \Theta} \mathbb{E}_\theta \left[ {Z^{(c)}}^2 \right]}
        + 2L_f \right)
        \left( d(\theta,\theta^\prime) + d(\Lambda,\Lambda^\prime) \right) .
    \end{align*}
    This proves that $\{h_\theta\}_{\theta \in \Theta}$ is Lipschitz continuous in both $\theta$ and $\Lambda$. 

    Consequently, by $V \geq 1$, $\alpha,\tilde{\alpha},\kappa,\tilde{\kappa} \in (0,1)$ and $\gamma \in (0,1]$, the family is also jointly locally Hölder continuous with parameters $((1, L_{h, \kappa, \alpha} V^\kappa, \alpha), (1, \tilde{L}_{h, \tilde{\kappa}, \tilde{\alpha}} V^{\tilde{\kappa}}, \tilde{\alpha}))$, and is bounded by $C_{h,\gamma} V^\gamma$ for some constant $C_{h,\gamma}>0$. 
\end{proof}

\subsubsection{Properties of $F_\theta$}

Define
\begin{equation*}
    F_\theta = P_\theta(\hat{h}_\theta^2) - (P_\theta \hat{h}_\theta)^2 .
\end{equation*}

\begin{lemma}
    \label{lemma_properties_F}

    Suppose that Assumptions \ref{assumption_lipschitz_continuity_functions_kernels}, \ref{assumption_lipschitz_continuity_conditional_expectation} and \ref{assumption_simultaneous_properties} hold.
    If $\sup_{\theta \in \Theta} \mathbb{E}_\theta \left[ {Z^{(c)}}^2 \right] < \infty$, then for any $\alpha,\tilde{\alpha},\kappa,\tilde{\kappa} \in (0,1)$ and $\gamma \in (0,1]$, the family of functions $\{F_\theta\}_{\theta \in \Theta}$ is $((1, L_{F, \kappa, \alpha} V^\kappa, \alpha), (1, \tilde{L}_{F, \tilde{\kappa}, \tilde{\alpha}} V^{\tilde{\kappa}}, \tilde{\alpha}))$-joint locally Hölder continuous, and bounded by $C_{F,\gamma} V^\gamma$ with corresponding positive constants.
\end{lemma}
\begin{proof}
    The properties of $\hat{h}_\theta$ follow directly from 
    Lemma \ref{lemma_h_lipschitz_continuous_bounded} and 
    Corollary \ref{corollary_V_Poisson_solution_robust_lipschitz_continuous_bounded}.
    The properties of $F_\theta$ are consequences of those of $\hat{h}_\theta$ 
    together with Corollaries \ref{corollary_n_step_h_robust_lipschitz_continuous_bounded} 
    and \ref{corollary_h_power_robust_lipschitz_continuous_bounded}.
\end{proof}

\subsubsection{Properties of $G_\theta$}

Define
\begin{align*}
    G_\theta(\Lambda) &= \mathbb{E}_\theta \left[ (T_1-\rho_\theta(X_1))^2{Z^{(c)}_1}^2 - h_\theta^2(\Lambda_0) \mid \Lambda_0=\Lambda \right]
    + \mathbb{E}_\theta \left\{ \left[ Z^{(u)} - \mathbb{E}_\theta \left[ Z^{(u)} \right] \right]^2 \right\} \\
    & \quad + 2 \mathbb{E}_\theta \left[ [(T_1-\rho_\theta(X_1))Z^{(c)}_1-h_\theta(\Lambda)][Z^{(u)}_1-\mathbb{E}_\theta Z^{(u)}_1] \mid \Lambda_0=\Lambda \right] \\
    &= \int \left\{ \left[ \rho_\theta^2(x) + (1-2\rho_\theta(x))g_\theta(\Lambda, x) \right] f_{{Z^{(c)}}^2,\theta}(x) + f_{{Z^{(u)}}^2,\theta}(x) \right\} \Gamma(\mathrm{d} x)
    - \left[ \int f_{Z^{(u)},\theta}(x) \Gamma(\mathrm{d} x) \right]^2 \\
    & \quad - h_\theta^2(\Lambda)
    + 2\mathbb{E}_{X \sim \Gamma} \left[ (g_\theta(\Lambda, X)-\rho_\theta(X)) \mathbb{E}_{(X, Z) \sim \Gamma_{X,Z,\theta}} \left[ (Z^{(u)}-\mathbb{E}_\theta Z^{(u)})Z^{(c)} \mid X \right] \right] \\
    &= \int \left\{ \left[ \rho_\theta^2(x) + (1-2\rho_\theta(x))g_\theta(\Lambda, x) \right] f_{{Z^{(c)}}^2,\theta}(x) + f_{{Z^{(u)}}^2,\theta}(x) \right. \\
    & \quad \left. + 2 (g_\theta(\Lambda, x)-\rho_\theta(x)) f_{(Z^{(u)}-\mathbb{E}_\theta Z^{(u)})Z^{(c)}, \theta}(x) \right\} \Gamma(\mathrm{d} x) \\
    & \quad - \left[ \int f_{Z^{(u)},\theta}(x) \Gamma(\mathrm{d} x) \right]^2 - h_\theta^2(\Lambda) .
\end{align*}

\begin{lemma}
    \label{lemma_properties_G}

    Let the Lyapunov function $V: \mathrm{X} \to [1, \infty)$.
    Suppose that Assumptions \ref{assumption_lipschitz_continuity_functions_kernels} and \ref{assumption_lipschitz_continuity_conditional_expectation} hold.
    If $\sup_{\theta \in \Theta} \mathbb{E}_\theta \left[ {Z^{(c)}}^4 \right] < \infty$ and $\sup_{\theta \in \Theta} \mathbb{E}_\theta \left[ {Z^{(u)}}^4 \right] < \infty$, then the family of functions $\{G_\theta\}_{\theta \in \Theta}$ is bounded and Lipschitz continuous with respect to $(\theta,\Lambda)$.
    In particular, for any $\alpha,\tilde{\alpha},\kappa,\tilde{\kappa} \in (0,1)$ and $\gamma \in (0,1]$, the family of functions $\{G_\theta\}_{\theta \in \Theta}$ is $((1, L_{G, \kappa, \alpha} V^\kappa, \alpha), (1, \tilde{L}_{G, \tilde{\kappa}, \tilde{\alpha}} V^{\tilde{\kappa}}, \tilde{\alpha}))$-joint locally Hölder continuous, and bounded by $C_{G,\gamma} V^\gamma$ with corresponding positive constants.
\end{lemma}
\begin{proof}
    For the boundedness of $G_\theta$, note that
    \begin{align*}
        G_\theta(\Lambda)
        &= \mathbb{E}_\theta \left[ \left\{
            (T_1-\rho_\theta(X_1)) Z^{(c)}_1 + Z^{(u)}_1 - \mathbb{E}_\theta \left[ Z^{(u)}_1 \right] \right. \right. \\
        & \quad\quad \left. \left. - \mathbb{E}_\theta \left[ (T_1-\rho_\theta(X_1)) Z^{(c)}_1 + Z^{(u)}_1 - \mathbb{E}_\theta \left[ Z^{(u)}_1 \right] \mid \Lambda_0 = \Lambda \right]
        \right\}^2\mid \Lambda_0 = \Lambda \right] \\
        & \leq \mathbb{E}_\theta \left[ \left\{
            (T_1-\rho_\theta(X_1)) Z^{(c)}_1 + Z^{(u)}_1 - \mathbb{E}_\theta \left[ Z^{(u)}_1 \right]
        \right\}^2 \mid \Lambda_0 = \Lambda \right] .
    \end{align*}
    By Assumption  \ref{assumption_lipschitz_continuity_conditional_expectation} and the moment conditions $\sup_{\theta \in \Theta} \mathbb{E}_\theta \left[ {Z^{(c)}}^4 \right] < \infty$ and $\sup_{\theta \in \Theta} \mathbb{E}_\theta \left[ {Z^{(u)}}^4 \right] < \infty$, we have
    \begin{align*}
        \left| G_\theta(\Lambda) \right|
        & \leq \int \left\{ f_{{Z^{(c)}}^2,\theta}(x) + f_{{Z^{(u)}}^2,\theta}(x)
        + 2 |f_{(Z^{(u)}-\mathbb{E}_\theta Z^{(u)})Z^{(c)}, \theta}(x)| \right\} \Gamma(\mathrm{d} x) \\
        & \leq \sup_{\theta \in \Theta} \mathbb{E}_\theta \left[ {Z^{(c)}}^2 \right]
        + \sup_{\theta \in \Theta} \mathbb{E}_\theta \left[ {Z^{(u)}}^2 \right]
        + 2 \sqrt{\sup_{\theta \in \Theta} \mathbb{E}_\theta \left[ {Z^{(c)}}^2 \right]
        \sup_{\theta \in \Theta} \mathbb{E}_\theta \left[ {Z^{(u)}}^2 \right]} \\
        & < \infty .
    \end{align*}

    The Lipschitz continuity with respect to $(\theta,\Lambda)$ follows from the Lipschitz continuity in Assumptions \ref{assumption_lipschitz_continuity_functions_kernels} and \ref{assumption_lipschitz_continuity_conditional_expectation}.
    Specifically, the Lipschitz continuity can be expressed as
    \begin{align*}
        & \quad |G_\theta(\Lambda) - G_{\theta^\prime}(\Lambda^\prime)| \\
        & \leq \int \left\{ \left| f_{{Z^{(c)}}^2,\theta}(x) \right| \left( |\rho_\theta^2(x) - \rho_{\theta^\prime}^2(x)| + |g_\theta(\Lambda, x) - g_{\theta^\prime}(\Lambda^\prime, x)| + 2| \rho_\theta(x)g_\theta(\Lambda, x) - \rho_{\theta^\prime}(x)g_{\theta^\prime}(\Lambda^\prime, x) | \right) \right. \\
        & \quad\quad + \left| \rho_{\theta^\prime}^2(x) + (1-2\rho_{\theta^\prime}(x))g_{\theta^\prime}(\Lambda^\prime, x) \right| \cdot |f_{{Z^{(c)}}^2,\theta}(x) - f_{{Z^{(c)}}^2,\theta^\prime}(x)| \\
        & \quad\quad + |f_{{Z^{(u)}}^2,\theta}(x) - f_{{Z^{(u)}}^2,\theta^\prime}(x)|
        + 2|f_{(Z^{(u)}-\mathbb{E}_\theta Z^{(u)})Z^{(c)}, \theta}(x)| \left( |g_\theta(\Lambda, x) - g_{\theta^\prime}(\Lambda^\prime, x)| + |\rho_\theta(x) - \rho_{\theta^\prime}(x)| \right) \\
        & \quad\quad + \left. 2|g_{\theta^\prime}(\Lambda^\prime, x)-\rho_{\theta^\prime}(x)| \cdot |f_{(Z^{(u)}-\mathbb{E}_\theta Z^{(u)})Z^{(c)}, \theta}(x) - f_{(Z^{(u)}-\mathbb{E}_{\theta^\prime} Z^{(u)})Z^{(c)}, \theta^\prime}(x)| \right\} \Gamma(\mathrm{d} x) \\
        & \quad + |h_\theta(\Lambda) - h_{\theta^\prime}(\Lambda^\prime)| |h_\theta(\Lambda) + h_{\theta^\prime}(\Lambda^\prime)| \\
        & \quad + \left| \int f_{Z^{(u)},\theta}(x) \Gamma(\mathrm{d} x) + \int f_{Z^{(u)},\theta^\prime}(x) \Gamma(\mathrm{d} x) \right|
        \int \left| f_{Z^{(u)},\theta}(x) - f_{Z^{(u)},\theta^\prime}(x) \right| \Gamma(\mathrm{d} x) \\
        & \leq \int \left\{ \left| f_{{Z^{(c)}}^2,\theta}(x) \right| \left( |(\rho_\theta(x) - \rho_{\theta^\prime}(x))(\rho_\theta(x) + \rho_{\theta^\prime}(x))| + |g_\theta(\Lambda, x) - g_{\theta^\prime}(\Lambda^\prime, x)| \right. \right. \\
        & \quad\quad\quad\quad\quad\quad + \left. 2 \left( |g_{\theta^\prime}(\Lambda^\prime, x)||\rho_\theta(x) - \rho_{\theta^\prime}(x)| + |\rho_\theta(x)| |g_\theta(\Lambda, x) - g_{\theta^\prime}(\Lambda^\prime, x)| \right) \right) \\
        & \quad\quad + 1 \cdot |f_{{Z^{(c)}}^2,\theta}(x) - f_{{Z^{(c)}}^2,\theta^\prime}(x)| + |f_{{Z^{(u)}}^2,\theta}(x) - f_{{Z^{(u)}}^2,\theta^\prime}(x)| \\
        & \quad\quad + 2|f_{(Z^{(u)}-\mathbb{E}_\theta Z^{(u)})Z^{(c)}, \theta}(x)| \left( |g_\theta(\Lambda, x) - g_{\theta^\prime}(\Lambda^\prime, x)| + |\rho_\theta(x) - \rho_{\theta^\prime}(x)| \right) \\
        & \quad\quad + \left. 2 \cdot |f_{(Z^{(u)}-\mathbb{E}_\theta Z^{(u)})Z^{(c)}, \theta}(x) - f_{(Z^{(u)}-\mathbb{E}_{\theta^\prime} Z^{(u)})Z^{(c)}, \theta^\prime}(x)| \right\} \Gamma(\mathrm{d} x) \\
        & \quad + |h_\theta(\Lambda) - h_{\theta^\prime}(\Lambda^\prime)| |h_\theta(\Lambda) + h_{\theta^\prime}(\Lambda^\prime)|
        + 2 \left[ \sup_{\theta \in \Theta} \mathbb{E}_\theta \left| {Z^{(u)}} \right| \right]
        \| f_{Z^{(u)},\theta} - f_{Z^{(u)},\theta^\prime} \|_{L^1(\Gamma)} \\
        & \leq \int \left| f_{{Z^{(c)}}^2,\theta}(x) \right| \left( 3|g_\theta(\Lambda, x) - g_{\theta^\prime}(\Lambda^\prime, x)| + 4|\rho_\theta(x) - \rho_{\theta^\prime}(x)| \right) \Gamma(\mathrm{d} x) \\
        & \quad + \int \left| f_{{Z^{(c)}}^2,\theta}(x) - f_{{Z^{(c)}}^2,\theta^\prime}(x) \right| \Gamma(\mathrm{d} x)
        + \int \left| f_{{Z^{(u)}}^2,\theta}(x) - f_{{Z^{(u)}}^2,\theta^\prime}(x) \right| \Gamma(\mathrm{d} x) \\
        & \quad + 2 \int \left|f_{(Z^{(u)}-\mathbb{E}_\theta Z^{(u)})Z^{(c)}, \theta}(x) \right| \left( |g_\theta(\Lambda, x) - g_{\theta^\prime}(\Lambda^\prime, x)| + |\rho_\theta(x) - \rho_{\theta^\prime}(x)| \right) \Gamma(\mathrm{d} x) \\
        & \quad + 2 \int \left| f_{(Z^{(u)}-\mathbb{E}_\theta Z^{(u)})Z^{(c)}, \theta}(x) - f_{(Z^{(u)}-\mathbb{E}_{\theta^\prime} Z^{(u)})Z^{(c)}, \theta^\prime}(x) \right| \Gamma(\mathrm{d} x) \\
        & \quad + |h_\theta(\Lambda) - h_{\theta^\prime}(\Lambda^\prime)| |h_\theta(\Lambda) + h_{\theta^\prime}(\Lambda^\prime)|
        + 2 \left[ \sup_{\theta \in \Theta} \mathbb{E}_\theta \left| {Z^{(u)}} \right| \right]
        \| f_{Z^{(u)},\theta} - f_{Z^{(u)},\theta^\prime} \|_{L^1(\Gamma)} \\
        & \leq 3 ||f_{{Z^{(c)}}^2,\theta}||_{L^2(\Gamma)} ||g_\theta(\Lambda, \cdot) - g_{\theta^\prime}(\Lambda^\prime, \cdot)||_{L^2(\Gamma)}
        + 4 ||f_{{Z^{(c)}}^2,\theta}||_{L^2(\Gamma)} ||\rho_\theta - \rho_{\theta^\prime}||_{L^2(\Gamma)} \\
        & \quad + ||f_{{Z^{(c)}}^2,\theta} - f_{{Z^{(c)}}^2,\theta^\prime}||_{L^1(\Gamma)}
        + ||f_{{Z^{(u)}}^2,\theta} - f_{{Z^{(u)}}^2,\theta^\prime}||_{L^1(\Gamma)} \\
        & \quad + 2 ||f_{(Z^{(u)}-\mathbb{E}_\theta Z^{(u)})Z^{(c)}, \theta}||_{L^2(\Gamma)} \left( ||g_\theta(\Lambda, \cdot) - g_{\theta^\prime}(\Lambda^\prime, \cdot)||_{L^2(\Gamma)} + ||\rho_\theta - \rho_{\theta^\prime}||_{L^2(\Gamma)} \right) \\
        & \quad + 2 ||f_{(Z^{(u)}-\mathbb{E} Z^{(u)})Z^{(c)}, \theta} - f_{(Z^{(u)}-\mathbb{E} Z^{(u)})Z^{(c)}, \theta^\prime}||_{L^1(\Gamma)} \\
        & \quad + |h_\theta(\Lambda) - h_{\theta^\prime}(\Lambda^\prime)| |h_\theta(\Lambda) + h_{\theta^\prime}(\Lambda^\prime)|
        + 2 \left[ \sup_{\theta \in \Theta} \mathbb{E}_\theta \left| {Z^{(u)}} \right| \right]
        \| f_{Z^{(u)},\theta} - f_{Z^{(u)},\theta^\prime} \|_{L^1(\Gamma)} \\
        & \leq 3 \sqrt{\sup_{\theta \in \Theta} \mathbb{E}_\theta \left[ {Z^{(c)}}^4 \right]} L_g \left( d(\theta, \theta^\prime) + d(\Lambda, \Lambda^\prime) \right)
        + 4 \sqrt{\sup_{\theta \in \Theta} \mathbb{E}_\theta \left[ {Z^{(c)}}^4 \right]} L_\rho d(\theta, \theta^\prime) \\
        & \quad + L_f d(\theta, \theta^\prime) + L_f d(\theta, \theta^\prime) \\
        & \quad + 2 \left[ \sup_{\theta \in \Theta} \mathbb{E}_\theta \left[ {Z^{(c)}}^4 \right]
        \sup_{\theta \in \Theta} \mathbb{E}_\theta \left[ {Z^{(u)}}^4 \right] \right]^{\frac{1}{4}}
        \left( L_g(d(\theta, \theta^\prime) + d(\Lambda, \Lambda^\prime)) + L_\rho d(\theta, \theta^\prime) \right) \\
        & \quad + 2 L_f d(\theta, \theta^\prime)
        + 2 \left[ \sup_{\theta \in \Theta} \mathbb{E}_\theta \left| {Z^{(c)}} \right| \right]
        L_h \left( d(\theta, \theta^\prime) + d(\Lambda, \Lambda^\prime) \right)
        + 2 \left[ \sup_{\theta \in \Theta} \mathbb{E}_\theta \left| {Z^{(u)}} \right| \right]
        L_f d(\theta, \theta^\prime) \\
        &= \left[
            (3 L_g + 4 L_\rho) \sqrt{\sup_{\theta \in \Theta} \mathbb{E}_\theta \left[ {Z^{(c)}}^4 \right]}
            + 2 (L_g + L_\rho) \left[ \sup_{\theta \in \Theta} \mathbb{E}_\theta \left[ {Z^{(c)}}^4 \right]
            \sup_{\theta \in \Theta} \mathbb{E}_\theta \left[ {Z^{(u)}}^4 \right] \right]^{\frac{1}{4}} \right. \\
        & \quad\quad \left. + 4L_f + 2L_h \sup_{\theta \in \Theta} \mathbb{E}_\theta \left| Z^{(c)} \right|
        + 2L_f \sup_{\theta \in \Theta} \mathbb{E}_\theta \left| Z^{(u)} \right|
        \right] d(\theta, \theta^\prime) \\
        & \quad + \left[
            3 L_g \sqrt{\sup_{\theta \in \Theta} \mathbb{E}_\theta \left[ {Z^{(c)}}^4 \right]}
            + 2 L_g \left[ \sup_{\theta \in \Theta} \mathbb{E}_\theta \left[ {Z^{(c)}}^4 \right]
            \sup_{\theta \in \Theta} \mathbb{E}_\theta \left[ {Z^{(u)}}^4 \right] \right]^{\frac{1}{4}} \right. \\
        & \quad\quad \left. + 2 L_h \sup_{\theta \in \Theta} \mathbb{E}_\theta \left| Z^{(c)} \right|
        \right] d(\Lambda, \Lambda^\prime) .
    \end{align*}

    This proves that $\{G_\theta\}_{\theta \in \Theta}$ is Lipschitz continuous in both $\theta$ and $\Lambda$. 

    Consequently, by $V \geq 1$, $\alpha,\tilde{\alpha},\kappa,\tilde{\kappa} \in (0,1)$ and $\gamma \in (0,1]$, the family is also jointly locally Hölder continuous with parameters $((1, L_{G, \kappa, \alpha} V^\kappa, \alpha), (1, \tilde{L}_{G, \tilde{\kappa}, \tilde{\alpha}} V^{\tilde{\kappa}}, \tilde{\alpha}))$, and is bounded by $C_{G,\gamma} V^\gamma$ for some constant $C_{G,\gamma}>0$. 
\end{proof}

\subsubsection{Properties of $H_\theta$}

Define
\begin{align*}
    & \quad H_\theta(\Lambda) \\
    &= \mathbb{E}_\theta \left[ \left[ (T_1-\rho_\theta(X_1)) f_{Z^{(c)},\theta}(X_1) - h_\theta(\Lambda_0) + f_{Z^{(u)},\theta}(X_1)-\mathbb{E}_\theta Z^{(u)} \right]
    \left[ \hat{h}_\theta(\Lambda_1) - (P_\theta\hat{h}_\theta)(\Lambda_0) \right] \mid \Lambda_0=\Lambda \right] \\
    &= \int \left[
    g_\theta(\Lambda, x) \left[ (1-\rho_\theta(x)) f_{Z^{(c)},\theta}(x)+f_{Z^{(u)},\theta}(x) \right]
    \hat{h}_\theta(\Lambda+\phi(x)/\rho_\theta(x)) \right. \\
    & \quad \left. + (1-g_\theta(\Lambda, x)) \left[ -\rho_\theta(x) f_{Z^{(c)},\theta}(x)+f_{Z^{(u)},\theta}(x) \right]
    \hat{h}_\theta(\Lambda-\phi(x)/(1-\rho_\theta(x))) \right] \Gamma(\mathrm{d} x) \\
    & \quad - [h_\theta(\Lambda)+\mathbb{E}_\theta Z^{(u)}](P_\theta\hat{h}_\theta)(\Lambda)
\end{align*}

\begin{lemma}
    \label{lemma_properties_H}

    Suppose that Assumptions \ref{assumption_lipschitz_continuity_functions_kernels}, \ref{assumption_lipschitz_continuity_conditional_expectation} and \ref{assumption_simultaneous_properties} hold.
    If $\sup_{\theta \in \Theta} \mathbb{E}_\theta \left[ {Z^{(c)}}^4 \right] < \infty$ and $\sup_{\theta \in \Theta} \mathbb{E}_\theta \left[ {Z^{(u)}}^4 \right] < \infty$, then for any $\alpha,\tilde{\alpha},\kappa,\tilde{\kappa} \in (0,1)$ and $\gamma \in (0,1]$, the family of functions $\{H_\theta\}_{\theta \in \Theta}$ is $((1, L_{H, \kappa, \alpha} V^\kappa, \alpha), (1, \tilde{L}_{H, \tilde{\kappa}, \tilde{\alpha}} V^{\tilde{\kappa}}, \tilde{\alpha}))$-joint locally Hölder continuous, and bounded by $C_{H,\gamma} V^\gamma$ with corresponding positive constants.
\end{lemma}
\begin{proof}
    The bound follows directly from Corollaries \ref{corollary_n_step_h_robust_lipschitz_continuous_bounded} and \ref{corollary_V_Poisson_solution_robust_lipschitz_continuous_bounded}.
    We therefore focus on the continuity properties.

    Note that from Assumptions \ref{assumption_sampling} and \ref{assumption_simultaneous_properties}, for any $a \in (0,1]$, there exist constants $\beta_a \in (0,1)$ and $b_a < \infty$, independent of $\theta$, such that for all $\Lambda \in \mathrm{X}$,
    \begin{align*}
        & \quad \int \left[ 
        V^a(\Lambda+\phi(x)/\rho_\theta(x))
        + V^a(\Lambda-\phi(x)/(1-\rho_\theta(x)))
        \right] \Gamma(\mathrm{d} x) \\
        & \leq \frac{1}{\iota} \int \left[ g_\theta(\Lambda, x)
        V^a(\Lambda+\phi(x)/\rho_\theta(x))
        \right] \Gamma(\mathrm{d} x)
        + \int \left[ \left\{ 1 - g_\theta(\Lambda, x) \right\}
        V^a(\Lambda-\phi(x)/(1-\rho_\theta(x)))
        \right] \Gamma(\mathrm{d} x) \\
        & = \frac{1}{\iota}P_\theta V^a(\Lambda)
        \leq \frac{1}{\iota} \left[ \beta_a V^a(\Lambda) + b_a \right]
        \leq \frac{\beta_a + b_a}{\iota} V^a(\Lambda) .
    \end{align*}

    Moreover, by Corollary \ref{corollary_V_Poisson_solution_robust_lipschitz_continuous_bounded}, for any $\alpha,\tilde{\alpha},\kappa,\tilde{\kappa} \in (0,1)$ and $\gamma \in (0,1]$, the family $\{\hat{h}_\theta\}_{\theta \in \Theta}$ is
    \begin{equation*}
        ((\Delta_{\hat{h}, \kappa, \alpha}, L_{\hat{h}, \kappa, \alpha} V^\kappa, \alpha),
        (\tilde{\Delta}_{\hat{h}, \tilde{\kappa}, \tilde{\alpha}}, \tilde{L}_{\hat{h}, \tilde{\kappa}, \tilde{\alpha}} V^{\tilde{\kappa}}, \tilde{\alpha}))
    \end{equation*}
    -joint locally Hölder continuous, and bounded by $C_{\hat{h},\gamma} V^\gamma(\Lambda)$ for some constant $\Delta_{\hat{h}, \kappa, \alpha}$, $\tilde{\Delta}_{\hat{h}, \tilde{\kappa}, \tilde{\alpha}}$, $L_{\hat{h}, \kappa, \alpha}$, $\tilde{L}_{\hat{h}, \tilde{\kappa}, \tilde{\alpha}}$ and $C_{\hat{h},\gamma} > 0$.

    We first split the bound of the difference $|H_\theta(\Lambda) - H_{\theta^\prime}(\Lambda^\prime)|$.
    \begin{align*}
        & |H_\theta(\Lambda) - H_{\theta^\prime}(\Lambda^\prime)| \\
        & \leq \int \left| g_\theta(\Lambda, x) \left[ (1-\rho_\theta(x)) f_{Z^{(c)},\theta}(x)+f_{Z^{(u)},\theta}(x) \right] \hat{h}_\theta(\Lambda+\phi(x)/\rho_\theta(x)) \right. \\
        & \quad\quad \left. - g_{\theta^\prime}(\Lambda^\prime, x) \left[ (1-\rho_{\theta^\prime}(x)) f_{Z^{(c)},\theta^\prime}(x)+f_{Z^{(u)},\theta^\prime}(x) \right] \hat{h}_{\theta^\prime}(\Lambda^\prime+\phi(x)/\rho_{\theta^\prime}(x)) \right| \Gamma(\mathrm{d} x) \\
        & \quad + \int \left| (1-g_\theta(\Lambda, x)) \left[ -\rho_\theta(x) f_{Z^{(c)},\theta}(x)+f_{Z^{(u)},\theta}(x) \right] \hat{h}_\theta(\Lambda-\phi(x)/(1-\rho_\theta(x))) \right. \\
        & \quad\quad \left. - (1-g_{\theta^\prime}(\Lambda^\prime, x)) \left[ -\rho_{\theta^\prime}(x) f_{Z^{(c)},\theta^\prime}(x)+f_{Z^{(u)},\theta^\prime}(x) \right] \hat{h}_{\theta^\prime}(\Lambda^\prime-\phi(x)/(1-\rho_{\theta^\prime}(x))) \right| \Gamma(\mathrm{d} x) \\
        & \quad + \left| [h_\theta(\Lambda)+\mathbb{E}_\theta Z^{(u)}](P_\theta\hat{h}_\theta)(\Lambda) - [h_{\theta^\prime}(\Lambda^\prime)+\mathbb{E}_{\theta^\prime} Z^{(u)}](P_{\theta^\prime}\hat{h}_{\theta^\prime})(\Lambda^\prime) \right| \\
        & \leq \int \left\{ \left| \left[ (1-\rho_\theta(x)) f_{Z^{(c)},\theta}(x)+f_{Z^{(u)},\theta}(x) \right] \hat{h}_\theta(\Lambda+\phi(x)/\rho_\theta(x)) \right| \cdot |g_\theta(\Lambda, x) - g_{\theta^\prime}(\Lambda^\prime, x)| \right. \\
        & \quad\quad + |g_{\theta^\prime}(\Lambda^\prime, x)| |\hat{h}_\theta(\Lambda+\phi(x)/\rho_\theta(x))| \left| \left(f_{Z^{(c)},\theta} - \rho_\theta f_{Z^{(c)},\theta} + f_{Z^{(u)},\theta}\right) - \left(f_{Z^{(c)},\theta^\prime} - \rho_{\theta^\prime} f_{Z^{(c)},\theta^\prime} + f_{Z^{(u)},\theta^\prime}\right) \right| \\
        & \quad\quad + |g_{\theta^\prime}(\Lambda^\prime, x)| \left| (1-\rho_{\theta^\prime}(x)) f_{Z^{(c)},\theta^\prime}(x)+f_{Z^{(u)},\theta^\prime}(x) \right| \left| \hat{h}_\theta(\Lambda+\phi(x)/\rho_\theta(x)) - \hat{h}_{\theta^\prime}(\Lambda^\prime+\phi(x)/\rho_{\theta^\prime}(x)) \right| \\
        & \quad\quad + \left| \left[ -\rho_\theta(x) f_{Z^{(c)},\theta}(x)+f_{Z^{(u)},\theta}(x) \right] \hat{h}_\theta(\Lambda-\phi(x)/(1-\rho_\theta(x))) \right| \cdot |g_\theta(\Lambda, x) - g_{\theta^\prime}(\Lambda^\prime, x)| \\
        & \quad\quad + |1-g_{\theta^\prime}(\Lambda^\prime, x)| |\hat{h}_\theta(\Lambda-\phi(x)/(1-\rho_\theta(x)))| \left| \left(-\rho_\theta f_{Z^{(c)},\theta} + f_{Z^{(u)},\theta}\right) - \left(-\rho_{\theta^\prime} f_{Z^{(c)},\theta^\prime} + f_{Z^{(u)},\theta^\prime}\right) \right| \\
        & \quad\quad + \left. |1-g_{\theta^\prime}(\Lambda^\prime, x)| \left| -\rho_{\theta^\prime}(x) f_{Z^{(c)},\theta^\prime}(x)+f_{Z^{(u)},\theta^\prime}(x) \right| \left| \hat{h}_\theta(\Lambda-\phi(x)/(1-\rho_\theta(x))) - \hat{h}_{\theta^\prime}(\Lambda^\prime-\phi(x)/(1-\rho_{\theta^\prime}(x))) \right| \right\} \Gamma(\mathrm{d} x) \\
        & \quad + \left| h_\theta(\Lambda)+\mathbb{E}_\theta Z^{(u)} \right| \left| (P_\theta\hat{h}_\theta)(\Lambda) - (P_{\theta^\prime}\hat{h}_{\theta^\prime})(\Lambda^\prime) \right| \\
        & \quad + \left| (P_{\theta^\prime}\hat{h}_{\theta^\prime})(\Lambda^\prime) \right| \left( |h_\theta(\Lambda) - h_{\theta^\prime}(\Lambda^\prime)| + |\mathbb{E}_\theta Z^{(u)}-\mathbb{E}_{\theta^\prime} Z^{(u)}| \right) \\ 
        & \leq \int \left( (|\hat{h}_\theta(\Lambda+\phi(x)/\rho_\theta(x))| + |\hat{h}_\theta(\Lambda-\phi(x)/(1-\rho_\theta(x)))|)
        (|f_{Z^{(c)},\theta}|+|f_{Z^{(u)},\theta}|) \right)
        |g_\theta(\Lambda, x) - g_{\theta^\prime}(\Lambda^\prime, x)| \Gamma(\mathrm{d} x) \\
        & \quad + \int ( |\hat{h}_\theta(\dots)| + |\hat{h}_{\theta^\prime}(\dots)| ) \left( |f_{Z^{(c)},\theta} - f_{Z^{(c)},\theta^\prime}|
        + |\rho_\theta(x) - \rho_{\theta^\prime}(x)||f_{Z^{(c)},\theta^\prime}(x)|
        + |f_{Z^{(u)},\theta} - f_{Z^{(u)},\theta^\prime}| \right) \Gamma(\mathrm{d} x) \\
        & \quad + \int \left(|f_{Z^{(c)},\theta^\prime}| + |f_{Z^{(u)},\theta^\prime}| \right)
        \left| \hat{h}_\theta(\Lambda+\phi(x)/\rho_\theta(x)) - \hat{h}_{\theta^\prime}(\Lambda^\prime+\phi(x)/\rho_{\theta^\prime}(x)) \right| \Gamma(\mathrm{d} x) \\
        & \quad + \int \left(|f_{Z^{(c)},\theta^\prime}| + |f_{Z^{(u)},\theta^\prime}| \right)
        \left| \hat{h}_\theta(\Lambda-\phi(x)/(1-\rho_\theta(x))) - \hat{h}_{\theta^\prime}(\Lambda^\prime-\phi(x)/(1-\rho_{\theta^\prime}(x))) \right| \Gamma(\mathrm{d} x) \\
        & \quad + |h_\theta(\Lambda)+\mathbb{E}_\theta Z^{(u)}| |(P_\theta\hat{h}_\theta)(\Lambda) - (P_{\theta^\prime}\hat{h}_{\theta^\prime})(\Lambda^\prime)| \\
        & \quad + |(P_{\theta^\prime}\hat{h}_{\theta^\prime})(\Lambda^\prime)| \left( |h_\theta(\Lambda) - h_{\theta^\prime}(\Lambda^\prime)| + |\mathbb{E}_\theta Z^{(u)}-\mathbb{E}_{\theta^\prime} Z^{(u)}| \right) .
    \end{align*}
    To facilitate the subsequent analysis, we decompose the integrals and differences into six terms, which we denote by $T^{(1)}, \dots, T^{(6)}$:
    \begin{align*}
        T^{(1)} &:= \int \left( (|\hat{h}_\theta(\Lambda+\phi(x)/\rho_\theta(x))| + |\hat{h}_\theta(\Lambda-\phi(x)/(1-\rho_\theta(x)))|)
        (|f_{Z^{(c)},\theta}|+|f_{Z^{(u)},\theta}|) \right)
        |g_\theta(\Lambda, x) - g_{\theta^\prime}(\Lambda^\prime, x)| \Gamma(\mathrm{d} x) \\
        T^{(2)} &:= \int ( |\hat{h}_\theta(\dots)| + |\hat{h}_{\theta^\prime}(\dots)| ) \left( |f_{Z^{(c)},\theta} - f_{Z^{(c)},\theta^\prime}|
        + |\rho_\theta(x) - \rho_{\theta^\prime}(x)||f_{Z^{(c)},\theta^\prime}(x)|
        + |f_{Z^{(u)},\theta} - f_{Z^{(u)},\theta^\prime}| \right) \Gamma(\mathrm{d} x) \\
        T^{(3)} &:= \int \left(|f_{Z^{(c)},\theta^\prime}| + |f_{Z^{(u)},\theta^\prime}| \right)
        \left| \hat{h}_\theta(\Lambda+\phi(x)/\rho_\theta(x)) - \hat{h}_{\theta^\prime}(\Lambda^\prime+\phi(x)/\rho_{\theta^\prime}(x)) \right| \Gamma(\mathrm{d} x) \\
        T^{(4)} &:= \int \left(|f_{Z^{(c)},\theta^\prime}| + |f_{Z^{(u)},\theta^\prime}| \right)
        \left| \hat{h}_\theta(\Lambda-\phi(x)/(1-\rho_\theta(x))) - \hat{h}_{\theta^\prime}(\Lambda^\prime-\phi(x)/(1-\rho_{\theta^\prime}(x))) \right| \Gamma(\mathrm{d} x) \\
        T^{(5)} &:= |h_\theta(\Lambda)+\mathbb{E}_\theta Z^{(u)}| |(P_\theta\hat{h}_\theta)(\Lambda) - (P_{\theta^\prime}\hat{h}_{\theta^\prime})(\Lambda^\prime)| \\
        T^{(6)} &:= |(P_{\theta^\prime}\hat{h}_{\theta^\prime})(\Lambda^\prime)| \left( |h_\theta(\Lambda) - h_{\theta^\prime}(\Lambda^\prime)| + |\mathbb{E}_\theta Z^{(u)}-\mathbb{E}_{\theta^\prime} Z^{(u)}| \right) .
    \end{align*}
    We aim to show that for any $\alpha, \kappa \in (0,1)$ and any $i \in \{1,\dots,6\}$, it holds that
    \begin{equation*}
        T^{(i)} \leq L V^\kappa(\Lambda) ([d(\theta, \theta^\prime)]^\alpha + [d(\Lambda, \Lambda^\prime)]^\alpha) ,
    \end{equation*}
    where the constant $L$ is independent of $\theta$, $\theta^\prime$, $\Lambda$ and $\Lambda^\prime$ when $d(\theta, \theta^\prime)$ and $d(\Lambda, \Lambda^\prime)$ are sufficiently small.

    For $T^{(1)}$, it suffices to use the Cauchy-Schwarz inequality, the bounds on $\hat{h}_\theta$, $f_{Z^{(c)},\theta}$ and $f_{Z^{(u)},\theta}$, and the Lipschitz continuity of $g_\theta$.
    For any $\gamma \in (0,\frac{1}{4}]$,
    \begin{align*}
        & \quad T^{(1)} \\
        & \leq \left\| \left( |\hat{h}_\theta(\dots)| + |\hat{h}_\theta(\dots)| \right) \left(|f_{Z^{(c)},\theta}|+|f_{Z^{(u)},\theta}|\right) \right\|_{L^2(\Gamma)}
        \left\| g_\theta(\Lambda, \cdot) - g_{\theta^\prime}(\Lambda^\prime, \cdot) \right\|_{L^2(\Gamma)} \\
        & \leq \left( \left\| | \hat{h}_\theta(\Lambda+\phi(\cdot)/\rho_\theta(\cdot)) | + | \hat{h}_\theta(\Lambda-\phi(\cdot)/(1-\rho_\theta(\cdot))) | \right\|_{L^4(\Gamma)} \right) \\
        & \quad\quad \left( \left\| f_{Z^{(c)},\theta} \right\|_{L^4(\Gamma)} + \left\| f_{Z^{(u)},\theta} \right\|_{L^4(\Gamma)} \right)
        \left\| g_\theta(\Lambda, \cdot) - g_{\theta^\prime}(\Lambda^\prime, \cdot) \right\|_{L^2(\Gamma)} \\
        & \leq C_{\hat{h},\gamma} \left( \left\| V^\gamma(\Lambda+\phi(\cdot)/\rho_\theta(\cdot))
        + V^\gamma(\Lambda-\phi(\cdot)/(1-\rho_\theta(\cdot))) \right\|_{L^4(\Gamma)} \right) \\
        & \quad\quad \left[ \left\{ \sup_{\theta \in \Theta} \mathbb{E}_\theta \left[ {Z^{(c)}}^4 \right] \right\}^{1/4}
        + \left\{ \sup_{\theta \in \Theta} \mathbb{E}_\theta \left[ {Z^{(u)}}^4 \right] \right\}^{1/4} \right]
        \left\| g_\theta(\Lambda, \cdot) - g_{\theta^\prime}(\Lambda^\prime, \cdot) \right\|_{L^2(\Gamma)} \\
        & \leq C_{\hat{h},\gamma}
        \left[ 8 \int \left[ 
        V^{4\gamma}(\Lambda+\phi(x)/\rho_\theta(x))
        + V^{4\gamma}(\Lambda-\phi(x)/(1-\rho_\theta(x)))
        \right] \Gamma(\mathrm{d} x) \right]^{\frac{1}{4}} \\
        & \quad\quad \left[ \left\{ \sup_{\theta \in \Theta} \mathbb{E}_\theta \left[ {Z^{(c)}}^4 \right] \right\}^{1/4}
        + \left\{ \sup_{\theta \in \Theta} \mathbb{E}_\theta \left[ {Z^{(u)}}^4 \right] \right\}^{1/4} \right]
        \left\| g_\theta(\Lambda, \cdot) - g_{\theta^\prime}(\Lambda^\prime, \cdot) \right\|_{L^2(\Gamma)} \\
        & \leq C_{\hat{h},\gamma}
        \left[ \frac{8(\beta_{4\gamma} + b_{4\gamma})}{\iota} \right]^{\frac{1}{4}}
        \left[ \left\{ \sup_{\theta \in \Theta} \mathbb{E}_\theta \left[ {Z^{(c)}}^4 \right] \right\}^{1/4}
        + \left\{ \sup_{\theta \in \Theta} \mathbb{E}_\theta \left[ {Z^{(u)}}^4 \right] \right\}^{1/4} \right]
        L_g V^\gamma(\Lambda) [ d(\theta, \theta^\prime) + d(\Lambda, \Lambda^\prime) ]
    \end{align*}
    Note that the result can be extended to cover the case $\gamma \in (0,1]$ because $V \geq 1$.

    For $T^{(2)}$, the term $\int ( |\hat{h}_\theta(\dots)| + |\hat{h}_{\theta^\prime}(\dots)| ) |\rho_\theta(x) - \rho_{\theta^\prime}(x)||f_{Z^{(c)},\theta^\prime}(x)| \Gamma(\mathrm{d} x)$ can be bounded using a method similar to that used for $T^{(1)}$.
    Because $\rho_\theta \in [\iota_\rho, 1-\iota_\rho]$ and $\left| \hat{h} \right| \lesssim V^\gamma$ for any $\gamma \in (0,1]$, the other term
    \begin{equation*}
        \int ( |\hat{h}_\theta(\dots)| + |\hat{h}_{\theta^\prime}(\dots)| ) \left( |f_{Z^{(c)},\theta} - f_{Z^{(c)},\theta^\prime}|
        + |f_{Z^{(u)},\theta} - f_{Z^{(u)},\theta^\prime}| \right) \Gamma(\mathrm{d} x)
    \end{equation*}
    can be bounded by a multiple of
    \begin{align*}
        & \quad \int e^{\lambda_1^\prime (\|\phi(x)\| + \|\Lambda\|)}
        \left( |f_{Z^{(c)},\theta} - f_{Z^{(c)},\theta^\prime}|
        + |f_{Z^{(u)},\theta} - f_{Z^{(u)},\theta^\prime}| \right)
        \Gamma(\mathrm{d} x) \\
        &= \left\| e^{\lambda_1^\prime (\|\phi(\cdot)\| + \|\Lambda\|)}
        \left| f_{Z^{(c)},\theta} - f_{Z^{(c)},\theta^\prime}| \right| \right\|_{L^1(\Gamma)}
        + \left\| e^{\lambda_1^\prime (\|\phi(\cdot)\| + \|\Lambda\|)}
        \left| f_{Z^{(u)},\theta} - f_{Z^{(u)},\theta^\prime} \right| \right\|_{L^1(\Gamma)} ,
    \end{align*}
    for any sufficiently small $\lambda_1^\prime > 0$.
    By Assumption \ref{assumption_lipschitz_continuity_conditional_expectation}, the other term can be further bounded by a multiple of $V^\gamma(\Lambda) d(\theta, \theta^\prime)$ for any sufficiently small $\gamma > 0$.
    Since $V(\Lambda) \geq 1$, the same bound extends to all $\gamma \in (0,1]$.

    For $T^{(3)}$ and $T^{(4)}$, it suffices to establish the bound for $T^{(3)}$, since the argument for $T^{(4)}$ follows analogously.
    For $T^{(3)}$, we have
    \begin{align*}
        T^{(3)}
        & \leq \left( \left\| f_{Z^{(c)},\theta^\prime} \right\|_{L^2(\Gamma)}
        + \left\| f_{Z^{(u)},\theta^\prime} \right\|_{L^2(\Gamma)} \right)
        \left\| \hat{h}_\theta(\Lambda+\phi(\cdot)/\rho_\theta(\cdot)) - \hat{h}_{\theta^\prime}(\Lambda^\prime+\phi(\cdot)/\rho_{\theta^\prime}(\cdot)) \right\|_{L^2(\Gamma)} \\
        & \leq \left( \sqrt{\sup_{\theta \in \Theta} \mathbb{E}_\theta \left[ {Z^{(c)}}^2 \right]}
        + \sqrt{\sup_{\theta \in \Theta} \mathbb{E}_\theta \left[ {Z^{(u)}}^2 \right]} \right)
        \left\| \hat{h}_\theta(\Lambda+\phi(\cdot)/\rho_\theta(\cdot)) - \hat{h}_{\theta^\prime}(\Lambda^\prime+\phi(\cdot)/\rho_{\theta^\prime}(\cdot)) \right\|_{L^2(\Gamma)} ,
    \end{align*}
    Hence, the main task reduces to controlling $\left\| \hat{h}_\theta(\Lambda+\phi(\cdot)/\rho_\theta(\cdot)) - \hat{h}_{\theta^\prime}(\Lambda^\prime+\phi(\cdot)/\rho_{\theta^\prime}(\cdot)) \right\|_{L^2(\Gamma)}$.
    To this end, we further decompose the expression into two parts.
    Let $\epsilon, p, q>0$ and $\gamma \in (0,1)$ with $2p\gamma \leq 1$ and $\frac{1}{q} + \frac{1}{p} = 1$, then
    \begin{align*}
        & \quad \int \left\{ \left[ \hat{h}_\theta(\Lambda+\phi(x)/\rho_\theta(x)) - \hat{h}_{\theta^\prime}(\Lambda^\prime+\phi(x)/\rho_{\theta^\prime}(x)) \right]^2 \right. \\
        & \quad\quad \left. \cdot \mathbb{I} \left( d(\Lambda+\phi(x)/\rho_\theta(x), \Lambda^\prime+\phi(x)/\rho_{\theta^\prime}(x))
        \geq \Delta_{\hat{h}, \kappa, \alpha} \right) \right\} \Gamma(\mathrm{d} x) \\
        & \leq \left\| | \hat{h}_\theta(\Lambda+\phi(\cdot)/\rho_\theta(\cdot)) | + | \hat{h}_\theta(\Lambda-\phi(\cdot)/(1-\rho_\theta(\cdot))) | \right\|_{L^{2p}(\Gamma)}^2 \\
        & \quad\quad \left[ P_{X \sim \Gamma} \left( d(\Lambda+\phi(X)/\rho_\theta(X), \Lambda^\prime+\phi(X)/\rho_{\theta^\prime}(X)) > \Delta_{\hat{h}, \kappa, \alpha} \right) \right]^{\frac{1}{q}} \\
        & \leq C_{\hat{h},\gamma}^2 \left[ 2^{2p-1} \int \left[ 
        V^{2p\gamma}(\Lambda+\phi(x)/\rho_\theta(x))
        + V^{2p\gamma}(\Lambda-\phi(x)/(1-\rho_\theta(x)))
        \right] \Gamma(\mathrm{d} x) \right]^{\frac{1}{p}} \\
        & \quad\quad \left[ \frac{\mathbb{E}_{X \sim \Gamma}
        \left[ [d(\Lambda+\phi(X)/\rho_\theta(X), \Lambda^\prime+\phi(X)/\rho_{\theta^\prime}(X))]^{2/(1+\epsilon)} \right]}{\Delta_{\hat{h}, \kappa, \alpha}^{2/(1+\epsilon)}} \right]^{\frac{1}{q}} \\
        & \leq C_{\hat{h},\gamma}^2
        \left[ \frac{2^{2p-1}(\beta_{2p\gamma} + b_{2p\gamma})}{\iota} \right]^{\frac{1}{p}} \\
        & \quad\quad V^{2\gamma}(\Lambda)
        \left[ \frac{2 [d(\Lambda, \Lambda^\prime)]^{2/(1+\epsilon)}
        + 2 \mathbb{E}_{X \sim \Gamma} \left[ \|\phi(X)\|^{2/(1+\epsilon)} |\rho_\theta(X) - \rho_{\theta^\prime}(X)|^{2/(1+\epsilon)}/\iota_\rho^{4/(1+\epsilon)} \right]}{\Delta_{\hat{h}, \kappa, \alpha}^{2/(1+\epsilon)}} \right]^{\frac{1}{q}} \\
        & \leq C_{\hat{h},\gamma}^2
        \left[ \frac{2^{2p-1}(\beta_{2p\gamma} + b_{2p\gamma})}{\iota} \right]^{\frac{1}{p}} \\
        & \quad\quad V^{2\gamma}(\Lambda)
        \left[ \frac{d(\Lambda, \Lambda^\prime)^{2/(1+\epsilon)}
        + [\mathbb{E} \|\phi(X)\|^{2/\epsilon}]^{\epsilon/(1+\epsilon)}
        L_\rho^{2/(1+\epsilon)} [d(\theta, \theta^\prime)]^{2/(1+\epsilon)}/\iota_\rho^{4/(1+\epsilon)} }{\Delta_{\hat{h}, \kappa, \alpha}^{2/(1+\epsilon)}} \right]^{\frac{1}{q}} \\
        & \leq C_{\hat{h},\gamma}^2
        \left[ \frac{2^{2p-1}(\beta_{2p\gamma} + b_{2p\gamma})}{\iota} \right]^{\frac{1}{p}}
        \left[ \frac{1 + [\mathbb{E} \|\phi(X)\|^{2/\epsilon}]^{\epsilon/(1+\epsilon)} L_\rho^{2/(1+\epsilon)}/\iota_\rho^{4/(1+\epsilon)} }{\Delta_{\hat{h}, \kappa, \alpha}^{2/(1+\epsilon)}} \right]^{\frac{1}{q}} \\
        & \quad\quad V^{2\gamma}(\Lambda)
        \left[ [d(\Lambda, \Lambda^\prime)]^{\frac{2}{q(1+\epsilon)}} + [d(\theta, \theta^\prime)]^{\frac{2}{q(1+\epsilon)}} \right] .
    \end{align*}
    Thus, taking $\alpha = \frac{1}{q(1+\epsilon)}$ and $\kappa = \gamma$, we can conclude that for any $\kappa \in \left(0,\frac{1}{2p}\right)$, there exists some $L > 0$ such that
    \begin{align*}
        & \left\| \left[ \hat{h}_\theta(\Lambda+\phi(\cdot)/\rho_\theta(\cdot)) - \hat{h}_{\theta^\prime}(\Lambda^\prime+\phi(\cdot)/\rho_{\theta^\prime}(\cdot)) \right] \right. \\
        & \quad \left. \cdot \mathbb{I} \left( d(\Lambda+\phi(\cdot)/\rho_\theta(\cdot), \Lambda^\prime+\phi(\cdot)/\rho_{\theta^\prime}(\cdot))
        \geq \Delta_{\hat{h}, \kappa, \alpha} \right) \right\|_{L^2(\Gamma)}
        \leq L V^\kappa(\Lambda) ([d(\theta, \theta^\prime)]^\alpha + [d(\Lambda, \Lambda^\prime)]^\alpha) .
    \end{align*}
    Since $q>1$ and $\epsilon>0$ can be chosen arbitrarily, the expression $\alpha = \frac{1}{q(1+\epsilon)}$ ranges over $(0,1)$.
    Hence, for any $\alpha \in (0,1)$, there exists an interval $(0,\gamma_{\mathrm{max}})$ such that for any $\kappa \in (0,\gamma_{\mathrm{max}})$, the above inequality holds.
    Moreover, for any fixed $\alpha \in (0,1)$, since $V \geq 1$, the above inequality also holds for any $\kappa \in (0,1)$.

    Next, let $p, q>0$, $\kappa$, $\alpha \in (0,1)$ with $2p\alpha \leq 2$, $2q\kappa \leq 1$ and $\frac{2}{q} + \frac{1}{p} = 1$, when $d(\theta, \theta^\prime) < \tilde{\Delta}_{\hat{h}, \kappa, \alpha}$,
    \begin{align*}
        & \quad \int \left\{ \left[ \hat{h}_\theta(\Lambda+\phi(x)/\rho_\theta(x)) - \hat{h}_{\theta^\prime}(\Lambda^\prime+\phi(x)/\rho_{\theta^\prime}(x)) \right]^2 \right. \\
        & \quad\quad \left. \cdot \mathbb{I} \left( d(\Lambda+\phi(x)/\rho_\theta(x), \Lambda^\prime+\phi(x)/\rho_{\theta^\prime}(x))
        < \Delta_{\hat{h}, \kappa, \alpha} \right) \right\} \Gamma(\mathrm{d} x) \\
        & \leq \int \left[ L_{\hat{h}, \kappa, \alpha}
        V^\kappa(\Lambda+\phi(x)/\rho_\theta(x))
        d(\Lambda+\phi(x)/\rho_\theta(x), \Lambda^\prime+\phi(x)/\rho_{\theta^\prime}(x))^\alpha \right. \\
        & \quad\quad \left. + \tilde{L}_{\hat{h}, \kappa, \alpha}
        V^\kappa(\Lambda+\phi(x)/\rho_\theta(x))
        [d(\theta, \theta^\prime)]^\alpha \right]^2 \Gamma(\mathrm{d} x) \\
        & \leq \int \max \left\{ L_{\hat{h}, \kappa, \alpha}^2, \tilde{L}_{\hat{h}, \kappa, \alpha}^2 \right\}
        V^{2\kappa}(\Lambda+\phi(x)/\rho_\theta(x)) \\
        & \quad\quad \left[ d[(\Lambda, \Lambda^\prime)]^\alpha + \|\phi(x)\|^\alpha |\rho_\theta(x)-\rho_{\theta^\prime}(x)|^\alpha / \iota_\rho^{2\alpha} + [d(\theta, \theta^\prime)]^\alpha \right]^2 \Gamma(\mathrm{d} x) \\
        & \leq 3 \max \left\{ L_{\hat{h}, \kappa, \alpha}^2, \tilde{L}_{\hat{h}, \kappa, \alpha}^2 \right\} \left[ [d(\Lambda, \Lambda^\prime)]^{2\alpha} + [d(\theta, \theta^\prime)]^{2\alpha} \right]
        \int V^{2\kappa}(\Lambda+\phi(x)/\rho_\theta(x)) \Gamma(\mathrm{d} x) \\
        & \quad + \frac{3 \max \left\{ L_{\hat{h}, \kappa, \alpha}^2, \tilde{L}_{\hat{h}, \kappa, \alpha}^2 \right\}}{\iota_\rho^{4\alpha}}
        \int V^{2\kappa}(\Lambda+\phi(x)/\rho_\theta(x))
        \|\phi(x)\|^{2\alpha} |\rho_\theta(x)-\rho_{\theta^\prime}(x)|^{2\alpha} \Gamma(\mathrm{d} x) \\
        & \leq 3 \max \left\{ L_{\hat{h}, \kappa, \alpha}^2, \tilde{L}_{\hat{h}, \kappa, \alpha}^2 \right\}
        \left[ [d(\Lambda, \Lambda^\prime)]^{2\alpha} + [d(\theta, \theta^\prime)]^{2\alpha} \right]
        \int V^{2\kappa}(\Lambda+\phi(x)/\rho_\theta(x)) \Gamma(\mathrm{d} x) \\
        & \quad + \frac{3 \max \left\{ L_{\hat{h}, \kappa, \alpha}^2, \tilde{L}_{\hat{h}, \kappa, \alpha}^2 \right\}}{\iota_\rho^{4\alpha}}
        \left[ \int V^{2q\kappa}(\Lambda+\phi(x)/\rho_\theta(x)) \Gamma(\mathrm{d} x) \right]^{\frac{1}{q}} \\
        & \quad\quad \left[ \int \|\phi(x)\|^{2q\alpha} \Gamma(\mathrm{d} x) \right]^{\frac{1}{q}}
        \left[ \int \left| \rho_\theta(x)-\rho_{\theta^\prime}(x) \right|^{2p\alpha} \Gamma(\mathrm{d} x) \right]^{\frac{1}{p}} \\
        & \leq 3 \max \left\{ L_{\hat{h}, \kappa, \alpha}^2, \tilde{L}_{\hat{h}, \kappa, \alpha}^2 \right\}
        \left[ [d(\Lambda, \Lambda^\prime)]^{2\alpha} + [d(\theta, \theta^\prime)]^{2\alpha} \right]
        \frac{\beta_{2\kappa} + b_{2\kappa}}{\iota} V^{2\kappa}(\Lambda) \\
        & \quad + \frac{3 \max \left\{ L_{\hat{h}, \kappa, \alpha}^2, \tilde{L}_{\hat{h}, \kappa, \alpha}^2 \right\}}{\iota_\rho^{4\alpha}}
        \left[ \frac{\beta_{2q\kappa} + b_{2q\kappa}}{\iota} V^{2q\kappa}(\Lambda) \right]^{\frac{1}{q}}
        \left[ \mathbb{E} \|\phi(X)\|^{2q\alpha} \right]
        \left[ L_\rho d(\theta, \theta^\prime) \right]^{2\alpha} \\
        & \leq L^2 V^{2\kappa}(\Lambda) \left[ [d(\Lambda, \Lambda^\prime)]^{2\alpha} + [d(\theta, \theta^\prime)]^{2\alpha} \right]
    \end{align*}
    for some constant $L>0$.
    Since $\alpha$ can range over $(0,1)$, for any $\alpha \in (0,1)$, there exists an interval $(0,\kappa_{\mathrm{max}})$ such that for any $\kappa \in (0,\kappa_{\mathrm{max}})$, the inequality holds that
    \begin{align*}
        & \left\| \left[ \hat{h}_\theta(\Lambda+\phi(\cdot)/\rho_\theta(\cdot)) - \hat{h}_{\theta^\prime}(\Lambda^\prime+\phi(\cdot)/\rho_{\theta^\prime}(\cdot)) \right] \right. \\
        & \quad \left. \cdot \mathbb{I} \left( d(\Lambda+\phi(\cdot)/\rho_\theta(\cdot), \Lambda^\prime+\phi(\cdot)/\rho_{\theta^\prime}(\cdot))
        < \Delta_{\hat{h}, \kappa, \alpha} \right) \right\|_{L^2(\Gamma)}
        \leq L V^\kappa(\Lambda) ([d(\theta, \theta^\prime)]^\alpha + [d(\Lambda, \Lambda^\prime)]^\alpha) .
    \end{align*}
    Moreover, for any fixed $\alpha \in (0,1)$, since $V \geq 1$, the above inequality also holds for any $\kappa \in (0,1)$ with some constant $L > 0$.

    The bound of $T^{(5)}$ follows from the boundedness of $h_\theta$ in Lemma \ref{lemma_h_lipschitz_continuous_bounded}, the finiteness of $\mathbb{E}_\theta Z^{(u)}$ and the joint locally Hölder continuity of $\{P_\theta \hat{h}_\theta\}_{\theta \in \Theta}$, which follows from Corollaries \ref{corollary_n_step_h_robust_lipschitz_continuous_bounded} and \ref{corollary_V_Poisson_solution_robust_lipschitz_continuous_bounded}

    The bound of $T^{(6)}$ follows from the Lipschitz continuity of $h_\theta$ in Lemma \ref{lemma_h_lipschitz_continuous_bounded}, the Lipschitz continuity of $f_{Z^{(u)},\theta}$ in Assumption \ref{assumption_lipschitz_continuity_conditional_expectation}, the bound of $P_\theta \hat{h}_\theta$, which follows from Corollaries \ref{corollary_n_step_h_robust_lipschitz_continuous_bounded} and \ref{corollary_V_Poisson_solution_robust_lipschitz_continuous_bounded}.

    In summary, the desired continuity and boundedness properties hold.
\end{proof}

\section{Lemmas for the Markov Chain}
\label{sec_lemmas_Markov_chain}

In this section, suppose that the parameter space $\Theta$ is a complete metric space with the metric $d(\theta, \theta^\prime)$ and the state space $\mathrm{X}$ is a Polish (complete separable metric) space with the metric $d(\Lambda,\Lambda^\prime)$.

\subsection{Definitions and Lemmas on Transition Kernels}
\label{subsec_definitions_lemmas_kernel}

\subsubsection{Definitions of Transition Kernels}

In this section, we formally introduce the definitions and assumptions regarding the continuity and stability of the transition kernel family $\{P_\theta\}_{\theta \in \Theta}$.
These definitions and assumptions are crucial for the subsequent analysis.
We begin by defining the necessary metrics and continuity properties.

\begin{definition}[Robust Wasserstein Metric]
    \label{definition_robust_wasserstein_metric}

    Let $\nu$ and $\nu^\prime$ be two finite measures. The \emph{robust Wasserstein metric} between them is defined as:
    \begin{align*}
        & \quad \mathcal{W}_{d, \delta}(\nu, \nu^\prime) \\
        & := \inf \left\{ \mathcal{W}_d \left( \mu, \mu^\prime \right)
        \mid \|\mu\| = \|\mu^\prime\|,
        \|\nu - \mu\| \leq \delta, \|\nu^\prime - \mu^\prime\| \leq \delta,
        \nu \geq \mu, \nu^\prime \geq \mu^\prime \right\} ,
    \end{align*}
    where $\mathcal{W}_d$ denotes the standard $d$-Wasserstein distance.
\end{definition}

\begin{definition}[Coupled Robust Lipschitz Continuity of Transition Kernels, Definition \ref{lemma_continuity_transition_kernels_main}]
    \label{definition_coupled_robust_lipschitz_continuity_kernels}

    Let $P$ and $Q$ be transition probability kernels on $(\mathrm{X}, \mathcal{X})$. 
    We say that $P$ and $Q$ are \emph{$(L_P,\tau,\epsilon)$-coupled robustly Lipschitz continuous} if there exists a coupling kernel
    \begin{equation*}
        K: \mathrm{X}^2 \times \mathcal{X}^{\otimes 2} \to [0,1]
    \end{equation*}
    such that, for all $\Lambda, \Lambda^\prime \in \mathrm{X}$,
    \begin{align*}
        K(\Lambda, \Lambda^\prime; A \times \mathrm{X})
        & \leq P(\Lambda, A), 
        && A \in \mathcal{X}, \\
        K(\Lambda, \Lambda^\prime; \mathrm{X} \times B)
        & \leq Q(\Lambda^\prime, B), 
        && B \in \mathcal{X},
    \end{align*}
    and the following bounds hold:
    \begin{align*}
        0 \leq 1 - K(\Lambda, \Lambda^\prime; \mathrm{X} \times \mathrm{X})
        & \leq L_P d(\Lambda, \Lambda^\prime) + \tau , \\
        \int d(u,v) K(\Lambda, \Lambda^\prime; \mathrm{d} u \times \mathrm{d} v)
        &\leq d(\Lambda, \Lambda^\prime) + \epsilon .
    \end{align*}
\end{definition}

\begin{definition}[Robust Lipschitz Continuity of a Family of Transition Kernels]
    \label{definition_robust_lipschitz_continuity_family_kernels}

    A family of transition probability kernels $\{P_\theta\}_{\theta \in \Theta}$ is \emph{robustly Lipschitz continuous} with a Lipschitz constant $L_P \geq 0$ if, for any parameters $\theta$, $\theta^\prime \in \Theta$, the kernels $P_\theta$ and $P_{\theta^\prime}$ are $(L_P, L_P d(\theta, \theta^\prime), L_P d(\theta, \theta^\prime))$-coupled robustly Lipschitz continuous.
\end{definition}

\subsubsection{Lemmas on Transition Kernels}

\begin{lemma}
    \label{lemma_existence_robust_wasserstein_metric}

    Let $\delta \geq 0$ be given.
    For any two finite measures $\nu$ and $\nu^\prime$, there exist decompositions $\nu = \mu_b + \mu_c$ and $\nu^\prime = \mu_b^\prime + \mu_c^\prime$ such that $\|\mu_c\|=\|\mu_c^\prime\| \leq \delta$ and
    \begin{equation*}
        \mathcal{W}_d\left(\mu_b, \mu_b^\prime\right) = \mathcal{W}_{d, \delta}(\nu, \nu^\prime) .
    \end{equation*}
    In addition, there is a coupling of $\left( \mu_b, \mu_b^\prime \right)$ which minimizes the total cost $\mathcal{W}_d\left(\mu_b, \mu_b^\prime\right)$ among all possible couplings $\left( \mu_b, \mu_b^\prime \right)$.
\end{lemma}
\begin{proof}
    The proof is similar to the proof of the existence of the optimal coupling (see the proof of Theorem 4.1 in \cite{villaniOptimalTransportOld2009} for details).

    The definition of $\mathcal{W}_{d,\delta}(\nu, \nu^\prime)$ in Definition \ref{definition_robust_wasserstein_metric} is given by
    \begin{align*}
        & \quad \mathcal{W}_{d, \delta}(\nu, \nu^\prime) \\
        & := \inf \left\{ \mathcal{W}_d\left(\mu, \mu^\prime\right)
        \mid \|\mu\| = \|\mu^\prime\|,
        \|\nu - \mu\| \leq \delta, \|\nu^\prime - \mu^\prime\| \leq \delta,
        \nu \geq \mu, \nu^\prime \geq \mu^\prime \right\} .
    \end{align*}

    Choose a minimizing sequence of decompositions $\left\{ \nu_b^{(n)},\nu_c^{(n)},\nu_b^{(n)\prime},\nu_c^{(n)\prime} \right\}_{n \geq 1}$ satisfying
    \begin{equation*}
        \nu=\nu_b^{(n)} + \nu_c^{(n)} , \quad
        \nu^\prime=\nu_b^{(n)\prime}+\nu_c^{(n)\prime},
    \end{equation*}
    with
    \begin{equation*}
        \|\nu_c^{(n)}\| = \|\nu_c^{(n)\prime}\| \leq \delta,
    \end{equation*}
    and
    \begin{equation*}
        \mathcal{W}_d \left(\nu_b^{(n)}, \nu_b^{(n)\prime}\right)
        \rightarrow \mathcal{W}_{d,\delta}(\nu,\nu^\prime) .
    \end{equation*}

    Since all measures are dominated by $\nu$ or $\nu^\prime$, the sequences $\{\nu_b^{(n)}\}$ and $\{\nu_b^{(n)\prime}\}$ are tight.
    By Prokhorov's theorem, along a subsequence (still denoted by $\nu_b^{(n)}, \nu_c^{(n)}, \nu_b^{(n) \prime}, \nu_c^{(n) \prime}$, without loss of generality), we have
    \begin{equation*}
        \nu_b^{(n)}\Rightarrow\mu_b,\quad \nu_b^{(n)\prime}\Rightarrow\mu_b^\prime,\quad
        \nu_c^{(n)}\Rightarrow\mu_c,\quad \nu_c^{(n)\prime}\Rightarrow\mu_c^\prime,
    \end{equation*}
    for some finite measures $\mu_b,\mu_b^\prime,\mu_c,\mu_c^\prime$, where $\Rightarrow$ denotes weak convergence of measures.
    Thus,
    \begin{equation*}
        \nu=\mu_b+\mu_c, \quad \nu^\prime=\mu_b^\prime+\mu_c^\prime,\quad \|\mu_c\|=\|\mu_c^\prime\| \leq \delta .
    \end{equation*}

    For each $n$, let $\pi^{(n)}$ be an optimal coupling between
    $\nu_b^{(n)}$ and $\nu_b^{(n)\prime}$, so that
    \begin{equation*}
        \mathcal{W}_d \left(\nu_b^{(n)}, \nu_b^{(n)\prime}\right)
        =\int d(x,y) \pi^{(n)}(\mathrm{d} x, \mathrm{d} y) .
    \end{equation*}
    By tightness, extract a weakly convergent subsequence $\pi^{(n)}\Rightarrow\pi$ with marginals $\mu_b$ and $\mu_b^\prime$.

    By Lemma 4.3 in \cite{villaniOptimalTransportOld2009},
    \begin{align*}
        & \quad \mathcal{W}_d \left( \mu_b, \mu_b^\prime \right)
        \leq \int d(x,y) \pi(\mathrm{d} x, \mathrm{d} y)
        \leq \liminf_{n \rightarrow \infty} \int d(x,y) \pi^{(n)}(\mathrm{d} x, \mathrm{d} y) \\
        &= \liminf_{n \rightarrow \infty}
        \mathcal{W}_d \left( \nu_b^{(n)}, \nu_b^{(n)\prime} \right)
        = \mathcal{W}_{d,\delta}(\nu,\nu^\prime) .
    \end{align*}
    Observe that the pair $(\mu_b,\mu_b^\prime)$ is an admissible decomposition in the definition of $\mathcal{W}_{d,\delta}(\nu,\nu^\prime)$.
    Hence,
    \begin{equation*}
        \mathcal{W}_{d,\delta}(\nu,\nu^\prime)
        \leq
        \mathcal{W}_d \left( \mu_b, \mu_b^\prime \right) .
    \end{equation*}
    Consequently, the equality $\mathcal{W}_{d,\delta}(\nu,\nu^\prime) = \mathcal{W}_d \left( \mu_b, \mu_b^\prime \right)$ holds.
\end{proof}

\begin{lemma}
    \label{lemma_decompose_bounded_complement}

    Let $\Delta > 0$ and $\delta \geq 0$ be given.
    For any two finite measures $\nu$ and $\nu^\prime$, there exist decompositions $\nu = \mu_b + \mu_c$ and $\nu^\prime = \mu_b^\prime + \mu_c^\prime$ such that
    \begin{enumerate}
        \item $\|\mu_b\|=\|\mu_b^\prime\|$ and $\|\mu_c\|=\|\mu_c^\prime\| \leq \delta + \frac{\mathcal{W}_{d, \delta} \left( \nu,\nu^\prime \right)}{\Delta}$,
        \item there exists a coupling $\pi_{\mathrm{cp}}$ of the finite measures $\mu_b$ and $\mu_b^\prime$ such that
        \begin{equation*}
            d(X,Y) \leq \Delta \quad \text{for almost all } (X,Y) \sim \pi_{\text{cp}}
        \end{equation*}
        and
        \begin{equation*}
            \mathbb{E}_{(X,Y) \sim \pi_{\text{cp}}} [d(X,Y)]
            \leq \mathcal{W}_{d, \delta} \left( \nu,\nu^\prime \right) .
        \end{equation*}
    \end{enumerate}
\end{lemma}
\begin{proof}
    By Lemma \ref{lemma_existence_robust_wasserstein_metric}, there exist decompositions
    \begin{equation*}
        \nu = \nu_b + \nu_c ,
        \quad \nu^\prime = \nu_b^\prime + \nu_c^\prime ,
    \end{equation*}
    where $\|\nu_c\| = \|\nu_c^\prime\| \leq \delta$, $\|\nu_b\| = \|\nu_b^\prime\|$, and $\nu_b, \nu_b^\prime, \nu_c, \nu_c^\prime \geq 0$ are finite measures, such that
    \begin{equation*}
        \mathcal{W}_d \left( \nu_b, \nu_b^\prime \right)
        = \mathcal{W}_{d,\delta}(\nu,\nu^\prime).
    \end{equation*}
    Moreover, Lemma \ref{lemma_existence_robust_wasserstein_metric} implies that there exists a coupling $\pi_{\mathrm{cp}}^\prime$ of the measures $\nu_b$ and $\nu_b^\prime$ satisfying
    \begin{equation*}
        \mathbb{E}_{(X,Y) \sim \pi_{\mathrm{cp}}^\prime} [d(X,Y)]
        = \mathcal{W}_{d,\delta}(\nu, \nu^\prime) .
    \end{equation*}

    Define the set $A := \{ (x,y) \mid d(x,y) \leq \Delta \}$. Then
    \begin{equation*}
        \pi_{\mathrm{cp}}^\prime(A^c)
        \leq \frac{\mathbb{E}_{(X,Y)\sim \pi_{\mathrm{cp}}^\prime}[d(X,Y)]}{\Delta}
        \leq \frac{\mathcal{W}_{d,\delta}(\nu,\nu^\prime)}{\Delta} .
    \end{equation*}

    Let $\nu_t$ and $\nu_t^\prime$ denote the marginal measures of the restricted coupling
    \begin{equation*}
        \pi_{\mathrm{cp}} := \pi_{\mathrm{cp}}^\prime \big|_A.
    \end{equation*}
    By construction, it holds that $\nu_b \geq \nu_t$, $\nu_b^\prime \geq \nu_t^\prime$ and $\| \nu_t \| = \| \nu_t^\prime \| = \| \pi_{\mathrm{cp}} \| = \pi_{\mathrm{cp}}^\prime(A)$.
    In addition,
    \begin{equation*}
        \mathbb{E}_{(X,Y) \sim \pi_{\mathrm{cp}}} [d(X,Y)]
        \leq \mathbb{E}_{(X,Y) \sim \pi_{\text{cp}}^\prime} [d(X,Y)]
        = \mathcal{W}_{d, \delta} \left( \nu,\nu^\prime \right) .
    \end{equation*}

    Now define
    \begin{align*}
        \mu_b &:= \nu_t , & \mu_b^\prime &:= \nu_t^\prime , \\
        \mu_c &:= \nu_c + [\nu_b - \nu_t], & \mu_c^\prime &:= \nu_c^\prime + [\nu_b^\prime - \nu_t^\prime] .
    \end{align*}

    By construction, we have $\|\mu_b\| = \|\mu_b^\prime\|$ and
    \begin{equation*}
        \|\mu_c\| = \|\mu_c^\prime\|
        = \| \nu_c \| + \| \nu_b - \nu_t \|
        = \| \nu_c \| + \pi_{\mathrm{cp}}^\prime(A^c)
        \leq \delta + \frac{\mathcal{W}_{d,\delta}(\nu,\nu^\prime)}{\Delta} .
    \end{equation*}

    Finally, by the definition of the set $A$, we have
    \begin{equation*}
        d(X,Y) \leq \Delta \quad \text{for almost all } (X,Y) \sim \pi_{\text{cp}} .
    \end{equation*}
    This completes the construction of the desired decompositions and coupling.
\end{proof}

\begin{lemma}
    \label{lemma_coupled_robustly_Lipschitz_continuous_robust_wasserstein_metric}

    Suppose that $P$ and $Q$ are $(L_P,\tau,\epsilon)$-coupled robustly Lipschitz continuous.
    Then, for any $\Lambda,\Lambda^\prime \in \mathrm{X}$, the robust Wasserstein metric between the probability measures $P(\Lambda, \cdot)$ and $Q(\Lambda^\prime, \cdot)$ satisfies
    \begin{equation*}
        \mathcal{W}_{d, L_P d(\Lambda, \Lambda^\prime) + \tau}
        \left( P(\Lambda, \cdot), Q(\Lambda^\prime, \cdot) \right)
        \leq d(\Lambda, \Lambda^\prime) + \epsilon .
    \end{equation*}
\end{lemma}
\begin{proof}
    This follows directly from the definitions of $(L_P,\tau,\epsilon)$-coupled robustly Lipschitz continuity and the robust Wasserstein metric.
    In particular, the coupling kernel $K$ in Definition \ref{definition_coupled_robust_lipschitz_continuity_kernels} guarantees the existence of a coupling $\pi = K(\Lambda, \Lambda^\prime; \cdot)$ such that the expected distance under $\pi$ satisfies
    \begin{equation*}
        \mathbb{E}_{(x,y) \sim \pi}[d(x,y)] \leq d(\Lambda, \Lambda^\prime) + \epsilon
    \end{equation*}
    and
    \begin{equation*}
        1 - \pi(\mathrm{X} \times \mathrm{X}) \leq L_P d(\Lambda, \Lambda^\prime) + \tau .
    \end{equation*}
    Denote the marginal measure of $\pi$ by $\mu$ and $\mu^\prime$.
    Then we have $\mu(\cdot) \leq P(\Lambda, \cdot)$, $\mu^\prime(\cdot) \leq Q(\Lambda^\prime, \cdot)$, and
    \begin{equation*}
        \mathcal{W}_{d, L_P d(\Lambda, \Lambda^\prime) + \tau}
        \leq \mathcal{W}_d \left( \mu, \mu^\prime \right)
        \leq \mathbb{E}_{(x,y) \sim \pi}[d(x,y)]
        \leq d(\Lambda, \Lambda^\prime) + \epsilon .
    \end{equation*}
    This completes the proof.
\end{proof}

\begin{theorem}
    \label{theorem_n_step_coupled_robustly_Lipschitz_continuous}
    Suppose that $P$ and $Q$ are $(L_P,\tau,\epsilon)$-coupled robustly Lipschitz continuous.
    Then there exists a coupling kernel 
    \begin{equation*}
        K: \mathrm{X}^2 \times \mathcal{X}^{\otimes 2} \to [0,1]
    \end{equation*}
    such that, for any $\Lambda$, $\Lambda^\prime \in \mathrm{X}$ and any $n \in \mathbb{N}$,
    \begin{align*}
        K^n(\Lambda, \Lambda^\prime; A \times \mathrm{X})
        & \leq P^n(\Lambda, A) , 
        && A \in \mathcal{X}, \\
        K^n(\Lambda, \Lambda^\prime; \mathrm{X} \times B)
        & \leq Q^n(\Lambda^\prime, B) , 
        && B \in \mathcal{X},
    \end{align*}
    and the following bounds hold:
    \begin{align*}
        0 \leq 1 - K^n(\Lambda, \Lambda^\prime; \mathrm{X} \times \mathrm{X})
        & \leq n L_P d(\Lambda, \Lambda^\prime) + n\tau + \frac{n(n-1)L_P \epsilon}{2} , \\
        \int d(u,v) K^n(\Lambda, \Lambda^\prime; \mathrm{d} u \times \mathrm{d} v)
        &\leq d(\Lambda, \Lambda^\prime) + n\epsilon .
    \end{align*}
\end{theorem}
\begin{proof}
    The assertions for $n = 0$ and $n = 1$ hold by definition.

    Assume the statements hold for $n=k$. We prove them for $n=k+1$.
    First,
    \begin{align*}
        & \quad K^{k+1}(\Lambda, \Lambda^\prime; A \times \mathrm{X})
        = \int K(\Lambda, \Lambda^\prime; \mathrm{d} x \times \mathrm{d} y)
        K^k(x, y; A \times \mathrm{X}) \\
        & \leq \int K(\Lambda, \Lambda^\prime; \mathrm{d} x \times \mathrm{d} y) P^k(x, A)
        = \int K(\Lambda, \Lambda^\prime; \mathrm{d} x \times \mathrm{X}) P^k(x, A) \\
        & \leq \int P(\Lambda, \mathrm{d} x) P^k(x, A)
        = P^{k+1}(\Lambda, A) .
    \end{align*}
    An analogous argument yields
    \begin{equation*}
        K^{k+1}(\Lambda, \Lambda^\prime; \mathrm{X} \times B)
        \leq Q^{k+1}(\Lambda^\prime, B) .
    \end{equation*}

    Clearly $K^{k+1}(\Lambda, \Lambda^\prime; \mathrm{X} \times \mathrm{X}) \leq 1$.
    Moreover,
    \begin{align*}
        & \quad K^{k+1}(\Lambda, \Lambda^\prime; \mathrm{X} \times \mathrm{X})
        = \int K^k(x, y; \mathrm{X} \times \mathrm{X})
        K(\Lambda, \Lambda^\prime; \mathrm{d} x \times \mathrm{d} y) \\
        & \geq \int \left[ 1 - k L_P d(x, y) - k\tau - \frac{k(k-1)L_P \epsilon}{2} \right]
        K(\Lambda, \Lambda^\prime; \mathrm{d} x \times \mathrm{d} y) \\
        &= \left[ 1 - k\tau - \frac{k(k-1)L_P \epsilon}{2} \right] K(\Lambda, \Lambda^\prime; \mathrm{X} \times \mathrm{X})
        - k L_P \int d(x, y) K(\Lambda, \Lambda^\prime; \mathrm{d} x \times \mathrm{d} y) \\
        & \geq \left[ 1 - k\tau - \frac{k(k-1)L_P \epsilon}{2} \right]
        - \left[ 1 - K(\Lambda, \Lambda^\prime; \mathrm{X} \times \mathrm{X}) \right]
        - k L_P \left[ d(\Lambda, \Lambda^\prime) + \epsilon \right] \\
        & \geq \left[ 1 - k\tau - \frac{k(k-1)L_P \epsilon}{2} \right]
        - \left[ L_P d(\Lambda, \Lambda^\prime) + \tau \right]
        - k L_P \left[ d(\Lambda, \Lambda^\prime) + \epsilon \right] \\
        &= 1 - (k+1)\tau - \frac{k(k+1)L_P \epsilon}{2} - (k+1) L_P d(\Lambda, \Lambda^\prime) ,
    \end{align*}
    where the inequality $(1-a)(1-b) \geq 1-a-b$ for $a,b \geq 0$ is used to obtain an intermediate bound in the second inequality.

    Finally,
    \begin{align*}
        & \quad \int d(u,v) K^{k+1}(\Lambda, \Lambda^\prime; \mathrm{d} u \times \mathrm{d} v)
        = \int \left[ K(\Lambda, \Lambda^\prime; \mathrm{d} x \times \mathrm{d} y) \int d(u,v) K^k(x, y; \mathrm{d} u \times \mathrm{d} v) \right] \\
        & \leq \int \left\{ K(\Lambda, \Lambda^\prime; \mathrm{d} x \times \mathrm{d} y) \left[ d(x,y) + k\epsilon \right] \right\}
        \leq \int \left[ K(\Lambda, \Lambda^\prime; \mathrm{d} x \times \mathrm{d} y) d(x,y) \right] + k\epsilon \\
        & \leq d(\Lambda, \Lambda^\prime) + (k+1)\epsilon .
    \end{align*}
    This completes the induction and the proof.
\end{proof}

\begin{corollary}[Corollary \ref{corollary_n_step_family_robustly_Lipschitz_continuous_main}]
    \label{corollary_n_step_family_robustly_Lipschitz_continuous}
    Let $\{P_\theta\}_{\theta \in \Theta}$ be a family of transition probability kernels that is robustly Lipschitz continuous with a Lipschitz constant $L_P \geq 0$.
    Then for any $\theta,\theta^\prime\in\Theta$ and any $n \in \mathbb{N}$, the probability kernels $P_\theta^n$ and $P_{\theta^\prime}^n$ are $(n L_P, n L_P d(\theta,\theta^\prime) + \frac{n(n-1) L_P^2}{2} d(\theta,\theta^\prime), n L_P d(\theta,\theta^\prime))$-coupled robustly Lipschitz continuous.
\end{corollary}

\begin{proof}
    By the definition of robust Lipschitz continuity, for any given $\theta$, $\theta^\prime \in \Theta$, the kernels $P_\theta$ and $P_{\theta^\prime}$ are $(L_P, \tau, \epsilon)$-coupled robustly Lipschitz continuous with
    \begin{equation*}
        \tau = L_P d(\theta,\theta^\prime), \quad \epsilon = L_P d(\theta,\theta^\prime) .
    \end{equation*}
    Hence, the condition of Theorem \ref{theorem_n_step_coupled_robustly_Lipschitz_continuous} are satisfied with these choices of $\tau$ and $\epsilon$.
    Applying that theorem yields the existence of a coupling kernel
    $K_{\theta,\theta^\prime}$ such that
    \begin{equation*}
        K_{\theta,\theta^\prime}^n(\Lambda,\Lambda^\prime;A\times\mathrm{X}) \leq P_\theta^n(\Lambda,A) , \quad
        K_{\theta,\theta^\prime}^n(\Lambda,\Lambda^\prime;\mathrm{X}\times B) \leq P_{\theta^\prime}^n(\Lambda^\prime,B) ,
    \end{equation*}
    together with the bounds
    \begin{equation*}
        0 \leq 1 - K_{\theta,\theta^\prime}^n(\Lambda,\Lambda^\prime;\mathrm{X}\times\mathrm{X})
        \leq n L_P d(\Lambda,\Lambda^\prime) + n\tau + \frac{n(n-1)L_P\epsilon}{2} ,
    \end{equation*}
    and
    \begin{equation*}
        \int d(u,v) K_{\theta,\theta^\prime}^n(\Lambda,\Lambda^\prime; \mathrm{d} u \times \mathrm{d} v)
        \leq d(\Lambda,\Lambda^\prime) + n\epsilon .
    \end{equation*}
    Substituting $\tau=\epsilon=L_P d(\theta,\theta^\prime)$ into these two inequalities proves the corollary.
\end{proof}

\subsection{Definitions and Assumptions on Functions}

\subsubsection{Definitions on Functions}

\begin{definition}[Local Hölder Continuity]
    Let $h: \mathrm{X} \to \mathbb{R}^m$ be a function.
    We say that $h$ is \emph{$(\Delta,\psi,\alpha)$-locally Hölder continuous} if for any $\Lambda,\Lambda^\prime \in \mathrm{X}$ with $d(\Lambda,\Lambda^\prime) < \Delta$, it holds that
    \begin{equation*}
        \|h(\Lambda)-h(\Lambda^\prime)\| \leq \psi(\Lambda) [d(\Lambda, \Lambda^\prime)]^\alpha ,
    \end{equation*}
    where $\Delta>0$, $\alpha \in (0,1]$, and
    $\psi:\mathrm{X}\to(0,\infty)$ is a nonnegative function depending on $\Lambda$.
\end{definition}

\subsubsection{Lemmas on Functions}

\begin{lemma}
    \label{lemma_robust_lipschitz_continuous_bounded}

    Let $\nu$ and $\nu^\prime$ be probability measures defined on the state space $\mathrm{X}$, and let $V: \mathrm{X} \to [1, \infty)$ be a measurable function.
    Suppose $\max\{\nu V^{a}, \nu^\prime V^{a}\} \leq C_{\nu,a}$ for any $a \in (0,1]$.
    Let $h$ be a function satisfying $(\Delta, L_h V^\kappa, \alpha)$-local Hölder continuity and $|h| \leq C_{h,\gamma} V^\gamma$ for some $\kappa$, $\gamma \in (0,1)$.
    Then for any $a \in (\gamma, 1]$,
    \begin{align*}
        & \quad | \nu h - \nu^\prime h | \\
        & \leq 2 C_{h,\gamma}
        \left[ \frac{C_{\nu,a}}{1-\frac{\gamma}{a}} + 1 \right]
        \left[ \delta + \frac{\mathcal{W}_{d, \delta} \left( \nu,\nu^\prime \right)}{\Delta} \right]^{1-\gamma/a}
        + L_h C_{\nu,1}^\kappa (\Delta+1)
        \mathcal{W}_{d, \delta}\left( \nu,\nu^\prime \right)^{\min \left\{ 1-\kappa,\alpha \right\}} .
    \end{align*}
\end{lemma}
\begin{proof}
    Apply the decomposition from Lemma \ref{lemma_decompose_bounded_complement} to write $\nu = \mu_b + \mu_c$ and $\nu^\prime = \mu_b^\prime + \mu_c^\prime$.
    Define $\epsilon = \delta + \frac{\mathcal{W}_{d, \delta} \left( \nu,\nu^\prime \right)}{\Delta}$.
    Then $\|\mu_c\|=\|\mu_c^\prime\| \leq \epsilon$.

    For $X \sim \nu$ and any $t>0$, we have the bound $P(V(X)>t) \leq \frac{C_{\nu,a}}{t^a}$.
    Since $|h(X)| \leq C_{h,\gamma} V^\gamma(X)$, it follows that
    \begin{equation*}
        P(|h(X)| > t)
        \leq P(C_{h,\gamma} V^\gamma(X) > t)
        = P \left( V(X) > \left( \frac{t}{C_{h,\gamma}} \right)^{\frac{1}{\gamma}} \right)
        \leq C_{\nu,a} C_{h,\gamma}^{\frac{a}{\gamma}} t^{-\frac{a}{\gamma}} .
    \end{equation*}
    Consequently,
    \begin{equation*}
        \mathbb{E} \left[ |h(X)| \mathbb{I}(|h(X)| > t) \right]
        = tP(|h(X)| > t) + \int_{t}^\infty P(|h(X)| > s) ds
        \leq \frac{C_{\nu,a} C_{h,\gamma}^{\frac{a}{\gamma}} t^{1-\frac{a}{\gamma}}}{1-\gamma/a} .
    \end{equation*}

    Therefore,
    \begin{equation*}
        |\mu_c h|
        \leq \mu_c [|h|\mathbb{I}(|h|>t)] + \mu_c [|h|\mathbb{I}(|h| \leq t)]
        \leq \nu [|h|\mathbb{I}(|h|>t)] + \mu_c t
        \leq \frac{C_{\nu,a} C_{h,\gamma}^{\frac{a}{\gamma}} t^{1-\frac{a}{\gamma}}}{1-\gamma/a}
        + t \epsilon .
    \end{equation*}
    When $\epsilon>0$, choosing $t=C_{h,\gamma}\epsilon^{-\gamma/a}$ reduces the bound to
    \begin{equation*}
        C_{h,\gamma}
        \left[ \frac{C_{\nu,a}}{1-\frac{\gamma}{a}} + 1 \right]
        \epsilon^{1-\gamma/a} .
    \end{equation*}
    Thus, we obtain
    \begin{equation*}
        |\mu_c h| \leq C_{h,\gamma}
        \left[ \frac{C_{\nu,a}}{1-\frac{\gamma}{a}} + 1 \right]
        \epsilon^{1-\gamma/a} .
    \end{equation*}
    An analogous bound holds for $|\mu_c^\prime h|$, and the same result applies when $\epsilon = 0$.

    By Lemma \ref{lemma_decompose_bounded_complement}, there exists a coupling $\pi_{\mathrm{cp}}$ of $\mu_b$ and $\mu_b^\prime$ supported on $\{d(x,y) \leq \Delta\}$ with $\pi_{\mathrm{cp}} d(X,Y) \leq \mathcal{W}_{d,\delta}(\nu,\nu^\prime)$.
    Using local Hölder continuity of $h$,
    \begin{align*}
        & \quad |\mu_b h - \mu_b^\prime h|
        \leq \pi_{\mathrm{cp}} |h(X)-h(Y)|
        \leq \pi_{\mathrm{cp}} \left\{ L_h V^\kappa(X) d(X,Y)^\alpha \right\} \\
        & \leq L_h (\pi_{\mathrm{cp}} V(X))^{\kappa}
        (\pi_{\mathrm{cp}} d(X,Y)^{\frac{\alpha}{1-\kappa}})^{1-\kappa}
        \leq L_h (\mu_b V(X))^\kappa \Delta^\alpha
        \left( \pi_{\mathrm{cp}} \left( \frac{d(X,Y)}{\Delta} \right)^{\min \left\{ 1,\frac{\alpha}{1-\kappa} \right\}} \right)^{1-\kappa} \\
        & \leq L_h C_{\nu,1}^\kappa \Delta^\alpha
        \left( \pi_{\mathrm{cp}} \left( \frac{d(X,Y)}{\Delta} \right) \right)^{\min \left\{ 1-\kappa,\alpha \right\}}
        \leq L_h C_{\nu,1}^\kappa \Delta^\alpha
        \left[ \mathcal{W}_{d, \delta} \left( \nu,\nu^\prime \right) \Delta^{-1} \right]^{\min \left\{ 1-\kappa,\alpha \right\}} \\
        & \leq L_h C_{\nu,1}^\kappa \Delta^{\max\{\alpha+\kappa-1,0\}}
        \mathcal{W}_{d, \delta} \left( \nu,\nu^\prime \right)^{\min \left\{ 1-\kappa,\alpha \right\}}
        \leq L_h C_{\nu,1}^\kappa (\Delta+1)
        \mathcal{W}_{d, \delta}\left( \nu,\nu^\prime \right)^{\min \left\{ 1-\kappa,\alpha \right\}} ,
    \end{align*}
    where Hölder's inequality and the inequality $\max\{\nu V, \nu^\prime V\} \leq C_{\nu,1}$ are used.

    Combining the bounds for $\mu_c$, $\mu_c^\prime$ and $\mu_b - \mu_b^\prime$, it holds that
    \begin{align*}
        | \nu h - \nu^\prime h |
        & \leq |\mu_c h| + |\mu_c^\prime h| + |\mu_b h - \mu_b^\prime h| \\
        & \leq 2 C_{h,\gamma}
        \left[ \frac{C_{\nu,a}}{1-\frac{\gamma}{a}} + 1 \right]
        \epsilon^{1-\gamma/a}
        + L_h C_{\nu,1}^\kappa (\Delta+1)
        \mathcal{W}_{d, \delta}\left( \nu,\nu^\prime \right)^{\min \left\{ 1-\kappa,\alpha \right\}} .
    \end{align*}
\end{proof}

\begin{corollary}
    \label{corollary_n_step_robust_lipschitz_continuous_bounded}

    Let $\{P_\theta\}_{\theta \in \Theta}$ be a family of transition probability kernels that is robustly Lipschitz continuous with a Lipschitz constant $L_P \geq 0$.
    Suppose that for any $a \in (0, 1]$, there exist constants $\beta_a \in (0, 1)$ and $b_a < \infty$, independent of $\theta$, such that for all $\Lambda \in \mathrm{X}$:
    \begin{equation*}
        P_\theta V^a(\Lambda) \leq \beta_a V^a(\Lambda) + b_a.
    \end{equation*}
    Assume that the function $h$ satisfies $(\Delta, L_h V^\kappa, \alpha)$-local Hölder continuity and $|h| \leq C_{h,\gamma} V^\gamma$ for some $\kappa$, $\gamma \in (0,1)$.
    Then for any $a \in (\gamma, 1]$ and $n \in \mathbb{N}$,
    \begin{align*}
        & \quad \left| P_\theta^n(\Lambda,h) - P_{\theta^\prime}^n(\Lambda^\prime,h) \right| \\
        & \leq 2 C_{h,\gamma}
        \left[ \frac{C_{\nu,a}}{1-\frac{\gamma}{a}} + 1 \right]
        \left[ L_P \left( nd(\Lambda,\Lambda^\prime) + \frac{2n + n(n-1)L_P}{2} d(\theta, \theta^\prime) \right)
        + \frac{d(\Lambda,\Lambda^\prime) + n L_Pd(\theta, \theta^\prime)}{\Delta} \right]^{1-\gamma/a} \\
        & \quad + L_h C_{\nu,1}^\kappa (\Delta+1)
        \left[ d(\Lambda,\Lambda^\prime) + n L_Pd(\theta, \theta^\prime) \right]^{\min \left\{ 1-\kappa,\alpha \right\}} ,
    \end{align*}
    where $C_{\nu,a} = \max \left\{ V^a(\Lambda),V^a(\Lambda^\prime),\frac{b_a}{1-\beta_a} \right\}$.
    When $\theta = \theta^\prime$,
    \begin{align*}
        & \quad \left| P_\theta^n(\Lambda,h) - P_\theta^n(\Lambda^\prime,h) \right| \\
        & \leq 2 C_{h,\gamma}
        \left[ \frac{C_{\nu,a}}{1-\frac{\gamma}{a}} + 1 \right]
        \left[ nL_Pd(\Lambda,\Lambda^\prime) + \frac{d(\Lambda,\Lambda^\prime)}{\Delta} \right]^{1-\gamma/a}
        + L_h C_{\nu,1}^\kappa (\Delta+1) [d(\Lambda,\Lambda^\prime)]^{\min \left\{ 1-\kappa,\alpha \right\}} .
    \end{align*}
\end{corollary}
\begin{proof}
    By Corollary \ref{corollary_n_step_family_robustly_Lipschitz_continuous}, the $n$-step kernels $P_\theta^n$ and $P_{\theta^\prime}^n$ are $(n L_P, n L_P d(\theta,\theta^\prime) + \frac{n(n-1) L_P^2}{2} d(\theta,\theta^\prime), n L_P d(\theta,\theta^\prime))$-coupled robustly Lipschitz continuous.
    
    Thus, Lemma \ref{lemma_coupled_robustly_Lipschitz_continuous_robust_wasserstein_metric} implies that
    \begin{equation*}
        \mathcal{W}_{d, n L_P d(\Lambda, \Lambda^\prime) + n L_P d(\theta,\theta^\prime) + \frac{n(n-1) L_P^2}{2} d(\theta,\theta^\prime)}
        \left( P(\Lambda, \cdot), Q(\Lambda^\prime, \cdot) \right)
        \leq d(\Lambda, \Lambda^\prime) + n L_P d(\theta,\theta^\prime) .
    \end{equation*}

    Combining with the setting $C_{\nu,a} = \max \left\{ V^a(\Lambda),V^a(\Lambda^\prime),\frac{b_a}{1-\beta_a} \right\}$, the inequality
    \begin{equation*}
        P_\theta V^a(\Lambda) \leq \beta_a V^a(\Lambda) + b_a
    \end{equation*}
    implies that
    \begin{equation*}
        \max\{ P_\theta^n(\Lambda,V^a), P_{\theta^\prime}^n(\Lambda^\prime,V^a)\} \leq C_{\nu,a} .
    \end{equation*}
    Applying Lemma \ref{lemma_robust_lipschitz_continuous_bounded} yields the asserted inequalities.
\end{proof}

\subsection{Definitions and Assumptions on Markov Chains and Functions with Parameters}
\label{subsec_definitions_assumptions_chains_functions}

This subsection aims to establish the joint local Hölder continuity of $\hat{h}_\theta$, $\pi_\theta h_\theta$, $P_\theta^n h_\theta$ and $h_\theta^n$.

\subsubsection{Assumptions}

\begin{definition}[Joint Local Hölder Continuity]
    A family of functions $\{h_\theta\}_{\theta \in \Theta}$ is said to be \emph{$((\Delta,\psi,\alpha), (\tilde{\Delta},\tilde{\psi},\tilde{\alpha}))$-jointly locally Hölder continuous} if it satisfies the following two conditions:
    \begin{enumerate}
        \item \textbf{(Continuity in state)} For each $\theta \in \Theta$, the function $h_\theta(\cdot)$ is $(\Delta, \psi, \alpha)$-locally Hölder continuous. That is, for some $\Delta > 0$, $\alpha \in (0, 1]$, and a function $\psi: \mathrm{X} \to (0,\infty)$,
        \begin{equation*}
            \|h_\theta(\Lambda) - h_\theta(\Lambda^\prime)\| \leq \psi(\Lambda) [d(\Lambda, \Lambda^\prime)]^\alpha
        \end{equation*}
        for all $\Lambda, \Lambda^\prime$ with $d(\Lambda, \Lambda^\prime) < \Delta$.

        \item \textbf{(Continuity in parameter)} For each state $\Lambda \in \mathrm{X}$, the mapping $\theta \mapsto h_\theta(\Lambda)$ is $(\tilde{\Delta}, \tilde{\psi}, \tilde{\alpha})$-locally Hölder continuous. That is, for some $\tilde{\Delta} > 0$, $\tilde{\alpha} \in (0, 1]$, and a function $\tilde{\psi}: \mathrm{X} \to (0,\infty)$,
        \begin{equation*}
            \|h_\theta(\Lambda) - h_{\theta^\prime}(\Lambda)\| \leq \tilde{\psi}(\Lambda) [d(\theta, \theta^\prime)]^{\tilde{\alpha}}
        \end{equation*}
        for all $\theta, \theta^\prime$ with $d(\theta, \theta^\prime) < \tilde{\Delta}$.
    \end{enumerate}
\end{definition}

\begin{assumption}[Simultaneous Stability, Ergodicity, and Regularity]
    \label{assumption_simultaneous_properties}

    The family of transition kernels $\{P_\theta\}_{\theta \in \Theta}$ and the associated Lyapunov function $V: \mathrm{X} \to [1, \infty)$ are assumed to satisfy the following conditions:
    \begin{enumerate}
        \item \label{item_assumption_simultaneous_properties_1} \textbf{(Simultaneous Geometric Ergodicity)}
        For each $\theta \in \Theta$, the transition kernel $P_\theta$ is positive Harris recurrent with a unique invariant probability $\pi_\theta$.
        Furthermore, for any $a \in (0,1]$, there exists a constant $L_a > 1$, independent of $\theta$, such that for all $n \geq 0$ and $\Lambda \in \mathrm{X}$:
        \begin{equation*}
            \| P_\theta^{n}(\Lambda, \cdot) - \pi_\theta \|_{V^a} \leq L_a(1 - L_a^{-1})^n V^a(\Lambda).
        \end{equation*}

        \item \label{item_assumption_simultaneous_properties_2} \textbf{(Simultaneous Drift Condition)}
        For any $a \in (0, 1]$, there exist constants $\beta_a \in (0, 1)$ and $b_a < \infty$, independent of $\theta$, such that for all $\Lambda \in \mathrm{X}$:
        \begin{equation*}
            P_\theta V^a(\Lambda) \leq \beta_a V^a(\Lambda) + b_a.
        \end{equation*}

        \item \label{item_assumption_simultaneous_properties_3} \textbf{(Robust Lipschitz Continuity of the Kernel)}
        The family of transition kernels $\{P_\theta\}_{\theta \in \Theta}$ is robustly Lipschitz continuous with a constant $L_P \geq 0$ that holds uniformly for all $\theta \in \Theta$.

        \item \label{item_assumption_simultaneous_properties_4} \textbf{(Regularity of the Lyapunov Function)}
        There exists a finite constant $L_{\ln V} > 1$ such that
        \begin{equation*}
            \sup_{d(\Lambda, \Lambda^\prime) \leq 1} \left| \ln V(\Lambda) - \ln V(\Lambda^\prime) \right|
            \leq \ln L_{\ln V} .
        \end{equation*}
    \end{enumerate}
    When $a = 1$, we suppress the subscripts and write $L = L_1$, $\beta = \beta_1$, and $b = b_1$.
\end{assumption}

Under Item \ref{item_assumption_simultaneous_properties_1} and \ref{item_assumption_simultaneous_properties_2} of Assumption \ref{assumption_simultaneous_properties}, it follows from the recursion
\begin{equation*}
    \mathbb{E} \left[ V^a(\Lambda_{n+1}) \right]
    \leq \beta_a \mathbb{E} \left[ V^a(\Lambda_n) \right] + b_a
\end{equation*}
for any $a \in (0,1]$ that
\begin{equation*}
    \pi_\theta V^a
    \leq \frac{b_a}{1-\beta_a}
    \quad \text{and} \quad
    \mathbb{E} \left[ V^a(\Lambda_n) \right] \leq \max\left\{ \mathbb{E} \left[ V^a(\Lambda_0) \right], \frac{b_a}{1-\beta_a} \right\} .
\end{equation*}
Also Note a useful inequality that $\frac{b_a}{1-\beta_a} \geq 1$ by $1 \leq \pi_\theta V^a \leq \frac{b_a}{1-\beta_a}$.

\subsubsection{Theorems}

\begin{theorem}
    \label{theorem_V_invariant_probability_Holder}

    Suppose that Assumption \ref{assumption_simultaneous_properties} holds.
    Assume that for some $\kappa \in (0,1)$ and $\gamma \in (0,1)$, the family of functions $\{h_\theta\}_{\theta \in \Theta}$ is $((\Delta, L_h V^\kappa, \alpha), (\tilde{\Delta}, \tilde{L}_h V^{\tilde{\kappa}}, \tilde{\alpha}))$-joint locally Hölder continuous, and bounded by $C_{h,\gamma} V^\gamma$.

    If $d(\theta, \theta^\prime) < \min\{\tilde{\Delta}, 1\}$, then for any chosen tuning parameter $\omega \in (0,\frac{1}{2})$, the following bound holds:
    \begin{equation*}
        \left| \pi_\theta h_\theta - \pi_{\theta^\prime} h_{\theta^\prime} \right|
        \leq L_{\pi h, \alpha_{\pi h, \omega}}
        \left[ d(\theta,\theta^\prime) \right]^{\alpha_{\pi h, \omega}} ,
    \end{equation*}
    where the Hölder exponent $\alpha_{\pi h, \omega}$ is defined as
    \begin{equation*}
        \alpha_{\pi h, \omega} = (1-2\omega) \min\{1-\gamma,1-\kappa,\alpha,\frac{\tilde{\alpha}}{1-2\omega}\} ,
    \end{equation*}
    and the constant $L_{\pi h, \omega}$ is finite.
    It is composed of the underlying parameters in Assumption \ref{assumption_simultaneous_properties}, with the basic quantity $C_\nu = b/(1-\beta)$ included, and is given by the expression:
    \begin{align*}
        & \quad L_{\pi h, \alpha_{\pi h, \omega}} \\
        &= L_{\pi h}(\Delta, \kappa, \tilde{\kappa}, \alpha, \tilde{\alpha}, L_h, \tilde{L}_h, \gamma, C_{h,\gamma}; \omega) \\
        &= \left\{
            2C_{h,\gamma} C_\nu^\gamma L_\gamma
            + 2 C_{h,\gamma} \left\{ \frac{C_\nu}{1-\gamma} + 1 \right\}
            L_P^{1-\gamma}\left(1+\frac{L_P}{2}+\frac{1}{\Delta}\right)^{1-\gamma}
            + L_h C_\nu^\kappa (\Delta+1)
            L_P^{\min \left\{ 1-\kappa,\alpha \right\}}
        \right\} \\
        & \quad \left\{ (1-L_\gamma^{-1})^{-1}
        + 2 + 2 \left[ \omega^2 e^2 \ln^2(1-L_\gamma^{-1}) \right]^{-1} \right\}
        + \frac{\tilde{L}_h b_{\tilde{\kappa}}}{1-\beta_{\tilde{\kappa}}} .
    \end{align*}
\end{theorem}
\begin{proof}
    We decompose the difference as
    \begin{equation*}
        \left| \pi_\theta h_\theta - \pi_{\theta^\prime} h_{\theta^\prime} \right|
        \leq \left| \pi_\theta h_\theta - \pi_{\theta^\prime} h_\theta \right|
        + \left| \pi_{\theta^\prime} h_{\theta^\prime} - \pi_{\theta^\prime} h_\theta \right| .
    \end{equation*}

    Choose $\Lambda$ with $V(\Lambda) \leq \frac{b}{1-\beta}$ and set $a=1$ in Corollary \ref{corollary_n_step_robust_lipschitz_continuous_bounded}.
    Combining Corollary \ref{corollary_n_step_robust_lipschitz_continuous_bounded} with Item \ref{item_assumption_simultaneous_properties_1} of Assumption \ref{assumption_simultaneous_properties}, we obtain that for any $n \in \mathbb{N}$,
    \begin{align*}
        & \quad \left| \pi_\theta h_\theta - \pi_{\theta^\prime} h_\theta \right| \\
        & \leq \left| \pi_\theta h_\theta - P_\theta^n(\Lambda,h_\theta) \right|
        + \left| P_\theta^n(\Lambda,h_\theta) - P_{\theta^\prime}^n(\Lambda,h_\theta) \right|
        + \left| \pi_{\theta^\prime} h_\theta - P_{\theta^\prime}^n(\Lambda,h_\theta) \right| \\
        & \leq 2C_{h,\gamma}V^\gamma(\Lambda)L_\gamma(1-L_\gamma^{-1})^n
        + 2 C_{h,\gamma}
        \left[ \frac{C_\nu}{1-\gamma} + 1 \right]
        \left[ L_P \left( \frac{2n + n(n-1)L_P}{2} d(\theta, \theta^\prime) \right)
        + \frac{n L_Pd(\theta, \theta^\prime)}{\Delta} \right]^{1-\gamma} \\
        & \quad + L_h C_{\nu,1}^\kappa (\Delta+1)
        \left[ n L_Pd(\theta, \theta^\prime) \right]^{\min \left\{ 1-\kappa,\alpha \right\}} \\
        & \leq \left\{
            2C_{h,\gamma}V^\gamma(\Lambda)L_\gamma
            + 2 C_{h,\gamma} \left\{ \frac{C_\nu}{1-\gamma} + 1 \right\}
            L_P^{1-\gamma}\left(1+\frac{L_P}{2}+\frac{1}{\Delta}\right)^{1-\gamma}
            + L_h C_\nu^\kappa (\Delta+1)
            L_P^{\min \left\{ 1-\kappa,\alpha \right\}}
        \right\} \\
        & \quad \left[
            (1-L_\gamma^{-1})^n
            + \left[ n^2d(\theta,\theta^\prime) \right]^{1-\gamma}
            + \left[ nd(\theta,\theta^\prime) \right]^{\min \left\{ 1-\kappa,\alpha \right\}}
        \right] .
    \end{align*}

    Now set $n=\lfloor \frac{\ln\epsilon}{\ln(1-L_\gamma^{-1})} \rfloor$, where $\epsilon = d(\theta,\theta^\prime) < \min\{\tilde{\Delta},1\} \leq 1$.
    Using the inequality $\ln\frac{1}{\epsilon} \leq \frac{1}{\omega e \epsilon^\omega}$ for $\omega > 0$, we deduce that for any $\omega \in (0, \frac{1}{2})$,
    \begin{align*}
        & \quad (1-L_\gamma^{-1})^n
        + \left[ n^2d(\theta,\theta^\prime) \right]^{1-\gamma}
        + \left[ nd(\theta,\theta^\prime) \right]^{\min \left\{ 1-\kappa,\alpha \right\}} \\
        & \leq (1-L_\gamma^{-1})^{-1}\epsilon
        + \left( \frac{\epsilon \ln^2 \epsilon}{\ln^2 (1-L_\gamma^{-1})} \right)^{1-\gamma}
        + \left( \frac{\epsilon \ln^2 \epsilon}{\ln^2 (1-L_\gamma^{-1})} \right)^{\min \left\{ 1-\kappa,\alpha \right\}} \\
        & \leq \left\{ (1-L_\gamma^{-1})^{-1}
        + 2 + 2 \left[ \omega^2 e^2 \ln^2(1-L_\gamma^{-1}) \right]^{-1} \right\}
        \epsilon^{(1-2\omega) \min\{1-\gamma,1-\kappa,\alpha\}} .
    \end{align*}

    On the other hand, when $d(\theta,\theta^\prime)<\min\{\tilde{\Delta},1\}$,
    \begin{equation*}
        \left| \pi_{\theta^\prime} h_{\theta^\prime} - \pi_{\theta^\prime} h_\theta \right|
        \leq \pi_\theta \left\{ \tilde{L}_h V^{\tilde{\kappa}}[d(\theta,\theta^\prime)]^{\tilde{\alpha}} \right\}
        \leq \frac{\tilde{L}_h b_{\tilde{\kappa}}}{1-\beta_{\tilde{\kappa}}} [d(\theta,\theta^\prime)]^{\tilde{\alpha}} .
    \end{equation*}
    Combining the previous bounds yields the asserted inequality and completes the proof.
\end{proof}

\begin{theorem}
    \label{theorem_V_n_step_robust_lipschitz_continuous_bounded}

    Suppose that Assumption \ref{assumption_simultaneous_properties} holds.
    Assume that for some $\kappa \in (0,1)$ and $\gamma \in (0,1)$, the family of functions $\{h_\theta\}_{\theta \in \Theta}$ is $((\Delta, L_h V^\kappa, \alpha), (\tilde{\Delta}, \tilde{L}_h V^{\tilde{\kappa}}, \tilde{\alpha}))$-joint locally Hölder continuous, and bounded by $C_{h,\gamma} V^\gamma$.

    Then when $n \in \mathbb{N}^*$, for any chosen tuning parameter $a \in (\gamma, 1]$, the family of functions $\{P_\theta^n h_\theta\}_{\theta \in \Theta}$ is bounded by $C_{P^n h,\gamma} V^\gamma(\Lambda)$ and is
    \begin{equation*}
        ((1, L_{P^n h, \kappa_{P^n h, a}, \alpha_{P^n h, a}} V^{\kappa_{P^n h, a}}, \alpha_{P^n h, a}),
        (\min\{\tilde{\Delta},1\}, L_{P^n h, \kappa_{P^n h, a}, \alpha_{P^n h, a}} V^{\kappa_{P^n h, a}}, \alpha_{P^n h, a}))
    \end{equation*}
    -joint locally Hölder continuous, where
    \begin{equation*}
        C_{P^n h,\gamma} = \frac{C_{h,\gamma} b_\gamma}{1-\beta_\gamma} ,
        \quad \kappa_{P^n h, a} = \max\{\kappa,\tilde{\kappa},a\} ,
        \quad \alpha_{P^n h, a} = \min\{1-\gamma/a,1-\kappa,\alpha,\tilde{\alpha}\}
    \end{equation*}
    and
    \begin{align*}
        & \quad L_{P^n h, \kappa_{P^n h, a}, \alpha_{P^n h, a}} \\
        &= L_{P^n h}(\Delta, \kappa, \tilde{\kappa}, \alpha, \tilde{\alpha}, L_h, \tilde{L}_h, \gamma, C_{h,\gamma}; a) \\
        &= 2 n^2 L_{\ln V} \left[ \frac{b_a}{1-\beta_a}
        + \frac{b_{\tilde{\kappa}}}{1-\beta_{\tilde{\kappa}}}
        + \left[\frac{b}{1-\beta} \right]^\kappa \right] \\
        & \quad \left\{
            \frac{4 C_{h,\gamma}}{1-\frac{\gamma}{a}}
            \left[ L_P + L_P^2 + \frac{L_P+1}{\Delta} \right]^{1-\gamma/a}
            + L_h (\Delta+1)(L_P+1)
            + \tilde{L}_h
        \right\} .
    \end{align*}
\end{theorem}
\begin{proof}
    Set $C_{\nu,a} = \max\{V^a(\Lambda),V^a(\Lambda^\prime),\frac{b_a}{1-\beta_a}\}$ for $a \in (0,1]$.
    We first bound the magnitude of $P_\theta^n(\Lambda,h_\theta)$:
    \begin{equation*}
        \left| P_\theta^n(\Lambda,h_\theta) \right|
        \leq C_{h,\gamma} \left| P_\theta^n(\Lambda,V^\gamma) \right|
        \leq C_{h,\gamma} \max \left\{ V^\gamma(\Lambda),\frac{b_\gamma}{1-\beta_\gamma} \right\}
        \leq \frac{C_{h,\gamma} b_\gamma}{1-\beta_\gamma} V^\gamma(\Lambda) ,
    \end{equation*}
    which yields the asserted control constant $C_{P^n h,\gamma}$.
    Now suppose $d(\theta,\theta^\prime)<\min\{\tilde{\Delta},1\}$ and set $a=1$ in Corollary \ref{corollary_n_step_robust_lipschitz_continuous_bounded}.
    Then
    \begin{align*}
        & \quad \left| P_\theta^n(\Lambda,h_\theta) - P_{\theta^\prime}^n(\Lambda^\prime,h_{\theta^\prime}) \right|
        \leq \left| P_\theta^n(\Lambda,h_\theta) - P_{\theta^\prime}^n(\Lambda^\prime,h_\theta) \right|
        + \left| P_{\theta^\prime}^n(\Lambda^\prime,h_\theta) - P_{\theta^\prime}^n(\Lambda^\prime,h_{\theta^\prime}) \right| \\
        & \leq 2 C_{h,\gamma}
        \left[ \frac{C_{\nu,a}}{1-\frac{\gamma}{a}} + 1 \right]
        \left[ L_P \left( nd(\Lambda,\Lambda^\prime) + \frac{2n + n(n-1)L_P}{2} d(\theta, \theta^\prime) \right)
        + \frac{d(\Lambda,\Lambda^\prime) + n L_Pd(\theta, \theta^\prime)}{\Delta} \right]^{1-\gamma/a} \\
        & \quad + L_h C_{\nu,1}^\kappa (\Delta+1)
        \left[ d(\Lambda,\Lambda^\prime) + n L_Pd(\theta, \theta^\prime) \right]^{\min \left\{ 1-\kappa,\alpha \right\}}
        + \tilde{L}_h C_{\nu,\tilde{\kappa}} [d(\theta,\theta^\prime)]^{\tilde{\alpha}} .
    \end{align*}

    The preceding expression is bounded by the product of
    \begin{equation*}
        \frac{4 C_{h,\gamma}}{1-\frac{\gamma}{a}}
        \left[ L_P + L_P^2 + \frac{L_P+1}{\Delta} \right]^{1-\gamma/a}
        + L_h (\Delta+1)(L_P+1)
        + \tilde{L}_h
    \end{equation*}
    and
    \begin{equation*}
        C_{\nu,a}+C_{\nu,1}^\kappa+C_{\nu,\tilde{\kappa}}
        \leq \left[ \frac{b_a}{1-\beta_a}
        + \frac{b_{\tilde{\kappa}}}{1-\beta_{\tilde{\kappa}}}
        + \left[\frac{b}{1-\beta} \right]^\kappa \right]
        \max\{V(\Lambda),V(\Lambda^\prime)\}^{\max\{\kappa,\tilde{\kappa},a\}}
    \end{equation*}
    together with
    \begin{equation*}
        n^2 \left[ d(\theta,\theta^\prime) \right]^{\min\{1-\gamma/a,1-\kappa,\alpha,\tilde{\alpha}\}}
        + n \left[ d(\Lambda,\Lambda^\prime) \right]^{\min\{1-\gamma/a,1-\kappa,\alpha,\tilde{\alpha}\}} ,
    \end{equation*}
    when $d(\theta,\theta^\prime)<\min\{\tilde{\Delta},1\}$ and $d(\Lambda,\Lambda^\prime) < 1$.

    Denote
    \begin{align*}
        C
        &= \left[ \frac{b_a}{1-\beta_a}
        + \frac{b_{\tilde{\kappa}}}{1-\beta_{\tilde{\kappa}}}
        + \left[\frac{b}{1-\beta} \right]^\kappa \right] \\
        & \quad \left\{
            \frac{4 C_{h,\gamma}}{1-\frac{\gamma}{a}}
            \left[ L_P + L_P^2 + \frac{L_P+1}{\Delta} \right]^{1-\gamma/a}
            + L_h (\Delta+1)(L_P+1)
            + \tilde{L}_h
        \right\}
    \end{align*}

    Hence, with $\alpha_{P^n h, a} = \min\{1-\gamma/a,1-\kappa,\alpha,\tilde{\alpha}\}$,
    \begin{align*}
        & \quad \left| P_\theta^n(\Lambda,h_\theta) - P_{\theta^\prime}^n(\Lambda^\prime,h_{\theta^\prime}) \right| \\
        & \leq
        n^2 C \max\{V(\Lambda),V(\Lambda^\prime)\}^{\max\{\kappa,\tilde{\kappa},a\}}
        \left[ [d(\Lambda,\Lambda^\prime)]^{\alpha_{P^n h, a}} + [d(\theta,\theta^\prime)]^{\alpha_{P^n h, a}} \right] \\
        & \leq
        2 n^2 C \max\{V(\Lambda),V(\Lambda^\prime)\}^{\max\{\kappa,\tilde{\kappa},a\}}
        \max \left\{d(\Lambda,\Lambda^\prime),d(\theta,\theta^\prime)\right\}^{\alpha_{P^n h, a}} .
    \end{align*}

    When $d(\Lambda, \Lambda^\prime) \leq 1$, Item \ref{item_assumption_simultaneous_properties_4} of Assumption \ref{assumption_simultaneous_properties} gives
    \begin{equation*}
        \max\{V(\Lambda),V(\Lambda^\prime)\}^{\max\{\kappa,\tilde{\kappa},a\}}
        \leq L_{\ln V} V(\Lambda)^{\max\{\kappa,\tilde{\kappa},a\}}
    \end{equation*}
    Combining the foregoing bounds yields the claimed inequalities and completes the proof.
\end{proof}

When $h_\theta$ is bounded by $C_{h,\gamma} V^\gamma(\Lambda)$ for some $C_{h,\gamma} > 0$ and $\gamma \in (0,1]$, Assumption \ref{assumption_simultaneous_properties} guarantees the geometric ergodicity of $P_\theta$.
In this case we define
\begin{equation*}
    \hat{h}_\theta = \sum_{n \geq 0} (P_\theta^n - \pi_\theta) h_\theta ,
\end{equation*}
which is well-defined under this assumption.
By construction, the function $\hat{h}_\theta$ solves the Poisson equation
\begin{equation*}
    \hat{h}_\theta - P_\theta \hat{h}_\theta
    = h_\theta - \pi_\theta h_\theta .
\end{equation*}
The properties of $\hat{h}_\theta$ are summarized in the following theorem.
\begin{theorem}
    \label{theorem_V_Poisson_solution_robust_lipschitz_continuous_bounded}

    Suppose that Assumption \ref{assumption_simultaneous_properties} holds.
    Assume that for some $\kappa \in (0,1)$ and $\gamma \in (0,1)$, the family of functions $\{h_\theta\}_{\theta \in \Theta}$ is $((\Delta, L_h V^\kappa, \alpha), (\tilde{\Delta}, \tilde{L}_h V^{\tilde{\kappa}}, \tilde{\alpha}))$-joint locally Hölder continuous, and bounded by $C_{h,\gamma} V^\gamma$.

    Then for any chosen tuning parameter $a \in (\gamma, 1]$, $\omega \in (0,\frac{1}{2})$ and $\tilde{\omega} \in (0,\frac{(1-2\omega)\min\{1-\gamma/a,1-\kappa,\alpha,\tilde{\alpha}\}}{3})$, the family of functions $\{\hat{h}_\theta\}_{\theta \in \Theta}$ is bounded by $C_{\hat{h},\gamma} V^\gamma(\Lambda)$ and is
    \begin{equation*}
        ((1, L_{\hat{h}, \kappa_{\hat{h}, a}, \alpha_{\hat{h}, a, \omega, \tilde{\omega}}} V^{\kappa_{\hat{h}, a}}, \alpha_{\hat{h}, a, \omega, \tilde{\omega}}),
        (\min\{\tilde{\Delta},1\}, L_{\hat{h}, \kappa_{\hat{h}, a}, \alpha_{\hat{h}, a, \omega, \tilde{\omega}}} V^{\kappa_{\hat{h}, a}}, \alpha_{\hat{h}, a, \omega, \tilde{\omega}}))
    \end{equation*}
    -joint locally Hölder continuous, where
    \begin{align*}
        C_{\hat{h},\gamma} &= C_{h,\gamma} L_\gamma^2 ,
        \quad
        \kappa_{\hat{h}, a} = \max\{\kappa,\tilde{\kappa},a,\gamma\} , \\
        \quad \alpha_{\hat{h}, a, \omega, \tilde{\omega}} &= (1-2\omega)\min\{1-\gamma/a,1-\kappa,\alpha,\tilde{\alpha}\}-3\tilde{\omega}
    \end{align*}
    and
    \begin{align*}
        & \quad L_{\hat{h}, \kappa_{\hat{h}, a}, \alpha_{\hat{h}, a, \omega, \tilde{\omega}}} \\
        &= L_{\hat{h}}(\Delta, \kappa, \tilde{\kappa}, \alpha, \tilde{\alpha}, L_h, \tilde{L}_h, \gamma, C_{h,\gamma}; a, \omega, \tilde{\omega}) \\
        &= L_h + \tilde{L}_h + L_{\pi h, \alpha_{\pi h, \omega}}
        + C_{h,\gamma} L_\gamma^2 (L_{\ln V}+1)
        + 2 \left[ L_{\pi h, \alpha_{\pi h, \omega}} + L_{\mathrm{step}, a} \right]
        \left[ \frac{1}{-\tilde{\omega} e \ln(1-L_\gamma^{-1})} \right]^3 .
    \end{align*}
    Here, the constant is further composed of $L_{\pi h, \alpha_{\pi h, \omega}}$ in Theorem \ref{theorem_V_invariant_probability_Holder} and $L_{\mathrm{step}, a}$, where the latter corresponds to $L_{P^n h, \kappa_{P^n h, a}, \alpha_{P^n h, a}}$ in Theorem \ref{theorem_V_n_step_robust_lipschitz_continuous_bounded}, but without the multiplicative factor $n^2$.
\end{theorem}
\begin{proof}
    Consider the difference
    \begin{equation*}
        \hat{h}_\theta(\Lambda) - \hat{h}_{\theta^\prime}(\Lambda^\prime)
        = \sum_{n=0}^\infty \left[ (P_\theta^n h_\theta)(\Lambda)
        - (P_{\theta^\prime}^n h_{\theta^\prime})(\Lambda^\prime)
        - \pi_\theta h_\theta + \pi_{\theta^\prime} h_{\theta^\prime} \right] .
    \end{equation*}
    When $\omega \in (0,\frac{1}{2})$, $a \in (\gamma, 1]$, $d(\theta,\theta^\prime)<\min\{\tilde{\Delta},1\}$ and $d(\Lambda, \Lambda^\prime)<1$, each term in the sum,
    \begin{equation*}
        \left[ (P_\theta^n h_\theta)(\Lambda) - \pi_\theta h_\theta \right]
        - \left[ (P_{\theta^\prime}^n h_{\theta^\prime})(\Lambda^\prime) - \pi_{\theta^\prime} h_{\theta^\prime} \right]
    \end{equation*}
    can be bounded as follows.
    \begin{enumerate}
        \item Exponential tail term (by Item \ref{item_assumption_simultaneous_properties_1} of Assumption \ref{assumption_simultaneous_properties}):
        \begin{align*}
            & \quad \left| \left[ (P_\theta^n h_\theta)(\Lambda) - \pi_\theta h_\theta \right]
            - \left[ (P_{\theta^\prime}^n h_{\theta^\prime})(\Lambda^\prime) - \pi_{\theta^\prime} h_{\theta^\prime} \right] \right| \\
            & \leq \left| (P_\theta^n h_\theta)(\Lambda) - \pi_\theta h_\theta \right|
            + \left| (P_{\theta^\prime}^n h_{\theta^\prime})(\Lambda^\prime) - \pi_{\theta^\prime} h_{\theta^\prime} \right| \\
            & \leq C_{h,\gamma}
            \left[V^\gamma(\Lambda)+V^\gamma(\Lambda^\prime)\right]
            L_\gamma(1-L_\gamma^{-1})^n
            \leq C_{h,\gamma} L_\gamma (L_{\ln V}+1) V^\gamma(\Lambda) (1-L_\gamma^{-1})^n .
        \end{align*}
        \item Difference of invariant expectations (by Theorem \ref{theorem_V_invariant_probability_Holder}):
        \begin{equation*}
            \left| \pi_\theta h_\theta - \pi_{\theta^\prime} h_{\theta^\prime} \right|
            \leq L_{\pi h, \alpha_{\pi h, \omega}}
            \left[ d(\theta,\theta^\prime) \right]^{\alpha_{\pi h, \omega}} .
        \end{equation*}
        \item Difference of $n$-step kernels (by Theorem \ref{theorem_V_n_step_robust_lipschitz_continuous_bounded} for $n \in \mathbb{N}^*$):
        \begin{equation*}
            \left| (P_\theta^n h_\theta)(\Lambda) - (P_{\theta^\prime}^n h_{\theta^\prime})(\Lambda^\prime) \right|
            \leq L_{P^n h, \kappa_{P^n h, a}, \alpha_{P^n h, a}} V^{\kappa_{P^n h, a}}(\Lambda)
            [d(\theta, \theta^\prime)]^{\alpha_{P^n h, a}} .
        \end{equation*}
    \end{enumerate}

    Denote
    \begin{equation*}
        q_{\mathrm{min}} = (1-2\omega)\min\{1-\gamma/a,1-\kappa,\alpha,\tilde{\alpha}\} \leq \min\left\{ \alpha_{\pi h, \omega}, \alpha_{P^n h, a} \right\} ,
    \end{equation*}
    and set $L_{\mathrm{step}, a} = \frac{L_{P^n h, \kappa_{P^n h, a}, \alpha_{P^n h, a}}}{n^2}$, which is a constant independent of $n$.
    Then for any $n \in \mathbb{N}^*$,
    \begin{align*}
        & \quad \left| \left[ (P_\theta^n h_\theta)(\Lambda) - \pi_\theta h_\theta \right]
        - \left[ (P_{\theta^\prime}^n h_{\theta^\prime})(\Lambda^\prime) - \pi_{\theta^\prime} h_{\theta^\prime} \right] \right| \\
        & \leq \left[ L_{\pi h, \alpha_{\pi h, \omega}}
        + L_{P^n h, \kappa_{P^n h, a}, \alpha_{P^n h, a}} \right] V^{\kappa_{P^n h, a}}(\Lambda)
        \left[ [d(\Lambda,\Lambda^\prime)]^{q_{\mathrm{min}}} + [d(\theta,\theta^\prime)]^{q_{\mathrm{min}}} \right] \\
        & \leq 2 n^2 \left[ L_{\pi h, \alpha_{\pi h, \omega}}
        + L_{\mathrm{step}, a} \right] V^{\kappa_{P^n h, a}}(\Lambda)
        \max \left\{d(\Lambda,\Lambda^\prime),d(\theta,\theta^\prime)\right\}^{q_{\mathrm{min}}} .
    \end{align*}

    Let $N=\lfloor \frac{\ln\epsilon}{\ln(1-L_\gamma^{-1})} \rfloor$, where $\epsilon = \max \left\{d(\Lambda,\Lambda^\prime),d(\theta,\theta^\prime)\right\} < 1$.
    Using the technical inequality $\ln\frac{1}{\epsilon} \leq \frac{1}{\tilde{\omega} e \epsilon^{\tilde{\omega}}}$ for $\tilde{\omega} > 0$, we split the summation over $n$ into two parts: $n \leq N$ and $n > N$.
    For the tail sum, observe that
    \begin{equation*}
        \sum_{n>N} (1-L_\gamma^{-1})^n = L_\gamma(1-L_\gamma^{-1})^{N+1} \leq L_\gamma \epsilon .
    \end{equation*}
    For the finite part, we have
    \begin{align*}
        \sum_{0 < n \leq N} n^2 \epsilon^{q_{\mathrm{min}}}
        \leq N^3 \epsilon^{q_{\mathrm{min}}}
        < \left[ \frac{\ln\epsilon}{\ln(1-L_\gamma^{-1})} \right]^3 \epsilon^{q_{\mathrm{min}}}
        < \left[ \frac{1}{-\tilde{\omega} e \ln(1-L_\gamma^{-1})} \right]^3 \epsilon^{q_{\mathrm{min}}-3\tilde{\omega}} .
    \end{align*}

    Combining the two parts with $\left| (P_\theta^0 h_\theta)(\Lambda) - (P_{\theta^\prime}^0 h_{\theta^\prime})(\Lambda^\prime) \right| = d(\Lambda, \Lambda^\prime) \leq \epsilon < 1$, the entire summation
    \begin{equation*}
        \hat{h}_\theta(\Lambda) - \hat{h}_{\theta^\prime}(\Lambda^\prime) = \sum_{n \geq 0} \left\{ \left[ (P_\theta^n h_\theta)(\Lambda) - \pi_\theta h_\theta \right]
        - \left[ (P_{\theta^\prime}^n h_{\theta^\prime})(\Lambda^\prime) - \pi_{\theta^\prime} h_{\theta^\prime} \right] \right\}
    \end{equation*}
    is bounded by
    \begin{align*}
        & \quad \left\{ L_h V^\kappa(\Lambda) \left[ d(\Lambda, \Lambda^\prime) \right]^\alpha
        + \tilde{L}_h V^{\tilde{\kappa}}(\Lambda) \left[ d(\theta, \theta^\prime) \right]^{\tilde{\alpha}}+ L_{\pi h, \alpha_{\pi h, \omega}} \left[ d(\theta,\theta^\prime) \right]^{\alpha_{\pi h, \omega}} \right\} \\
        & \quad + C_{h,\gamma} L_\gamma (L_{\ln V}+1) V^\gamma(\Lambda)
        \sum_{n>N} (1-L_\gamma^{-1})^n \\
        & \quad + 2 \left[ L_{\pi h, \alpha_{\pi h, \omega}} + L_{\mathrm{step}, a} \right]
        V^{\kappa_{P^n h, a}}(\Lambda)
        \sum_{0 < n \leq N} n^2 \epsilon^{q_{\mathrm{min}}} \\
        & \leq \left\{ L_h V^\kappa(\Lambda) \epsilon^\alpha
        + \tilde{L}_h V^{\tilde{\kappa}}(\Lambda) \epsilon^{\tilde{\alpha}}+ L_{\pi h, \alpha_{\pi h, \omega}} \epsilon^{\alpha_{\pi h, \omega}} \right\} \\
        & \quad + C_{h,\gamma} L_\gamma^2 (L_{\ln V}+1) V^\gamma(\Lambda) \epsilon \\
        & \quad + 2 \left[ L_{\pi h, \alpha_{\pi h, \omega}} + L_{\mathrm{step}, a} \right]
        \left[ \frac{1}{-\tilde{\omega} e \ln(1-L_\gamma^{-1})} \right]^3
        V^{\kappa_{P^n h, a}}(\Lambda)
        \epsilon^{q_{\mathrm{min}}-3\tilde{\omega}} \\
        & \leq \left\{ L_h + \tilde{L}_h + L_{\pi h, \alpha_{\pi h, \omega}}
        + C_{h,\gamma} L_\gamma^2 (L_{\ln V}+1)
        + 2 \left[ L_{\pi h, \alpha_{\pi h, \omega}} + L_{\mathrm{step}, a} \right]
        \left[ \frac{1}{-\tilde{\omega} e \ln(1-L_\gamma^{-1})} \right]^3 \right\} \\
        & \quad V^{\max\{\kappa_{P^n h, a},\gamma\}}(\Lambda)
        \epsilon^{q_{\mathrm{min}}-3\tilde{\omega}} .
    \end{align*}
\end{proof}

\subsubsection{Corollaries}

\begin{corollary}
    \label{corollary_n_step_h_robust_lipschitz_continuous_bounded}

    Suppose that Assumption \ref{assumption_simultaneous_properties} holds.
    Assume that for any $\alpha,\tilde{\alpha},\kappa,\tilde{\kappa} \in (0,1)$ and $\gamma \in (0,1]$, the family of functions $\{h_\theta\}_{\theta \in \Theta}$ is $((\Delta_{h, \kappa, \alpha}, L_{h, \kappa, \alpha} V^\kappa, \alpha), (\tilde{\Delta}_{h, \tilde{\kappa}, \tilde{\alpha}}, \tilde{L}_{h, \tilde{\kappa}, \tilde{\alpha}} V^{\tilde{\kappa}}, \tilde{\alpha}))$-joint locally Hölder continuous, and bounded by $C_{h,\gamma} V^\gamma$.
    
    Then for any $n \in \mathbb{N}^*$, any $\alpha,\tilde{\alpha},\kappa,\tilde{\kappa} \in (0,1)$ and $\gamma \in (0,1]$, the family of functions $\{P_\theta^n h_\theta\}_{\theta \in \Theta}$ is also
    \begin{equation*}
        ((\Delta_{P^n h, \kappa, \alpha}, L_{P^n h, \kappa, \alpha} V^\kappa, \alpha),
        (\tilde{\Delta}_{P^n h, \tilde{\kappa}, \tilde{\alpha}}, \tilde{L}_{P^n h, \tilde{\kappa}, \tilde{\alpha}} V^{\tilde{\kappa}}, \tilde{\alpha}))
    \end{equation*}
    -joint locally Hölder continuous, and bounded by $C_{P^n h,\gamma} V^\gamma(\Lambda)$ for some constant $\Delta_{P^n h, \kappa, \alpha}$, $\tilde{\Delta}_{P^n h, \tilde{\kappa}, \tilde{\alpha}}$, $L_{P^n h, \kappa, \alpha}$, $\tilde{L}_{P^n h, \tilde{\kappa}, \tilde{\alpha}}$ and $C_{P^n h,\gamma} > 0$.
\end{corollary}
\begin{proof}
    By Theorem \ref{theorem_V_n_step_robust_lipschitz_continuous_bounded}, the family $\{P_\theta^n h_\theta\}_{\theta \in \Theta}$ is bounded by $C_{P^n h,\gamma} V^\gamma(\Lambda)$ for any $\gamma \in (0,1]$.
    
    Moreover, to align the exponents with arbitrary prescribed $\alpha,\tilde{\alpha},\kappa,\tilde{\kappa} \in (0,1)$, define
    \begin{align*}
        \tilde{\alpha}^\prime, \alpha^\prime &= \max\{\alpha, \tilde{\alpha}\} , \\
        \tilde{\kappa}^\prime, \kappa^\prime &= \min\{\kappa, \tilde{\kappa}, 1-\max\{\alpha, \tilde{\alpha}\}\} , \\
        a &= \min\{\kappa, \tilde{\kappa}\} , \\
        \gamma^\prime &= \min\{\kappa, \tilde{\kappa}\} [1-\max\{\alpha, \tilde{\alpha}\}] .
    \end{align*}
    Since $\{h_\theta\}_{\theta \in \Theta}$ can be $((\Delta_{h, \kappa^\prime, \alpha^\prime}, L_{h, \kappa^\prime, \alpha^\prime} V^{\kappa^\prime}, \alpha^\prime), (\tilde{\Delta}_{h, \tilde{\kappa}^\prime, \tilde{\alpha}^\prime}, \tilde{L}_{h, \tilde{\kappa}^\prime, \tilde{\alpha}^\prime} V^{\tilde{\kappa}^\prime}, \tilde{\alpha}^\prime))$-joint locally Hölder continuous, the above choice is admissible.

    From this construction and Theorem \ref{theorem_V_n_step_robust_lipschitz_continuous_bounded}, we obtain that $\{P_\theta^n h_\theta\}_{\theta \in \Theta}$ is
    \begin{equation*}
        ((1, L_{P^n h, \kappa_{P^n h, a}, \alpha_{P^n h, a}} V^{\kappa_{P^n h, a}}, \alpha_{P^n h, a}),
        (\min\{\tilde{\Delta}_{h, \tilde{\kappa}, \tilde{\alpha}},1\}, L_{P^n h, \kappa_{P^n h, a}, \alpha_{P^n h, a}} V^{\kappa_{P^n h, a}}, \alpha_{P^n h, a}))
    \end{equation*}
    -joint locally Hölder continuous, where
    \begin{equation*}
        \kappa_{P^n h, a} = \max\{\kappa^\prime,\tilde{\kappa}^\prime,a\}
        \leq \min\{\kappa, \tilde{\kappa}\} ,
        \quad \alpha_{P^n h, a} = \min\{1-\gamma^\prime/a,1-\kappa^\prime,\alpha^\prime,\tilde{\alpha}^\prime\}
        \geq \max\{\alpha, \tilde{\alpha}\} .
    \end{equation*}
    Thus, $\{P_\theta^n h_\theta\}_{\theta \in \Theta}$ is 
    \begin{equation*}
        ((1, L_{P^n h, \kappa_{P^n h, a}, \alpha_{P^n h, a}} V^\kappa, \alpha),
        (\min\{\tilde{\Delta}_{h, \tilde{\kappa}, \tilde{\alpha}},1\}, L_{P^n h, \kappa_{P^n h, a}, \alpha_{P^n h, a}} V^{\tilde{\kappa}}, \tilde{\alpha}))
    \end{equation*}
    -joint locally Hölder continuous.
\end{proof}

\begin{corollary}
    \label{corollary_V_Poisson_solution_robust_lipschitz_continuous_bounded}

    Suppose that Assumption \ref{assumption_simultaneous_properties} holds.
    Assume that for any $\alpha,\tilde{\alpha},\kappa,\tilde{\kappa} \in (0,1)$ and $\gamma \in (0,1]$, the family of functions $\{h_\theta\}_{\theta \in \Theta}$ is $((\Delta_{h, \kappa, \alpha}, L_{h, \kappa, \alpha} V^\kappa, \alpha), (\tilde{\Delta}_{h, \tilde{\kappa}, \tilde{\alpha}}, \tilde{L}_{h, \tilde{\kappa}, \tilde{\alpha}} V^{\tilde{\kappa}}, \tilde{\alpha}))$-joint locally Hölder continuous, and bounded by $C_{h,\gamma} V^\gamma$.
    
    Then for any $\alpha,\tilde{\alpha},\kappa,\tilde{\kappa} \in (0,1)$ and $\gamma \in (0,1]$, the family of functions $\{\hat{h}_\theta\}_{\theta \in \Theta}$ is also
    \begin{equation*}
        ((\Delta_{\hat{h}, \kappa, \alpha}, L_{\hat{h}, \kappa, \alpha} V^\kappa, \alpha),
        (\tilde{\Delta}_{\hat{h}, \tilde{\kappa}, \tilde{\alpha}}, \tilde{L}_{\hat{h}, \tilde{\kappa}, \tilde{\alpha}} V^{\tilde{\kappa}}, \tilde{\alpha}))
    \end{equation*}
    -joint locally Hölder continuous, and bounded by $C_{\hat{h},\gamma} V^\gamma(\Lambda)$ for some constant $\Delta_{\hat{h}, \kappa, \alpha}$, $\tilde{\Delta}_{\hat{h}, \tilde{\kappa}, \tilde{\alpha}}$, $L_{\hat{h}, \kappa, \alpha}$, $\tilde{L}_{\hat{h}, \tilde{\kappa}, \tilde{\alpha}}$ and $C_{\hat{h},\gamma} > 0$.
\end{corollary}
\begin{proof}
    By Theorem \ref{theorem_V_Poisson_solution_robust_lipschitz_continuous_bounded}, the family $\{\hat{h}_\theta\}_{\theta \in \Theta}$ is bounded by $C_{\hat{h},\gamma} V^\gamma(\Lambda)$ for any $\gamma \in (0,1]$.
    
    Moreover, to align the exponents with arbitrary prescribed $\alpha,\tilde{\alpha},\kappa,\tilde{\kappa} \in (0,1)$, define
    \begin{align*}
        \tilde{\omega} &= \frac{1-\max\{\alpha, \tilde{\alpha}\}}{6} , \\
        \omega &= \frac{1-\max\{\alpha, \tilde{\alpha}\}}{8} , \\
        \tilde{\alpha}^\prime, \alpha^\prime &= \frac{\max\{\alpha, \tilde{\alpha}\} + 3\tilde{\omega}}{1-2\omega}
        = \frac{2+2\max\{\alpha, \tilde{\alpha}\}}{3+\max\{\alpha, \tilde{\alpha}\}} , \\
        \tilde{\kappa}^\prime, \kappa^\prime &= \min\{\kappa, \tilde{\kappa}, 1-\frac{\max\{\alpha, \tilde{\alpha}\} + 3\tilde{\omega}}{1-2\omega}\} , \\
        a &= \min\{\kappa, \tilde{\kappa}\} , \\
        \gamma^\prime &= \min\{\kappa, \tilde{\kappa}\} \left[ 1-\frac{\max\{\alpha, \tilde{\alpha}\} + 3\tilde{\omega}}{1-2\omega} \right] .
    \end{align*}
    Since $\{h_\theta\}_{\theta \in \Theta}$ can be $((\Delta_{h, \kappa^\prime, \alpha^\prime}, L_{h, \kappa^\prime, \alpha^\prime} V^{\kappa^\prime}, \alpha^\prime), (\tilde{\Delta}_{h, \tilde{\kappa}^\prime, \tilde{\alpha}^\prime}, \tilde{L}_{h, \tilde{\kappa}^\prime, \tilde{\alpha}^\prime} V^{\tilde{\kappa}^\prime}, \tilde{\alpha}^\prime))$-joint locally Hölder continuous, the above choice is admissible.

    From this construction and Theorem \ref{theorem_V_Poisson_solution_robust_lipschitz_continuous_bounded}, we obtain that $\{\hat{h}_\theta\}_{\theta \in \Theta}$ is
    \begin{equation*}
        ((1, L_{\hat{h}, \kappa_{\hat{h}, a}, \alpha_{\hat{h}, a, \omega, \tilde{\omega}}} V^{\kappa_{\hat{h}, a}}, \alpha_{\hat{h}, a, \omega, \tilde{\omega}}),
        (\min\{\tilde{\Delta}_{h, \tilde{\kappa}, \tilde{\alpha}},1\}, L_{\hat{h}, \kappa_{\hat{h}, a}, \alpha_{\hat{h}, a, \omega, \tilde{\omega}}} V^{\kappa_{\hat{h}, a}}, \alpha_{\hat{h}, a, \omega, \tilde{\omega}}))
    \end{equation*}
    -joint locally Hölder continuous, where
    \begin{align*}
        \kappa_{\hat{h}, a} &= \max\{\kappa^\prime,\tilde{\kappa}^\prime,a,\gamma^\prime\}
        \leq \min\{\kappa, \tilde{\kappa}\} , \\
        \quad \alpha_{\hat{h}, a, \omega, \tilde{\omega}} &= (1-2\omega)\min\{1-\gamma^\prime/a,1-\kappa^\prime,\alpha^\prime,\tilde{\alpha}^\prime\}-3\tilde{\omega}
        \geq \max\{\alpha, \tilde{\alpha}\} .
    \end{align*}
    Thus, $\{\hat{h}_\theta\}_{\theta \in \Theta}$ is 
    \begin{equation*}
        ((1, L_{\hat{h}, \kappa_{\hat{h}, a}, \alpha_{\hat{h}, a, \omega, \tilde{\omega}}} V^\kappa, \alpha),
        (\min\{\tilde{\Delta}_{h, \tilde{\kappa}, \tilde{\alpha}},1\}, L_{\hat{h}, \kappa_{\hat{h}, a}, \alpha_{\hat{h}, a, \omega, \tilde{\omega}}} V^{\tilde{\kappa}}, \tilde{\alpha}))
    \end{equation*}
    -joint locally Hölder continuous.
\end{proof}

\begin{corollary}
    \label{corollary_h_power_robust_lipschitz_continuous_bounded}

    Suppose that Assumption \ref{assumption_simultaneous_properties} holds.
    Assume that for any $\alpha,\tilde{\alpha},\kappa,\tilde{\kappa} \in (0,1)$ and $\gamma \in (0,1]$, the family of functions $\{h_\theta\}_{\theta \in \Theta}$ is $((\Delta_{h, \kappa, \alpha}, L_{h, \kappa, \alpha} V^\kappa, \alpha), (\tilde{\Delta}_{h, \tilde{\kappa}, \tilde{\alpha}}, \tilde{L}_{h, \tilde{\kappa}, \tilde{\alpha}} V^{\tilde{\kappa}}, \tilde{\alpha}))$-joint locally Hölder continuous, and bounded by $C_{h,\gamma} V^\gamma$.
    
    Then for any $\alpha,\tilde{\alpha},\kappa,\tilde{\kappa} \in (0,1)$ and $\gamma \in (0,1]$, the family of functions $\{h_\theta^n\}_{\theta \in \Theta}$ is also
    \begin{equation*}
        ((\Delta_{h^n, \kappa, \alpha}, L_{h^n, \kappa, \alpha} V^\kappa, \alpha),
        (\tilde{\Delta}_{h^n, \tilde{\kappa}, \tilde{\alpha}}, \tilde{L}_{h^n, \tilde{\kappa}, \tilde{\alpha}} V^{\tilde{\kappa}}, \tilde{\alpha}))
    \end{equation*}
    -joint locally Hölder continuous, and bounded by $C_{h^n,\gamma} V^\gamma(\Lambda)$ for some constant $\Delta_{h^n, \kappa, \alpha}$, $\tilde{\Delta}_{h^n, \tilde{\kappa}, \tilde{\alpha}}$, $L_{h^n, \kappa, \alpha}$, $\tilde{L}_{h^n, \tilde{\kappa}, \tilde{\alpha}}$ and $C_{h^n,\gamma} > 0$.
\end{corollary}
\begin{proof}
    Fix $\gamma \in (0,1]$.
    Since $|h_\theta| \leq C_{h,\gamma} V^\gamma$, it follows that $|h_\theta^n| \leq C_{h,\gamma}^n V^{n\gamma}$.
    Moreover, for any $\Lambda,\Lambda^\prime$ and $\theta,\theta^\prime$, we have
    \begin{align*}
        \left| (h_\theta^n)(\Lambda) - (h_{\theta^\prime}^n)(\Lambda^\prime) \right|
        & \leq \left| h_\theta(\Lambda) - h_{\theta^\prime}(\Lambda^\prime) \right|
        \left[
            \sum_{k=0}^{n-1} |h_\theta|^k(\Lambda)|h_{\theta^\prime}|^{n-1-k}(\Lambda^\prime)
        \right] \\
        & \leq n C_{h,\gamma}^{n-1} \left| h_\theta(\Lambda) - h_{\theta^\prime}(\Lambda^\prime) \right| \max\{V^\gamma(\Lambda), V^\gamma(\Lambda^\prime)\}^{n-1} .
    \end{align*}

    Therefore, the joint local Hölder continuity of $\{h_\theta\}_{\theta \in \Theta}$ immediately transfers to the joint local Hölder continuity of $\{h_\theta^n\}_{\theta \in \Theta}$ with suitably modified constants, completing the proof.
\end{proof}

\subsection{Lemmas for LLN}
\label{subsec_lemmas_LLN}

\begin{assumption}
    \label{assumption_theta_difference_average_convergence}

    For some $\alpha_0>0$,
    \begin{equation*}
        \frac{1}{N} \sum_{n=0}^{N-1} d^{\alpha_0}(\theta_n, \theta_{n+1}) \mathbb{I}(d(\theta_n, \theta_{n+1}) < 1) \xrightarrow{\mathbb{P}} 0
        \quad \text{and} \quad
        \frac{1}{N} \sum_{n=0}^{N-1}
        \mathbb{I} \left( d(\theta_n,\theta_{n+1}) \geq 1 \right) \xrightarrow{\mathbb{P}} 0 ,
    \end{equation*}
\end{assumption}

\begin{assumption}
    \label{assumption_theta_difference_average_convergence_large_deviation}

    Building upon Assumption \ref{assumption_theta_difference_average_convergence}, suppose there exists some $q \in (0,1]$ such that, for any $\delta > 0$, there exists a constant $c_\delta > 0$ satisfying
    \begin{align*}
        & P\left( \frac{1}{N} \sum_{n=0}^{N-1} d^{\alpha_0}(\theta_n, \theta_{n+1}) \mathbb{I}(d(\theta_n, \theta_{n+1}) < 1) > \delta
        \right) < c_\delta N^{-q} , \\
        & P\left( \frac{1}{N} \sum_{n=0}^{N-1}
        \mathbb{I} \left( d(\theta_n,\theta_{n+1}) \geq 1 \right) > \delta
        \right) < c_\delta N^{-q} .
    \end{align*}
\end{assumption}

\begin{lemma}
    \label{lemma_convergence_average_enter}
    
    Suppose that Assumption \ref{assumption_simultaneous_properties} holds.
    Assume that for any $\alpha,\tilde{\alpha},\kappa,\tilde{\kappa} \in (0,1)$ and $\gamma \in (0,1]$, the family of functions $\{F_\theta\}_{\theta \in \Theta}$ is $((\Delta_{F, \kappa, \alpha}, L_{F, \kappa, \alpha} V^\kappa, \alpha), (\tilde{\Delta}_{F, \tilde{\kappa}, \tilde{\alpha}}, \tilde{L}_{F, \tilde{\kappa}, \tilde{\alpha}} V^{\tilde{\kappa}}, \tilde{\alpha}))$-joint locally Hölder continuous, and that $|F_\theta| \leq C_{F,\gamma} V^\gamma$.
    
    If the step sizes of the allocation parameter sequence $\{\theta_n\}$ satisfy Assumption \ref{assumption_theta_difference_average_convergence}, then the limit holds that
    \begin{equation*}
        \frac{1}{N} \sum_{n=0}^{N-1} F_{\theta_n}(\Lambda_n)
        -\frac{1}{N} \sum_{n=0}^{N-1} \pi_{\theta_n}F_{\theta_n}
        \xrightarrow{\mathbb{P}} 0 .
    \end{equation*}

    In addition, under Assumption \ref{assumption_theta_difference_average_convergence_large_deviation} on the step sizes of the allocation parameter sequence $\{\theta_n\}$, for any $\delta>0$, there exists a constant $c_\delta > 0$ such that
    \begin{equation*}
        P \left( \left| \frac{1}{N} \sum_{n=0}^{N-1} F_{\theta_n}(\Lambda_n)
        -\frac{1}{N} \sum_{n=0}^{N-1} \pi_{\theta_n}F_{\theta_n} \right|
        > \delta \right) < c_\delta N^{-q} ,
    \end{equation*}
    with $q \in (0,1]$ defined in Assumption \ref{assumption_theta_difference_average_convergence_large_deviation}.
\end{lemma}
\begin{proof}[Proof of Lemma \ref{lemma_convergence_average_enter}]
    By Item \ref{item_assumption_simultaneous_properties_2} of Assumption \ref{assumption_simultaneous_properties}, Lemma \ref{lemma_convergence_V_as} implies that $V(\Lambda_n) = O_P(1)$, $\frac{1}{N} \sum_{n=1}^N V(\Lambda_n) = O_P(1)$, and for any $a \in (0,1)$, $\sum_{n=1}^\infty n^{-1/a} V(\Lambda_n) < \infty$ almost surely.
    These properties will play an essential role in the subsequent analysis.

    For any $\theta \in \Theta$, define the function
    \begin{equation*}
        \hat{F}_\theta := \sum_{n \geq 0} P_\theta^n\left\{ F_\theta - \pi_\theta F_\theta \right\} .
    \end{equation*}
    It is well-defined under Item \ref{item_assumption_simultaneous_properties_1} of Assumption \ref{assumption_simultaneous_properties}.

    We decompose the target expression into a sum:
    \begin{equation*}
        \frac{1}{N} \sum_{n=0}^{N-1} F_{\theta_n}(\Lambda_n)
        -\frac{1}{N} \sum_{n=0}^{N-1} \pi_{\theta_n}F_{\theta_n}
        = T_{N,1} + T_{N,2} + T_{N,3} ,
    \end{equation*}
    where
    \begin{align*}
        T_{N,1} &:= \frac{1}{N} \sum_{n=0}^{N-2} \left[ \hat{F}_{\theta_n}(\Lambda_{n+1})
        -(P_{\theta_n}\hat{F}_{\theta_n})(\Lambda_n) \right] , \\
        T_{N,2} &:= \frac{1}{N} \sum_{n=0}^{N-2} \left[ \hat{F}_{\theta_{n+1}}(\Lambda_{n+1})
        -\hat{F}_{\theta_n}(\Lambda_{n+1}) \right] , \\
        T_{N,3} &:= \frac{1}{N} \left[ \hat{F}_{\theta_0}(\Lambda_0) - \hat{F}_{\theta_{N-1}}(\Lambda_{N-1})
        + F_{\theta_{N-1}}(\Lambda_{N-1}) - \pi_{\theta_{N-1}}F_{\theta_{N-1}} \right] .
    \end{align*}
    We show each term converges to zero in probability.

    First, $\{T_{N,1}\}_{N \in \mathbb{N}^*}$ is a martingale.
    To apply a law of large numbers for martingales, we bound its conditional moments.
    We use the Jensen's inequality and the conclusion in Corollary \ref{corollary_V_Poisson_solution_robust_lipschitz_continuous_bounded} that for any $a \in (0,1)$,
    \begin{align}
        & \quad \mathbb{E} \left[ \left| \hat{F}_{\theta_n}(\Lambda_{n+1})-(P_{\theta_n}\hat{F}_{\theta_n})(\Lambda_n) \right|^{1/a} \mid \mathcal{F}_n \right] \label{proof_lemma_convergence_average_enter_1} \\
        & \leq 2^{1/a-1} \mathbb{E} \left[ \left| \hat{F}_{\theta_n}(\Lambda_{n+1}) \right|^{1/a} 
        + \left| (P_{\theta_n}\hat{F}_{\theta_n})(\Lambda_n) \right|^{1/a} \mid \mathcal{F}_n \right] \notag \\
        & \leq 2^{1/a-1} \mathbb{E} \left[ \left| \hat{F}_{\theta_n} \right|^{1/a}(\Lambda_{n+1}) 
        + \left( P_{\theta_n}\left|\hat{F}_{\theta_n}\right|^{1/a} \right)(\Lambda_n) \mid \mathcal{F}_n \right] \notag \\
        &= 2^{1/a} \left( P_{\theta_n}\left|\hat{F}_{\theta_n}\right|^{1/a} \right)(\Lambda_n)
        \leq 2^{1/a} \left[ P_{\theta_n} \left( C_{\hat{F}, a} V^a \right)^{1/a} \right] (\Lambda_n) \notag \\
        & \leq 2^{1/a} C_{\hat{F}, a}^{1/a} (\beta+b) V(\Lambda_n) \notag .
    \end{align}

    By the assumption $\sum_{n=1}^\infty n^{-1/a} V(\Lambda_n) < \infty$ almost surely, we have
    \begin{equation*}
        \sum_{n=1}^\infty n^{-1/a}
        \mathbb{E} \left[ \left\{\left[ \hat{F}_{\theta_n}(\Lambda_{n+1})-(P_{\theta_n}\hat{F}_{\theta_n})(\Lambda_n) \right] \right\}^{1/a} \mid \mathcal{F}_n \right] < \infty .
    \end{equation*}
    By Theorem 2.18 in \cite{hallMartingaleLimitTheory1980}, it holds that
    \begin{equation*}
        T_{N,1}
        = \frac{1}{N} \sum_{n=0}^{N-2} \left[ \hat{F}_{\theta_n}(\Lambda_{n+1})
        -(P_{\theta_n}\hat{F}_{\theta_n})(\Lambda_n) \right] \rightarrow 0 \quad \text{a.s.}
    \end{equation*}

    Next, we prove that $T_{N,2} \xrightarrow{\mathbb{P}} 0$.
    Based on Corollary \ref{corollary_V_Poisson_solution_robust_lipschitz_continuous_bounded}, for any $\tilde{\kappa} \in (0,1)$ and $\tilde{\alpha} \in (0,1)$, there exists some positive constants $C_{\hat{F},\tilde{\kappa}}$, $\tilde{L}_{\hat{F}, \tilde{\kappa}, \tilde{\alpha}}$ and $\tilde{\Delta}_{\hat{F}, \tilde{\kappa}, \tilde{\alpha}}$ such that
    \begin{equation*}
        |\hat{F}_\theta(\Lambda)| \leq C_{\hat{F},\tilde{\kappa}} V^{\tilde{\kappa}}(\Lambda) ,
    \end{equation*}
    and when $d(\theta,\theta^\prime) < \tilde{\Delta}_{\hat{F}, \tilde{\kappa}, \tilde{\alpha}}$,
    \begin{equation*}
        \left| \hat{F}_\theta(\Lambda) - \hat{F}_{\theta^\prime}(\Lambda) \right|
        \leq \tilde{L}_{\hat{F}, \tilde{\kappa}, \tilde{\alpha}} V^{\tilde{\kappa}}(\Lambda) [d(\theta,\theta^\prime)]^{\tilde{\alpha}} .
    \end{equation*}

    Thus, by taking $\tilde{\kappa} = \frac{\alpha_0}{2(\alpha_0+1)} \in (0,1)$ and $\tilde{\alpha} = \frac{\alpha_0}{\alpha_0+1} \in (0,1)$, we have
    \begin{align*}
        & \quad \left| \frac{1}{N} \sum_{n=0}^{N-2} \left[ \hat{F}_{\theta_{n+1}}(\Lambda_{n+1})
        -\hat{F}_{\theta_n}(\Lambda_{n+1}) \right] \right|
        \leq \frac{1}{N} \sum_{n=0}^{N-2} \left|
        \hat{F}_{\theta_{n+1}}(\Lambda_{n+1})
        -\hat{F}_{\theta_n}(\Lambda_{n+1}) \right| \\
        & \leq \frac{1}{N} \sum_{n=0}^{N-2} \left\{
        \left| \hat{F}_{\theta_{n+1}}(\Lambda_{n+1})
        -\hat{F}_{\theta_n}(\Lambda_{n+1}) \right|
        \mathbb{I} \left( d(\theta_n,\theta_{n+1}) < \min \{ 1, \tilde{\Delta}_{\hat{F}, \frac{\alpha_0}{2(\alpha_0+1)}, \frac{\alpha_0}{\alpha_0+1}} \} \right) \right. \\
        & \quad \left. + \left[ \left| \hat{F}_{\theta_{n+1}}(\Lambda_{n+1}) \right|
        + \left| \hat{F}_{\theta_n}(\Lambda_{n+1}) \right| \right]
        \mathbb{I} \left( d(\theta_n,\theta_{n+1}) \geq \min \{ 1, \tilde{\Delta}_{\hat{F}, \frac{\alpha_0}{2(\alpha_0+1)}, \frac{\alpha_0}{\alpha_0+1}} \} \right) \right\} \\
        & \leq \frac{1}{N} \sum_{n=0}^{N-2}
        \tilde{L}_{\hat{F}, \frac{\alpha_0}{2(\alpha_0+1)}, \frac{\alpha_0}{\alpha_0+1}} V^{\frac{\alpha_0}{2(\alpha_0+1)}}(\Lambda_{n+1}) [d(\theta_n,\theta_{n+1})]^{\frac{\alpha_0}{\alpha_0+1}}
        \mathbb{I} \left( d(\theta_n,\theta_{n+1}) < \min \{ 1, \tilde{\Delta}_{\hat{F}, \frac{\alpha_0}{2(\alpha_0+1)}, \frac{\alpha_0}{\alpha_0+1}} \} \right) \\
        & \quad + \frac{1}{N} \sum_{n=0}^{N-2}
        2 C_{\hat{F},\frac{\alpha_0}{2(\alpha_0+1)}} V^{\frac{\alpha_0}{2(\alpha_0+1)}}(\Lambda_{n+1})
        \mathbb{I} \left( d(\theta_n,\theta_{n+1}) \geq \min \{ 1, \tilde{\Delta}_{\hat{F}, \frac{\alpha_0}{2(\alpha_0+1)}, \frac{\alpha_0}{\alpha_0+1}} \} \right) .
    \end{align*}

    Because
    \begin{align}
        & \quad \frac{1}{N} \sum_{n=0}^{N-2}
        \tilde{L}_{\hat{F}, \frac{\alpha_0}{2(\alpha_0+1)}, \frac{\alpha_0}{\alpha_0+1}} V^{\frac{\alpha_0}{2(\alpha_0+1)}}(\Lambda_{n+1}) [d(\theta_n,\theta_{n+1})]^{\frac{\alpha_0}{\alpha_0+1}}
        \mathbb{I} \left( d(\theta_n,\theta_{n+1}) < \min \{ 1, \tilde{\Delta}_{\hat{F}, \frac{\alpha_0}{2(\alpha_0+1)}, \frac{\alpha_0}{\alpha_0+1}} \} \right) \label{proof_lemma_convergence_average_enter_2} \\
        & \leq \tilde{L}_{\hat{F}, \frac{\alpha_0}{2(\alpha_0+1)}, \frac{\alpha_0}{\alpha_0+1}}
        \left[ \frac{1}{N} \sum_{n=0}^{N-2} V^{\frac{1}{2}}(\Lambda_{n+1}) \right]^{\frac{\alpha_0}{\alpha_0+1}}
        \left[ \frac{1}{N} \sum_{n=0}^{N-2} [d(\theta_n,\theta_{n+1})]^{\alpha_0} \mathbb{I} \left( d(\theta_n,\theta_{n+1}) < 1 \right) \right]^{\frac{1}{\alpha_0+1}} \notag \\
        & \xrightarrow{\mathbb{P}} 0 \notag
    \end{align}
    and
    \begin{align}
        & \quad \frac{1}{N} \sum_{n=0}^{N-2}
        2 C_{\hat{F},\frac{\alpha_0}{2(\alpha_0+1)}} V^{\frac{\alpha_0}{2(\alpha_0+1)}}(\Lambda_{n+1})
        \mathbb{I} \left( d(\theta_n,\theta_{n+1}) \geq \min \{ 1, \tilde{\Delta}_{\hat{F}, \frac{\alpha_0}{2(\alpha_0+1)}, \frac{\alpha_0}{\alpha_0+1}} \} \right) \label{proof_lemma_convergence_average_enter_3} \\
        & \leq 2 C_{\hat{F},\frac{\alpha_0}{2(\alpha_0+1)}} \left[ \frac{1}{N} \sum_{n=0}^{N-2} V^{\frac{1}{2}}(\Lambda_{n+1}) \right]^{\frac{\alpha_0}{\alpha_0+1}}
        \left[ \frac{1}{N} \sum_{n=0}^{N-2}
        \mathbb{I} \left( d(\theta_n,\theta_{n+1}) \geq \min \{ 1, \tilde{\Delta}_{\hat{F}, \frac{\alpha_0}{2(\alpha_0+1)}, \frac{\alpha_0}{\alpha_0+1}} \} \right) \right]^{\frac{1}{\alpha_0+1}} \notag \\
        & \leq 2 C_{\hat{F},\frac{\alpha_0}{2(\alpha_0+1)}} \left[ \frac{1}{N} \sum_{n=0}^{N-2} V^{\frac{1}{2}}(\Lambda_{n+1}) \right]^{\frac{\alpha_0}{\alpha_0+1}} \notag \\
        & \quad \left[ \frac{1}{N} \sum_{n=0}^{N-2}
        \mathbb{I} \left( d(\theta_n,\theta_{n+1}) \geq 1 \right)
        +  \frac{1}{N} \sum_{n=0}^{N-2} \left[ \frac{ d(\theta_n,\theta_{n+1})}{\tilde{\Delta}_{\hat{F}, \frac{\alpha_0}{2(\alpha_0+1)}, \frac{\alpha_0}{\alpha_0+1}}} \right]^{\alpha_0}
        \mathbb{I} \left( d(\theta_n,\theta_{n+1}) < 1 \right)
        \right]^{\frac{1}{\alpha_0+1}} \notag \\
        & \xrightarrow{\mathbb{P}} 0 \notag ,
    \end{align}
    it follows that
    \begin{equation*}
        |T_{N,2}|
        = \left| \frac{1}{N} \sum_{n=0}^{N-2} \left[ \hat{F}_{\theta_{n+1}}(\Lambda_{n+1})
        -\hat{F}_{\theta_n}(\Lambda_{n+1}) \right] \right|
        \xrightarrow{\mathbb{P}} 0 .
    \end{equation*}
    
    Finally, for the term $T_{N,3}$, the assumption in Lemma \ref{lemma_convergence_average_enter} implies that
    \begin{equation*}
        \left| F_{\theta_{N-1}}(\Lambda_{N-1}) \right|
        \leq C_{F,a} V^a(\Lambda_{N-1}) ,
    \end{equation*}
    Corollary \ref{corollary_V_Poisson_solution_robust_lipschitz_continuous_bounded} implies that
    \begin{equation*}
        \left| \hat{F}_{\theta_{N-1}}(\Lambda_{N-1}) \right|
        \leq C_{\hat{F},a} V^a(\Lambda_{N-1}) ,
    \end{equation*}
    and Assumption \ref{assumption_simultaneous_properties} implies that
    \begin{equation*}
        \left| \pi_{\theta_{N-1}}F_{\theta_{N-1}} \right|
        \leq C_{F,1} \pi_\theta V
        \leq C_{F,1} \frac{b}{1-\beta} .
    \end{equation*}
    Therefore, we can show that $T_{N,3} \xrightarrow{\mathbb{P}} 0$ from the condition that $V(\Lambda_n)$ is $O_P(1)$.

    The convergence of all three terms completes the proof of the first part.

    For the second part, regarding the large deviation of $\frac{1}{N} \sum_{n=0}^{N-1} F_{\theta_n}(\Lambda_n)-\frac{1}{N} \sum_{n=0}^{N-1} \pi_{\theta_n}F_{\theta_n}$, we proceed to establish the large deviation bounds for each component term in its decomposition, namely $T_{N,1}$, $T_{N,2}$, and $T_{N,3}$.
    \begin{enumerate}
        \item Because \eqref{proof_lemma_convergence_average_enter_1} holds and $\mathbb{E}[V(\Lambda_n)]$ is uniformly bounded, each increment of the martingale $T_{N,1}$ has a bounded $L^{p_1}$ norm for any $p_1 \geq 2$.
        By Lemma \ref{lemma_martingale_large_deviation}, it follows that for any $p_1 \geq 2$, $P(|T_{N,1}| > \delta) = O(N^{-p_1/2})$.
        \item Lemma \ref{lemma_V_large_deviation} with $p=2$ implies that for some $x>0$ and $c>0$, it holds that
        \begin{equation*}
            P\left( \sum_{n=1}^N V^{\frac{1}{2}}(\Lambda_n) > N x \right)
            \leq \frac{c}{N} .
        \end{equation*}
        Combining with Assumption \ref{assumption_theta_difference_average_convergence_large_deviation}, we can conclude that \eqref{proof_lemma_convergence_average_enter_2} and \eqref{proof_lemma_convergence_average_enter_3} admit a large deviation bound of order $O(N^{-q})$, with $q \in (0,1]$ defined in Assumption \ref{assumption_theta_difference_average_convergence_large_deviation}. Thus, we can conclude that $P(|T_{N,2}| > \delta) = O(N^{-q})$.
        \item Because $N |T_{N,3}|$ is bounded by
        \begin{equation*}
            C_{\hat{F},a} V^a(\Lambda_0) + (C_{F,a}+C_{\hat{F},a}) V^a(\Lambda_{N-1})+C_{F,1} \frac{b}{1-\beta} ,
        \end{equation*}
        which has a bounded expectation for $N \geq 1$.
        Therefore, by Markov's inequality, $P(|T_{N,3}| > \delta) = O(\frac{1}{N})$.
    \end{enumerate}
    Hence, all the component terms in the decomposition exhibit polynomially decaying large deviations, implying that the total deviation is of order $O(N^{-q})$, with $q \in (0,1]$ defined in Assumption \ref{assumption_theta_difference_average_convergence_large_deviation}.
\end{proof}

\begin{lemma}
    \label{lemma_allocation_convergence}

    Suppose that $\mathbb{E}_{(X,Y(t)) \sim \Gamma_{X,Y(t)}} \left[ \left[ m(X, Y(t), t) \right]^2 \right] < \infty$ for $t \in \{0, 1\}$, and that the allocation function $g_\theta$ satisfies $\pi_\theta \left[ g_\theta(\cdot, x) \right] = \rho_\theta(x)$ for $\Gamma$-a.e. $x$.
    Under Assumptions \ref{assumption_theta_compact}, \ref{assumption_lipschitz_continuity_functions_kernels} and \ref{assumption_simultaneous_properties}, if the step sizes of the allocation parameter sequence $\{\theta_n\}$ satisfy Assumption \ref{assumption_theta_difference_average_convergence}, then
    \begin{align*}
        & \frac{M_n}{n}
        = \frac{1}{n} \sum_{i=1}^n
        \left[ \frac{\rho^{\mathrm{ref}}(T_i \mid X_i)}{\rho_{\theta_{i-1}}(T_i \mid X_i)} m(X_i, Y_i(T_i), T_i) \right] \\
        & \quad \xrightarrow{\mathbb{P}} M = \mathbb{E}_{(X,Y(1),Y(0)) \sim \Gamma_{X,\bm{Y}}} \left[ \rho^{\mathrm{ref}}(X) m(X, Y(1), 1) + [1 - \rho^{\mathrm{ref}}(X)] m(X, Y(0), 0) \right] .
    \end{align*}

    In addition, under Assumption \ref{assumption_theta_difference_average_convergence_large_deviation} on the step sizes of the allocation parameter sequence $\{\theta_n\}$, for any $\delta>0$, there exists a constant $c_\delta > 0$ such that
    \begin{equation*}
        P \left( \left| \frac{M_n}{n} - M \right|
        > \delta \right) < c_\delta n^{-q} ,
    \end{equation*}
    with $q \in (0,1]$ defined in Assumption \ref{assumption_theta_difference_average_convergence_large_deviation}.
\end{lemma}
\begin{proof}
    Let us define the function $h_\theta(\Lambda)$ as the conditional expectation of a single term in $M_n$, given by
    \begin{equation*}
        h_\theta(\Lambda) =
        \mathbb{E}_{(X,Y(1),Y(0)) \sim \Gamma_{X,\bm{Y}}} \left[ \sum_{T=0}^{1} \frac{\rho^{\mathrm{ref}}(T \mid X) g_\theta(T \mid \Lambda, X)}
        {\rho_\theta(T \mid X)} m(X, Y(T), T) \right] .
    \end{equation*}
    By Assumption \ref{assumption_sampling}, $h_\theta(\Lambda)$ is uniformly bounded as follows:
    \begin{align*}
        h_\theta(\Lambda)
        & \leq \mathbb{E}_{(X,Y(1),Y(0)) \sim \Gamma_{X,\bm{Y}}} \left[
        \frac{ |m(X, Y(1), 1)| + |m(X, Y(0), 0)| }{\iota} \right] \\
        & \leq \frac{\mathbb{E}_{(X,Y(1)) \sim \Gamma_{X,Y(1)}} \left| m(X, Y(1), 1) \right|
        + \mathbb{E}_{(X,Y(0)) \sim \Gamma_{X,Y(0)}} \left| m(X, Y(0), 0) \right|}{\iota}
        < \infty .
    \end{align*}
    Let $C$ denote this finite upper bound.

    The proof proceeds by decomposing the term $M_n$ and establishing two key convergence results.
    Let us define an auxiliary sequence $M_n^\prime$ as
    \begin{equation*}
        M_n^\prime
        = \sum_{i=1}^n
        \left[ h_{\theta_{i-1}}(\Lambda_{i-1}) \right] .
    \end{equation*}
    The sequence $\{M_n - M_n^\prime\}$ forms a martingale with respect to the filtration $\{\mathcal{F}_i\}$, because for any $i \in \mathbb{N}^*$, the conditional expectation of each term satisfies
    \begin{equation*}
        \mathbb{E} \left[ \frac{\rho^{\mathrm{ref}}(T_i \mid X_i)}{\rho_{\theta_{i-1}}(T_i \mid X_i)} m(X_i, Y_i(T_i), T_i)
        - h_{\theta_{i-1}}(\Lambda_{i-1})
        \mid \mathcal{F}_{i-1} \right] = 0 .
    \end{equation*}

    Next, we will now show that
    \begin{enumerate}
        \item $\frac{M_n - M_n^\prime}{n} \rightarrow 0$ almost surely,
        \item $\frac{M_n^\prime}{n} \xrightarrow{\mathbb{P}} M = \mathbb{E}_{(X,Y(1),Y(0)) \sim \Gamma_{X,\bm{Y}}} \left[ \rho^{\mathrm{ref}}(X) m(X, Y(1), 1) + [1 - \rho^{\mathrm{ref}}(X)] m(X, Y(0), 0) \right]$.
    \end{enumerate}

    To prove the first claim, we invoke a strong law of large numbers for martingales.
    This requires verifying a moment condition on the martingale differences.
    Let the constant $a=1/2$.
    The conditional second moment of each difference term is bounded as
    \begin{align}
        & \quad \mathbb{E} \left[ \left| \frac{\rho^{\mathrm{ref}}(T_i \mid X_i)}{\rho_{\theta_{i-1}}(T_i \mid X_i)} m(X_i, Y_i(T_i), T_i)
        - h_{\theta_{i-1}}(\Lambda_{i-1}) \right|^{1/a}
        \mid \mathcal{F}_{i-1} \right] \label{eq_proof_lemma_allocation_convergence_1} \\
        & \leq 2^{1/a-1} \mathbb{E} \left[ \left| \frac{\rho^{\mathrm{ref}}(T_i \mid X_i)}{\rho_{\theta_{i-1}}(T_i \mid X_i)} m(X_i, Y_i(T_i), T_i) \right|^{1/a}
        + \left| h_{\theta_{i-1}}(\Lambda_{i-1}) \right|^{1/a} \mid \mathcal{F}_{i-1} \right] \notag \\
        & \leq 2^{1/a-1} \mathbb{E} \left[ \frac{\left| m(X_i, Y_i(1), 1) \right|^{1/a} + \left| m(X_i, Y_i(0), 0) \right|^{1/a}}{\iota^{1/a}}
        + C^{1/a} \mid \mathcal{F}_{i-1} \right] \notag .
    \end{align}
    Because the right-hand side of \eqref{eq_proof_lemma_allocation_convergence_1} is a constant independent of the index $i$, it holds that
    \begin{equation*}
        \sum_{i=1}^\infty i^{-1/a}
        \mathbb{E} \left[ \left| \frac{\rho^{\mathrm{ref}}(T_i \mid X_i)}{\rho_{\theta_{i-1}}(T_i \mid X_i)} m(X_i, Y_i(T_i), T_i)
        - h_{\theta_{i-1}}(\Lambda_{i-1}) \right|^{1/a}
        \mid \mathcal{F}_{i-1} \right]
        < \infty .
    \end{equation*}
    By Theorem 2.18 in \cite{hallMartingaleLimitTheory1980}, we can conclude that
    \begin{equation*}
        \frac{M_n - M_n^\prime}{n}
        = \frac{1}{n} \sum_{i=1}^{n} \left[ \frac{\rho^{\mathrm{ref}}(T_i \mid X_i)}{\rho_{\theta_{i-1}}(T_i \mid X_i)} m(X_i, Y_i(T_i), T_i)
        - h_{\theta_{i-1}}(\Lambda_{i-1}) \right]
        \rightarrow 0 \quad \text{a.s.}
    \end{equation*}
    
    We now turn to prove
    \begin{equation*}
        \frac{M_n^\prime}{n} \xrightarrow{\mathbb{P}}
        M = \mathbb{E}_{(X,Y(1),Y(0)) \sim \Gamma_{X,\bm{Y}}} \left[ \rho^{\mathrm{ref}}(X) m(X, Y(1), 1) + [1 - \rho^{\mathrm{ref}}(X)] m(X, Y(0), 0) \right] .
    \end{equation*}
    Lemma \ref{lemma_convergence_average_enter} establishes this limit under a collection of technical conditions because each term of the limit in Lemma \ref{lemma_convergence_average_enter}, $\pi_\theta h_\theta$, is
    \begin{align*}
        \pi_\theta h_\theta
        &= \mathbb{E}_{(X,Y(1),Y(0)) \sim \Gamma_{X,\bm{Y}}, \Lambda \sim \pi_\theta} \left[ \sum_{T=0}^{1} \frac{\rho^{\mathrm{ref}}(T \mid X) g_\theta(T \mid \Lambda, X)}
        {\rho_\theta(T \mid X)} m(X, Y(T), T) \right] \\
        &= \mathbb{E}_{(X,Y(1),Y(0)) \sim \Gamma_{X,\bm{Y}}} \left[
        \frac{\rho^{\mathrm{ref}}(X) \pi_\theta \left[ g_\theta(\cdot, X) \right]}
        {\rho_\theta(X)} m(X, Y(1), 1) \right. \\
        & \quad \left. + \frac{ \left\{ 1-\rho^{\mathrm{ref}}(X) \right\} \left\{ 1-\pi_\theta \left[ g_\theta(\cdot, X) \right] \right\} }
        {1-\rho_\theta(X)} m(X, Y(0), 0) \right] \\
        & = \mathbb{E}_{(X,Y(1),Y(0)) \sim \Gamma_{X,\bm{Y}}} \left[ \rho^{\mathrm{ref}}(X) m(X, Y(1), 1) + [1 - \rho^{\mathrm{ref}}(X)] m(X, Y(0), 0) \right] ,
    \end{align*}
    which is a constant independent of $\theta$.
    Since the conditions of the lemma already guarantee a subset of these requirements of Lemma \ref{lemma_convergence_average_enter} and the boundedness of $h_\theta$, it remains to verify the remaining one that for any $\alpha,\tilde{\alpha},\kappa,\tilde{\kappa} \in (0,1)$, the family of functions $\{h_\theta\}_{\theta \in \Theta}$ is $((\Delta_{h, \kappa, \alpha}, L_{h, \kappa, \alpha} V^\kappa, \alpha), (\tilde{\Delta}_{h, \tilde{\kappa}, \tilde{\alpha}}, \tilde{L}_{h, \tilde{\kappa}, \tilde{\alpha}} V^{\tilde{\kappa}}, \tilde{\alpha}))$-joint locally Hölder continuous.

    For brevity, it suffices to show that $\mathbb{E}_{(X,Y(1)) \sim \Gamma_{X,Y(1)}} \left[ \frac{\rho^{\mathrm{ref}}(X) g_\theta(\Lambda, X)}{\rho_\theta(X)} m(X, Y(1), 1) \right]$ is Lipschitz continuous with respect to $(\theta,\Lambda)$.
    To this end, we examine the difference between the expectations at $(\theta,\Lambda)$ and $(\theta^\prime,\Lambda^\prime)$:
    \begin{align*}
        & \quad \left| \mathbb{E}_{(X,Y(1)) \sim \Gamma_{X,Y(1)}} \left[ \frac{\rho^{\mathrm{ref}}(X) g_\theta(\Lambda, X)}{\rho_\theta(X)} m(X, Y(1), 1) \right] \right. \\
        & \quad\quad \left. - \mathbb{E}_{(X,Y(1)) \sim \Gamma_{X,Y(1)}} \left[ \frac{\rho^{\mathrm{ref}}(X) g_{\theta^\prime}(\Lambda^\prime, X)}{\rho_{\theta^\prime}(X)} m(X, Y(1), 1) \right] \right| \\
        & \leq \mathbb{E}_{(X,Y(1)) \sim \Gamma_{X,Y(1)}} \left|
        \frac{\rho^{\mathrm{ref}}(X) g_\theta(\Lambda, X)}{\rho_\theta(X)} m(X, Y(1), 1)
        - \frac{\rho^{\mathrm{ref}}(X) g_{\theta^\prime}(\Lambda^\prime, X)}{\rho_{\theta^\prime}(X)} m(X, Y(1), 1) \right| \\
        & \leq \sqrt{\mathbb{E}_{(X,Y(1)) \sim \Gamma_{X,Y(1)}} |m(X, Y(1), 1)|^2}
        \sqrt{\mathbb{E}_{X \sim \Gamma} \left[ \left[ \frac{g_\theta(\Lambda, X)}{\rho_\theta(X)}
        - \frac{g_{\theta^\prime}(\Lambda^\prime, X)}{\rho_{\theta^\prime}(X)} \right]^2 \right]} \\
        & \leq \sqrt{\mathbb{E}_{(X,Y(1)) \sim \Gamma_{X,Y(1)}} |m(X, Y(1), 1)|^2}
        \sqrt{\mathbb{E}_{X \sim \Gamma} \left[ \left[ \frac{g-g^\prime}{\rho}
        - \frac{g^\prime(\rho^\prime-\rho)}{\rho\rho^\prime} \right]^2 \right]} \\
        & \leq \sqrt{\mathbb{E}_{(X,Y(1)) \sim \Gamma_{X,Y(1)}} |m(X, Y(1), 1)|^2}
        \sqrt{ \frac{2}{\iota^2} \mathbb{E}_{X \sim \Gamma} \left[ g-g^\prime \right]^2
        + \frac{2}{\iota^4} \mathbb{E}_{X \sim \Gamma} \left[ \rho-\rho^\prime \right]^2 } \\
        & \leq \sqrt{\mathbb{E}_{(X,Y(1)) \sim \Gamma_{X,Y(1)}} |m(X, Y(1), 1)|^2} \frac{2}{\iota^2} \left[ \left\| g_{\theta}(\Lambda,\cdot) - g_{\theta^\prime}(\Lambda^\prime,\cdot) \right\|_{L^2(\Gamma)}
        + \left\| \rho_{\theta} - \rho_{\theta^\prime} \right\|_{L^2(\Gamma)} \right] \\
        & \leq \sqrt{\mathbb{E}_{(X,Y(1)) \sim \Gamma_{X,Y(1)}} |m(X, Y(1), 1)|^2} \frac{2(L_g+L_\rho)}{\iota^2}
        \left[ d(\theta, \theta^\prime) + d(\Lambda, \Lambda^\prime) \right] .
    \end{align*}
    These steps above use the Cauchy-Schwarz inequality, the bound in Assumption \ref{assumption_sampling}, and finally the Lipschitz continuity of $g$ and $\rho$ in Assumption \ref{assumption_lipschitz_continuity_functions_kernels}.
    This confirms that $h_\theta(\Lambda)$ is Lipschitz continuous with respect to $(\theta, \Lambda)$, thus satisfying the conditions of Lemma \ref{lemma_convergence_average_enter}.

    We have shown that $\frac{M_n - M_n'}{n} \to 0$ almost surely and $\frac{M_n'}{n} \xrightarrow{\mathbb{P}} M$.
    Therefore, we can conclude that
    \begin{align*}
        & \frac{M_n}{n}
        = \frac{1}{n} \sum_{i=1}^n
        \left[ \frac{\rho^{\mathrm{ref}}(T_i \mid X_i)}{\rho_{\theta_{i-1}}(T_i \mid X_i)} m(X_i, Y_i(T_i), T_i) \right] \\
        & \quad \xrightarrow{\mathbb{P}} M = \mathbb{E}_{(X,Y(1),Y(0)) \sim \Gamma_{X,\bm{Y}}} \left[ \rho^{\mathrm{ref}}(X) m(X, Y(1), 1) + [1 - \rho^{\mathrm{ref}}(X)] m(X, Y(0), 0) \right] .
    \end{align*}

    For the second part of the lemma, regarding the large deviation of $\frac{M_n}{n}$, we proceed to establish the large deviation bounds for each component term in its decomposition, namely $\frac{M_n - M_n^\prime}{n}$ and $\frac{M_n^\prime}{n}$.
    \begin{enumerate}
        \item Because $\mathbb{E}_{(X,Y(t)) \sim \Gamma_{X,Y(t)}} \left[ \left[ m(X, Y(t), t) \right]^2 \right] < \infty$, each increment of the martingale $\frac{M_n - M_n^\prime}{n}$ has a bounded $L^2$ norm.
        By Lemma \ref{lemma_martingale_large_deviation}, it follows that $P(|\frac{M_n - M_n^\prime}{n}| > \delta) = O(\frac{1}{n})$.
        \item Lemma \ref{lemma_convergence_average_enter} establishes a large deviation bound for $\frac{M_n^\prime}{n}$, such that
        \begin{equation*}
        P \left( \left| \frac{M_n^\prime}{n} - M \right|
        > \delta \right) < c_\delta n^{-q} .
    \end{equation*}
    \end{enumerate}
    Hence, all the component terms in the decomposition exhibit polynomially decaying large deviations, implying that the total deviation is of order $O(n^{-q})$, with $q \in (0,1]$ defined in Assumption \ref{assumption_theta_difference_average_convergence_large_deviation}.
\end{proof}

\subsection{Lemmas for Simultaneous Geometric Ergodicity}

The following lemma is based on Lemma 2.3 in \cite{fortConvergenceAdaptiveInteracting2011}.
\begin{lemma}
    \label{lemma_simultaneous_geometric_ergodicity}
    Assume that for all $\theta \in \Theta$, $P_\theta$ is a $\pi$-irreducible and aperiodic transition kernel on $\mathrm{X}$.
    Moreover, there exist some constants $b<\infty, \delta \in (0,1)$, $\beta \in (0,1)$, a probability measure $\nu$ on X and a function $V: \mathrm{X} \to [1,+\infty)$, such that for any $\theta \in \Theta$,
    \begin{align*}
        P_\theta V & \leq \beta V+b, \\
        P_\theta(x, \cdot) & \geq \delta \nu(\cdot) \mathbb{I}_{\left\{V \leq c\right\}}(x), \quad
        c := 2 b\left(1-\beta\right)^{-1} .
    \end{align*}
    Then there are some universal constants $C$ and $\gamma$ such that for any $\theta \in \Theta$, there exists a probability distribution $\pi_\theta$ such that $\pi_\theta P_\theta=\pi_\theta$, $\pi_\theta(V) \leq b\left(1-\beta\right)^{-1}$ and the inequality
    \begin{equation*}
        \left\|P_\theta^n(x, \cdot)-\pi_\theta\right\|_V \leq L(1-L^{-1})^n V(x)
    \end{equation*}
    holds with $L = C\left\{b \vee \delta^{-1} \vee\left(1-\beta\right)^{-1} \vee c\right\}^\gamma$.
\end{lemma}

\section{Additional Lemmas}

The following lemma provides a large-deviation inequality for martingales.
\begin{lemma}[Theorem 3.6 in \cite{lesigneLargeDeviationsMartingales2001}]
    \label{lemma_martingale_large_deviation}
    Let $\left\{ X_i \right\}_{1 \leq i \leq n}$ be a finite sequence of martingale differences where $X_i \in L^p, 2 \leq p < \infty, \left\| X_i \right\|_p < M <\infty$ for all $i$.
    Let $x>0$. Then
    \begin{equation*}
        P\left( \left| S_n \right| > n x \right)
        \leq \left(18 p q^{1 / 2}\right)^p \frac{M^p}{x^p} \frac{1}{n^{p / 2}} ,
    \end{equation*}
    where $q$ is the real number for which $1 / p+1 / q=1$.
\end{lemma}

The following lemma is a slight modification of the corresponding lemma in \cite{fangGeneralNonMarkovianFramework2026}.
\begin{lemma}
    \label{lemma_convergence_V_as}

    Let $V \geq 0$ be a function such that its initial expectation is finite, i.e., $\mathbb{E}[V(\Lambda_0)] < \infty$.
    Suppose that the inequality $(P_\theta V)(\Lambda) \leq \beta V(\Lambda)+b$ holds uniformly for any $\theta \in \Theta$.
    Then,
    \begin{enumerate}
        \item $\sup_{n \geq 0} \mathbb{E}[V(\Lambda_n)] \leq \max \left\{ \mathbb{E}[V(\Lambda_0)], \frac{b}{1-\beta} \right\}$,
        \item $V(\Lambda_n) = O_P(1)$,
        \item $\frac{1}{N} \sum_{n=1}^N V(\Lambda_n) = O_P(1)$,
        \item for any $p>1$, $\sum_{n=1}^\infty n^{-p} V(\Lambda_n) < \infty$ almost surely.
    \end{enumerate}
\end{lemma}
\begin{proof}
    Iterating the inequality $(P_\theta V)(\Lambda) \leq \beta V(\Lambda)+b$ via the law of total expectation yields, for any $n \geq 0$,
    \begin{align*}
        \mathbb{E}[V(\Lambda_n)]
        & \leq \beta \mathbb{E}[V(\Lambda_{n-1})] + b
        \leq \beta^2 \mathbb{E}[V(\Lambda_{n-2})] + b(1+\beta) \\
        & \leq \dots
        \leq \beta^n \mathbb{E}[V(\Lambda_0)] + (1-\beta^n) \frac{b}{1-\beta}
        \leq \max \left\{ \mathbb{E}[V(\Lambda_0)], \frac{b}{1-\beta} \right\} .
    \end{align*}
    Since $\mathbb{E}[V(\Lambda_0)]$ is finite, this establishes a uniform bound on the moments that $\sup_{n \geq 0} \mathbb{E}[V(\Lambda_n)] \leq C < \infty$.

    The three results are direct consequences of this uniform moment bound.
    \begin{enumerate}
        \item By Markov's inequality, $\mathbb{P}(V(\Lambda_n) > M) \leq C/M$, which implies that $V(\Lambda_n) = O_P(1)$.
        \item Again by Markov's inequality,
        \begin{equation*}
            \mathbb{P}\left(\frac{1}{N} \sum_{n=1}^N V(\Lambda_n) > M\right)
            \leq \frac{1}{M N} \sum_{n=1}^N \mathbb{E}[V(\Lambda_n)]
            \leq \frac{C}{M} .
        \end{equation*}
        \item By the Monotone Convergence Theorem, for any $p>1$,
        \begin{equation*}
            \mathbb{E}\left[ \sum_{n=1}^\infty n^{-p} V(\Lambda_n) \right] = \sum_{n=1}^\infty n^{-p} \mathbb{E}[V(\Lambda_n)] \leq C \sum_{n=1}^\infty n^{-p} < \infty .
        \end{equation*}
        A non-negative random variable with finite expectation is finite almost surely, so the series converges almost surely.
    \end{enumerate}
    This completes the proof.
\end{proof}

\begin{lemma}
    \label{lemma_V_large_deviation}

    Let $V \geq 0$ be a measurable function satisfying $\mathbb{E}[V(\Lambda_0)] < \infty$.
    Suppose that for some $p \in [2,\infty)$, there exist constants $\beta, \beta_{1/p} \in (0,1)$, and $b, b_{1/p} > 0$ such that, uniformly over all $\theta \in \Theta$,
    \begin{equation*}
        (P_\theta V)(\Lambda) \leq \beta V(\Lambda) + b
        \quad \text{and} \quad
        (P_\theta V^{1/p})(\Lambda) \leq \beta_{1/p} V^{1/p}(\Lambda) + b_{1/p}.
    \end{equation*}
    Then for any $x > 0$,
    \begin{align*}
        & P\left( \sum_{n=1}^N V^{\frac{1}{p}}(\Lambda_n) > N \left( \frac{2b_{\frac{1}{p}}}{1-\beta_{\frac{1}{p}}}+x \right) \right) \\
        & \quad \leq \left(18 p q^{1 / 2}\right)^p \frac{2^p \left[ \max \left\{ \mathbb{E}[V(\Lambda_0)], \frac{b}{1-\beta} \right\} \right]}{(1-\beta_{\frac{1}{p}})^p x^p} \frac{1}{N^{p / 2}}
        + \frac{\mathbb{E}[V(\Lambda_0)] \beta_{\frac{1}{p}}^p}{N^p b_{\frac{1}{p}}^p} ,
    \end{align*}
    where $q$ is the conjugate exponent of $p$, that is, $1/p + 1/q = 1$.
\end{lemma}
\begin{proof}
    Lemma \ref{lemma_martingale_large_deviation} implies that for any $x>0$,
    \begin{equation*}
        P\left( \sum_{n=1}^N \left[ V^{\frac{1}{p}}(\Lambda_n) - \mathbb{E} \left[ V^{\frac{1}{p}}(\Lambda_n) \mid \mathcal{F}_{n-1} \right] \right] > N x \right)
        \leq \left(18 p q^{1 / 2}\right)^p \frac{M^p}{x^p} \frac{1}{N^{p / 2}}
    \end{equation*}
    with $M = 2 \left[ \max \left\{ \mathbb{E}[V(\Lambda_0)], \frac{b}{1-\beta} \right\} \right]^{\frac{1}{p}}$, which follows from Lemma \ref{lemma_convergence_V_as} and the inequality
    \begin{equation*}
        \mathbb{E} \left| V^{\frac{1}{p}}(\Lambda_n) - \mathbb{E} \left[ V^{\frac{1}{p}}(\Lambda_n) \mid \mathcal{F}_{n-1} \right] \right|^p
        \leq 2^{p-1} \mathbb{E} \left| V^{\frac{1}{p}}(\Lambda_n) \right|^p
        + 2^{p-1} \mathbb{E} \left| \mathbb{E} \left[ V^{\frac{1}{p}}(\Lambda_n) \mid \mathcal{F}_{n-1} \right] \right|^p
        \leq 2^p \mathbb{E}[V(\Lambda_n)] .
    \end{equation*}
    Moreover, by $(P_\theta V^{\frac{1}{p}})(\Lambda) \leq \beta_{\frac{1}{p}} V^{\frac{1}{p}}(\Lambda)+b_{\frac{1}{p}}$ for any $\theta \in \Theta$, it holds that
    \begin{equation*}
        \sum_{n=1}^N \left[ V^{\frac{1}{p}}(\Lambda_n) - \mathbb{E} \left[ V^{\frac{1}{p}}(\Lambda_n) \mid \mathcal{F}_{n-1} \right] \right]
        \geq \sum_{n=1}^N \left[ (1-\beta_{\frac{1}{p}}) V^{\frac{1}{p}}(\Lambda_n) \right] - N b_{\frac{1}{p}} - \beta_{\frac{1}{p}} V^{\frac{1}{p}}(\Lambda_0) .
    \end{equation*}
    Thus,
    \begin{align*}
        & \quad P\left( \sum_{n=1}^N \left[ V^{\frac{1}{p}}(\Lambda_n) - \mathbb{E} \left[ V^{\frac{1}{p}}(\Lambda_n) \mid \mathcal{F}_{n-1} \right] \right] > N x \right) \\
        & \geq P\left( \sum_{n=1}^N \left[ (1-\beta_{\frac{1}{p}}) V^{\frac{1}{p}}(\Lambda_n) \right] - N b_{\frac{1}{p}} - \beta_{\frac{1}{p}} V^{\frac{1}{p}}(\Lambda_0) > N x \right)
    \end{align*}
    Combining with
    \begin{equation*}
        P \left( V^{\frac{1}{p}}(\Lambda_0) \geq N x \right)
        \leq \frac{\mathbb{E}[V(\Lambda_0)]}{N^p x^p} ,
    \end{equation*}
    we can conclude that for any $x > 0$,
    \begin{align*}
        & \quad P\left( \sum_{n=1}^N V^{\frac{1}{p}}(\Lambda_n) > N \left( \frac{2b_{\frac{1}{p}}}{1-\beta_{\frac{1}{p}}}+x \right) \right) \\
        &= P\left( \sum_{n=1}^N \left[ (1-\beta_{\frac{1}{p}}) V^{\frac{1}{p}}(\Lambda_n) \right] - N b_{\frac{1}{p}} - \beta_{\frac{1}{p}} V^{\frac{1}{p}}(\Lambda_0)
        > N b_{\frac{1}{p}} + (1-\beta_{\frac{1}{p}}) N x - \beta_{\frac{1}{p}} V^{\frac{1}{p}}(\Lambda_0) \right) \\
        & \leq P\left( \sum_{n=1}^N \left[ (1-\beta_{\frac{1}{p}}) V^{\frac{1}{p}}(\Lambda_n) \right] - N b_{\frac{1}{p}} - \beta_{\frac{1}{p}} V^{\frac{1}{p}}(\Lambda_0)
        > (1-\beta_{\frac{1}{p}}) N x \right) \\
        & \quad + P \left( \beta_{\frac{1}{p}} V^{\frac{1}{p}}(\Lambda_0) \geq N b_{\frac{1}{p}} \right) \\
        & \leq \left(18 p q^{1 / 2}\right)^p \frac{2^p \left[ \max \left\{ \mathbb{E}[V(\Lambda_0)], \frac{b}{1-\beta} \right\} \right]}{(1-\beta_{\frac{1}{p}})^p x^p} \frac{1}{N^{p / 2}}
        + \frac{\mathbb{E}[V(\Lambda_0)] \beta_{\frac{1}{p}}^p}{N^p b_{\frac{1}{p}}^p} .
    \end{align*}
\end{proof}

\section{Experiments}

In this subsection, we present the simulation results of several randomization procedures.
Consider the case where the covariate vector $X$ is $3$-dimensional.
To construct a covariate vector that is both correlated and bounded, we proceed as follows:
\begin{enumerate}
    \item Generate a random vector from a multivariate normal distribution from a correlated multivariate normal distribution with covariance matrix
    \begin{equation*}
        \begin{pmatrix}
            1 & 0.3 & 0.2 \\
            0.3 & 1 & 0.4 \\
            0.2 & 0.4 & 1
        \end{pmatrix} .
    \end{equation*}

    \item Transform the first component monotonically to take values in $\{-1, 0, 1\}$, with probabilities $0.25$, $0.5$, and $0.25$, respectively.

    \item Transform the second component $x$ to $\max\{\min\{x, 2\}, -2\}/2$ to ensure boundedness.

    \item Transform the third component monotonically to have a uniform distribution on $[-1, 1]$.
\end{enumerate}

The working model used for estimation is
\begin{equation}
    Y(T) = T (\alpha_1 + X_1 \gamma_1) + (1-T) (\alpha_0 + X_1 \gamma_0) + X_2 \beta_2 + X_3 \beta_3 + \epsilon ,
    \label{eq_working_model_experiment}
\end{equation}
where $\epsilon$ denotes the random noise.
Moreover, the feature map $\phi$ used in the covariate imbalance and in \eqref{eq_CARA_allocation_function} is defined as $\phi: x \mapsto (1,x^T)^T$.
Thus, we are concerned with the balance of the treatment group sizes and the first moments of the covariate vector $X$.
The randomization procedures in this subsection are implemented with the following components:
\begin{enumerate}
    \item When the number of allocated units $n \le 20$, the allocation probability is fixed at $0.5$.
    When $n$ exceeds the threshold $20$, the allocation probability is set equal to the targeted allocation ratio itself (direct) or determined by the allocation function \eqref{eq_CARA_allocation_function} (balance).
    In the expression of \eqref{eq_CARA_allocation_function}, the constants $p_\theta = \frac{1}{\max_x \rho_\theta(x)}$, $C_\theta = \frac{\sqrt{3+1}}{(\max_x \rho_\theta(x))(1-(\max_x \rho_\theta(x)))}$ and $C_\Lambda = 1$.
    Moreover, when using \eqref{eq_CARA_allocation_function}, the parameter update mechanism is given by \eqref{eq_allocation_parameter_update_mechanism_2} with $C_{\mathrm{clip}, n} = n^{-1/2}$.
    Otherwise, the parameter is set equal to the estimate of the model parameter, following the framework in \cite{zhangAsymptoticPropertiesCovariateadjusted2007}.
    
    \item The model parameter is estimated using either the weighted M-estimator in \eqref{eq_eta_estimation} with $\rho^{\mathrm{ref}} \equiv 1/2$ (weighted) or the M-estimator without weights (unweighted).
    
    \item The targeted allocation ratio $\rho(x)$ is specified as 
    \begin{equation*}
        \rho(x) =
        \begin{cases}
            0.5, & \text{(CRD)},\\
            \text{plogistic}\left(\hat{\mu}_{\theta,1}(x)-\hat{\mu}_{\theta,0}(x)\right), & \text{(logistic)} , \\
            \text{probit}\left(\hat{\mu}_{\theta,1}(x)-\hat{\mu}_{\theta,0}(x)\right), & \text{(probit)} ,
        \end{cases}
    \end{equation*}
    where
    \begin{align*}
        \text{plogistic}(x) &= \min \left\{ \max \left\{ \frac{1}{1+\exp(-x/2)},0.2 \right\}, 0.8 \right\} , \\
        \text{probit}(x) &= \min \left\{ \max \left\{ \Phi(x/3),0.2 \right\} ,0.8 \right\} , \\
        \Phi(x) &= \int_{-\infty}^{x} \frac{1}{\sqrt{2\pi}} e^{-t^2/2} \mathrm{d} t .
    \end{align*}
    $\hat{\mu}_{\theta,1}(x)$ and $\hat{\mu}_{\theta,0}(x)$ denote the expected responses of a unit with covariate $x$ under the treatment and control group, respectively, evaluated at the current allocation parameter $\theta$.
    Note that larger responses are considered preferable and by \eqref{eq_working_model_experiment}, the maximum $\max_x \rho_\theta(x)$ equals $0.5$ (CRD), $\text{plogistic}(|\alpha_1 - \alpha_0| + |\gamma_1 - \gamma_0|)$ (logistic), and $\text{probit}(|\alpha_1 - \alpha_0| + |\gamma_1 - \gamma_0|)$ (probit).
    When the first setting of the targeted allocation ratio is used, the CBARA procedure reduces to CAR or simple randomization, whereas the latter two correspond to the classical targeted allocation ratio settings within the CARA framework \cite{bandyopadhyayAdaptiveDesignsNormal2001,zhangAsymptoticPropertiesCovariateadjusted2007,zhuCovariateadjustedResponseAdaptive2015}.
\end{enumerate}

The primary performance measures considered in this simulation include
\begin{enumerate}
    \item the magnitude of the imbalance vector,
    \item the convergence rate of the imbalance of the additional covariate,
    \item and the mean squared error (MSE) of the ATE estimator.
\end{enumerate}
The additional covariate vector we consider is set to
\begin{equation*}
    Z^* = \tanh \left( 0.8 X_1 + 0.5 X_2^2 - 0.3 X_3 + 0.1 X_1 X_3 \right) + \epsilon,
\end{equation*}
where $\epsilon \sim \mathcal{N}(0, 0.2^2)$ is independent of $X$.
The true response model is defined for two scenarios: scenario A with only a treatment-by-$X_1$ interaction, and scenario B with additional treatment-by-$X_2$ and treatment-by-$X_3$ interactions.
For scenario A, the potential outcomes are
\begin{align*}
    Y(1) &= 4.5 + 4.7 X_1 + 2.9 X_2 + 1.4 X_3 , \\
    Y(0) &= 7.5 + 1.7 X_1 + 2.9 X_2 + 1.4 X_3 .
\end{align*}
For scenario B, the potential outcomes are
\begin{align*}
    Y(1) &= 4.5 + 4.7 X_1 - 0.6 X_2 - 0.6 X_3 , \\
    Y(0) &= 7.5 + 1.7 X_1 + 2.9 X_2 + 1.4 X_3 .
\end{align*}
For each case, the sample sizes are $n=200, 800$, and we conduct the simulation for $N=2000$ times.
The simulation results are given in Tables \ref{table_imbalance} and \ref{table_estimation}.

\begin{table}[!tbp]
    \centering
    \caption{Average responses, average norms of the imbalance vectors $\|\Lambda_n\|$ and $\|\Psi_n\|$ under two allocation mechanisms.}
    \label{table_imbalance}
    \begin{adjustbox}{scale=0.9}
    \begin{threeparttable}
    \begin{tabular}{cccccccccc}
        \hline
        Size & Model & Procedure & Estimation &
        $\text{Response}_D$ & $\text{Response}_B$ &
        $\Lambda_D$ & $\Lambda_B$ &
        $\Psi_D$ & $\Psi_B$ \\
        \hline
        200 & A & CRD & unweighted & 5.997 & 5.992 & 37.354 & 8.061 & 11.610 & 5.814 \\
            &   &     & weighted   & 6.012 & 6.000 & 37.128 & 7.875 & 11.457 & 5.835 \\
        \cline{3-10}
            &   & logistic & unweighted & 6.806 & 6.812 & 43.432 & 18.066 & 12.817 & 8.375 \\
            &   &          & weighted   & 6.802 & 6.804 & 44.009 & 18.076 & 12.953 & 8.527 \\
        \cline{3-10}
            &   & probit   & unweighted & 6.813 & 6.801 & 43.932 & 18.281 & 13.430 & 8.548 \\
            &   &          & weighted   & 6.811 & 6.822 & 43.529 & 18.264 & 13.414 & 8.468 \\
        \hline
        200 & B & CRD & unweighted & 5.994 & 6.000 & 37.502 & 8.141 & 11.808 & 5.889 \\
            &   &     & weighted   & 6.002 & 6.005 & 37.891 & 7.909 & 11.331 & 5.740 \\
        \cline{3-10}
            &   & logistic & unweighted & 6.759 & 6.754 & 44.142 & 18.127 & 13.509 & 8.474 \\
            &   &          & weighted   & 6.766 & 6.763 & 43.587 & 18.021 & 13.172 & 8.447 \\
        \cline{3-10}
            &   & probit   & unweighted & 6.769 & 6.767 & 43.880 & 18.244 & 13.320 & 8.633 \\
            &   &          & weighted   & 6.769 & 6.763 & 43.838 & 17.917 & 13.160 & 8.276 \\
        \hline
        800 & A & CRD & unweighted & 6.001 & 5.996 & 74.839 & 8.044 & 24.134 & 10.774 \\
            &   &     & weighted   & 6.002 & 5.998 & 75.309 & 8.004 & 23.744 & 11.078 \\
        \cline{3-10}
            &   & logistic & unweighted & 6.876 & 6.879 & 88.679 & 19.856 & 26.820 & 14.430 \\
            &   &          & weighted   & 6.876 & 6.879 & 88.996 & 19.859 & 27.613 & 14.120 \\
        \cline{3-10}
            &   & probit   & unweighted & 6.877 & 6.877 & 87.228 & 20.300 & 26.544 & 13.676 \\
            &   &          & weighted   & 6.875 & 6.877 & 88.124 & 19.920 & 26.438 & 14.240 \\
        \hline
        800 & B & CRD & unweighted & 6.002 & 6.001 & 75.292 & 7.976 & 23.459 & 10.961 \\
            &   &     & weighted   & 5.999 & 5.998 & 74.935 & 8.073 & 23.724 & 11.119 \\
        \cline{3-10}
            &   & logistic & unweighted & 6.832 & 6.832 & 88.997 & 19.921 & 26.700 & 14.445 \\
            &   &          & weighted   & 6.831 & 6.830 & 88.886 & 19.655 & 26.625 & 13.876 \\
        \cline{3-10}
            &   & probit   & unweighted & 6.834 & 6.841 & 89.585 & 19.779 & 27.681 & 13.934 \\
            &   &          & weighted   & 6.831 & 6.836 & 88.633 & 20.053 & 26.876 & 14.296 \\
        \hline
    \end{tabular}
    \begin{tablenotes}
        \item
        $\text{Response}_U$, average response;
        $\Lambda_U$, average norm of the imbalance vectors $\|\Lambda_n\|$;
        $\Psi_U$, average norm of the imbalances $\|\Psi_n\|$ of the additional covariate $Z^*$.
        \item
        The subscript $U \in \{D,B\}$ indicates the direct allocation mechanism and the balance allocation mechanism, respectively.
    \end{tablenotes}
    \end{threeparttable}
    \end{adjustbox}
\end{table}

\begin{table}[!tbp]
    \centering
    \caption{Average overall standard deviations of the targeted allocation ratios across allocation steps, mean squared errors (MSEs) of the IPW estimators under two allocation mechanisms.}
    \label{table_estimation}
    \begin{adjustbox}{scale=1.0}
    \begin{threeparttable}
    \begin{tabular}{cccccccc}
        \hline
        Size & Model & Procedure & Estimation &
        $\text{TargetSD}_D$ & $\text{TargetSD}_B$ &
        $MSE_{\mathrm{W},D}$ & $MSE_{\mathrm{W},B}$ \\
        \hline
        200 & A & CRD & unweighted & 0     & 0     & 0.970 & 0.071 \\
            &   &     & weighted   & 0     & 0     & 0.995 & 0.070 \\
        \cline{3-8}
            &   & logistic & unweighted & 0.141 & 0.141 & 1.033 & 0.200 \\
            &   &          & weighted   & 0.142 & 0.141 & 1.025 & 0.189 \\
        \cline{3-8}
            &   & probit   & unweighted & 0.142 & 0.142 & 1.056 & 0.195 \\
            &   &          & weighted   & 0.142 & 0.142 & 1.022 & 0.194 \\
        \hline
        200 & B & CRD & unweighted & 0     & 0     & 0.895 & 0.080 \\
            &   &     & weighted   & 0     & 0     & 0.913 & 0.082 \\
        \cline{3-8}
            &   & logistic & unweighted & 0.111 & 0.111 & 0.950 & 0.190 \\
            &   &          & weighted   & 0.111 & 0.112 & 0.947 & 0.191 \\
        \cline{3-8}
            &   & probit   & unweighted & 0.111 & 0.111 & 0.932 & 0.189 \\
            &   &          & weighted   & 0.111 & 0.112 & 1.013 & 0.193 \\
        \hline
        800 & A & CRD & unweighted & 0     & 0     & 0.241 & 0.012 \\
            &   &     & weighted   & 0     & 0     & 0.250 & 0.012 \\
        \cline{3-8}
            &   & logistic & unweighted & 0.133 & 0.133 & 0.273 & 0.022 \\
            &   &          & weighted   & 0.133 & 0.134 & 0.273 & 0.022 \\
        \cline{3-8}
            &   & probit   & unweighted & 0.133 & 0.133 & 0.256 & 0.023 \\
            &   &          & weighted   & 0.133 & 0.134 & 0.263 & 0.024 \\
        \hline
        800 & B & CRD & unweighted & 0     & 0     & 0.224 & 0.016 \\
            &   &     & weighted   & 0     & 0     & 0.212 & 0.016 \\
        \cline{3-8}
            &   & logistic & unweighted & 0.090 & 0.090 & 0.247 & 0.025 \\
            &   &          & weighted   & 0.091 & 0.091 & 0.227 & 0.026 \\
        \cline{3-8}
            &   & probit   & unweighted & 0.088 & 0.088 & 0.238 & 0.024 \\
            &   &          & weighted   & 0.089 & 0.089 & 0.236 & 0.025 \\
        \hline
    \end{tabular}
    \begin{tablenotes}
        \item
        $\text{TargetSD}_U$, average overall standard deviation of the targeted allocation ratios across allocation steps;
        $MSE_{\mathrm{W},U}$, MSE of the ATE IPW estimator.
        \item
        The subscript $U \in \{D,B\}$ indicates the direct allocation mechanism and the balance allocation mechanism, respectively.
    \end{tablenotes}
    \end{threeparttable}
    \end{adjustbox}
\end{table}

Table \ref{table_imbalance} indicates that, under the last two targeted allocation ratio settings, the average response is higher than under CRD.
Across all sample sizes, model specifications, and estimation methods, the average responses under direct allocation and the balance allocation mechanism are nearly identical.
This indicates that the balance allocation mechanism achieves improved covariate balance without sacrificing the average response.
Furthermore, under the balance allocation mechanism, the norm of the primary covariate imbalance vector $\|\Lambda_n\| = O_P(1)$, whereas under direct allocation it grows at rate $O_P(\sqrt{n})$.
The imbalance of the additional covariate, $\|\Psi_n\|$, grows at $O_P(\sqrt{n})$ under both allocation mechanisms, but is reduced under the balance allocation mechanism.
These results indicate that applying the balance allocation mechanism significantly improves the balance of $X$ and also partially reduces imbalance for the additional covariate.

Table \ref{table_estimation} shows that the logistic and probit targeted allocation ratios under the CBARA procedure exhibit similar variability.
The choice of estimation method does not materially affect the overall variations of the allocation ratios across allocation steps, although weighted methods are theoretically more robust.
Moreover, improved covariate balance leads to lower mean squared errors for the IPW estimators of the ATE.
As the sample size increases, the efficiency gains of the balance allocation mechanism become more pronounced.
This corresponds to the fact that $\|\Lambda_n\| = O_P(1)$ under the balance mechanism, which provides a clear advantage compared with the direct allocation mechanism in relatively large samples.
Overall, the balance allocation mechanism improves covariate balance and substantially enhances the accuracy of ATE estimators.

\FloatBarrier

\bibliographystyle{plain}
\bibliography{Library.bib}

@article{alettiNonparametricCovariateadjustedResponseadaptive2018,
  title = {Nonparametric Covariate-Adjusted Response-Adaptive Design Based on a Functional Urn Model},
  author = {Aletti, Giacomo and Ghiglietti, Andrea and Rosenberger, William F.},
  year = 2018,
  journal = {The Annals of Statistics},
  volume = {46},
  number = {6B},
  pages = {3838--3866},
  publisher = {Institute of Mathematical Statistics},
  doi = {10.1214/17-AOS1677}
}

@article{baldiantogniniCovariateadaptiveBiasedCoin2011,
  title = {The Covariate-Adaptive Biased Coin Design for Balancing Clinical Trials in the Presence of Prognostic Factors},
  author = {Baldi Antognini, A. and Zagoraiou, M.},
  year = 2011,
  journal = {Biometrika},
  volume = {98},
  number = {3},
  pages = {519--535},
  doi = {10.1093/biomet/asr021}
}

@article{baldiantogniniEfficientCovariateAdaptiveDesign2024,
  title = {The {{Efficient Covariate-Adaptive Design}} for High-Order Balancing of Quantitative and Qualitative Covariates},
  author = {Baldi Antognini, Alessandro and Frieri, Rosamarie and Zagoraiou, Maroussa and Novelli, Marco},
  year = 2024,
  journal = {Statistical Papers},
  volume = {65},
  number = {1},
  pages = {19--44},
  doi = {10.1007/s00362-022-01381-1}
}

@phdthesis{ballouResponseadaptiveCovariatebalancedRandomization2015,
  title = {A {{Response-adaptive}} Covariate-Balanced Randomization for Multi-Arm Clinical Trials},
  author = {Ballou, Cassandra M.},
  year = 2015,
  doi = {10.6083/M4PR7TX9},
  school = {Oregon Health \& Science University}
}

@article{bandyopadhyayAdaptiveDesignsNormal2001,
  title = {Adaptive Designs for Normal Responses with Prognostic Factors},
  author = {Bandyopadhyay, Uttam and Biswas, Atanu},
  year = 2001,
  journal = {Biometrika},
  volume = {88},
  number = {2},
  pages = {409--419},
  doi = {10.1093/biomet/88.2.409}
}

@article{bhattacharyaClassOptimalType2015,
  title = {On a {{Class}} of {{Optimal Type Covariate Adjusted Response Adaptive Allocations}} for {{Normal Treatment Responses}}},
  author = {Bhattacharya, Rahul and Bandyopadhyay, Uttam},
  year = 2015,
  journal = {Austrian Journal of Statistics},
  volume = {44},
  number = {4},
  pages = {53--65},
  doi = {10.17713/ajs.v44i4.69}
}

@article{biswasClassCovariateAdjustedResponseAdaptive2018,
  title = {A Class of {{Covariate-Adjusted Response-Adaptive Allocation Designs}} for {{Multitreatment Binary Response Trials}}},
  author = {Biswas, Atanu and Bhattacharya, Rahul},
  year = 2018,
  journal = {Journal of Biopharmaceutical Statistics},
  volume = {28},
  number = {5},
  pages = {809--823},
  doi = {10.1080/10543406.2018.1485683}
}

@article{biswasClassOptimalCovariateadjusted2016,
  title = {On a Class of Optimal Covariate-Adjusted Response Adaptive Designs for Survival Outcomes},
  author = {Biswas, Atanu and Bhattacharya, Rahul and Park, Eunsik},
  year = 2016,
  journal = {Statistical Methods in Medical Research},
  volume = {25},
  number = {6},
  pages = {2444--2456},
  doi = {10.1177/0962280214524177}
}

@article{bugniInferenceCovariateadaptiveRandomization2018,
  title = {Inference under Covariate-Adaptive Randomization},
  author = {Bugni, Federico A. and Canay, Ivan A. and Shaikh, Azeem M.},
  year = 2018,
  journal = {Journal of the American Statistical Association},
  volume = {113},
  number = {524},
  pages = {1784--1796},
  doi = {10.1080/01621459.2017.1375934}
}

@article{chambazInferenceTargetedGroupSequential2014,
  title = {Inference in {{Targeted Group}}-{{Sequential Covariate}}-{{Adjusted Randomized Clinical Trials}}},
  author = {Chambaz, Antoine and Van Der Laan, Mark J.},
  year = 2014,
  journal = {Scandinavian Journal of Statistics},
  volume = {41},
  number = {1},
  pages = {104--140},
  doi = {10.1111/sjos.12013}
}

@article{chambazTargetedSequentialDesign2017,
  title = {Targeted Sequential Design for Targeted Learning Inference of the Optimal Treatment Rule and Its Mean Reward},
  author = {Chambaz, Antoine and Zheng, Wenjing and Van Der Laan, Mark J.},
  year = 2017,
  journal = {The Annals of Statistics},
  volume = {45},
  number = {6},
  pages = {2537--2564},
  publisher = {Institute of Mathematical Statistics},
  doi = {10.1214/16-AOS1534}
}

@article{cheungCovariateadjustedResponseadaptiveDesigns2014,
  title = {Covariate-Adjusted Response-Adaptive Designs for Generalized Linear Models},
  author = {Cheung, Siu Hung and Zhang, Li-Xin and Hu, Feifang and Chan, Wai Sum},
  year = 2014,
  journal = {Journal of Statistical Planning and Inference},
  volume = {149},
  pages = {152--161},
  doi = {10.1016/j.jspi.2014.02.006}
}

@article{chimisovAirMarkovChain2018,
  title = {Air {{Markov Chain Monte Carlo}}},
  author = {Chimisov, Cyril and Latuszynski, Krzysztof and Roberts, Gareth},
  year = 2018,
  publisher = {arXiv},
  doi = {10.48550/arXiv.1801.09309},
  journal = {arXiv preprint arXiv:1801.09309}
}

@article{fangGeneralNonMarkovianFramework2026,
  title = {A {{General}} ({{Non-Markovian}}) {{Framework}} for {{Covariate Adaptive Randomization}}: {{Achieving Balance While Eliminating}} the {{Shift}}},
  shorttitle = {A {{General}} ({{Non-Markovian}}) {{Framework}} for {{Covariate Adaptive Randomization}}},
  author = {Fang, Hengjia and Ma, Wei},
  year = 2026,
  publisher = {arXiv},
  doi = {10.48550/arXiv.2602.22648},
  journal = {arXiv preprint arXiv:2602.22648}
}

@unpublished{fangSupplementCBARACovariateBalancedandAdjusted2026,
  title = {Supplement to ``{{CBARA}}: {{Covariate-Balanced-and-Adjusted Response-Adaptive Randomization}}''},
  author = {Fang, Hengjia and Ma, Wei},
  year = 2026
}

@book{follandRealAnalysisModern1999,
  title = {Real Analysis: Modern Techniques and Their Applications},
  shorttitle = {Real Analysis},
  author = {Folland, G. B.},
  year = 1999,
  series = {Pure and Applied Mathematics},
  edition = {2nd},
  publisher = {Wiley},
  address = {New York},
  isbn = {978-0-471-31716-6},
  lccn = {QA300 .F67 1999}
}

@article{fortCentralLimitTheorem2014,
  title = {A Central Limit Theorem for Adaptive and Interacting {{Markov}} Chains},
  author = {Fort, G. and Moulines, E. and Priouret, P. and Vandekerkhove, P.},
  year = 2014,
  journal = {Bernoulli},
  volume = {20},
  number = {2},
  pages = {457--485},
  doi = {10.3150/12-BEJ493}
}

@article{fortConvergenceAdaptiveInteracting2011,
  title = {Convergence of Adaptive and Interacting {{Markov}} Chain {{Monte Carlo}} Algorithms},
  author = {Fort, G. and Moulines, E. and Priouret, P.},
  year = 2011,
  journal = {The Annals of Statistics},
  volume = {39},
  number = {6},
  pages = {3262--3289},
  doi = {10.1214/11-AOS938}
}

@article{gaoResponseAdaptiveRandomizationProcedure2024,
  title = {Response-{{Adaptive Randomization Procedure}} in {{Clinical Trials}} with {{Surrogate Endpoints}}},
  author = {Gao, Jingya and Hu, Feifang and Ma, Wei},
  year = 2024,
  journal = {Statistics in Medicine},
  volume = {43},
  number = {30},
  pages = {5911--5921},
  doi = {10.1002/sim.10286}
}

@book{hallMartingaleLimitTheory1980,
  title = {Martingale Limit Theory and Its Application},
  author = {Hall, Peter and Heyde, Christopher Charles},
  year = 1980,
  series = {Probability {{And Mathematical Statistics}} [Unnumbered]},
  publisher = {Academic Press},
  address = {New York},
  isbn = {978-0-12-319350-6},
  lccn = {519.287}
}

@article{huAsymptoticPropertiesCovariateadaptive2012,
  title = {Asymptotic Properties of Covariate-Adaptive Randomization},
  author = {Hu, Yanqing and Hu, Feifang},
  year = 2012,
  journal = {The Annals of Statistics},
  volume = {40},
  number = {3},
  pages = {1794--1815},
  doi = {10.1214/12-AOS983}
}

@article{huDoublyAdaptiveBiased2008,
  title = {Doubly Adaptive Biased Coin Designs with Delayed Responses},
  author = {Hu, Feifang and Zhang, Li-Xin and Cheung, Siu H. and Chan, Wai S.},
  year = 2008,
  journal = {Canadian Journal of Statistics},
  volume = {36},
  number = {4},
  pages = {541--559},
  doi = {10.1002/cjs.5550360404},
  copyright = {http://onlinelibrary.wiley.com/termsAndConditions\#vor}
}

@article{huMultiArmCovariateAdaptiveRandomization2023,
  title = {Multi-{{Arm Covariate-Adaptive Randomization}}},
  author = {Hu, Feifang and Ye, Xiaoqing and Zhang, Li-Xin},
  year = 2023,
  journal = {Science China Mathematics},
  volume = {66},
  number = {1},
  pages = {163--190},
  doi = {10.1007/s11425-020-1954-y}
}

@article{huTheoryCovariateadaptiveDesigns2020,
  title = {On the Theory of Covariate-Adaptive Designs},
  author = {Hu, Feifang and Zhang, Li-Xin},
  year = 2020,
  publisher = {arXiv},
  journal = {arXiv preprint arXiv:2004.02994}
}

@article{huUnifiedFamilyCovariateAdjusted2015,
  title = {A {{Unified Family}} of {{Covariate-Adjusted Response-Adaptive Designs Based}} on {{Efficiency}} and {{Ethics}}},
  author = {Hu, Jianhua and Zhu, Hongjian and Hu, Feifang},
  year = 2015,
  journal = {Journal of the American Statistical Association},
  volume = {110},
  number = {509},
  pages = {357--367},
  doi = {10.1080/01621459.2014.903846}
}

@article{imaiCovariateBalancingPropensity2014,
  title = {Covariate {{Balancing Propensity Score}}},
  author = {Imai, Kosuke and Ratkovic, Marc},
  year = 2014,
  journal = {Journal of the Royal Statistical Society Series B: Statistical Methodology},
  volume = {76},
  number = {1},
  pages = {243--263},
  doi = {10.1111/rssb.12027},
  copyright = {https://academic.oup.com/journals/pages/open\_access/funder\_policies/chorus/standard\_publication\_model}
}

@article{lesigneLargeDeviationsMartingales2001,
  title = {Large Deviations for Martingales},
  author = {Lesigne, Emmanuel and Voln{\'y}, Dalibor},
  year = 2001,
  journal = {Stochastic Processes and their Applications},
  volume = {96},
  number = {1},
  pages = {143--159},
  doi = {10.1016/S0304-4149(01)00112-0}
}

@article{liuBalancingUnobservedCovariates2022,
  title = {Balancing {{Unobserved Covariates With Covariate-Adaptive Randomized Experiments}}},
  author = {Liu, Yang and Hu, Feifang},
  year = 2022,
  journal = {Journal of the American Statistical Association},
  volume = {117},
  number = {538},
  pages = {875--886},
  doi = {10.1080/01621459.2020.1825450}
}

@article{liuImpactsUnobservedCovariates2023,
  title = {The Impacts of Unobserved Covariates on Covariate-Adaptive Randomized Experiments},
  author = {Liu, Yang and Hu, Feifang},
  year = 2023,
  journal = {The Annals of Statistics},
  volume = {51},
  number = {5},
  pages = {1895--1920},
  publisher = {Institute of Mathematical Statistics},
  doi = {10.1214/23-AOS2308}
}

@article{liuPropertiesCovariateadaptiveRandomization2025,
  title = {The Properties of Covariate-Adaptive Randomization Procedures with Possibly Unequal Allocation Ratio},
  author = {Liu, Xiao and Hu, Feifang and Ma, Wei},
  year = 2025,
  journal = {The Annals of Applied Statistics},
  volume = {19},
  number = {2},
  pages = {907--925},
  publisher = {Institute of Mathematical Statistics},
  doi = {10.1214/25-AOAS2023}
}

@article{maNewUnifiedFamily2024,
  title = {A {{New}} and {{Unified Family}} of {{Covariate Adaptive Randomization Procedures}} and {{Their Properties}}},
  author = {Ma, Wei and Li, Ping and Zhang, Li-Xin and Hu, Feifang},
  year = 2024,
  journal = {Journal of the American Statistical Association},
  volume = {119},
  number = {545},
  pages = {151--162},
  doi = {10.1080/01621459.2022.2102986}
}

@article{maStatisticalInferenceCovariateAdaptive2020,
  title = {Statistical {{Inference}} for {{Covariate-Adaptive Randomization Procedures}}},
  author = {Ma, Wei and Qin, Yichen and Li, Yang and Hu, Feifang},
  year = 2020,
  journal = {Journal of the American Statistical Association},
  volume = {115},
  number = {531},
  pages = {1488--1497},
  doi = {10.1080/01621459.2019.1635483}
}

@article{maTestingHypothesesCovariateAdaptive2015,
  title = {Testing {{Hypotheses}} of {{Covariate-Adaptive Randomized Clinical Trials}}},
  author = {Ma, Wei and Hu, Feifang and Zhang, Li-Xin},
  year = 2015,
  journal = {Journal of the American Statistical Association},
  volume = {110},
  number = {510},
  pages = {669--680},
  doi = {10.1080/01621459.2014.922469}
}

@article{meurerSimulationVariousRandomization2016,
  title = {Simulation of Various Randomization Strategies for a Clinical Trial in Sickle Cell Disease},
  author = {Meurer, William J. and Connor, Jason T. and Glassberg, Jeffrey},
  year = 2016,
  journal = {Hematology},
  volume = {21},
  number = {4},
  pages = {241--247},
  doi = {10.1080/10245332.2015.1101966}
}

@book{meynMarkovChainsStochastic2009,
  title = {Markov {{Chains}} and {{Stochastic Stability}}},
  author = {Meyn, S. P. and Tweedie, R. L.},
  year = 2009,
  series = {Communications and Control Engineering Series},
  edition = {2nd},
  publisher = {Cambridge University Press},
  address = {Cambridge ; New York},
  isbn = {978-0-521-73182-9},
  lccn = {QA274.7 .M49 2009}
}

@book{nocedalNumericalOptimization2006,
  title = {Numerical Optimization},
  author = {Nocedal, Jorge and Wright, Stephen J.},
  year = 2006,
  series = {Springer Series in Operations Research},
  edition = {2nd},
  publisher = {Springer},
  address = {New York},
  isbn = {978-0-387-30303-1},
  lccn = {QA402.5 .N62 2006}
}

@article{phienQuantitativeResultsLipschitz2012,
  title = {Some Quantitative Results on {{Lipschitz}} Inverse and Implicit Functions Theorems},
  author = {Phien, Phan},
  year = 2012,
  publisher = {arXiv},
  doi = {10.48550/arXiv.1204.4916},
  journal = {arXiv preprint arXiv:1204.4916}
}

@article{pocockSequentialTreatmentAssignment1975,
  title = {Sequential Treatment Assignment with Balancing for Prognostic Factors in the Controlled Clinical Trial},
  author = {Pocock, Stuart J. and Simon, Richard},
  year = 1975,
  journal = {Biometrics},
  volume = {31},
  number = {1},
  eprinttype = {jstor},
  pages = {103--115},
  doi = {10.2307/2529712}
}

@article{rosenbergerCovariateAdjustedResponseAdaptiveDesigns2001,
  title = {Covariate-{{Adjusted Response-Adaptive Designs}} for {{Binary Response}}},
  author = {Rosenberger, William F. and Vidyashankar, A. N. and Agarwal, Deepak K.},
  year = 2001,
  journal = {Journal of Biopharmaceutical Statistics},
  volume = {11},
  number = {4},
  pages = {227--236},
  doi = {10.1081/BIP-120008846}
}

@article{rosenbergerHandlingCovariatesDesign2008,
  title = {Handling {{Covariates}} in the {{Design}} of {{Clinical Trials}}},
  author = {Rosenberger, William F. and Sverdlov, Oleksandr},
  year = 2008,
  journal = {Statistical Science},
  volume = {23},
  number = {3},
  pages = {404--419},
  publisher = {Institute of Mathematical Statistics},
  doi = {10.1214/08-STS269}
}

@book{rosenbergerRandomizationClinicalTrials2016,
  title = {Randomization in Clinical Trials: Theory and Practice},
  shorttitle = {Randomization in Clinical Trials},
  author = {Rosenberger, William F. and Lachin, John M.},
  year = 2016,
  series = {Wiley Series in Probability and Statistics},
  edition = {Second},
  publisher = {John Wiley \& Sons, Inc},
  address = {Hoboken, New Jersey},
  doi = {10.1002/9781118742112},
  isbn = {978-1-118-74224-2 978-1-118-74211-2 978-1-118-74215-0}
}

@article{shaoTheoryTestingHypotheses2010,
  title = {A Theory for Testing Hypotheses under Covariate-Adaptive Randomization},
  author = {Shao, J. and Yu, X. and Zhong, B.},
  year = 2010,
  journal = {Biometrika},
  volume = {97},
  number = {2},
  pages = {347--360},
  doi = {10.1093/biomet/asq014}
}

@book{sverdlovModernAdaptiveRandomized2016,
  title = {Modern Adaptive Randomized Clinical Trials: Statistical and Practical Aspects},
  shorttitle = {Modern Adaptive Randomized Clinical Trials},
  editor = {Sverdlov, Oleksan},
  year = 2016,
  publisher = {CRC Press, Taylor \& Francis Group},
  address = {Boca Raton, FL},
  isbn = {978-1-4822-3989-8}
}

@article{tavesMinimizationNewMethod1974,
  title = {Minimization: {{A}} New Method of Assigning Patients to Treatment and Control Groups},
  shorttitle = {Minimization},
  author = {Taves, Donald R.},
  year = 1974,
  journal = {Clinical Pharmacology \& Therapeutics},
  volume = {15},
  number = {5},
  pages = {443--453},
  doi = {10.1002/cpt1974155443}
}

@book{vaartAsymptoticStatistics2007,
  title = {Asymptotic Statistics},
  author = {van der Vaart, Aad W.},
  year = 2007,
  series = {Cambridge Series on Statistical and Probabilistic Mathematics},
  edition = {8th},
  number = {3},
  publisher = {Cambridge University Press},
  address = {Cambridge New York Port Melbourne New Delhi Singapore},
  doi = {10.1017/CBO9780511802256},
  isbn = {978-0-521-49603-2 978-0-511-80225-6}
}

@book{villaniOptimalTransportOld2009,
  title = {Optimal {{Transport}}: {{Old}} and {{New}}},
  shorttitle = {Optimal {{Transport}}},
  author = {Villani, C{\'e}dric},
  year = 2009,
  series = {Grundlehren Der Mathematischen {{Wissenschaften}}},
  number = {338},
  publisher = {Springer},
  address = {Berlin},
  isbn = {978-3-540-71049-3},
  lccn = {QA402.5 .V538 2009}
}

@article{weiApplicationUrnModel1978,
  title = {An Application of an Urn Model to the Design of Sequential Controlled Clinical Trials},
  author = {Wei, L. J.},
  year = 1978,
  journal = {Journal of the American Statistical Association},
  volume = {73},
  number = {363},
  pages = {559--563},
  doi = {10.1080/01621459.1978.10480054}
}

@article{yangSequentialCovariateadjustedRandomization2024,
  title = {Sequential Covariate-Adjusted Randomization via Hierarchically Minimizing {{Mahalanobis}} Distance and Marginal Imbalance},
  author = {Yang, Haoyu and Qin, Yichen and Li, Yang and Hu, Feifang},
  year = 2024,
  journal = {Biometrics},
  volume = {80},
  number = {2},
  pages = {ujae047},
  doi = {10.1093/biomtc/ujae047},
  copyright = {https://academic.oup.com/journals/pages/open\_access/funder\_policies/chorus/standard\_publication\_model}
}

@article{yuanBayesianResponseadaptiveCovariatebalanced2011,
  title = {A {{Bayesian}} Response-adaptive Covariate-balanced Randomization Design with Application to a Leukemia Clinical Trial},
  author = {Yuan, Ying and Huang, Xuelin and Liu, Suyu},
  year = 2011,
  journal = {Statistics in Medicine},
  volume = {30},
  number = {11},
  pages = {1218--1229},
  doi = {10.1002/sim.4218},
  copyright = {http://onlinelibrary.wiley.com/termsAndConditions\#vor}
}

@article{zelenRandomizationStratificationPatients1974,
  title = {The Randomization and Stratification of Patients to Clinical Trials},
  author = {Zelen, M.},
  year = 1974,
  journal = {Journal of Chronic Diseases},
  volume = {27},
  number = {7-8},
  pages = {365--375},
  doi = {10.1016/0021-9681(74)90015-0}
}

@article{zhangAsymptoticPropertiesCovariateadjusted2007,
  title = {Asymptotic Properties of Covariate-Adjusted Response-Adaptive Designs},
  author = {Zhang, Li-Xin and Hu, Feifang and Cheung, Siu Hung and Chan, Wai Sum},
  year = 2007,
  journal = {The Annals of Statistics},
  volume = {35},
  number = {3},
  pages = {1166--1182},
  publisher = {Institute of Mathematical Statistics},
  doi = {10.1214/009053606000001424}
}

@article{zhangAsymptoticPropertiesMultitreatment2023,
  title = {Asymptotic Properties of Multi-Treatment Covariate Adaptive Randomization Procedures for Balancing Observed and Unobserved Covariates},
  author = {Zhang, Li-Xin},
  year = 2023,
  publisher = {arXiv},
  journal = {arXiv preprint arXiv:2305.13842}
}

@article{zhangNewFamilyCovariateAdjusted2009,
  title = {A {{New Family}} of {{Covariate-Adjusted Response Adaptive Designs}} and {{Their Properties}}},
  author = {Zhang, Li-Xin and Hu, Fei-fang},
  year = 2009,
  journal = {Applied Mathematics-A Journal of Chinese Universities},
  volume = {24},
  number = {1},
  pages = {1--13},
  doi = {10.1007/s11766-009-0001-6}
}

@article{zhangOnlineMetaLevelAdaptiveDesign2025,
  title = {An {{Online Meta-Level Adaptive-Design Framework}} with {{Targeted Learning Inference}}: {{Applications}} to {{Evaluating}} and {{Utilizing Surrogate Outcomes}} in {{Adaptive Designs}}},
  shorttitle = {An {{Online Meta-Level Adaptive-Design Framework}} with {{Targeted Learning Inference}}},
  author = {Zhang, Wenxin and Hudson, Aaron and Petersen, Maya and van der Laan, Mark},
  year = 2025,
  publisher = {arXiv},
  doi = {10.48550/arXiv.2408.02667},
  journal = {arXiv preprint arXiv:2408.02667}
}

@inproceedings{zhangStatisticalInferenceMEstimators2021,
  title = {Statistical {{Inference}} with {{M-Estimators}} on {{Adaptively Collected Data}}},
  booktitle = {Advances in {{Neural Information Processing Systems}}},
  author = {Zhang, Kelly and Janson, Lucas and Murphy, Susan},
  year = 2021,
  volume = {34},
  pages = {7460--7471},
  publisher = {Curran Associates, Inc.}
}

@article{zhaoIncorporatingCovariatesInformation2022,
  title = {Incorporating Covariates Information in Adaptive Clinical Trials for Precision Medicine},
  author = {Zhao, Wanying and Ma, Wei and Wang, Fan and Hu, Feifang},
  year = 2022,
  journal = {Pharmaceutical Statistics},
  volume = {21},
  number = {1},
  pages = {176--195},
  doi = {10.1002/pst.2160}
}

@article{zhuCovariateadjustedResponseAdaptive2015,
  title = {Covariate-adjusted Response Adaptive Designs Incorporating Covariates with and without Treatment Interactions},
  author = {Zhu, Hongjian},
  year = 2015,
  journal = {Canadian Journal of Statistics},
  volume = {43},
  number = {4},
  pages = {534--553},
  doi = {10.1002/cjs.11260},
  copyright = {http://onlinelibrary.wiley.com/termsAndConditions\#vor}
}

\end{document}